%% file: main.tex
\documentclass[acmsmall, nonacm]{acmart}

\setcopyright{cc}
\acmPrice{}
\acmDOI{}
\acmYear{}
\copyrightyear{2023}
\acmSubmissionID{}
\acmVolume{}
\acmNumber{}
\acmArticle{}
\acmMonth{}

\usepackage{booktabs}   
\usepackage{subcaption} 

\usepackage[colorinlistoftodos,disable]{todonotes}
\newcommand{\todoin[3]}[]{
  \todo[inline, caption={#2}, #1]{%
    \begin{minipage}{\textwidth-4pt}%
      #3%
    \end{minipage}%
  }%
}
\newcommand{\note}[1]{\todo[color=green!40]{{#1}}}
\newcommand{\notein}[3][]{
  \todo[color=green!40, inline, caption={#2}, #1]{%
    \begin{minipage}{\textwidth-4pt}%
      #3%
    \end{minipage}%
  }%
}

\usepackage[english]{babel}
\usepackage[utf8]{inputenc}
\usepackage{changepage}
\usepackage{amsmath}
\usepackage{thmtools}
\usepackage{thm-restate}
\usepackage{stmaryrd}
\usepackage{graphicx, xcolor}
\usepackage{cleveref}
\usepackage{listings}
\usepackage{tikz}
\usepackage[inline]{enumitem}
\usepackage{hyperref}

\usepackage{xspace}

\AtEndPreamble{%
\theoremstyle{acmdefinition}
\newtheorem{remark}[theorem]{Remark}}

\newcommand{\derive}[2]{\genfrac{}{}{0.5pt}{0}{\begin{gathered}#1\end{gathered}}{#2}}

\newcommand{\reduce}[1]{\xrightarrow{#1}}
\newcommand{\reduceTwoLabel}[2]{{#1}\leadsto{#2}}
\newcommand{\reduceTwo}[2]{\xrightarrow{\reduceTwoLabel{#1}{#2}}}
\newcommand{\reduceTwoS}[2]{\xrightarrow{\overline{\reduceTwoLabel{#1}{#2}}}}
\newcommand{\curl}[1]{\{#1\}}
\newcommand{\mbC}{\mathbb{C}}
\newcommand{\serverRed}[4]{\langle #1 \rangle \reduceTwoS{#2}{#3} \langle #4 \rangle}
\newcommand{\serverRedInt}[4]{\langle #1, #2, #3 \rangle \downarrow #4}
\newcommand{\tx}[1]{\text{#1}}
\newcommand{\ttt}[1]{\texttt{#1}}
\newcommand{\tsc}[1]{\textsc{#1}}
\newcommand{\Rule}[3]{\derive{\begin{gathered}#1\end{gathered}}{#2}\ \textsc{#3}}
\newcommand{\setcomp}[2]{\curl{#1\ |\ #2}}
\newcommand{\unit}{\ttt{()}}
\newcommand{\wildcard}{\ttt{"*"}}
\newcommand{\types}[2]{#1 : #2}
\newcommand{\typesTo}[3]{#1 \vdash \types{#2}{#3}}

\newcommand{\subtypes}[2]{#1 <: #2}
\newcommand{\subtypesTo}[3]{#1 \vdash \subtypes{#2}{#3}}
\newcommand{\eqtypesTo}[3]{#1 \vdash #2 =:= #3}
\newcommand{\mutualTo}[3]{#1 \vdash #2 \rightleftharpoons #3}

\newcommand{\memberOf}[2]{#1 \mathrel{\in_G} #2}

\newcommand{\disjoint}[3]{#1 \vdash #2 \diamond #3}
\newcommand{\envJ}[1]{\vdash #1\ \tx{env}}
\newcommand{\typeJ}[2]{#1 \vdash #2\ \tx{type}}
\newcommand{\bigT}{\mathcal{T}}

\newcommand{\singleton}[1]{\underline{#1}}
\newcommand{\subs}[3]{#1[#2 \mapsto #3]}
\newcommand{\typeUnion}[2]{{#1} \cup {#2}}
\newcommand{\match}[3]{#1\ttt{ match }\curl{#2}_{#3}}
\newcommand{\letExp}[2]{\ttt{let }#1\ttt{ = }#2\ttt{ in }}
\newcommand{\encoded}[1]{\llbracket #1 \rrbracket}
\newcommand{\dual}[1]{\overline{#1}}
\newcommand{\nil}{\ttt{nil}}
\newcommand{\concat}[2]{{#1}\mathbin{\ttt{{:}\!{:}}}{#2}}
\newcommand{\listHead}{\ttt{head }}
\newcommand{\listTail}{\ttt{tail }}
\newcommand{\abs}[2]{\lambda \types{#1}{#2}.\ }
\newcommand{\tabs}[2]{\lambda \subtypes{#1}{#2}.\ }
\newcommand{\quant}[2]{\forall \subtypes{#1}{#2}.\ }
\newcommand{\quantT}[1]{\forall #1.\ }

\newcommand{\Chan}[1]{\tx{Chan}[#1]}
\newcommand{\ServerRef}[1]{\tx{ServerRef}[#1]}
\newcommand{\TableEntry}{\tx{TableEntry}}
\newcommand{\PPPPEntity}{\tx{P4Entity}}
\newcommand{\hole}{\ttt{[\,]}}
\newcommand{\listT}[1]{[#1]}
\newcommand{\stringTerm}[1]{\text{"\texttt{#1}"}}

\input{macros.tex}

\begin{document}

\title{\ourDSL: a Verified API for P4 Control Plane Programs (Technical Report)}         
\titlenote{This is an extended version of the paper published at OOPSLA 2023,
           available at: \url{https://doi.org/10.1145/3622866}}             


\author{Jens Kanstrup Larsen}
\orcid{0009-0006-4039-3808}             
\affiliation{
  \position{PhD student}
  \department{DTU Compute}              
  \institution{Technical University of Denmark} 
  \streetaddress{Richard Petersens Plads, Bygning 321}
  \city{Kongens Lyngby}
  \postcode{2800}
  \country{Denmark}                    
}
\email{jekla@dtu.dk}          

\author{Roberto Guanciale}
\orcid{0000-0002-8069-6495}             
\affiliation{
  \position{Associate professor}
  \department{Department of Computer Science}              
  \institution{KTH Royal Institute of Technology}            
  \streetaddress{Lindstedtsvägen 5, Plan 5}
  \city{Stockholm}
  \postcode{114 28}
  \country{Sweden}                    
}
\email{robertog@kth.se}          

\author{Philipp Haller}
\orcid{0000-0002-2659-5271}             
\affiliation{
  \position{Associate professor}
  \department{Department of Computer Science}              
  \institution{KTH Royal Institute of Technology}            
  \streetaddress{Lindstedtsvägen 5, Plan 5}
  \city{Stockholm}
  \postcode{114 28}
  \country{Sweden}                    
}
\email{phaller@kth.se}          

\author{Alceste Scalas}
\orcid{0000-0002-1153-6164}             
\affiliation{
  \position{Associate professor}
  \department{DTU Compute}              
  \institution{Technical University of Denmark} 
  \streetaddress{Richard Petersens Plads, Bygning 321}
  \city{Kongens Lyngby}
  \postcode{2800}
  \country{Denmark}                    
}
\email{alcsc@dtu.dk}          

\input{abstract.tex}



\maketitle

\input{intro.tex}
\input{background.tex}
\input{language.tex}
\input{typesystem.tex}
\input{semantics.tex}
\input{results.tex}
\input{implementation.tex}
\input{case-study.tex}
\input{related.tex}
\input{conclusion.tex}

\begin{acks}                            
  This work was partially supported by the DTU Nordic Five Tech Alliance grant
  ``Safe and secure software-defined networks in P4'' %
  and %
  the \grantsponsor{HorizonEurope}{Horizon Europe}{https://research-and-innovation.ec.europa.eu/funding/funding-opportunities/funding-programmes-and-open-calls/horizon-europe_en}
  grant no.~\grantnum{HorizonEurope}{101093006} ``TaRDIS.''
\end{acks}

\bibliography{main}

\appendix
\input{appendix-defs.tex}
\input{appendix.tex}


\end{document}

%% file: macros.tex
\newcommand{\hlight}[1]{
  \setlength\fboxsep{1pt}%
  \colorbox{yellow}{{#1}}}
\newcommand{\hlightM}[1]{\hlight{${#1}$}}
\newcommand{\fSub}{$F_{<:}$\xspace}

\newcommand{\dom}[1]{\operatorname{dom}\!\left({#1}\right)}

\newcommand{\reduceStar}{\mathrel{\rightarrow}^{\!*}}

\newcommand{\change}[3]{#3}
\newcommand{\changeNoMargin}[1]{#1}

\makeatletter
\@ifclasswith{acmart}{review}{%
  \makeatletter
  \patchcmd{\@addmarginpar}{\ifodd\c@page}{\ifodd\c@page\@tempcnta\m@ne}{}{}
  \makeatother
  \reversemarginpar
  \newcounter{change}
  \setcounter{change}{0}
  \renewcommand{\change}[3]{%
    \refstepcounter{change}%
    {\marginpar{%
      {\color{red}\scriptsize\textbf{(\thechange)\label{#1} {#2}}}}}{{\color{red}{#3}%
    }}%
  }%
  \renewcommand{\changeNoMargin}[1]{{\color{red}{#1}}}%
}
\makeatother

\newcommand{\ourFormalLang}{$F_{\text{P4R}}$\xspace}
\newcommand{\ourDSL}{\texttt{P4R-Type}\xspace}

%% file: abstract.tex
\begin{abstract}
    Software-Defined Networking (SDN) significantly simplifies programming,
    reconfiguring, and optimizing network devices, such as switches and routers.
    The \emph{de facto} standard for programmming SDN devices is the P4 language.
    However, the flexibility and power of P4, and SDN more generally, gives rise
    to important risks. As a number of incidents at major cloud providers have
    shown, errors in SDN programs can compromise the availability of networks,
    leaving them in a non-functional state. The focus of this paper are errors
    in control-plane programs that interact with P4-enabled network devices via
    the standardized P4Runtime API. For clients of the P4Runtime API it is easy
    to make mistakes that lead to catastrophic failures, despite the use of
    Google's Protocol Buffers as an interface definition language.

    This paper proposes \ourDSL, a novel verified P4Runtime API for Scala that
    performs static checks for P4 control plane operations, ruling out
    mismatches between P4 tables, allowed actions, and action parameters. As a
    formal foundation of \ourDSL, we present the \ourFormalLang calculus and its
    typing system, which ensure that well-typed programs never get stuck by
    issuing invalid P4Runtime operations. We evaluate the safety and flexibility
    of \ourDSL with 3 case studies. To the best of our knowledge, this is the
    first work that formalises P4Runtime control plane applications,
    and a typing discipline ensuring the correctness of P4Runtime operations.

\end{abstract}


%% file: intro.tex
\section{Introduction}
\label{sec:intro}

\emph{Software-Defined Networking (SDN)} is a modern approach to network
management where network traffic is handled by highly flexible and programmable
devices (such as switches and routers).  The high-level design of SDN is based
on a separation of the \emph{data plane} and the \emph{control plane}:
\begin{itemize}
  \item the \emph{data plane} is responsible for processing network packets by
    applying a set of programmable rules.  For example, a network router may be
    programmed to match network packets based on their content, and then
    forward, rewrite, or drop them;
  \item the \emph{control plane} is responsible for setting up and updating the
    packet processing rules executed in the data plane.  For example, a company
    may enact a new network management policy by running a control plane program
    that reconfigures their routers with new rules.
\end{itemize}
%
This separation simplifies network management and enables network administrators
to quickly and easily reconfigure and optimize network traffic flows.%


The de facto Open Source standard for SDN is P4~\cite{p4spec}. In P4, the
the data plane is programmed by specifying packet processing \emph{tables} which
select the \emph{actions} to perform when a network packet matches certain
patterns.  The P4 standard also defines a control plane API (called P4Runtime~\cite{p4rtspec}) for writing programs that query or alter the configuration of
P4-enabled network devices.

\changeNoMargin{
Unfortunately, the power and ease of automation of SDN
come with risks: a mistake in an SDN program can leave a network in a
non-functional state.  Indeed, erroneous configuration changes have compromised
the availability of entire regions of large cloud providers~\cite{Sharwood16}.%
}
\change{change:rg:1}{Add motivation}{
A recent study by~\citeN{bhardwaj2021comprehensive} shows that
38.8\% of SDN bugs are triggered when the controller
\emph{``attempts to process system configurations''}
--- i.e.~read, add, update, delete table entries;
the authors add that \emph{``this fact is astounding because a critical
motivation for SDN is to move towards automation and eliminate
configuration-based errors.''}
}
\changeNoMargin{
In this paper, we focus on statically preventing a specific form of P4Runtime
controller bug: attempting to read/insert/modify/delete P4 table entries that
do not conform to the actual table layout of the P4 data plane. Such erroneous
attempts are not statically checked by the official, weakly-typed P4Runtime
API, as we explain below. Preventing this form of bug does not avert all
possible P4 configuration processing bugs (e.g. a P4Runtime controller may
insert a well-formed but incorrect routing table entry, or omit or delete a
necessary entry) --- but it provides a baseline correctness guarantee towards
more thorough static verification of P4Runtime applications (that we discuss as
future work in \Cref{sec:conclusions}).
}

\subsection*{The Problem with Weakly-Typed P4Runtime APIs}

For a concrete example of how mistakes could happen, consider
\Cref{fig:build-table-error} (left): it is based on the P4
documentation~\cite{p4tutorial}, and shows a control plane program written in
Python using the official P4Runtime API.  The program is connected to a
P4-enabled switch, and inserts a new entry (i.e.~a packet processing rule) into
a table called \verb|IPv4_table|, meaning: \emph{``if a packet has destination
address \texttt{10.0.1.1}, then perform the action \texttt{IPv6\_forward} with
the given parameters.''}  (We provide more details about P4
in~\Cref{sec:background}.)

\begin{figure}
\begin{minipage}{0.42\linewidth}
\begin{lstlisting}[language=Python,escapechar=\%,basicstyle=\scriptsize\ttfamily,frame=single,numbersep=1mm]
table_entry = buildTableEntry(
  table_name = %{\color{red}"IPv4\_table"}%,
  match_fields = {
    %{\color{blue}"IPv4\_dst\_addr""}%: ("10.0.1.1", 32)
  },
  action_name = %{\color{red}"IPv6\_forward"}%,
  action_params = {
    %{\color{blue}"mac\_dst"}%: "08:00:00:00:01:00",
    %{\color{blue}"port"}%: 1
  })
switch.WriteTableEntry(table_entry)
\end{lstlisting}
\end{minipage}
\hfill
\begin{minipage}{0.57\linewidth}
\begin{lstlisting}[language=Scala,escapechar=\%,basicstyle=\scriptsize\ttfamily,frame=single,numbersep=1mm]
insert(switch,
  TableEntry(
    %{\color{red}"IPv4\_table"}%,
    Some(%{\color{blue}"ipv4\_dst\_addr"}%, LPM(bytes(10,0,1,1), 32)),
    %{\color{red}"IPv6\_forward"}%, %{\hlight{\color{red}// ERR: invalid action in IPv4\_table}}%
    ((%{\color{blue}"mac\_dst"}%, bytes(8,0,0,0,1,0)),
     (%{\color{blue}"port"}%, 1)),
  )
)
\end{lstlisting}
\end{minipage}
\vspace{-3mm}
\caption{Example of control plane P4 programs.  Left: a Python program using the
official P4Runtime API.  Right: the equivalent Scala 3 program using verified
API \ourDSL.}
\label{fig:build-table-error}
\end{figure}

The Python program in \Cref{fig:build-table-error} contains an error: the table
\verb|IPv4_table| in the switch does \emph{not} allow for an action called
\verb|IPv6_forward| (although that action may be allowed by other tables in the
same switch).  The P4Runtime Python API detects this discrepancy at run-time,
and throws an exception --- which may cause the program to fail half-way during
a series of related P4 rule updates, leaving the network configuration in an
inconsistent state.  The same program may have other problems: e.g.~does the
intended action for \verb|IPV4_table| actually take two parameters?  Is one of
such parameters actually called \verb|mac_dst|?  Again, the official
P4Runtime Python API would only spot these issues at run-time, by throwing
exceptions.

As this example shows, it is all too easy to make mistakes when writing control
plane programs in scripting languages (like Python) that don't perform static
checks to ensure the validity of P4Runtime operations.  However, statically
detecting such errors is not trivial: to prevent errors without being
overly-restrictive, the static checks must take into account the actual
\emph{dependencies} between the packet processing tables available in a
P4-enabled device, the actions allowed by each specific table, and the
parameters expected by each specific action.

\change{change:requirements}{Add requirements}{%
Our objective is to design and develop a strongly-typed P4Runtime API that
addresses the issues above, while satisfying \textbf{three key requirements}:

\begin{enumerate}[label={(\textbf{R{\arabic*}})}]
  \item\label{requirement:formal-foundation} the API must have a formal
    foundation for proving that well-typed programs never get stuck by issuing
    invalid P4Runtime operations or receiving unexpected responses;
  \item\label{requirement:existing-prog-lang} the API must be written and usable
    in an \emph{existing} programming language --- i.e.~the implementation of
    the formal results (from requirement \ref{requirement:formal-foundation})
    must not depend on a bespoke programming language nor type checker;
  \item\label{requirement:codegen} if the API depends on code
    generation, the amount of generated code must be minimal.
\end{enumerate}%
}

\subsection*{Our Proposal: \ourDSL and its Formal Foundation \ourFormalLang}

This paper proposes \ourDSL, a novel verified P4Runtime API for Scala 3 that
performs \emph{static} checks for P4 control plane operations, ruling out
mismatches between P4 tables, allowed actions, and action parameters.  Programs
written with \ourDSL look like the one shown in \Cref{fig:build-table-error}
(right): albeit similar to its Python equivalent, the \ourDSL program does
\emph{not} compile, because (thanks to its type constraints) the off-the-shelf
Scala 3 compiler can spot that the action on line 5 is not valid for the table
\verb|IPv4_Table|.  The Scala 3 compiler can also similarly spot any discrepancy
between a selected action and the supplied parameters.

\ourDSL has a formal foundation: \ourFormalLang, a calculus and typing system
allowing us to state and prove that \emph{``well-typed \ourFormalLang programs
never perform invalid P4Runtime operations''} (like the mistake in
\Cref{fig:build-table-error}).  \ourFormalLang is specifically designed for
implementation as a Scala 3 API, and for enabling the ``Python-like''
P4Runtime programs shown in \Cref{fig:build-table-error}.

To the best of our knowledge, this is the first work that formalises control
plane applications based on P4Runtime, and a typing discipline to ensure the
correctness of P4Runtime operations.

\subsection*{Contributions and Outline of the Paper}

\change{change:clarify-contribs}{}{%
After a background and overview (\Cref{sec:background}), we introduce our main
contributions:

\begin{enumerate}
\item The first formal model of P4Runtime networks
  (\Cref{sec:modelling-p4runtime-networks}) consisting of clients written in our
  novel formal language \ourFormalLang (\Cref{sec:p4runtime-client-syntax}) and
  servers with different configurations (\Cref{sec:p4runtime-server-model})
  interacting with each other (\Cref{sec:p4runtime-networks}).

\item A typing discipline for \ourFormalLang (\Cref{sec:type-system}) ensuring
  that if a client is well-typed w.r.t.~the configuration of the surrounding
  P4Runtime network servers (under the server-configuration-to-type encoding we
  introduce in \Cref{def:server-config-encoding}), then the client will never
  perform invalid P4Runtime operations nor get stuck
  (\Cref{theorem:preservation,theorem:progress}).  To ensure that these results
  translate into a verified P4Runtime client API in an \emph{existing}
  programming language (as per requrement~\ref{requirement:existing-prog-lang}
  above), we equip the \ourFormalLang typing system with a limited form of
  type-level computation based on \emph{match types} \cite{BlanvillainBKO22} and
  \emph{singleton types}, both available in Scala 3.  Our development of
  \ourFormalLang also contributes a novel combination of
  \begin{enumerate*}[label=\emph{(\roman*)}]
    \item match types \emph{without} default cases,
    \item structural subtyping, and
    \item singleton types%
  \end{enumerate*}: the details and challenges are explained in
  \Cref{remark:differences-blanvillain22}.  (Besides, our theory and results are
  not Scala-specific and can be embedded e.g. in dependently-typed languages
  like Coq.)

\item The first implementation of a verified P4Runtime API, called \ourDSL
  (\Cref{sec:implementation}) and published as companion artifact of this
  paper. \ourDSL is based on the formalisation and results of \ourFormalLang, is
  written and usable in Scala 3, and only depends on a small amount of
  autogenerated type definitions (based on our server-configuration-to-type
  encoding in \Cref{def:server-config-encoding}): therefore, \ourDSL satisfies
  the requirements \ref{requirement:formal-foundation},
  \ref{requirement:existing-prog-lang}, and \ref{requirement:codegen} above.  We
  demonstrate the features of \ourDSL with 3 case studies (\Cref{sec:case-study}),
  and discuss the drawbacks of alternative approaches
  (\Cref{why-match-singleton-types}).
\end{enumerate}
}

We discuss the related work in \Cref{sec:related-work} and conclude in
\Cref{sec:conclusions}.

%% file: background.tex
\section{Background and Overview}
\label{sec:background}

We now provide an overview of Software Defined Networks
(\Cref{sec:overview-sdn}), the P4 data plane (\Cref{sec:overview-p4}), and P4Runtime
(\Cref{sec:overview-p4runtime}), followed by a bird's eye view of our
approach (\Cref{sec:overview-approach}).

\subsection{Software Defined Networks}
\label{sec:overview-sdn}

Software Defined Networking (SDN) is an umbrella that covers several technologies to sopport dynamic and programmable network reconfigurations. SDN can be used to improve network performance (e.g.~intelligent load balancers~\cite{belgaum2020systematic}), efficiency (e.g. network resource virtualisation and partitioning among customers of a cloud provider~\cite{ordonez2017network}), and security (AI based anomaly detection systems~\cite{garg2019hybrid}).
As mentioned in \Cref{sec:intro}, an SDN
consists (roughly speaking) of at least
two architectural components:
\begin{itemize}
\item \emph{data plane} devices with direct control of packet processing ---
  e.g.~network interface cards, or switches, or a network of such devices; and
\item a centralised or distributed \emph{controller}, which is in charge of
  interacting, via an \emph{interface}, with the data plane devices to manage
  network flows.
\end{itemize}

\subsection{Programmable Data Plane and the P4 Language}
\label{sec:overview-p4}

For many years SDN data plane elements were implemented with fixed-function application-specific integrated circuits (ASICs),
with very limited programmability.  In fact, programmable switches were two orders of magnitude slower than the corresponding fixed-function ASICs. However, newer programmable switches can run as fast as fixed-function ones. The key for this improvement is the usage of dedicated programmable accelerators, called Network Processing Units (NPUs), and FPGAs.
Programmable data-processing enables the support of customised network protocols, for example VPN-aware data processing and in-line packet inspection.
NPUs and FPGAs cannot be programmed using general purpose languages. Hence, high-speed data plane must be programmed with dedicated programming languages.
Recently, P4~\cite{p4spec} has risen as the main Domain Specific Language for data plane. P4 can be compiled to a variety of targets, including NPUs (e.g.~Intel Tofino), FPGAs, and software switches.
The key form of configuration for a P4 program are \emph{tables}, which are manipulated
by the control plane.

The P4 fragment below defines the tables \verb|IPv4_table| and \verb|IPv6_table|,
with an ``if'' statement that inspect the header of an incoming network packet and selects one of the two tables.
When the program executes \verb|IPv4_table.apply()|, the P4 system performs 3 steps:
\begin{itemize}
  \item it computes a \emph{key} value from the network packet being processed.
  In this case, the key is the IPv4 destination address of the packet;
\end{itemize}

\smallskip
\noindent%
\begin{minipage}{0.45\linewidth}
\vspace{1mm}
\begin{lstlisting}[escapechar=\%,basicstyle=\scriptsize\ttfamily,frame=single]
table %{\color{blue}"IPv4\_table"% {
  key = { hdr.ip.IPv4_dst_addr: lpm; }
  actions = { Drop_action;
              IPv4_forward; }
}

table% {\color{blue}"IPv6\_table"% {
  key = { hdr.ip.IPv6_dst_addr: lpm; }
  actions = { Drop_action;
              IPv6_forward; }
}
...
if (hdr.ip.version == 4w4)
  IPv4_table.apply();
else
  IPv6_table.apply();
\end{lstlisting}
\end{minipage}
\begin{minipage}{0.53\linewidth}
\begin{itemize}
  \item it \emph{matches} the packet against the entries
    in \verb|IPv4_table|. This matching operation can use several algorithms: in this case, \verb|IPv4_table| uses the
    longest-prefix match (\verb|lpm|), hence system chooses the table entry with the longest IP prefix matching the key. Each table entry represents
    a packet processing rule: it specifies a value to match against the packet, an \emph{action}
    to execute after a successful match, and its arguments;
  \item when a table entry matches the packet,
    the P4 system executes the corresponding action,
    which is a function that must be defined in the P4 program itself.
\end{itemize}
\end{minipage}%
\vspace{2mm}

\noindent
In this example, the definition of \verb|IPv4_table| says that a table entry can
select one of two possible actions (\verb|Drop_action| and \verb|IPV4_forward|,
defined below) to execute after a packet match:


\begin{lstlisting}[escapechar=\%,basicstyle=\scriptsize\ttfamily,frame=single]
action Drop_action() {
  outCtrl.outputPort = DROP_PORT;
}
action IPV4_forward(EthernetAddress mac_dst, PortId port) {
  packet.ethernet.dstAddr = mac_dst;
  packet.ip.version = 4w4;
  packet.ip.ttl = headers.ip.ttl - 1;
  outCtrl.outputPort = port;
}

\end{lstlisting}
\verb|Drop_action| does not require any argument and simply forwards the packet to a ``port'' that drops (i.e.~discards) it. The \verb|IPv4_forward| action requires two arguments: therefore, when a table entry
in \verb|IPv4_table| wants to invoke the action \verb|IPv4_forward|, the table entry must also specify a destination Ethernet address and a port.
In the following section we briefly discuss exemples that violate these constraints.






\subsection{P4Runtime and P4Info Metadata Files}
\label{sec:overview-p4runtime}

%
%
%
Today, applications that control P4-enabled devices use a control plane API
called P4Runtime~\cite{p4rtspec}: the control application (acting as a
P4Runtime client) connects to P4 device (which acts as a P4Runtime server) and
issues API calls to query and modify the device configuration.

Thanks to a portable design based on Google's Protobuf interface definition language, P4Runtime programs may be written in any programming language with Protobuf support --- and the official P4Runtime API implementation is written in Python \todo{Reference}.
The use of general-purpose programming languages for control plane applications is possible because their performance is less critical the one of the data plane;
moreover, general purpose languages allow for reusing existing software stacks and support a wide range of application domains.

In the usual workflow, when a P4 data plane program is compiled, it yields two
outputs:
\begin{enumerate}
  \item a packet-processing ``executable'' deployed on a P4-enabled device
    (e.g. on a switch); and
  \item a \emph{P4Info metadata file}, which summarises all the entities defined
  in the P4 program --- in particular, its tables and actions.  Each entity has
  a numeric identifier.
\end{enumerate}
To interact with a P4-enabled device (e.g.~to add entries to its tables), a
P4Runtime program uses the P4Info metadata corresponding to the P4 program
deployed on the device.


\begin{figure}
\begin{lstlisting}[escapechar=\%,basicstyle=\scriptsize\ttfamily,frame=single]
"tables": [ {    "preamble": {"name": %{\color{blue}"IPv4\_table"}%},
              "matchFields": [ {"name": %{\color{blue}"IPv4\_dst\_addr"}%, "matchType": "LPM"} ],
               "actionRefs": [ {"id": 1010} , {"id": 3030} ] },
            {    "preamble": {"name": %{\color{blue}"IPv6\_table"}%},
              "matchFields": [ {"name": %{\color{blue}"IPv6\_dst\_addr"}%, "matchType": "LPM"} ],
               "actionRefs": [ {"id": 2020} , {"id": 3030} ] } ],
"actions": [ { "preamble": {"id": 1010, "name": %{\color{blue}"IPv4\_forward"}%},
                 "params": [ {"name": %{\color{blue}"mac\_dst"}%, "bitwidth": 48},
                             {"name": %{\color{blue}"port"}%, "bitwidth": 9} ] },
             { "preamble": {"id": 2020, "name": %{\color{blue}"IPv6\_forward"}%},
                 "params": [ {"name": %{\color{blue}"mac\_dst"}%, "bitwidth": 48},
                             {"name": %{\color{blue}"port"}%, "bitwidth": 9 } ] },
             { "preamble": {"id": 3030, "name": %{\color{blue}"Drop\_action"}%},
                 "params": [ ] } ]
\end{lstlisting}
\vspace{-3mm}%
\caption{Example P4Info metadata file with the tables and actions of the P4 program in \Cref{sec:overview-p4}.
\changeNoMargin{
    For brevity, we only show actions IDs and omit tables IDs.
  }}
\label{fig:p4info}
\end{figure}

\Cref{fig:p4info} shows an example of P4Info metadata file for the P4 program in
\Cref{sec:overview-p4}.  From this metadata we can see that a P4 device running
that program has two tables: \verb|IPv4_table| and \verb|IPv6_table|.  Each
table has one key that is an address of the corresponding IP protocol version.
The entries of \verb|IPv4_table| and \verb|IPv6_table| can invoke actions
\verb|IPv4_forward| and \verb|IPv6_forward| (respectively) and must provide a
MAC address and port as action's arguments.  All table entries can invoke
\verb|Drop_action|, which has no parameters.

P4Runtime applications can change the configuration of a P4-enabled device by
adding, updating, and deleting table entries.  P4Runtime applications can also
read the table contents, possibly using \emph{wildcards} to filter the results.
As shown by the Python program in Figure~\ref{fig:build-table-error}, it is easy
to make mistakes if the language does not perform static checks on table
updates.  Specifically, the P4Info metadata in \Cref{fig:p4info} says that any
new entry added to \verb|IPv4_table| cannot use the \verb|IPv6_forward| action
--- which is the mistake highlighted in \Cref{fig:build-table-error}.

\subsection{An Overview of Our Approach}
\label{sec:overview-approach}

To address the issues described above, we propose \ourDSL: a verified P4Runtime
API for Scala 3, with a tool to translate P4Info metadata into Scala types.  Our
approach is depicted below.  

\begin{center}
  \includegraphics[width=0.75\linewidth]{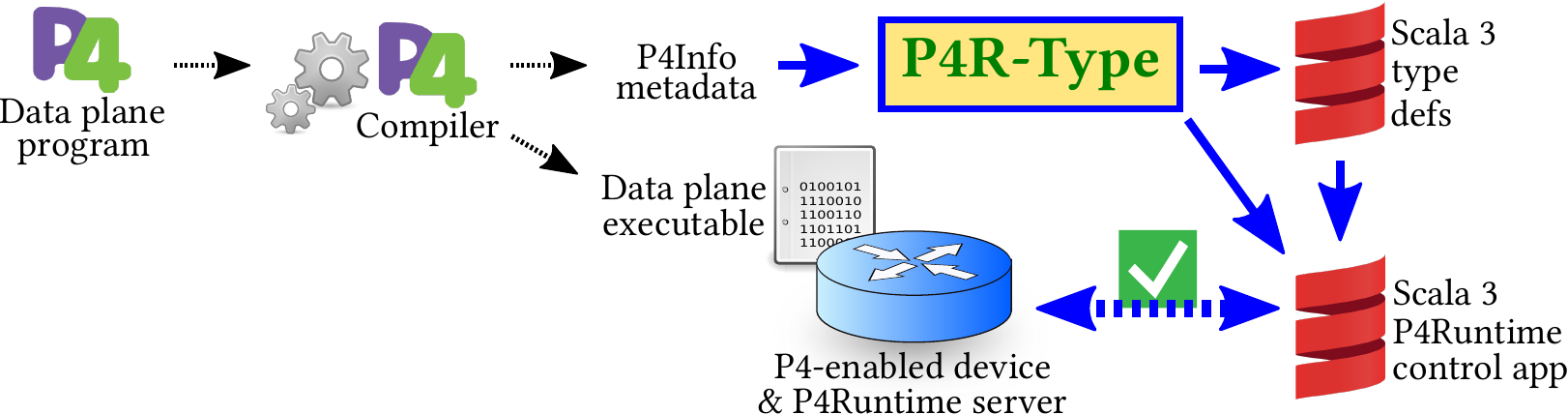}%
\end{center}

As usual, a P4 data plane program is compiled and deployed on one or more
P4-enabled network devices; the program's P4Info metadata is made available for
P4Runtime applications.  This is where \ourDSL comes into play: a programmer
can write a P4Runtime control application by importing (1) the \ourDSL library,
and (2) a set of type definitions automatically generated from P4Info metadata.
If the \ourDSL-based application type-checks, then it will be able to connect to
a P4 device (acting as P4Runtime server)
and perform P4Runtime operations that never violate the

\vspace{0.5mm}%
\noindent%
\begin{minipage}{0.55\linewidth}
\includegraphics[width=\linewidth]{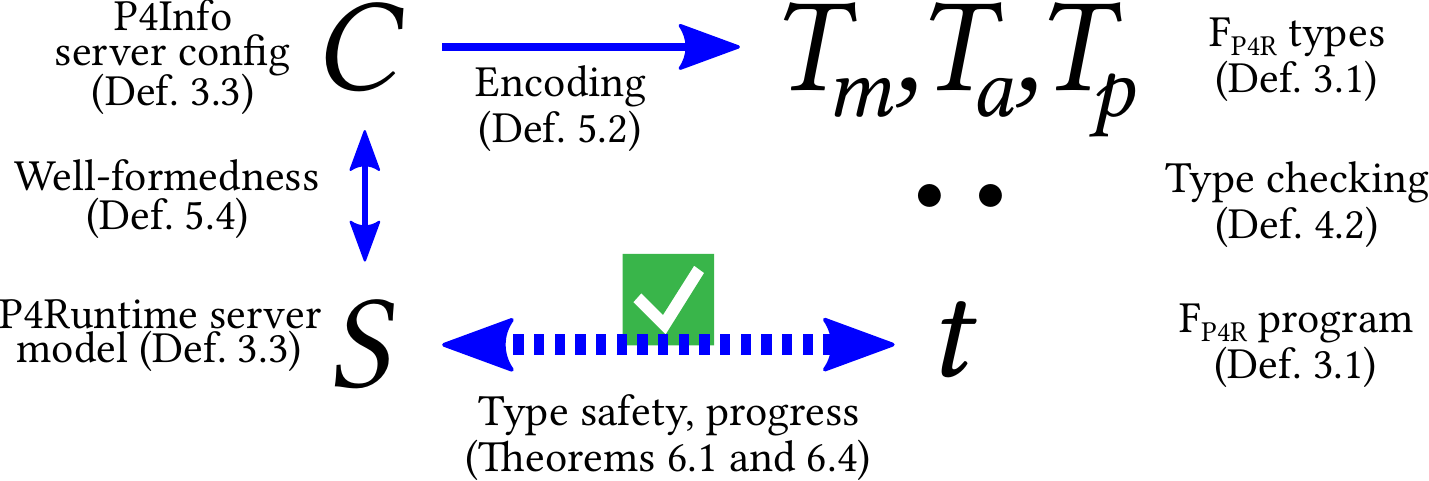}
\end{minipage}
\hfill
\begin{minipage}{0.44\linewidth}
  device configuration --- provided that the
  configuration matches the P4Info metadata.
  The design and implementation of \ourDSL is based on a formal model allowing us
  to reason about the behaviour of P4Runtime client applications and P4 devices
  acting as P4Runtime servers.
  \vspace{.5mm}
\end{minipage}

\noindent%
Our formal model is outlined above: a P4Runtime server $S$ holds tables
and actions that are well-formed w.r.t.~a configuration $C$ (which represents
P4Info metadata).
We define an encoding from a P4Info configuration $C$
into a set of types for \ourFormalLang: a formal
calculus describing P4Runtime client applications%
\change{change:rg:9}{Add Scala 3 features}{.
We design the typing discipline of \ourFormalLang
to include match types and singleton types, which are also present in
the Scala 3 typing system: this design allows us to implement our results as a
Scala 3 API (i.e., \ourDSL) reflecting the typing constraints of \ourFormalLang.%
}
Then, we prove our \Cref{theorem:preservation,theorem:progress}: %
if a \ourFormalLang program $t$ type-checks with types
encoded from a P4Info configuration $C$, then $t$ will interact correctly with
any P4Runtime server $S$ that is well-formed w.r.t.~$C$.



%% file: language.tex
\change{change:clarify-section-title}{}{%
\section{A Model of P4Runtime Clients, Servers, and Networks}
\label{sec:modelling-p4runtime-networks}
\label{sec:language}
}%

We now illustrate how we model P4Runtime networks consisting of P4-enabled
devices (acting as servers), and control applications (the clients) that connect
and modify the devices' P4 table entries.
In \Cref{sec:p4runtime-client-syntax} we introduce the syntax of \ourFormalLang,
a formal language for modelling P4Runtime client programs, with the capability of
connecting to P4Runtime servers and performing P4Runtime operations.
In \Cref{sec:p4runtime-server-model} we model P4Runtime servers by focusing on
their internal configuration, i.e.~their P4 tables, packet matching methods, and
actions.
In \Cref{sec:p4runtime-networks} we formalise a P4Runtime network as a parallel
composition of P4Runtime clients and servers.

We introduce the semantics of \ourFormalLang programs, servers, and networks
later on (in \Cref{sec:semantics}) after introducing the typing system (in
\Cref{sec:type-system}).


\subsection{The \ourFormalLang Language for P4Runtime Clients}
\label{sec:p4runtime-client-syntax}

In \Cref{def:p4runtime-client-syntax} below we introduce the syntax of the
\ourFormalLang language and its types.  \ourFormalLang is designed as an
extension of \fSub (System F with subtyping~\cite{Cardelli1994SystemFSub})
augmented with:
\begin{itemize}
  \item \textbf{P4Runtime-specific operations}: primitives and types for server
   addresses and channels;
  \item \textbf{singleton types}, i.e.~types inhabited by exactly one value; and
  \item \textbf{match types}, introducing the capability of performing
   type-level computations and refining the result type of a pattern matching.
   Our match types are based on the work of~\citeN{BlanvillainBKO22} (which in
   turn formalises the corresponding feature of the Scala 3 programming
   language) %
   \change{change:refer-blanvillain-diffs}{}{%
   --- but our adaptation includes significant differences: we discuss
   them later on, in \Cref{remark:differences-blanvillain22}.%
   }
\end{itemize}

\begin{definition}[Syntax of \ourFormalLang]
  \label{def:p4runtime-client-syntax}
  The syntax of \ourFormalLang terms $t$ and types $T$ is shown in
  \Cref{fig:language-syntax-terms-types}%
  \change{change:remove-overbar-syntax}{}{ --- where $I$ (used to index records
    and pattern matching terms, and record and match types) represents a finite,
    non-empty set containing sequential natural numbers $1,2,\ldots$%
  }
  Moreover,
  \Cref{fig:language-syntax-sugar-types} introduces some frequently-used
  syntactic abbreviations.
\end{definition}

\begin{figure}
  \newcommand{\rulesSep}{0.75mm}%
  \[
  \small
  \begin{array}{r@{\;\;}rcl@{\quad}l}
    \text{Term} & t & \Coloneqq & v & \text{(Value)}\\
    & & | & x \;\;|\;\; y \;\;|\;\; z \;\;|\;\; \ldots & \text{(Variable)}\\
    & & | & \concat{t}{t} \;\;|\;\; \listHead t \;\;|\;\; \listTail t & \text{(List constructor, head, tail)}\\
    & & | & \text{\changeNoMargin{$\{f_i=t_{i}\}_{i \in I}$}} \;\;|\;\; t.f & \text{(Record constructor, field selection)}\\
    & & | & t\ t \;\;|\;\; t\ T& \text{(Term application, type application)}\\
    & & | & \letExp{x}{t} t & \text{(Binder)}\\
    & & | & \text{\changeNoMargin{$\match{t}{x_i : T_i \Rightarrow t_i}{i \in I}$}} & \text{(Pattern matching)}\\
    & & | & \hlightM{\mathit{op}} & \hlightM{\text{(P4Runtime operation)}}\\[\rulesSep]
    \hlightM{\text{P4Runtime op.}} & \hlightM{\mathit{op}} & \Coloneqq & \hlightM{\ttt{Connect}(w)} & \hlightM{\text{(Connect to server)}}\\
    & & | & \hlightM{\ttt{Read}(w,w)} & \hlightM{\text{(Read entries)}}\\
    & & | & \hlightM{\ttt{Insert}(w,w)} & \hlightM{\text{(Insert an entry)}}\\
    & & | & \hlightM{\ttt{Modify}(w,w)} & \hlightM{\text{(Modify entries)}}\\
    & & | & \hlightM{\ttt{Delete}(w,w)} & \hlightM{\text{(Delete entries)}}\\[\rulesSep]
    \text{P4Runtime arg.} & w & \Coloneqq & v \;\;|\;\; x \;\;|\;\; y \;\;|\;\; z \;\;|\;\; \ldots & \text{(Value or variable)}\\[\rulesSep]
    \text{Value} & v & \Coloneqq & \hlightM{v_G} & \hlightM{\text{(Ground value)}}\\
    & & | & \concat{v}{v} & \text{(List value)}\\
    & & | & \text{\changeNoMargin{$\{f_i=v_i\}_{i \in I}$}} & \text{(Record value)}\\
    & & | & \abs{x}{T} t & \text{(Lambda abstraction)}\\
    & & | & \tabs{X}{T} t & \text{(Type abstraction)}\\[\rulesSep]
    \text{Ground value} & v_G & \Coloneqq & \unit & \text{(Unit value)}\\
    & & | & 0 \;\;|\;\; 1 \;\;|\;\; 2 \;\;|\;\; \ldots & \text{(Integer)}\\
    & & | & \ttt{true} \;\;|\;\; \ttt{false} & \text{(Boolean)}\\
    & & | & \stringTerm{Hello} \;\;|\;\; \stringTerm{World} \;\;|\;\; \ldots & \text{(String)}\\
    & & | & \hlightM{\ttt{b}(42) \;\;|\;\; \ttt{b}(192,168,1,1) \;\;|\;\; \ldots} & \hlightM{\text{(Bytes)}}\\
    & & | & \hlightM{a_{T_m,T_a,T_p}} & \hlightM{\text{(Server address)}}\\
    & & | & \hlightM{s_{T_m,T_a,T_p}} & \hlightM{\text{(Client-server channel)}}\\[\rulesSep]
    & & | & \concat{v_G}{v_G} \;\;|\;\; \nil & \text{(List of ground values)}\\
    & & | & \text{\changeNoMargin{$\{f_i=v_{Gi}\}_{i \in I}$}} & \text{(Record of ground values)}\\
    \text{Type} & T & \Coloneqq & \top & \text{(Top type)}\\
    & & | & \tx{Int} \;|\; \tx{Bool} \;|\; \tx{String} \;|\; \tx{Unit}  \;|\; \hlight{\tx{Bytes}} \;|  \; \ldots & \text{(Basic types)}\\
    & & | & \hlightM{\ServerRef{T,T,T}} \;\;|\;\; \hlightM{\Chan{T,T,T}} & \hlightM{\text{(P4Runtime server address \& channel)}}\\
    & & | & \text{\changeNoMargin{$\{\types{f_i}{T_i}\}_{i \in I}$}} \;\;|\;\; [T] & \text{(Record and list types)}\\
    & & | & T \rightarrow T & \text{(Function type)}\\
    & & | & X \;\;|\;\; \quant{X}{T} T \;\;|\;\; T\ T & \text{(Type var., bounded quant., application)}\\
    & & | & \typeUnion{T}{T} & \text{(Union type)}\\
    & & | & \hlightM{\singleton{v_G}} & \hlightM{\text{(Singleton type for ground value $v_G$)}}\\
    & & | & \hlightM{\text{\changeNoMargin{$\match{T}{T_i \Rightarrow T'_i}{i \in I}$}}} & \hlightM{\text{(Match type)}}\\
  \end{array}
  \]
  \caption{Syntax of \ourFormalLang terms and types.  Non-standard extensions to \fSub are \hlight{highlighted}.}
  \label{fig:language-syntax-terms}
  \label{fig:language-syntax-types}
  \label{fig:language-syntax-terms-types}
\end{figure}

Most of \Cref{def:p4runtime-client-syntax} is based on standard \fSub constructs
and extensions (in particular, lists and records).  The key deviations are the
\hlight{highlighted} constructs in \Cref{fig:language-syntax-terms-types}:
\begin{itemize}
  \item a \textbf{P4Runtime operation $\mathit{op}$} allows a client to connect
    to a P4Runtime server, and query or change the entries in its configuration;
  \item a \textbf{ground value $v_G$} is a value that does not contain
    lambda nor type abstractions.  A ground value is a ``simple'' value
    (e.g.~string, integer, \ldots), or a list or record of ground values.  For
    each ground value $v_G$, there is a \textbf{singleton type
    $\singleton{v_G}$} only inhabited by $v_G$ itself;
  \item a \textbf{byte string $\ttt{b}(\ldots)$} is the byte representation of a
    sequence of integers, and has type $\tx{Bytes}$;
  \item a \textbf{server address $a_{T_m,T_a,T_p}$} represents a handle for
    connecting to a P4Runtime server (in practice, it represents its IP address
    and TCP port). A server address $a_{T_m,T_a,T_p}$ has a corresponding
    \textbf{server address type} $\ServerRef{T_m,T_a,T_p}$, where the type
    parameters reflect information available in the server's P4Info
    file:\footnote{%
      The instantiation of the type parameters $T_m,T_a,T_p$ is detailed later,
      in \Cref{eg:typable-untypable-p4runtime-ops}, \Cref{def:server-config-encoding} and \Cref{eg:server-config-types-encoded}.%
    }
    \begin{itemize}
      \item $T_m$ describes the \emph{matches} of each table in the server
        configuration;
      \item $T_a$ describes the \emph{actions} that could be performed after
        a network packet is matched;
      \item $T_p$ describes the types of the \emph{parameters} expected by each
      action.
    \end{itemize}
    For brevity, we will often just write $a$, omitting the subscript;
  \item a \textbf{client-server connection $s_{T_m,T_a,T_p}$} represents a
    communication channel (created after establishing a connection) that a
    P4Runtime server and client use to communicate.  A connection value has a
    corresponding \textbf{channel type} $\Chan{T_m,T_a,T_p}$, whose type
    arguments have the same meaning outlined above for
    $\ServerRef{T_m,T_a,T_p}$.  For brevity, we will often just write $s$,
    omitting the subscript;
  \item a \textbf{match type} describes a type-level pattern matching, as
    illustrated in \Cref{eg-match-type} below.
\end{itemize}

\begin{example}[Match Types]
  \label{eg-match-type}
  Consider the following match type:

  \smallskip\centerline{$
    \tx{Int}\; \texttt{match}\; \curl{\tx{Int} \Rightarrow \tx{Bool},\; \tx{String} \Rightarrow \tx{Unit}}%
  $}\smallskip

\change{change:rg:10}{Add explanation of continuation types}{
  We call the types $\tx{Bool}$ and $\tx{Unit}$ (i.e., the types of the expressions that can be executed after a match case is selected) \emph{continuation types} of the match type.}
  This match type ``reduces'' to the continuation type $\tx{Bool}$, because its
  guard matches the type \tx{Int}.  (More precisely, in \Cref{sec:type-system}
  we will see that the match type and the selected continuation are
  subytping-equivalent.)
  Now consider the following match type abstracted over type $T$:
  \[
  \begin{array}{rcl}
      \tx{Table\_Actions} &\triangleq& \quantT{T} T \texttt{ match } \curl{\\
      & &\qquad \singleton{\ttt{"IPv4\_table"}} \Rightarrow \singleton{\ttt{"IPv4\_forward"}} \cup \singleton{\ttt{"drop"}}\\
      & &\qquad \singleton{\ttt{"IPv6\_table"}} \Rightarrow \singleton{\ttt{"IPv6\_forward"}} \cup \singleton{\ttt{"drop}}\\
      & &\qquad \singleton{\wildcard}           \Rightarrow \singleton{\tx{\wildcard}} \;\;}
  \end{array}
  \]
  Intuitively, 
  the type\; $\tx{Table\_Actions}\ \singleton{\ttt{"IPv4\_table"}}$
  (i.e.~type $\tx{Table\_Actions}$ applied to the singleton type $\singleton{\ttt{"IPv4\_table"}}$)
  \;``reduces'' to the continuation type $\singleton{\ttt{"IPv4\_forward"}} \cup
  \singleton{\ttt{"drop"}}$ (which is a union of singleton types).  Instead,\; $\tx{Table\_Actions}\
  \singleton{\ttt{"*"}}$ \;``reduces'' to the singleton type
  $\singleton{\ttt{"*"}}$.
\end{example}

\begin{figure}
  \[
    \small
  \begin{array}{rcl}
    \forall X. T                         & \triangleq & \quant{X}{\top} T\\
    \tx{Option}                          & \triangleq & \quantT{A}\typeUnion{\curl{\tx{some} : A}}{\curl{\tx{none} : \tx{Unit}}}\\
    \TableEntry                          & \triangleq & \quantT{T_m}\quantT{T_a}\quantT{T_p}\quantT{X_n}\quant{X_a}{T_a\ X_n}
                                                \left\{\begin{array}{@{}l@{\;}l@{}}
                                                         \tx{name} : X_n, & \tx{matches} : T_m\ X_n,\\
                                                         \tx{action} : X_a, & \tx{params} : T_p\ X_a
                                                       \end{array}\right\}\\
    \PPPPEntity & \triangleq & \quantT{T_m}\quantT{T_a}\quantT{T_p}\quantT{X_n}\quant{X_a}{T_a\ X_n}
                           \left(\begin{array}{@{}l@{}}
                              \TableEntry\ T_m\ T_a\ T_p\ X_n\ X_a\\
                              \vphantom{x} \cup \ldots
                           \end{array}\right)
  \end{array}
  \vspace{-3mm}
  \]
  \caption{Syntactic sugar and abbreviations for \ourFormalLang types (based on \Cref{fig:language-syntax-types}).}
  \label{fig:language-syntax-sugar-types}
\end{figure}

We complement the syntax of \ourFormalLang with the abbreviations in
\Cref{fig:language-syntax-sugar-types} --- which introduces, in particular, these
P4Runtime-specific aliases:
\begin{itemize}
  \item ``$\TableEntry$'' is a record type describing a P4 table entry,
    as outlined in \Cref{sec:overview-p4,sec:overview-p4runtime}:
    \begin{itemize}
    \item the field ``name'' is a table name, ``action'' is the action to
      perform (using the given ``params'') whenever one of the ``matches'' is
      satisfied;
    \item the type arguments enforce dependencies between fields, as we will see
      in \Cref{sec:type-system};
    \end{itemize}
  \item ``$\PPPPEntity$'' is the union of different record types,
    each one representing a different type of P4 entity that could
    be returned when issuing a ``read'' request to a P4Runtime server.
    In this paper we focus on entities of type $\TableEntry$
    (other types of entities are elided).
\end{itemize}

\subsection{P4Runtime Server Model}
\label{sec:p4runtime-server-model}

We now introduce our abstract model of P4Runtime server
(\Cref{def:p4runtime-server-model}). The intuition is: a server $S$ models the
behaviour of a P4Runtime server with configuration $C$ --- which, in turn,
represents the information available in the server's P4Info metadata file, as
outlined in \Cref{sec:overview-p4,sec:overview-p4runtime}.

\begin{definition}[P4Runtime Server Model]
  \label{def:p4runtime-server-model}
  \label{def:p4runtime-server-config}
  \label{def:p4runtime-server-entity}
  A \emph{P4Runtime server $S$} is a 4-tuple
  $\langle C,E,a,K \rangle$ where:
  \begin{itemize}
    \item the \textbf{configuration $C$} is a mapping that contains the following fields:
      \begin{itemize}
        \item $\mathit{table\_matches}$ (abbreviated $\mathit{tm}$),
          a mapping from table names to
          \textbf{match fields}, which in turn consist of:
          \begin{itemize}
            \item $\mathit{name}$, the name of the network packet field to inspect;
            \item $\mathit{type}$, the type of packet matching algorithm used for this packet field
               (either \ttt{Exact},
                \ttt{Ternary}, \ttt{LPM}, \ttt{Range}, or \ttt{Optional}
                \cite{p4spec}).
          \end{itemize}
        \item $\mathit{table\_actions}$ (abbreviated $\mathit{ta}$),
          mapping table names to sets of allowed action names;
        \item $\mathit{action\_params}$ (abbreviated $\mathit{ap}$),
          mapping action names to sets of
          \textbf{action parameters}:%
          \begin{itemize}
            \item $\mathit{name}$, the name of the parameter;
            \item $\mathit{bitwidth}$, the size of the parameter.
          \end{itemize}
        \end{itemize}
    \item $E$ is a set of \textbf{P4Runtime entities $e$} that can be hosted on
      a P4Runtime server.  The main type of entity is a \textbf{table
      entry},\footnote{P4Runtime models several other types of entries, but they
      are out of scope for this work.} which is a record consisting of:
      \begin{itemize}
        \item $\mathit{table\_name}$, the name of table which owns the entry;
        \item $\mathit{field\_matches}$, a set of records describing packet matching rules:
          \begin{itemize}
              \item $\mathit{name}$, the name of the network packet field to inspect;
              \item a set of additional key-value entries, depending on the type
                of packet matching algorithm used for this field (for example,
                when using the matching algorithm type
                \ttt{Range}, the key-value entries are called $'\ttt{low}'$ and $'\ttt{high}'$);
          \end{itemize}
        \item $\mathit{action\_name}$, the name of the action that the entry
          applies upon network packet match;
        \item $\mathit{action\_args}$, a set of \textbf{action argument} records,
          which in turn contains:
          \begin{itemize}
            \item $\mathit{name}$, the name of the associated parameter;
            \item $\mathit{value}$, the value provided as the argument;
          \end{itemize}
      \end{itemize}
    \item $a$ is the \textbf{address} where the server listens for connections;
    \item $K$ is a set of \textbf{channels}: the active connections between the
      server and its clients.
    \end{itemize}
\end{definition}


As mentioned in the opening of \Cref{sec:modelling-p4runtime-networks}, the
P4Runtime server model formalised in \Cref{def:p4runtime-server-model} focuses
on the configuration of a P4Runtime server, by abstracting from other
implementation details (e.g.~its implementation language).  An example
server configuration can be seen in \Cref{fig:p4runtime-server-config-ex}.

\begin{figure}
  \[
  \small
  \begin{array}{rcl}
    C.\mathit{table\_matches} &=& \curl{\ttt{"IPv4\_table"} \mapsto \curl{\mathit{name} \mapsto \ttt{"IPv4\_dst\_addr"}, \mathit{type} \mapsto \ttt{LPM}},\\
                              & &\;\;\ttt{"IPv6\_table"} \mapsto \curl{\mathit{name} \mapsto \ttt{"IPv6\_dst\_addr"}, \mathit{type} \mapsto \ttt{LPM}}}\\
    C.\mathit{table\_actions} &=& \curl{\ttt{"IPv4\_table"} \mapsto [\ttt{"IPv4\_forward"}, \ttt{"Drop"}],\\
                              & &\;\;\ttt{"IPv6\_table"} \mapsto [\ttt{"IPv6\_forward"}, \ttt{"Drop"}]}\\
    C.\mathit{action\_params} &=& \curl{\ttt{"IPv4\_forward"} \mapsto [\; \curl{\mathit{name} \mapsto \ttt{"mac\_dst"}, \mathit{bitwidth} \mapsto 48},\\
                        & & \phantom{\;\;\ttt{"IPv4\_forward"} \mapsto [\;} \curl{\mathit{name} \mapsto \ttt{"port"}, \mathit{bitwidth} \mapsto 9}\;]\\
                              & &\;\;\ttt{"IPv6\_forward"} \mapsto [\; \curl{\mathit{name} \mapsto \ttt{"mac\_dst"}, \mathit{bitwidth} \mapsto 48},\\
                     & & \phantom{\;\;\ttt{"IPv6\_forward"} \mapsto  [\;}\curl{\mathit{name} \mapsto \ttt{"port"}, \mathit{bitwidth} \mapsto 9}\;]\\
                              & &\;\;\ttt{"Drop"} \mapsto []\;}
  \end{array}
  \]
  \vspace{-3mm}
  \caption{Example of a P4Runtime server configuration $C$ (by
    \Cref{def:p4runtime-server-config}). This JSON-like structure models the
    P4Info metadata describing the configuration of an actual P4 device (as
    outlined in \Cref{sec:overview-p4,sec:overview-p4runtime}).}
  \label{fig:p4runtime-server-config-ex}
\end{figure}

\subsection{P4Runtime Networks}
\label{sec:p4runtime-networks}

We conclude this section by formalising the syntax of a network of P4Runtime
servers and clients.

\begin{definition}[P4Runtime Network]
  \label{def:p4runtime-network-syntax}
  A \emph{P4Runtime network} is a parallel composition of clients (i.e.~terms
  $t$ of the grammar in \Cref{def:p4runtime-client-syntax}) and servers
  (i.e.~tuples $S$ conforming to \Cref{def:p4runtime-server-model}):
  \[
      N \quad \Coloneqq \quad t \;\;\;|\;\;\; S \;\;\;|\;\;\; N\,|\,N
  \]
\end{definition}


%% file: typesystem.tex
\section{The \ourFormalLang Typing System}
\label{sec:type-system}

\Cref{def:typing-env,def:type-system} formalise the typing system of
\ourFormalLang.  The typing system is based on System~\fSub~\cite{Cardelli1994SystemFSub}
extended with singleton types and match types
\cite{BlanvillainBKO22}, plus new typing rules for the P4Runtime-specific
operations we introduced in \Cref{def:p4runtime-client-syntax}.

\begin{definition}[Typing Environment]
  \label{def:typing-env}
  A \emph{typing environment $\Gamma$} is a mapping from term or type variables
  to types, that we syntactically represent as follows:
  \[
    \begin{array}{rcll}
    \Gamma & \Coloneqq & \emptyset & \text{(Empty typing environment)} \\
           & | & \Gamma,\, \types{x}{T} & \text{(Term variable $x$ has type $T$)} \\ 
           & | & \Gamma,\, \subtypes{X}{T} & \text{(Type variable $X$ has upper bound $T$)} \\ 
    \end{array}
  \]
\end{definition}

\begin{definition}[The \ourFormalLang Typing System]
  \label{def:type-system}
  The \emph{\ourFormalLang typing system} consists of the following
  mutually-defined, inductive judgements:
  \[
    \begin{array}{c@{\quad}l@{\hspace{-10mm}}l}
      \envJ{\Gamma} & \text{($\Gamma$ is a valid typing environment)} & \text{(\Cref{fig:typing-env-valid})}\\
      \typeJ{\Gamma}{T} & \text{($T$ is a valid type in $\Gamma$)} & \text{(\Cref{fig:type-valid})}\\
      \disjoint{\Gamma}{T}{T'} & \text{(Types $T$ and $T'$ are disjoint in $\Gamma$)} & \text{(\Cref{def:type-disjoint})}\\
      \subtypesTo{\Gamma}{T}{T'} & \text{($T$ is subtype of $T'$ in $\Gamma$ --- assuming $\typeJ{\Gamma}{T}$ and $\typeJ{\Gamma}{T'}$)} & \text{(\Cref{fig:subtyping-rules})}\\
      \eqtypesTo{\Gamma}{T}{T'} & \text{($T$ and $T'$ are subtyping-equivalent in $\Gamma$, i.e.~$\subtypesTo{\Gamma}{T}{T'}$ and $\subtypesTo{\Gamma}{T'}{T}$)}\\
      \typesTo{\Gamma}{t}{T} & \text{($t$ has type $T$ in $\Gamma$ --- assuming $\typeJ{\Gamma}{T}$)} & \text{(\Cref{fig:typing-rules})}
    \end{array}
  \]
\end{definition}

\begin{figure}
\[
\begin{array}{c}
    \Rule
        {}{\envJ{\emptyset}}
        {Env-$\emptyset$}\qquad
    \Rule
        {\typeJ{\Gamma}{T} \quad x \not\in \dom{\Gamma}}
        {\envJ{\Gamma, \types{x}{T}}}
        {Env-V}\qquad
    \Rule
        {\typeJ{\Gamma}{T} \quad X \not\in \dom{\Gamma}}
        {\envJ{\Gamma, X <: T}}
        {Env-TV}
\end{array}
\]
\vspace{-3mm}
\caption{Typing environment validity rules.}
\label{fig:typing-env-valid}
\end{figure}

\begin{figure}
    \newcommand{\rowSep}{4mm}
    \[
    \small
    \begin{array}{c}
    \hlightM{\Rule
        {\envJ{\Gamma}}
        {\typeJ{\Gamma}{\singleton{v_G}}}
        {Type-Val}}\quad
    \Rule
        {\envJ{\Gamma}}
        {\typeJ{\Gamma}{\tx{Int}}{}}
        {Type-Int}\quad
    \Rule
        {\envJ{\Gamma}}
        {\typeJ{\Gamma}{\tx{Bool}}{}}
        {Type-Bool}\quad
    \Rule
        {\envJ{\Gamma}}
        {\typeJ{\Gamma}{\tx{Unit}}{}}
        {Type-Unit}\\[\rowSep]
    \Rule
        {\envJ{\Gamma}}
        {\typeJ{\Gamma}{\tx{String}}{}}
        {Type-String}\qquad
    \Rule
        {\envJ{\Gamma}}
        {\typeJ{\Gamma}{\tx{Bytes}}{}}
        {Type-Bytes}\qquad
    \Rule
        {\envJ{\Gamma, \subtypes{X}{T}, \Gamma'}}
        {\typeJ{\Gamma, \subtypes{X}{T}, \Gamma'}{X}}
        {Type-TV}\\[\rowSep]
    \hlightM{\Rule
        {\typeJ{\Gamma}{T_1} \quad \typeJ{\Gamma}{T_2} \quad \typeJ{\Gamma}{T_3}}
        {\typeJ{\Gamma}{\ServerRef{T_1,T_2,T_3}}}
        {Type-SR}}\quad
    \hlightM{\Rule
        {\typeJ{\Gamma}{T_1} \quad \typeJ{\Gamma}{T_2} \quad \typeJ{\Gamma}{T_3}}
        {\typeJ{\Gamma}{\Chan{T_1,T_2,T_3}}}
        {Type-Chan}}\\[\rowSep]
    \Rule
        {\typeJ{\Gamma}{T_1} \quad \typeJ{\Gamma}{T_2}}
        {\typeJ{\Gamma}{T_1 \rightarrow T_2}}
        {Type-Abs}\qquad
    \Rule
        {\typeJ{\Gamma, \subtypes{X}{T_1}}{T_2}}
        {\typeJ{\Gamma}{\quant{X}{T_1} T_2}}
        {Type-TAbs}\\[\rowSep]
    \Rule
        {\typeJ{\Gamma}{T_1} \quad \typeJ{\Gamma}{T_2} \\
         \eqtypesTo{\Gamma}{T_1}{\quant{X}{T_{11}} T_{12}} \quad \typeJ{\Gamma}{\quant{X}{T_{11}} T_{12}} \quad \subtypesTo{\Gamma}{T_2}{T_{11}}}
        {\typeJ{\Gamma}{T_1\ T_2}}
        {Type-TApp}\\[\rowSep]
    \Rule
        {\forall i \in I : \typeJ{\Gamma}{T_i}}
        {\typeJ{\Gamma}{\curl{f_i : T_i}_{i \in I}}}
        {Type-Rec}\qquad
    \Rule
        {\typeJ{\Gamma}{T}}
        {\typeJ{\Gamma}{[T]}}
        {Type-List}\qquad
    \Rule
        {\typeJ{\Gamma}{T_1} \quad \typeJ{\Gamma}{T_2}}
        {\typeJ{\Gamma}{\typeUnion{T_1}{T_2}}}
        {Type-Union}\\[\rowSep]
    \hlightM{\Rule
        {\typeJ{\Gamma}{T_s} \quad \forall i \in I : \typeJ{\Gamma}{T_i} \quad \forall i \in I : \typeJ{\Gamma}{T'_i}}
        {\typeJ{\Gamma}{\match{T_s}{T_i \Rightarrow T'_i}{i \in I}}}
        {Type-Match}}
\end{array}
\]
\caption{Type validity rules. Non-standard extensions to \fSub are \hlight{highlighted}.}
\label{fig:type-valid}
\end{figure}

Most type validity rules in \Cref{fig:type-valid} are standard. The exceptions
are \hlight{highlighted}:
\begin{itemize}
  \item by rule \textsc{Type-Val}, any ground value $v_G$ (i.e.~any value that
    does \emph{not} contain lambda or type abstractions, by
    \Cref{def:p4runtime-client-syntax}) has a corresponding singleton type
    $\singleton{v_G}$;
  \item rules \textsc{Type-SR} and \textsc{Type-Chan} say that our new server
    address and client-server channel types are well-formed if all their type
    arguments are well-formed;

  \item \change{change:type-match}{}{%
      by rule \textsc{Type-Match}, a match type is well-formed if the scrutinee
      type ($T_s$), the types being matched against ($T_i$), and the
      continuation types $(T'_i)$ are well-formed.%
  }
\end{itemize}

\begin{figure}
    \newcommand{\rowSep}{4mm}
    \[
    \small
    \begin{array}{@{}c@{}}
    \Rule
        {}
        {\subtypesTo{\Gamma}{T}{T}}
        {ST-Refl}\qquad
    \Rule
        {}
        {\subtypesTo{\Gamma}{T}{\top}}
        {ST-$\top$}\qquad
    \hlightM{\Rule
        {\memberOf{v_G}{T}} 
        {\subtypesTo{\Gamma}{\singleton{v_G}}{T}}
        {ST-Val}}\qquad
    \Rule
        {\envJ{\Gamma, \subtypes{X}{T}, \Gamma'}}
        {\subtypesTo{\Gamma, \subtypes{X}{T}, \Gamma'}{X}{T}}
        {ST-Var}\\[\rowSep]
    \Rule
        {\subtypesTo{\Gamma}{T_1}{T_2}\quad\subtypesTo{\Gamma}{T_2}{T_3}}
        {\subtypesTo{\Gamma}{T_1}{T_3}}
        {ST-Trans}\qquad
    \Rule
        {}{\eqtypesTo{\Gamma}{\typeUnion{T_1}{T_2}}{\typeUnion{T_2}{T_1}}}
        {ST-$\cup$Comm}\\[\rowSep]
    \Rule
        {}{\eqtypesTo{\Gamma}{\typeUnion{(\typeUnion{T_1}{T_2})}{T_3}}{\typeUnion{T_1}{(\typeUnion{T_2}{T_3})}}}
        {ST-$\cup$Assoc}\qquad
    \Rule
        {\subtypesTo{\Gamma}{T_1}{T_2}}
        {\subtypesTo{\Gamma}{[T_1]}{[T_2]}}
        {ST-List}\\[\rowSep]
    \Rule
        {\subtypesTo{\Gamma}{T_1}{T_2}}
        {\subtypesTo{\Gamma}{T_1}{\typeUnion{T_2}{T_3}}}
        {ST-$\cup$R}\quad
    \Rule
        {\forall i \in I:\ \subtypesTo{\Gamma}{T_i}{T_i'}}
        {\subtypesTo{\Gamma}{\curl{f_i : T_i}_{i \in I}}{\curl{f_i : T_i'}_{i \in I}}}
        {ST-Rec}\\[\rowSep]
    \hlightM{\Rule
        {\subtypesTo{\Gamma}{T_s}{T_k} \quad \forall i \in I: i < k \Rightarrow \disjoint{\Gamma}{T_s}{T_i}}
        {\eqtypesTo{\Gamma}{T_s\ttt{ match }\curl{T_i \Rightarrow T_i'}_{i \in I}}{T_k'}}
        {ST-Match1}}\\[\rowSep]
    \hlightM{\Rule
        {\subtypesTo{\Gamma}{T_s}{T_s'} \quad \forall i \in I: \subtypesTo{\Gamma}{T_i'}{T_i''}}
        {\subtypesTo{\Gamma}{T_s\ttt{ match }\curl{T_i \Rightarrow T_i'}_{i \in I}}{T_s'\ttt{ match }\curl{T_i \Rightarrow T_i''}_{i \in I}}}
        {ST-Match2}}\\[\rowSep]
    \Rule
        {\subtypesTo{\Gamma}{T_1'}{T_1}\quad\subtypesTo{\Gamma}{T_2}{T_2'}}
        {\subtypesTo{\Gamma}{T_1 \rightarrow T_2}{T_1' \rightarrow T_2'}}
        {ST-Abs}\qquad
    \Rule
        {\subtypesTo{\Gamma}{T_1'}{T_1}\quad\subtypesTo{\Gamma, X <: T_1'}{T_2}{T_2'}}
        {\subtypesTo{\Gamma}{(\forall X <: T_1.\ T_2)}{(\forall X <: T_1'.\ T_2')}}
        {ST-TAbs}\\[\rowSep]
    \Rule
        {\subtypesTo{\Gamma}{T_2}{T_1}}
        {\eqtypesTo{\Gamma}{(\quant{X}{T_1} T_0)\ T_2}{\subs{T_0}{X}{T_2}}}
        {ST-App}
    \quad
    \begin{minipage}{0.5\linewidth}
        \small
        \(
          \text{where:}
          \begin{array}{@{}r@{\;}c@{\;}l@{}}
              \subs{\emptyset}{X}{T}                   &=& \emptyset\\
              \subs{(\Gamma, \types{x}{T'})}{X}{T}     &=& (\subs{\Gamma}{X}{T}), \types{x}{\subs{T'}{X}{T}}\\
              \subs{(\Gamma, \subtypes{X'}{T'})}{X}{T} &=& (\subs{\Gamma}{X}{T}), \subtypes{X'}{\subs{T'}{X}{T}}
          \end{array}
        \)
    \end{minipage}
    \end{array}
\]
\caption{Subtyping rules.}
\label{fig:subtyping-rules}
\end{figure}

The subtyping rules in \Cref{fig:subtyping-rules} are also standard, with the
following \hlight{highlighted} exceptions:

\begin{itemize}
  \item by rule \textsc{ST-Val}, if a ground value $v_G$ belongs to type $T$
    (by the relation ``$\memberOf{v_G}{T}$'' defined in \Cref{sec:memberof-relation}),
    then the singleton type $\singleton{v_G}$ is subtype of $T$.  For example:
    since $42 \in \tx{Int}$, then we have $\subtypesTo{\Gamma}{\singleton{42}}{\tx{Int}}$
    (assuming $\typeJ{\Gamma}{\tx{Int}}$);
  \item rule \textsc{ST-Match1} (adapted from \cite{BlanvillainBKO22}) says that
    a match type is subtyping-equivalent to the continuation type $T'_k$ if all
    match type cases before $k$ (i.e.~all $T_i$ with $i < k$) are
    \emph{disjoint} form $T_k$, according to \Cref{def:type-disjoint} below;
  \item rule \textsc{ST-Match2} (also adapted from \cite{BlanvillainBKO22}) says
    that match types are covariant in both the type being matched, and the
    continuation types.
\end{itemize}

The type disjointness judgement\; $\disjoint{\Gamma}{T_1}{T_2}$ \;(used in rule
\textsc{ST-Match1}) is formalised in \Cref{def:type-disjoint} below: the
intuition is that two types $T_1$ and $T_2$ are disjoint when they have no
common subtypes, hence there exists no value that can have both types $T_1$ and
$T_2$.

\begin{definition}[Disjointness of Types]
  \label{def:type-disjoint}
  Two types \emph{$T_1$ and $T_2$ are disjoint in $\Gamma$}, written
  $\disjoint{\Gamma}{T_1}{T_2}$, iff:
  \begin{enumerate}
    \item $\typeJ{\Gamma}{T_1}$ \;and\; $\typeJ{\Gamma}{T_2}$; \;and
    \item \(
      \not\exists T_3 :\; \typeJ{\Gamma}{T_3} \;\text{ and }\; \subtypesTo{\Gamma}{T_3}{T_1} \;\text{ and }\; \subtypesTo{\Gamma}{T_3}{T_2}
    \)
  \end{enumerate}
\end{definition}

\begin{example}[Subtyping and Disjointness in Match Types]
  \label{eg:match-subtyping-disjointness}
  Consider the following match type:

  \smallskip\centerline{\(
    \quantT{X} \match{X}{\tx{Int} \Rightarrow \singleton{42}, \;\tx{Bool} \Rightarrow \singleton{\stringTerm{Hello}}}{}
  \)}\smallskip

  \noindent%
  The type\;
  $(\quantT{X} \match{X}{\tx{Int} \Rightarrow \singleton{42}, \tx{Bool} \Rightarrow \singleton{\stringTerm{Hello}}}{})\ \singleton{\ttt{true}}$%
  \;is subtyping-equivalent (i.e.~``reduces'') to
  $\singleton{\stringTerm{Hello}}$, by the subtyping rule \textsc{ST-Match1}
  in \Cref{fig:subtyping-rules}.  The rule
  first checks whether $\singleton{\ttt{true}}$ is a subtype of $\tx{Int}$,
  which it is not.  Since it is also disjoint from $\tx{Int}$
  (the two types do not share a common subtype),
  the rule then proceeds to the next case. Here,
  $\singleton{\ttt{true}}$ is a subtype of $\tx{Bool}$, and so
  the type ``reduces'' to the case $\singleton{\stringTerm{Hello}}$.
\end{example}

\begin{figure}
    \[
    \small
    \newcommand{\rowSep}{4mm}
    \begin{array}{@{}c@{}}
        \hlightM{\Rule
            {}
            {\typesTo{\Gamma}{v_G}{\singleton{v_G}}}
            {T-Val}}\qquad
        \Rule
            {\forall i \in I: \typesTo{\Gamma}{t_i}{T_i}}
            {\typesTo{\Gamma}{\curl{f_i=t_i}_{i \in I}}{\curl{f_i : T_i}_{i \in I}}}
            {T-Rec}\qquad
        \Rule
            {\typesTo{\Gamma}{t_0}{T_0} \quad \typesTo{\Gamma}{t_1}{\listT{T_0}}}
            {\typesTo{\Gamma}{\concat{t_0}{t_1}}{\listT{T_0}}}
            {T-List}\\[\rowSep]
        \Rule
            {\typesTo{\Gamma}{t}{\listT{T}}}
            {\typesTo{\Gamma}{\listHead t}{\tx{Option}\ T}}
            {T-Head}\qquad
        \Rule
            {\typesTo{\Gamma}{t}{\listT{T}}}
            {\typesTo{\Gamma}{\listTail t}{\tx{Option}\ \listT{T}}}
            {T-Tail}\qquad
        \Rule
            {\envJ{\Gamma, \types{x}{T}, \Gamma'}}
            {\typesTo{\Gamma, \types{x}{T}, \Gamma'}{x}{T}}
            {T-Var}\\[\rowSep]
        \Rule
            {\typesTo{\Gamma}{t_0}{T_0} \quad \typesTo{\Gamma, \types{x}{T_0}}{t_1}{T_1}}
            {\typesTo{\Gamma}{\letExp{x}{t_0} t_1}{T_1}}
            {T-Let}\quad
        \Rule
            {\typesTo{\Gamma,\types{x}{T_1}}{t}{T_2}}
            {\typesTo{\Gamma}{\abs{x}{T_1} t}{T_1 \rightarrow T_2}}
            {T-Abs}\quad
        \Rule
            {\typesTo{\Gamma}{t_1}{T_1 \rightarrow T_2} \quad \typesTo{\Gamma}{t_2}{T_1}}
            {\typesTo{\Gamma}{t_1\ t_2}{T_2}}
            {T-App}\\[\rowSep]
        \Rule
            {\typesTo{\Gamma,\subtypes{X}{T_1}}{t}{T_2}}
            {\typesTo{\Gamma}{(\tabs{X}{T_1} t)}{(\quant{X}{T_1} T_2)}}
            {T-TAbs}\qquad
        \Rule
            {\typesTo{\Gamma}{t_1}{\quant{X}{T_1} T_0} \quad \subtypesTo{\Gamma}{T_2}{T_1}}
            {\typesTo{\Gamma}{t_1\ T_2}{(\quant{X}{T_1} T_0)\ T_2}}
            {T-TApp}\\[\rowSep]
        \Rule
            {\typesTo{\Gamma}{t_r}{\curl{\types{f_i}{T_i}}_{i \in I}} \quad k \in I}
            {\typesTo{\Gamma}{t_r.f_k}{T_k}}
            {T-Field}\qquad
        \Rule
            {\typesTo{\Gamma}{t}{T} \quad \subtypesTo{\Gamma}{T}{T'}}
            {\typesTo{\Gamma}{t}{T'}}
            {T-Sub}\\[\rowSep]
        \hlightM{\Rule
            {\typesTo{\Gamma}{t_s}{T_s} \quad \forall i \in I: \typesTo{\Gamma, \types{x_i}{T_i}}{t_i}{T_i'} \quad \subtypesTo{\Gamma}{T_s}{\cup_{i \in I} T_i}}
            {\typesTo{\Gamma}{t_s \ttt{ match }\curl{x_i : T_i \Rightarrow t_i}_{i \in I}}{T_s \ttt{ match } \curl{T_i \Rightarrow T_i'}_{i \in I}}}
            {T-Match}}\\[\rowSep]
        \hlightM{\Rule
            {\typesTo{\Gamma}{t}{\ServerRef{T_m,T_a,T_p}}}
            {\typesTo{\Gamma}{\ttt{Connect}(t)}{\Chan{T_m,T_a,T_p}}}
            {T-OpC}}\\[\rowSep]
        \hlightM{\Rule
            {\typesTo{\Gamma}{t_c}{\Chan{T_m,T_a,T_p}} \quad \typesTo{\Gamma}{t_e}{\PPPPEntity\ T_m\ T_a\ T_p\ X_n\ X_a}}
            {\typesTo{\Gamma}{\ttt{Read}(t_c,t_e)}{[\PPPPEntity\ T_m\ T_a\ T_p\ X_n\ X_a]}}
            {T-OpR}}\\[\rowSep]
        \hlightM{\Rule
            {\typesTo{\Gamma}{t_c}{\Chan{T_m,T_a,T_p}} \quad \typesTo{\Gamma}{t_e}{\PPPPEntity\ T_m\ T_a\ T_p\ X_n\ X_a}}
            {\typesTo{\Gamma}{\ttt{Insert}(t_c,t_e)}{\tx{Bool}}}
            {T-OpI}}\\[\rowSep]
        \hlightM{\Rule
            {\typesTo{\Gamma}{t_c}{\Chan{T_m,T_a,T_p}} \quad \typesTo{\Gamma}{t_e}{\PPPPEntity\ T_m\ T_a\ T_p\ X_n\ X_a}}
            {\typesTo{\Gamma}{\ttt{Modify}(t_c,t_e)}{\tx{Bool}}}
            {T-OpM}}\\[\rowSep]
        \hlightM{\Rule
            {\typesTo{\Gamma}{t_c}{\Chan{T_m,T_a,T_p}} \quad \typesTo{\Gamma}{t_e}{\PPPPEntity\ T_m\ T_a\ T_p\ X_n\ X_a}}
            {\typesTo{\Gamma}{\ttt{Delete}(t_c,t_e)}{\tx{Bool}}}
            {T-OpD}}
    \end{array}
    \]
\caption{Typing rules for \ourFormalLang terms. Non-standard extensions to \fSub are \hlight{highlighted}.}
\label{fig:typing-rules}
\end{figure}

Finally, \Cref{fig:typing-rules} includes several (\hlight{highlighted})
non-standard typing rules for \ourFormalLang terms:

\begin{itemize}
  \item by rule \textsc{T-Val}, a ground value $v_G$ is typed by the singleton
    type $\singleton{v_G}$.  E.g.~$42$ has type $\singleton{42}$, hence (via the
    subsumption rule \textsc{T-Sub} and \textsc{ST-Val} in
    \Cref{fig:subtyping-rules}) we also have that $42$ has type $\tx{Int}$;
  \item by rule \textsc{T-Match} (adapted from \cite{BlanvillainBKO22}), a
    pattern matching term is typed with a match type of a similar shape. The
    clause ``$\subtypesTo{\Gamma}{T_s}{\cup_{i \in I} T_i}$'' ensures that
    pattern matching is exhaustive;
  \item by rule \textsc{T-OpC}, a connection operation to a server address $t$
    of type $\ServerRef{T_m,T_a,T_p}$ has channel type $\Chan{T_m,T_a,T_p}$ ---
    i.e. the type of the channel returned by the connection maintains
    type-level information about the server configuration;
  \item by rule \textsc{T-OpR}, the query operation $\ttt{Read}(t_c,t_e)$ is typed
    as follows:
    \begin{enumerate}
      \item the query argument $t_e$ has type $\PPPPEntity$
        (\Cref{fig:language-syntax-sugar-types}) applied to type parameters that
        match those of the type of $t_c$ (expected to be a channel).
        Intuitively, this means that $t_e$ can only be
        a P4 entity supported by the P4Runtime server connected over $t_c$; and
      \item the read operation returns a list of type $\PPPPEntity$ applied to
        type arguments that match those of the type of $t_c$.  Intuitively,
        this means that the returned list is expected to only contain entities supported by the P4Runtime server
        connected via $t_c$;
    \end{enumerate}
  \item rules \textsc{T-OpI}, \textsc{T-OpM}, and \textsc{T-OpD} have type
    constraints similar to \textsc{T-OpR} above: their argument $t_e$ must be a
    P4 entity supported by the server connected over channel $t_c$.  All
    these operations return a boolean value (indicating whether the operation had
    an effect).
\end{itemize}

\begin{example}[Typable and Untypable \ourFormalLang Operations]
  \label{eg:typable-untypable-p4runtime-ops}
  \change{change:example-self-contained}{}{%
  Consider the following types:\footnote{%
    \changeNoMargin{%
      You may notice a similarity between the types used in
      \Cref{eg:typable-untypable-p4runtime-ops} and the P4Runtime server
      configuration in \Cref{fig:p4runtime-server-config-ex}:
      indeed, those types capture the constraints of that server configuration.
      We will reprise the topic in
      \Cref{sec:p4runtime-server-model-semantics}.%
      }
    }
  \[
  \begin{array}{@{}r@{\;\;}c@{\;\;}l@{}}
    T_m &=& \quantT{T} T \ttt{ match }\{\\
    & &\qquad \singleton{\stringTerm{IPv4\_table}} \Rightarrow \curl{\mathit{name} : \singleton{\stringTerm{IPv4\_dst\_addr}}, \mathit{value} : \tx{Bytes}, \mathit{prefixLen} : \tx{Bytes}}\\
    & &\qquad \singleton{\stringTerm{IPv6\_table}} \Rightarrow \curl{\mathit{name} : \singleton{\stringTerm{IPv6\_dst\_addr}}, \mathit{value} : \tx{Bytes}, \mathit{prefixLen} : \tx{Bytes}} \;\}\\
    T_a &=& \quantT{T} T \ttt{ match }\{\\
    & &\qquad \singleton{\stringTerm{IPv4\_table}} \Rightarrow \typeUnion{\singleton{\stringTerm{IPv4\_forward}}}{\singleton{\stringTerm{Drop}}}\\
    & &\qquad \singleton{\stringTerm{IPv6\_table}} \Rightarrow \typeUnion{\singleton{\stringTerm{IPv6\_forward}}}{\singleton{\stringTerm{Drop}}} \;\}\\
    T_p &=& \quantT{A} A \ttt{ match }\{\\
    & &\qquad \singleton{\stringTerm{IPv4\_forward}} \Rightarrow \curl{\mathit{mac\_dst} : \tx{Bytes}, \mathit{port} : \tx{Bytes}}\\
    & &\qquad \singleton{\stringTerm{IPv6\_forward}} \Rightarrow \curl{\mathit{mac\_dst} : \tx{Bytes}, \mathit{port} : \tx{Bytes}}\\
    & &\qquad \singleton{\stringTerm{Drop}} \Rightarrow \tx{Unit}\;\}\\
  \end{array}
  \]
  Assume that the connection $s$ has type $\Chan{T_m,T_a,T_p}$.
  }%
  Then, we can type the following P4Runtime operation:
  \begin{align*}
    \small%
    \ttt{Insert}(s, \curl{ & \stringTerm{IPv4\_table},\;\; \curl{\mathit{name} = \stringTerm{IPv4\_dst\_addr}, \mathit{value} = \ttt{b}(10,1,0,0), \mathit{prefixLen} = 32},\\
    &\stringTerm{IPv4\_forward},\;\; \curl{\mathit{mac\_dst} = \ttt{b}(8,0,0,0,10,1), \mathit{port} = \ttt{b}(1)} \;})
  \end{align*}
  \changeNoMargin{%
    The typing succeeds because,
    by rule \textsc{T-OpI} in \Cref{fig:typing-rules},
    the types $T_m,T_a,T_p$ used to instantiate
    the type of $s$ (i.e.~$\Chan{T_m,T_a,T_p}$) can correctly type the
    record passed as second argument of \ttt{Insert}.  The type of such a record
    is\; $\PPPPEntity\ T_m\ T_a\ T_p\ \singleton{\stringTerm{IPv4\_table}}\ \singleton{\stringTerm{IPv4\_forward}}$\; (see \Cref{fig:language-syntax-sugar-types}),
    and its type-checking succeeds because:
    \begin{enumerate}
    \item $T_m\ \singleton{\stringTerm{IPv4\_table}}$ \;reduces to a record type
      $\curl{\mathit{name}: \singleton{\stringTerm{IPv4\_dst\_addr}},\; \ldots}$
      that types the record instance
      $\curl{\mathit{name} = \stringTerm{IPv4\_dst\_addr},\; \ldots}$
      (i.e.~the table match fields);
    \item $T_a\ \singleton{\stringTerm{IPv4\_table}}$ \;reduces to the union type
      $\typeUnion{\singleton{\stringTerm{IPv4\_forward}}\,}{\,\singleton{\stringTerm{Drop}}}$,
      which types the string $\stringTerm{IPv4\_forward}$
      (i.e.~the name of the action applied upon a packet match);
    \item $T_p\ \singleton{\stringTerm{IPv4\_forward}}$ \;reduces to a record type
      $\curl{\mathit{mac\_dst} : \tx{Bytes},\; \ldots}$
      that types the record instance $\curl{\mathit{mac\_dst} = \ldots}$
      (i.e.~the action parameters).%
    \end{enumerate}
    }

    Instead, the following P4Runtime operation cannot be typed:
  \begin{align*}
    \small%
    \ttt{Insert}(s, \curl{ & \stringTerm{IPv4\_table},\;\; \curl{\mathit{name} = \stringTerm{IPv4\_dst\_addr}, \mathit{value} = \ttt{b}(10,1,0,0), \mathit{prefixLen} = 32},\\
    &\stringTerm{IPv6\_forward},\;\; \curl{\mathit{mac\_dst} = \ttt{b}(8,0,0,0,10,1), \mathit{port} = \ttt{b}(1)} })
  \end{align*}
  The typing fails because $\text{\changeNoMargin{$T_a$}}\ \singleton{\stringTerm{IPv4\_table}}$
  reduces to $\typeUnion{\singleton{\stringTerm{IPv4\_forward}}}{\singleton{\stringTerm{Drop}}}$,
  which cannot type the string $\stringTerm{IPv6\_forward}$.
  \changeNoMargin{Intuitively this means that, according to type $T_a$,
  $\stringTerm{IPv6\_forward}$
  is an invalid action for the table $\stringTerm{IPv4\_forward}$.}%
\end{example}

\change{change:blanvillain22-diff}{}{%
\begin{remark}[Differences with \citeN{BlanvillainBKO22}]
\label{remark:differences-blanvillain22}

Our formulation of match types differs from the original presentation by
\citeN{BlanvillainBKO22} in 3 significant aspects: these differences are
non-trivial and interplay with each other in subtle ways, making our
formalisation and proofs quite challenging.

\begin{enumerate}
\item \citeN{BlanvillainBKO22} use a \emph{nominal} type system which models
  class hierarchies, abstracting from class fields and data.  Instead, we need
  data in order to represent P4Runtime tables in \ourFormalLang and in our
  results; moreover, our implementation (\Cref{sec:implementation}) does not
  make significant use of class hierarchies.  Therefore, unlike
  \citeN{BlanvillainBKO22}, we adopt standard data types (records, lists\ldots)
  with \emph{structural} typing and subtyping, and we support singleton types
  --- and consequently, we adapt the match-typing-related rules accordingly.

\item Unlike \citeN{BlanvillainBKO22}, our match types do \emph{not} include a
  mandatory default case.  With the default case, a match type can be
  ``reduced'' (i.e. proven subtype-equivalent) to the type in its default case,
  if the scrutinee type does not match any other case.  We removed the mandatory
  default case because it is not needed (and is actually undesirable) for our
  modelling of P4Runtime table types.  Moreover, the Scala 3 compiler does
  \emph{not} require programmers to specify a default case in their match types
  --- and since our API \ourDSL leverages this feature, we formalised the
  typing system of \ourFormalLang accordingly.  A default match type case can be
  obtained (when needed) by adding a branch that matches the top type $\top$.

\item Correspondingly, our match terms do \emph{not} include a mandatory
  default case (unlike \citeN{BlanvillainBKO22}).  Consequently, our typing rule
  \textsc{T-Match} (\Cref{fig:typing-rules}) has an additional constraint
  w.r.t.~\citeN{BlanvillainBKO22}: the scrutinee type must be a subtype of the
  union of all case types, thus ensuring that the pattern matching is exhaustive
  (the Scala 3 compiler performs similar checks). Notably, match term
  exhaustiveness is needed to prove progress (\Cref{theorem:progress}); instead,
  \citeN{BlanvillainBKO22} do not check match term exhaustiveness because their
  default match case ensures that a match term can always be reduced.
\end{enumerate}
\end{remark}
}%


%% file: semantics.tex
\section{Semantics of \ourFormalLang Programs and P4Runtime Networks}
\label{sec:semantics}

In this section we formalise the semantics of 
\ourFormalLang programs (\Cref{sec:p4runtime-client-semantics}), P4Runtime servers
(\Cref{sec:p4runtime-server-model-semantics}), and networks of clients and servers
(\Cref{sec:p4runtime-network-semantics}).

\subsection{Semantics of \ourFormalLang Programs}
\label{sec:p4runtime-client-semantics}

We introduce the labelled transition system (LTS) semantics of
\ourFormalLang. \Cref{def:p4runtime-client-semantics} below formalises an
\emph{early} semantics, where each transition label denotes either an internal
computation ($\tau$), or a possible input/output interaction with the
surrounding environment.  This style of \emph{early} LTS semantics is inspired
by the $\pi$-calculus \cite{Sangiorgi01}, and
\change{change:why-early-semantics}{}{%
allows us to formalise and reason about the interactions between \ourFormalLang
programs and P4Runtime servers (formalised later in
\Cref{def:p4runtime-network-semantics}) while keeping the respective syntax and
semantics decoupled.%
}%

\begin{definition}[Semantics of \ourFormalLang]
  \label{def:p4runtime-client-semantics}
  Assume a predicate ``$v \in T$'' which holds iff value $v$ belongs to
  type $T$.  We define the \emph{labelled
  transition system (LTS) semantics} of \ourFormalLang as a
  transition relation $t \reduce{\alpha} t'$, where the label $\alpha$ is defined as:
  \[
  \begin{array}{r@{\;\;}r@{\;\;}c@{\;\;}l@{\quad}l}
    \text{Transition label} & \alpha & \Coloneqq & \tau & \text{(Internal transition)}\\
    & & | & \reduceTwoLabel{\tx{connect}(a)}{s} & \text{(Connect to server address $a$, getting channel $s$)}\\
    & & | & \reduceTwoLabel{\tx{read}(s,v)}{v'} & \text{(Perform query $v$ on channel $s$, getting result $v'$)}\\
    & & | & \reduceTwoLabel{\tx{insert}(s,v)}{v'} & \text{(Insert $v$ on channel $s$, getting result $v'$)}\\
    & & | & \reduceTwoLabel{\tx{modify}(s,v)}{v'} & \text{(Modify $v$ on channel $s$, getting result $v'$)}\\
    & & | & \reduceTwoLabel{\tx{delete}(s,v)}{v'} & \text{(Delete $v$ on channel $s$, getting result $v'$)}\\
  \end{array}
  \]
  
  The transition relation $t \reduce{\alpha} t'$ is defined in
  \Cref{fig:p4runtime-client-semantics}, where the context transition rule \textsc{E-$\mbC$} uses
  an \emph{evaluation context $\mbC$} (defined below) which represents a
  \ourFormalLang term 
  with one hole $\hole$:

  \smallskip
  \centerline{$
  \begin{array}{rcl}
    \mbC & \Coloneqq &\hole \;\;|\;\; \concat{\mbC}{t} \;\;|\;\; \concat{v}{\mbC} \;\;|\;\; \listHead{\mbC} \;\;|\;\;\listTail{\mbC} \;\;|\;\; \ttt{let }x\ttt{ = }\mbC\ttt{ in }t\\
    & | & \mbC\;t \;\;|\;\; v\;\mbC \;\;|\;\; \mbC\;T \;\;|\;\; \mbC.f \;\;|\;\; \mbC \ttt{ match }\{\types{x_i}{T_i} \Rightarrow t_i\}_{i \in I}\\
    & | & \curl{f_i = \gamma_i}_{i \in I} \quad \tx{where\; } \exists k \in I: \forall i \in I:\begin{cases}i < k \text{ \;implies\; } \gamma_i = v_i\\i = k \text{ \;implies\; } \gamma_i = \mbC\\i > k \text{ \;implies\; } \gamma_i = t_i\end{cases}
  \end{array}
  $}
\end{definition}

\begin{figure}
  \newcommand{\rowSep}{1.2mm}
  \begin{gather*}
      \Rule{}{(\abs{x}{T} t)\ v \reduce{\tau} \subs{t}{x}{v}}{E-App}\qquad
      \Rule{}{(\tabs{X}{T_1} t)\ T_2 \reduce{\tau} \subs{t}{X}{T_2}}{E-TApp}\\[\rowSep]
      \Rule{}{\letExp{x}{v} t \reduce{\tau} \subs{t}{x}{v}}{E-Let}\qquad
      \Rule{v = \curl{f_i=v_i}_{i \in I} \quad k \in I}{v.f_k \reduce{\tau} v_k}{E-Field}\\[\rowSep]
      \Rule{}{\listHead \concat{v_0}{v_1} \reduce{\tau} \curl{\tx{some} = v_0}}{E-Head1}\qquad
      \Rule{}{\listHead \nil \reduce{\tau} \curl{\tx{none} = \unit}}{E-Head2}\\[\rowSep]
      \Rule{}{\listTail \concat{v_0}{v_1} \reduce{\tau} \curl{\tx{some} = v_1}}{E-Tail1}\qquad
      \Rule{}{\listTail \nil \reduce{\tau} \curl{\tx{none} = \unit}}{E-Tail2}\\[\rowSep]
    \Rule
          {k \in I \quad
           v \in T_k \quad
           \forall j \in I: j < k \implies
           v \not\in T_j
          }
          {v \ttt{ match }\curl{\types{x_i}{T_i} \Rightarrow t_i}_{i \in I} \;\reduce{\tau}\; \subs{t_k}{x_k}{v}}
          {E-Match}\\[\rowSep]
    \hlightM{\Rule
        {a \in \ServerRef{T_m, T_a, T_p} \quad s \in \Chan{T_m, T_a, T_p}}
        {\ttt{Connect}(a)  \reduceTwo{\tx{connect}(a)}{s} s}
        {E-Connect}}\\[\rowSep]
    \hlightM{\Rule
        {s \in \Chan{T_m,T_a,T_p} \quad v' \in \listT{\PPPPEntity\ T_m\ T_a\ T_p\ X_n\ X_a}}
        {\ttt{Read}(s,v)   \reduceTwo{\tx{read}(s,v)}{v'} {v'}}
        {E-Read}}\\[\rowSep]
    \hlightM{\Rule
        {v' \in \tx{Bool}}
        {\ttt{Insert}(s,v) \reduceTwo{\tx{insert}(s,v)}{v'} {v'}}
        {E-Insert}}
    \qquad
    \hlightM{\Rule
        {v' \in \tx{Bool}}
        {\ttt{Modify}(s,v) \reduceTwo{\tx{modify}(s,v)}{v'} {v'}}
        {E-Modify}}\\[\rowSep]
    \hlightM{\Rule
        {v' \in \tx{Bool}}
        {\ttt{Delete}(s,v) \reduceTwo{\tx{delete}(s,v)}{v'} {v'}}
        {E-Delete}}
    \qquad
    \Rule
        {t \reduce{\alpha} t'}
        {\mbC[t] \reduce{\alpha} \mbC[t']}
        {E-$\mbC$}
  \end{gather*}
  \caption{LTS semantics of \ourFormalLang terms.  Non-standard extensions to \fSub are \hlight{highlighted}.}
  \label{fig:p4runtime-client-semantics}
\end{figure}

Most rules in \Cref{def:p4runtime-client-semantics} are standard, except for the
ones \hlight{highlighted} in \Cref{fig:p4runtime-client-semantics}:
\begin{itemize}
  \item by rule \textsc{E-Connect}, the term $\ttt{Connect}(a)$ transitions by
    producing a channel $s$, whose type conforms to the type of the server
    address $a$.  The transition label ``$\reduceTwoLabel{\tx{connect}(a)}{s}$''
    means that the term is trying to interact with the surrounding environment:
    hence, as we will see in \Cref{sec:p4runtime-server-model-semantics}, the
    client expects a P4Runtime server to emit the dual label
    ``$\dual{\reduceTwoLabel{\tx{connect}(a)}{s}}$'' --- meaning that the server
    is listening on address $a$ and can produce channel $s$;
  \item by rule \textsc{E-Read}, the term $\ttt{Read}(s,v)$ transitions by
    producing a value $v'$, which is a list of $\PPPPEntity$ instances
    (\Cref{fig:language-syntax-sugar-types}) whose type conforms to the type of
    channel $s$.  The transition label means that the term expects to interact
    with a P4Runtime server on channel $s$;
  \item rules \textsc{E-Insert}, \textsc{E-Modify}, and \textsc{E-Delete} work
    similarly, and produce a boolean value describing whether the operation
    had an effect or not.
\end{itemize}

\change{change:explain-record-reduction}{}{%
\ourFormalLang terms are evaluated from left to right, using the evaluation
contexts $\mbC$ in \Cref{def:p4runtime-client-semantics}.  For instance, the
last case ``$\curl{f_i = \gamma_i}_{i \in I}$'' represents a record whose fields
$f_i$ are indexed by $i \in I$, where $I$ is a set of consecutive natural
numbers $1..n$ (as per \Cref{def:p4runtime-client-syntax}); all fields to the
left of $f_k$ (for some $k \in I$) are already fully-evaluated into values
$v_i$; the field $f_k$ is a context with a hole, which is going to be evaluated
next; and all fields to the right of $f_k$ are arbitrary terms $t_i$, which may
be evaluated after $f_k$.%
}

\subsection{Semantics of P4Runtime Servers}
\label{sec:p4runtime-server-model-semantics}

To define our P4Runtime server semantics (in
\Cref{def:p4runtime-server-semantics} later on), we need to ensure that a server
$S$ will only answer to well-typed requests from its clients, and that the
server entities are well-typed w.r.t.~the server configuration $C$.  To this
end, we formalise an encoding of a server configuration $C$ into \ourFormalLang
types (\Cref{def:server-config-encoding} below).  Intuitively, this describes
how to turn the P4Info metadata of a P4 device into a set of types describing
the device tables, actions, etc.

\begin{definition}[Encoding of a Server Configuration into \ourFormalLang Types]
  \label{def:server-config-encoding}
  Given a P4Runtime server configuration $C$, we define the \emph{encoding}
  $\encoded{\cdots}$ of its entries into \ourFormalLang types in \Cref{fig:p4runtime-server-config-enc}.
\end{definition}

\begin{figure}
  \(%
  \begin{array}{@{}rcl@{}}
  \encoded{C.\mathit{table\_matches}} &=& \quantT{T} T \ttt{ match } \{\\
                                      & & \qquad \singleton{t} \Rightarrow \encoded{C.\mathit{table\_matches}(t)},\\
                                      & & \qquad \singleton{\wildcard} \Rightarrow \singleton{\wildcard}\\
                                      & & \}_{\mathit{t} \in dom(C.\mathit{table\_matches})}\\
\encoded{C.\mathit{table\_matches}(t)} &=& \typeUnion{\curl{\mathit{mf}.name : \encoded{\mathit{mf}.type}}_{\mathit{mf} \in C.\mathit{table\_matches}(t)}}{\singleton{\wildcard}}\\
                \encoded{\ttt{Exact}} &=& \curl{\mathit{value} : \tx{Bytes}}\\
              \encoded{\ttt{Ternary}} &=& \tx{Option}\ \curl{\mathit{value} : \tx{Bytes}, \mathit{mask} : \tx{Bytes}}\\
                  \encoded{\ttt{LPM}} &=& \tx{Option}\ \curl{\mathit{value} : \tx{Bytes}, \mathit{prefixLen} : \tx{Int}}\\
                \encoded{\ttt{Range}} &=& \tx{Option}\ \curl{\mathit{low} : \tx{Bytes}, \mathit{high} : \tx{Bytes}}\\
             \encoded{\ttt{Optional}} &=& \tx{Option}\ \curl{\mathit{value} : \tx{Bytes}}\\
  \encoded{C.\mathit{table\_actions}} &=& \quantT{T} T \ttt{ match } \{\\
                                      & & \qquad \singleton{t} \Rightarrow \encoded{C.\mathit{table\_actions}(t)},\\
                                      & & \qquad \singleton{\wildcard} \Rightarrow \singleton{\wildcard}\\
                                      & & \}_{\mathit{t} \in dom(C.\mathit{table\_actions})}\\
\encoded{C.\mathit{table\_actions}(t)} &=& \bigcup_{a \in C.\mathit{table\_actions}(t)} \singleton{a}\\
  \encoded{C.\mathit{action\_params}} &=& \quantT{A} A \ttt{ match } \{\\
                                      & & \qquad \singleton{a} \Rightarrow \encoded{C.\mathit{action\_params}(a)},\\
                                      & & \qquad \singleton{\wildcard} \Rightarrow \tx{Unit}\\
                                      & & \}_{\mathit{a} \in dom(C.\mathit{action\_params})}\\
\encoded{C.\mathit{action\_params}(a)} &=& \begin{cases} \tx{Unit} &\tx{if } C.\mathit{action\_params}(a) = \emptyset \\\curl{p.\mathit{name} : \tx{Bytes}}_{p \in C.\mathit{action\_params}(a)} & \tx{otherwise}\end{cases}
  \end{array}
  \)
\caption{Definition of the encoding operation $\encoded{\cdots}$ from P4Runtime configurations to \ourFormalLang types.}
\label{fig:p4runtime-server-config-enc}
\end{figure}

\begin{example}[Server Configuration Representation as \ourFormalLang Types]
    \label{eg:server-config-types-encoded}
    Consider the P4Runtime server configuration in
    \Cref{fig:p4runtime-server-config-ex}: by \Cref{def:server-config-encoding},
    its encoding into the \ourFormalLang types is shown in
    \Cref{fig:p4runtime-server-config-ex-enc}.
    \change{change:fig12-refer-example}{}{%
      (The same types are also used in
      \Cref{eg:typable-untypable-p4runtime-ops}, where they are called
      $T_m,T_a,T_p$.)%
    }%
\end{example}

\begin{figure}
  \[\small
  \begin{array}{@{}r@{\;}c@{\;}l@{}}
    \encoded{C.\mathit{table\_matches}} &=& \quantT{T} T \ttt{ match }\{\\
    & &\qquad \singleton{\ttt{"IPv4\_table"}} \Rightarrow \curl{\mathit{name} : \singleton{\ttt{"IPv4\_dst\_addr"}}, \mathit{value} : \tx{Bytes}, \mathit{prefixLen} : \tx{Bytes}}\\
    & &\qquad \singleton{\ttt{"IPv6\_table"}} \Rightarrow \curl{\mathit{name} : \singleton{\ttt{"IPv6\_dst\_addr"}}, \mathit{value} : \tx{Bytes}, \mathit{prefixLen} : \tx{Bytes}} \;\}\\
    \encoded{C.\mathit{table\_actions}} &=& \quantT{T} T \ttt{ match }\{\\
    & &\qquad \singleton{\ttt{"IPv4\_table"}} \Rightarrow \typeUnion{\singleton{\ttt{"IPv4\_forward"}}}{\singleton{\ttt{"Drop"}}}\\
    & &\qquad \singleton{\ttt{"IPv6\_table"}} \Rightarrow \typeUnion{\singleton{\ttt{"IPv6\_forward"}}}{\singleton{\ttt{"Drop"}}} \;\}\\
    \encoded{C.\mathit{action\_params}} &=& \quantT{A} A \ttt{ match }\{\\
    & &\qquad \singleton{\ttt{"IPv4\_forward"}} \Rightarrow \curl{\mathit{mac\_dst} : \tx{Bytes}, \mathit{port} : \tx{Bytes}}\\
    & &\qquad \singleton{\ttt{"IPv6\_forward"}} \Rightarrow \curl{\mathit{mac\_dst} : \tx{Bytes}, \mathit{port} : \tx{Bytes}}\\
    & &\qquad \singleton{\ttt{"Drop"}} \Rightarrow \tx{Unit}\;\}\\
  \end{array}
  \]
  \vspace{-3mm}
  \caption{The encoding of the server configuration $C$ in
  \Cref{fig:p4runtime-server-config-ex} into \ourFormalLang types.
  \changeNoMargin{(The same types are also used in
  \Cref{eg:typable-untypable-p4runtime-ops}, where they are called
  $T_m,T_a,T_p$.)}}
  \label{fig:p4runtime-server-config-ex-enc}
\end{figure}

From now on, we will assume that each P4Runtime server is \emph{well-formed} by
\Cref{def:p4runtime-server-well-formed} below: it means that each entity belongs
to the $\PPPPEntity$ type (\Cref{fig:language-syntax-sugar-types}) instantiated
with type parameters that correspond to the type-encoded server configuration
(by \Cref{def:server-config-encoding}).

\begin{definition}[P4Runtime Entity Conformity and Server Well-Formedness]
  \label{def:p4runtime-server-well-formed}
  \label{def:p4runtime-entity-conformity}
  A P4Runtime entity $e$ \emph{conforms} to a server configuration $C$ iff:
  \[
    \exists X_n, X_a:
    e \in \PPPPEntity\ \encoded{C.\mathit{table\_matches}}\ \encoded{C.\mathit{table\_actions}}\ \encoded{C.\mathit{action\_params}}\ X_n\ X_a
  \]
  %
  The predicate $\mathit{Conforms}(e,C)$ holds iff entity $e$ conforms to the
  configuration $C$.
  A P4Runtime server \emph{$\langle C,E,a,K \rangle$ is well-formed}
  iff $\forall e \in E : \mathit{Conforms}(e, C)$.
\end{definition}

The key insight behind \Cref{def:p4runtime-server-well-formed} is that the
instantiation of $\PPPPEntity$ can only reduce to an actual type if the argument
$X_n$ is a valid table name in $C$, and if $X_a$ is a valid action for table
$X_n$.
\note{An example would not hurt}
\Cref{def:p4runtime-server-well-formed} directly leads to the following property,
which will allow us to prove the results in \Cref{sec:results}: if a client
sends to the server a well-typed value $v$, the server will consider it conformant.

\begin{proposition}[Conformance of Well-Typed Values]
    \label{def:config-enc-equiv}
    For any server $S = \langle C,E,a,K \rangle$ and any value $v$, we have:
    \begin{equation*}
        \mathit{Conforms}(v, C) \iff \typesTo{\emptyset}{v}{\PPPPEntity\ \encoded{C.\mathit{table\_matches}}\ \encoded{C.\mathit{table\_actions}}\ \encoded{C.\mathit{action\_params}}\ X_n\ X_a}
    \end{equation*}
\end{proposition}

\begin{definition}[P4Runtime Server Semantics]
  \label{def:p4runtime-server-semantics}
  We define the \emph{semantics of a P4Runtime server $S$} as a relation $S
  \reduce{\dual{\alpha}} S'$ (where $\alpha$ is from
  \Cref{def:p4runtime-client-semantics}) inductively defined by the rules in
  \Cref{fig:p4runtime-server-semantics}.
\end{definition}

\begin{figure}
\newcommand{\rowSep}{2mm}
\begin{gather*}
    \Rule
        {s \tx{ is a fresh channel} \\[-1mm]
         s \in \tx{Chan}[\encoded{C.\mathit{table\_matches}}, \encoded{C.\mathit{table\_actions}}, \encoded{C.\mathit{action\_params}}]}
        {\serverRed{C, E, a, K}{\tx{connect}(a)}{s}{C, E, a, K \cup \curl{s}}}
        {Sv-Connect}\\[\rowSep]
    \Rule
        {s \in K \quad \mathit{Conforms}(v,C) \quad \serverRedInt{C}{E}{\tx{read}(v)}{v'}}
        {\serverRed{C, E, a, K}{\tx{read}(s,v)}{v'}{C, E, a, K}}
        {Sv-Read}\\[\rowSep]
    \Rule
        {s \in K \quad \mathit{Conforms}(v,C) \quad \serverRedInt{C}{E}{\tx{insert}(v)}{(E',v')}}
        {\serverRed{C, E, a, K}{\tx{insert}(s,v)}{v'}{C, E', a, K}}
        {Sv-Insert}\\[\rowSep]
    \Rule
        {s \in K \quad \mathit{Conforms}(v,C) \quad \serverRedInt{C}{E}{\tx{modify}(v)}{(E',v')}}
        {\serverRed{C, E, a, K}{\tx{modify}(s,v)}{v'}{C, E', a, K}}
        {Sv-Modify}\\[\rowSep]
    \Rule
        {s \in K \quad \mathit{Conforms}(v,C) \quad \serverRedInt{C}{E}{\tx{delete}(v)}{(E',v')}}
        {\serverRed{C, E, a, K}{\tx{delete}(s,v)}{v'}{C, E', a, K}}
        {Sv-Delete}
\end{gather*}
\caption{LTS semantics of a P4Runtime server.}
\label{fig:p4runtime-server-semantics}
\end{figure}

The P4Runtime server semantics in \Cref{def:p4runtime-server-semantics} show how
the internal configuration of a P4Runtime server evolves, and how the server
responds to queries from clients. The semantics are based on the
\citeN{p4rtspec}.

The server semantics focus on checking the conformance of requests from the
clients, and computes a response using an abstract evaluation predicate
``$\serverRedInt{\_}{\_}{\_}{\_}$'': the details of this predicate are not
crucial --- but we assume that it always yields a well-typed response, i.e. a
boolean or an entity that conforms to the server configuration $C$.\footnote{%
  For reference, the semantics of the predicate ``$\serverRedInt{C}{E}{\tx{read}(v)}{v'}$''
  are available in the appendix, in \Cref{fig:p4runtime-server-semantics-internal-read-eval}.}

\begin{itemize}
  \item By rule \textsc{Sv-Connect}, a server listening on address $a$ can
    accept a client connection by generating a unique channel instance $s$,
    adding $s$ to the set of established connections $K$, and producing a
    transition label $\dual{\reduceTwoLabel{{\tx{connect}(a)}}{s}}$.
    Importantly, the channels $s$ belongs to a \ourFormalLang channel type whose
    type arguments are obtained by encoding the server configuration $C$
    (by the encoding $\encoded{\cdots}$ in \Cref{def:server-config-encoding}).
  \item By rule \textsc{Sv-Read}, the server can handle a client's read request
    by emitting a label $\dual{\reduceTwoLabel{{\tx{read}(s,v)}}{v'}}$, provided
    that the connection $s$ belongs to the set of established connections $K$,
    and the query argument $v$ conforms to the server configuration (by
    \Cref{def:p4runtime-entity-conformity});
  \item Rules \textsc{Sv-Insert}, \textsc{Sv-Modify}, and \textsc{Sv-Delete}
    have requirements similar to \textsc{Sv-Read}, except that they can modify
    the server entities in $E$ --- e.g.~by adding or removing P4 table
    entries.
\end{itemize}


\subsection{Semantics of P4Runtime Networks}
\label{sec:p4runtime-network-semantics}

We now formalise the semantics of the P4Runtime networks introduced in
\Cref{def:p4runtime-network-syntax}.

\begin{definition}[P4Runtime Network Semantics]
  \label{def:p4runtime-network-semantics}
  The \emph{LTS semantics of a P4Runtime network}
  is defined by the following rules, where $\alpha$ ranges over the labels
  introduced in \Cref{def:p4runtime-client-semantics}: (for brevity, we omit
  the symmetric rules)
  \[
    \Rule
        {N \reduce{\alpha} N'}
        {N \,|\, N'' \;\reduce{\alpha}\; N' \,|\, N''}
        {Net-$\alpha$}
    \qquad
    \Rule
        {N_1 \reduce{\dual{\alpha}} N_1' \quad N_2 \reduce{\alpha} N_2'}
        {N_1 \,|\, N_2 \;\reduce{\tau}\; N_1' \,|\, N_2'}
        {Net-Comm}
  \]
  We often write $N \reduce{} N'$ instead of $N \reduce{\tau} N'$, and
  $\reduceStar$ for the reflexive and transitive closure of\, $\reduce{}$.
\end{definition}

By \Cref{def:p4runtime-network-syntax}, a network $N$ is a parallel composition
of any number of P4Runtime clients and servers.  According to the semantics in
\Cref{def:p4runtime-network-semantics}, a network $N$ can perform a transition
$\alpha$ even when composed with another network $N''$ (rule
\textsc{Net-$\alpha$}); and if two networks fire dual labels $\alpha$ and
$\dual{\alpha}$, then they can synchronise when composed, producing a
$\tau$-transition: this allows a P4Runtime client and server to interact, as
illustrated in \Cref{eg:p4runtime-network} below.

\begin{example}[A Simple P4Runtime Network]
  \label{eg:p4runtime-network}
  We give a brief example of how a network reduction could look using our semantics.
  Consider the \ourFormalLang term:

  \smallskip\centerline{$
    \letExp{c}{\ttt{Connect}(a)}\ttt{Insert}(c,v)
  $}\smallskip

  This term attempts to connect to a P4Runtime server $a$ and insert a value $v$
  (a P4 table entry).   If we compose this \ourFormalLang term with a P4Runtime
  server, the resulting network reduces as:
  \change{change:rg:11}{Fixed transition and show reduction of let}{
  \[
    \small
    \derive
      {\derive
        {s \tx{ is a fresh channel} \\[-1mm]
         s \in \Chan{\encoded{C.\mathit{tm}}, \encoded{C.\mathit{ta}}, \encoded{C.\mathit{ap}}}}
        {\langle C,E,a,K \rangle \;\reduceTwoS{\tx{connect}(a)}{s}\; \langle C,E,a,K \cup \curl{s} \rangle}
      \quad
      \derive
        {a \in \ServerRef{T_m,T_a,T_p} \quad s \in \Chan{T_m,T_a,T_p}}
        {\ttt{Connect}(a) \;\reduceTwo{\tx{connect}(a)}{s}\; s}}
      {\langle C,E,a,K \rangle \,|\, \letExp{c}{\ttt{Connect}(a)}\ttt{Insert}(c,v) \;\reduce{\tau}\; \langle C,E,a,K \cup \curl{s} \rangle \,|\, \letExp{c}{s}\ttt{Insert}(c,v) }
  \]
  Then, the ``let'' expression substitutes the variable $c$ with the channel
  $s$:
  \[
    \small
    \langle C,E,a,K \cup \curl{s} \rangle \,|\, \letExp{c}{s}\ttt{Insert}(c,v)
    \;\reduce{\tau}\; \langle C,E,a,K \cup \curl{s} \rangle \,|\, \ttt{Insert}(s,v)
  \]

    Finally,
}
  the network may reduce as follows: (here, let $K' = K \cup \curl{s}$ and
  $E' = E \cup \curl{v}$)
  \begin{equation*}
    \derive
      {\derive
        {s \in K' \quad \mathit{Conforms}(v, C) \quad \serverRedInt{C}{E}{\tx{insert}(v)}{\ttt{true}}}
        {\langle C,E,a,K' \rangle \;\reduceTwoS{\tx{insert}(s,v)}{\ttt{true}}\; \langle C,E',a,K' \rangle}
      \quad
      \derive
        {\ttt{true} \in \tx{Bool}}
        {\ttt{Insert}(s,v) \;\reduceTwo{\tx{insert}(s,v)}{\ttt{true}}\; \ttt{true}}}
      {\langle C,E,a,K' \rangle \,|\, \ttt{Insert}(s,v) \;\reduce{\tau}\; \langle C,E',a,K' \rangle \,|\, \ttt{true}}
  \end{equation*}

  The above example succeeds because the inserted value $v$ conforms to the
  server configuration $C$ --- but this may not always happen: in particular, if
  $\mathit{Conforms}(v, C)$ does \emph{not} hold (e.g.~because $v$
  refers to a P4 table name that does not exist in $C$), then the server and
  client cannot synchronise and they get stuck.
  To prevent this situation, we introduced the type system for \ourFormalLang in \Cref{sec:type-system}, and we prove its properties in \Cref{sec:results}.
\end{example}


%% file: results.tex
\section{Results}
\label{sec:results}
\label{sec:type-system-properties}

We now present our main results: well-typed \ourFormalLang terms enjoy
type preservation (\Cref{theorem:preservation}) and progress
(\Cref{theorem:progress}). Type preservation ensures that if a
well-typed \ourFormalLang term $t$ performs any number of
$\tau$-transitions (either with internal computations, or by
interacting with P4Runtime servers in a surrounding network) and
becomes $t'$, then $t'$ is also well-typed with the same type.

\begin{restatable}[Type preservation]{theorem}{lemPreservation}
    \label{theorem:preservation}
    If\, ${\typesTo{\Gamma}{t}{T}}$ and $t \mid N \reduceStar t' \mid N'$, then ${\typesTo{\Gamma}{t'}{T}}$.
\end{restatable}

Our progress result (\Cref{theorem:progress}) ensures that if a well-typed
\ourFormalLang term $t$ is not a value, then it can perform
further $\tau$-transitions --- either through internal computation, or by
interacting with a P4Runtime server in a surrounding network.  This second
situation means that a well-typed client and server never get stuck, and
therefore:
\begin{itemize}
  \item a well-typed client will never attempt to read or update P4 entities
    unsupported by the server, and
  \item a well-typed client will support all values returned by a server after
    any operation.
\end{itemize}

For this result to hold, we must ensure that all clients attempting to connect
to a server will succeed, and that any client-server channel has a client-side
type that matches the server configuration.  To guarantee this, we consider
\emph{well-typed networks} according to \Cref{def:well-typed-network} below. But
first, we introduce a notion of \emph{network congruence}
(\Cref{def:network-congruence} below) allowing us to abstract from the order in
which clients and servers appear in a parallel composition.

\begin{definition}[Network Congruence]
  \label{def:network-congruence}
  $\equiv$ is the least congruence between networks such that:
  \[
    N_1 \mid N_2 \,\equiv\, N_2 \mid N_1
    \qquad
    (N_1 \mid N_2) \mid N_3 \,\equiv\, N_1 \mid (N_2 \mid N_3)
  \]
\end{definition}

\begin{definition}[Well-typed Network]
  \label{def:well-typed-network}
  We say that \emph{a network $N$ is well-typed} iff for all P4Runtime clients
  $t$ such that $N \equiv t \mid N_0$ (for some $N_0$), we have:
  \todoin{Revise notation}{Make sure the server notation below is correct}
  \begin{enumerate}
    \item $\typesTo{\emptyset}{t}{T}$ (for some $T$);
    \item for all server addresses $a$ occurring in $t$:
      \begin{itemize}
        \item there is exactly one server $S = \langle C, E, a, K \rangle$
            such that $N_0 \equiv S \mid N_1$; and
        \item $a \in \tx{ServerRef}\!\left[\encoded{C.\mathit{table\_matches}}, \encoded{C.\mathit{table\_actions}}, \encoded{C.\mathit{action\_params}}\right]$
      \end{itemize}
    \item for all client-server channels $s$ occurring in $t$:
      \begin{itemize}
        \item there is exactly one server $S = \langle C, E, a, K \rangle$
          with $s \in K$, and such that $N_0 \equiv S \mid N_1$; and
        \item $s \in \tx{Chan}\!\left[\encoded{C.\mathit{table\_matches}}, \encoded{C.\mathit{table\_actions}}, \encoded{C.\mathit{action\_params}}\right]$
      \end{itemize}
    \end{enumerate}
\end{definition}

We now have all ingredients to formalise progress (\Cref{theorem:progress}), and
the resulting \Cref{corollary:type-soundness}: well-typed networks only stop
reducing when all P4Runtime clients terminate successfully.

\begin{restatable}[Progress]{theorem}{progress}
    \label{theorem:progress}
    Take any well-typed network $N$, and take any P4Runtime client $t$
    such that $N \equiv t \mid N_0$ (for some $N_0$). Then either:
    \begin{itemize}
      \item $t$ is fully-reduced into a value; or
      \item $t \reduce{} t'$, and correspondingly,
        $N \reduce{} N' \equiv t' \mid N_0$ with $N'$ well-typed; or
      \item there is a server $S$ such that $N_0 \equiv S \mid N_1$ and $t \mid S \reduce{} t' \mid S'$,
        and correspondingly, $N \reduce{} N' \equiv t' \mid S' \mid N_1$ with
        $N'$ well-typed.
    \end{itemize}
\end{restatable}

\begin{corollary}[Type soundness]
  \label{corollary:type-soundness}
  Take any well-typed network $N$.  If\, $N \reduceStar N'$ and $N'$ cannot
  perform further $\tau$-transitions, then all P4Runtime clients in $N'$ are
  fully-reduced into values.
\end{corollary}

%% file: implementation.tex
\section{Implementation of \ourDSL: a Scala 3 API based on \ourFormalLang}
\label{sec:implementation}

We now outline the implementation \ourDSL, our verified API for programming
P4Runtime client applications, based on our formalisation
of \ourFormalLang and its typing system (\Cref{sec:language,sec:type-system}).
\change{change:impl-mention-artifact}{}{%
\ourDSL is published as companion artifact of this paper, and its latest version
is available at:

\smallskip\centerline{%
    \url{https://github.com/JensKanstrupLarsen/P4R-Type/}
}\smallskip%
}

Our typing system (\Cref{sec:type-system}) is designed to take advantage of
Scala 3 features (in particular, match types \cite{BlanvillainBKO22}): this
naturally leads to implementing \ourDSL as a Scala 3 API. Consequently, the
interactions between a client using \ourDSL and one or more P4 devices have the
properties presented in \Cref{sec:results}: all read/insert/modify/delete
operations are type-safe, and they enjoy progress and preservation (if
both client and device use the same P4Info file).

The implementation of \ourDSL consists of:
\emph{(1)} a type-parametric API for P4Runtime operations (\texttt{connect},
    \texttt{read}, \texttt{insert}, etc.) (\Cref{sec:implementation:api}), and
\emph{(2)} a software tool that turns a P4Info file into a set of Scala 3
    types
    which constrain the \ourDSL API
    (\Cref{sec:implementation:p4info-to-scala}).%

\subsection{Type-Parametric API for P4Runtime Operations}
\label{sec:implementation:api}

The \ourDSL API consists of the five P4Runtime operations detailed in
\Cref{sec:language}: \ttt{connect}, \ttt{read}, \ttt{insert}, \ttt{modify}, and
\ttt{delete}.  We implement these operations as methods equipped with the
strict type parameters shown in \Cref{fig:typing-rules} (rules \textsc{T-OpC},
\textsc{T-OpI}, \textsc{T-OpM}, \textsc{T-OpD}). The operations closely
correspond to the operations in the P4Runtime protobuf API~\cite{p4rtspec}.
Under the hood, these methods use the loosely-typed P4Runtime protobuf
specification and RPC,\footnote{\url{https://github.com/p4lang/p4runtime}}
with (de-)serialisation from/to Scala objects based on the ScalaPB
library:\footnote{\url{https://scalapb.github.io/}}

\begin{itemize}
  \item \ttt{connect} uses the \texttt{StreamChannel} RPC to establish a
    connection; 
  \item \ttt{read} uses the \texttt{Read} RPC to read table entries from the
    server;
  \item \ttt{insert}, \ttt{modify}, and \ttt{delete} use the \texttt{Write} RPC
    to update the server.
\end{itemize}

The signature of the API methods also align with the formal API:
\begin{align*}
  \letExp{\ttt{read}}{&\lambda T_m.\; \lambda T_a.\; \lambda T_p.\; \tabs{X_n}{\tx{TableName}} \tabs{X_a}{T_a\ X_n}\\
                      &\abs{c}{\tx{Chan}[T_m,T_a,T_p]} \\
                      &\abs{x}{\curl{\tx{name} : X_n,\; \tx{matches} : T_m\ X_n,\; \tx{action} : X_a,\; \tx{params} : T_p\ X_a}}\ttt{Read}(c,x)\\
                      &\qquad}\ldots
\end{align*}
\begin{lstlisting}[language=Scala]
  def read[TM[_], TA[_], TP[_]]
    (c: FP4Channel[TM, TA, TP], tableEntry: FP4TableEntry[TM, TA, TP, _, _])
    : Seq[FP4TableEntry[TM, TA, TP, _, _]] 
    = ...
\end{lstlisting}
\noindent
In the code snippet above, the two types \ttt{FP4Channel} and
\ttt{FP4TableEntry} are also part of \ourDSL.
Each of these types take the same type parameters as
their equivalents in \Cref{fig:language-syntax-types,fig:language-syntax-sugar-types}; such type
parameters are usually constrained by the context and inferred by the Scala 3 compiler,
hence the user does not need to write them explicitly. The \ttt{FP4Channel} type is simply
a case class that contains the table entry values (table name, parameters,
etc.), while the \ttt{FP4Channel} is an abstract class containing the methods
for serialization (\ttt{toProto}) and deserialization (\ttt{fromProto}).

\subsection{Translation of P4 Device Configuration Metadata (P4Info) into Scala 3 Types}
\label{sec:implementation:p4info-to-scala}

\ourDSL includes a tool that implements the encoding in
\Cref{def:server-config-encoding}: the tool takes a P4Info file
(representing a P4 device's tables, actions, \ldots) and generates three Scala 3
types, which can be used to instantiate the type parameters $T_m, T_a, T_p$ (see
\Cref{sec:language,sec:type-system}) to guarantee type safety and progress. Such
generated types are called \texttt{TableMatchFields}, \texttt{TableActions}, and
\texttt{ActionParams}:
\begin{itemize}
  \item type \texttt{TableMatchFields} can instantiate $T_m$, and maps table
    names to their match fields;
  \item type \texttt{TableActions} can instantiate $T_a$, and maps table names
    to their action names;
  \item type \texttt{ActionParams} can instantiate $T_p$, and maps action names
    to their parameter types.
\end{itemize}
A programmer can use \ourDSL to connect to a P4 device and obtain a typed
channel constrained by the 3 types above (akin to $\Chan{T_m,T_a,T_p}$ in
\Cref{sec:language}); when using our type-parametric API
(\Cref{sec:implementation:api}) on this typed channel, only operations
compatible with the P4 device can be performed; otherwise, a type error occurs
(just like our type system in \Cref{sec:type-system} prevents invalid
operations).

We now illustrate in more detail the \ourDSL-generated types that can
instantiate the type parameters $T_m,T_a,T_p$, using
\Cref{fig:p4runtime-server-config-ex-enc} as an example.

\medskip\noindent%
\textbf{The type parameter $T_m$ (match fields of a P4 table)}
can be instantiated with the higher-kinded type \texttt{TableMatchFields}, which
takes a parameter \ttt{TN} (expected to be a known table name).
\begin{lstlisting}[language=Scala]
  type TableMatchFields[TN] = TN match
      case "IPv4_table" => ("IPv4_dst_addr", P4.LPM)
      case "IPv6_table" => ("IPv6_dst_addr", P4.LPM)
      case "*" => "*"
\end{lstlisting}
The type above matches a table name \ttt{TN} with one of
the known table names (represented as singleton string types) and yields tuple
types pairing \texttt{TN}'s field names 
with their type of packet match (\texttt{P4.Exact}, \texttt{P4.Ternary},
\texttt{P4.LPM}, \ldots which are types provided by \ourDSL).  As per P4Runtime
standard, table fields can be optionally undefined, unless they perform a
\texttt{P4.Exact} packet match.


\medskip\noindent%
\textbf{The type parameter $T_a$ (P4 table actions)} can be instantiated with type
\ttt{TableAction}, that matches a table name \ttt{TN}
to yield the valid actions for \ttt{TN} (which may include the wildcard \ttt{*}).
\begin{lstlisting}[language=Scala]
type TableActions[TN] <: ActionName = TN match
    case "IPv4_table" => "IPv4_forward" | "Drop" | "*"
    case "IPv6_table" => "IPv6_forward" | "Drop" | "*"
    case "*" => "*"
\end{lstlisting}


\medskip\noindent%
\textbf{The type parameter $T_p$ (action parameters)} can be instantiated with
type \ttt{ActionParams}, that matches an action name \ttt{AN}
to yield the parameter types for \ttt{AN}. Each parameter type is a tuple with
the name of the parameter (as a singleton string type) and the value type.
\begin{lstlisting}[language=Scala]
type ActionParams[AN] = AN match
    case "IPv4_forward" => (("mac_dst", ByteString), ("port", ByteString))
    case "IPv6_forward" => (("mac_dst", ByteString), ("port", ByteString))
    case "Drop" => Unit
    case "*" => Unit
\end{lstlisting}

%

\medskip%
All three types above also accept a \emph{wildcard} singleton type \ttt{"*"} as
a parameter, representing the request of querying all/any table match fields,
actions, or parameters.



%% file: case-study.tex
\change{change:case-studies-title-intro}{}{%
\section{Case Studies and Discussion of Alternative Designs}
\label{sec:case-study}

In this section we demonstrate the usefulness of having compile-time checked
P4Runtime queries, by illustrating three case studies implemented using \ourDSL.
We discuss one case study in detail (update of multiple switches, in
\Cref{sec:case-study-multiple-switches}) and outline two more
(port forwarding and load balancing, in \Cref{sec:case-study-port-fw,sec:case-study-load-bal}):
these applications derive and extend the tunnelling example in
the P4Runtime tutorials,\footnote{%
  \url{https://github.com/p4lang/tutorials/tree/master/exercises/p4runtime}%
} %
and are all included in the software artifact that accompanies this paper.

\subsection{Updating a Network with Multiple P4 Switches}
\label{sec:case-study-multiple-switches}
}%

\begin{figure}
    \begin{tikzpicture}
        \node[circle, draw] (h1) at (-1, 2) {N1};
        \node[circle, draw] (h2) at (5, 2) {N2};
        \node[circle, draw] (h3) at (-1, 1) {N3};
        \node[circle, draw] (h4) at (5, 1) {N4};
        \node[rectangle, draw] (s1) at (1, 2) {S1};
        \node[rectangle, draw] (s2) at (3, 2) {S2};
        \node[rectangle, draw] (s3) at (1, 1) {S4};
        \node[rectangle, draw] (s4) at (3, 1) {S4};
        \draw[dashed] (h1) -- (s1);
        \draw[dashed] (h2) -- (s2);
        \draw[dashed] (h3) -- (s3);
        \draw[dashed] (h4) -- (s4);
        \draw (s1) -- (s2);
        \draw (s1) -- (s3);
        \draw (s1) -- (s4);
        \draw (s2) -- (s4);
        \draw (s2) -- (s3);
        \draw (s3) -- (s4);
    \end{tikzpicture}
    \caption{Network topology used in the case studies
    (\Cref{sec:case-study}): N1--N4 are networks, and S1--S4 are switches.}
    \label{fig:case-study-topology}
\end{figure}
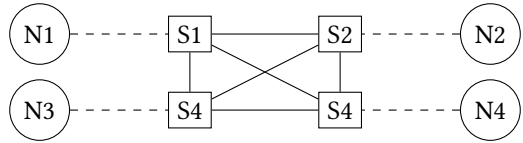

\Cref{fig:case-study-topology} shows the case study network. It contains four
networks (N1--N4) which are connected through the bridge established by the
switches (S1--S4). Switch S1 and S2 use the same subnet mask (\ttt{10.1.0.0}),
as do switch S3 and S4 (\ttt{10.2.0.0}).%
\begin{figure}
    \centering
    \begin{minipage}{.45\textwidth}
    \begin{lstlisting}[basicstyle=\tiny\ttfamily,frame=tlrb]{config1.p4}
control Process(...) {
  action drop() { ... }
  action ipv4_forward(macAddr_t dstAddr,
                      egressSpec_t port) {
    ...
  }
  table firewall {
    key = { hdr.ipv4.dstAddr : lpm; }
    actions = { drop; NoAction; }
  }
  table ipv4_lpm {
    key = { hdr.ipv4.dstAddr : lpm; }
    actions = { ipv4_forward; NoAction; }
  }
  apply {
    if (hdr.ipv4.isValid()) {
      if (firewall.apply().miss) {
        ipv4_lpm.apply();
} } } }
    \end{lstlisting}
    \end{minipage}\hfill
    \begin{minipage}{.45\textwidth}
    \begin{lstlisting}[basicstyle=\tiny\ttfamily,frame=tlrb]{config2.p4}
control Process(...) {
  action drop() { ... }
  action forward_packet(egressSpec_t port) {
    ...
  }
  table firewall {
    key = { hdr.ipv4.dstAddr : lpm; }
    actions = { drop; NoAction; }
  }
  table ipv4_table {
    key = { hdr.ipv4.srcAddr : exact;
            hdr.ipv4.dstAddr : lpm; }
    actions = { forward_packet; drop; }
  }
  apply {
    if (hdr.ipv4.isValid()) {
      if (firewall.apply().miss) {
        ipv4_table.apply();
} } } }
    \end{lstlisting}
    \end{minipage}
    \centering
    \vspace{-5mm}
    \caption{The packet-processing sections of the P4 files of switches S1 and S3 (left) and S2 and S4 (right).}
    \label{fig:case-study-p4-files}
\end{figure}
Each switch is configured with a general firewall table for all network
traffic, as well as a more specific IPv4 forwarding table for its own subnet.
For this reason, the switches use different configuration files, shown in
\Cref{fig:case-study-p4-files}. All of the switches should share the same
entries for the firewall table. Switch S1 is the master switch for forwarding
rules related to subnet \ttt{10.1.0.0}, while switch S3 is the master switch
for forwarding rules related to subnet \ttt{10.2.0.0}, meaning that S2 and S4
should replicate their table entries, respectively.

The replication of table entries must be done periodically by an external
controller. For this case study, we implement a controller in \ourDSL
that performs this replication, which should:
\begin{enumerate}
  \item Insert a set of \ttt{firewall} table entries into all four switches.
  \item Read all entries from the \ttt{ipv4\_lpm} table on S1,
    then insert them into S2.
  \item Read all entries from the \ttt{ipv4\_table} table on S3,
    then insert them into S4.
\end{enumerate}

\begin{figure}
  \begin{lstlisting}[language=Scala]
val s1 = config1.connect(0, "127.0.0.1", 50051)
val s2 = config1.connect(1, "127.0.0.1", 50052)
val s3 = config2.connect(2, "127.0.0.1", 50053)
val s4 = config2.connect(3, "127.0.0.1", 50054)

// Writing firewall entries
for (s <- List(s1,s2,s3,s4))
   for (ip <- List(bytes(10,0,42,43), bytes(10,0,13,0), bytes(10,0,37,0)))
       insert(s, FP4TableEntry("Process.firewall",
                               Some("hdr.ipv4.dstAddr", FP4_LPM(ip, 32)),
                               "Process.drop", () ))

// Reading config1 IPv4 entries
val c1_ipv4_entries = read(s1, FP4TableEntry("Process.ipv4_lpm", "*", "*", ()))

// Writing config1 IPv4 entries
for (entry <- c1_ipv4_entries) insert(s2, entry)

// Reading config2 IPv4 entries
val c2_ipv4_entries = read(s3, FP4TableEntry("Process.ipv4_table", "*", "*", ()))

// Writing config2 IPv4 entries
for (entry <- c2_ipv4_entries) insert(s4, entry)
  \end{lstlisting}
  \vspace{-3mm}
  \caption{The replication program written in \ourDSL.}
  \label{fig:case-study-dsl}
\end{figure}

\change{change:rg:3}{Add possible errors}{
When a programmer uses our \ourDSL API, the Scala compiler
spots several possible errors that may occur when updating multiple
switches with different P4 configurations:
\begin{itemize}
  \item Using non-existent table or action names (e.g., due to typos)
  \item   Inserting the wrong type of entries in a table (e.g., wrong number of match
        fields)
        \item   Using an existing action in an existing table that does not support it (e.g., an entry in \ttt{firewall} referencing \ttt{ipv4\_forward})
        \item   Passing the wrong type of arguments to an action (e.g., an entry in \ttt{ipv4\_lpm} referencing action \ttt{ipv4\_forward},
        but passing only one argument)
\end{itemize}
}

\paragraph{Generated Types in Scala 3.} Using the types generated by the tool, the replication program written in
\ourDSL is shown in \Cref{fig:case-study-dsl}. Note that the API
interface is relatively minimal and similar to the Python API. For instance,
compare the \ttt{insert} call in line 9--11 to the Python code in \Cref{fig:build-table-error}.
The difference here is that an error like the one shown in \Cref{fig:build-table-error}
would be caught at compile time by the Scala 3 type system. For example, using \stringTerm{Process.forward\_packet} instead of \stringTerm{Process.drop}
on line 11 would yield a type error: \emph{``a value of type \ttt{s.TA[("Process.firewall")]}
is required''}.

On lines 1-4, the connection to each switch is established. Note that the
\ttt{connect} methods are specific to each configuration, unlike the other
P4Runtime operations which are part of a generic package: \ttt{connect} returns
an \ttt{FP4Channel} instance with predefined type parameters, which in turn
constrain the \texttt{read}/\texttt{insert}/\texttt{modify}/\texttt{delete}
operations that can be performed on that channel.
%
%
%
Consider e.g.~lines 9--11 in \Cref{fig:case-study-dsl}: in the \ttt{insert}
call, the \ttt{tableEntry} parameter is constrained to only accept table entries
that satisfy the switch configuration of channel \texttt{s}.  Since \texttt{s}
ranges over a list of channels having two different types of switches
(\ttt{config1} and \ttt{config2}), such entries must be valid in \emph{both} switch
configurations. Since both configurations share a \stringTerm{Process.firewall}
table, the program compiles. Instead, if an otherwise valid entry for e.g. the
\stringTerm{Process.ipv4\_lpm} table is provided, the code would not compile,
as that table is defined in \ttt{config1} but not in \ttt{config2}.


\change{change:additional-applications}{New applications}{%
\subsection{Port Forwarding Management}
\label{sec:case-study-port-fw}

We implemented a control plane program for \emph{port forwarding}, which is a
Network Address Translations (NAT) service typically offered e.g.~by Internet
routers.  We use the same topology as in \Cref{fig:case-study-topology}, but we
assume that N1, N2, and N3 are local networks, while N4 is an external network.
The goal is to allow external clients to connect to servers hosted in the
internal networks.  To this end, S4 applies a set of NAT rules saying e.g.~that:
\begin{itemize}
  \item each packet received on the external S4 interface, and having
    destination IP address \ttt{1.2.3.4} and port \ttt{42}, should be translated
    to have destination IP address \ttt{10.1.0.4} and port 1042
    (and vice versa for the internal S4 interface).
\end{itemize}
We developed a program (using \ourDSL) that offers a command line interface to
connect to S4 and query, add, and delete its NAT rules.  The program reads and
modifies two P4 tables called \texttt{nat\_ingress} and \texttt{nat\_egress}
containing the translations for incoming and outgoing packets.
Translated packets are then forwarded according to the entries of a table called
\ttt{ipv4\_forward} (similar to the one used in
\Cref{sec:case-study-multiple-switches}).

\subsection{Load Balancing}
\label{sec:case-study-load-bal}
We implemented a control plane program for load balancing packet transfers. We
use the same topology as in \Cref{fig:case-study-topology}, and the goal is
for S1 to equally destribute all packets bound for N4 between its outgoing
ports to S2, S3 and S4.
To implement this, we use a P4 entity called \emph{counter},\footnote{%
  \changeNoMargin{%
    Counters are not modelled in \ourFormalLang; they can be
    easily added e.g.~as a new case to the union type of $\PPPPEntity$
    (\Cref{fig:language-syntax-sugar-types}).%
  }%
} which can be incremented by the data plane and read by the control plane.  We
configure the data plane of S1 with one counter per output port, and rules that
increment a counter every time a packet is forwarded through the corresponding
port.  Our control plane program then periodically reads the counter values (using
the \ourDSL method \texttt{readCounter}, similar to \texttt{read} for P4
tables)
and updates the packet forwarding rules (using the \ourDSL method \texttt{modify}).%
}

\change{change:why-match-types}{}{%
\subsection{On the Role of Match Types and Singleton Types}
\label{why-match-singleton-types}

We now discuss whether our results could be achieved with a different design
that, while still satisfying requirements \ref{requirement:formal-foundation},
\ref{requirement:existing-prog-lang}, and \ref{requirement:codegen} in
\Cref{sec:intro}, would not rely on match types nor singleton types, and would
be therefore less tied to the Scala 3 programming language.  Let us consider the
case study in \Cref{sec:case-study-multiple-switches}, and how we could address
it in a subset of Scala 3 \emph{without} match nor singleton types.

To ensure that the table entries described in a P4Info file are constructed
correctly, we would need to generate a dedicated data type for each table, with
argument types capturing the constraints on actions and parameters. We would also
need to constrain channel types to valid table entry types, to ensure that
\texttt{read}/\texttt{insert}/\texttt{modify}/\texttt{delete} only use table
entries of the correct type.  E.g.~in the case of the first P4Info metadata in
\Cref{fig:case-study-p4-files} we might generate a set of type definitions
like:%
}

\begin{lstlisting}[language=Scala,basicstyle=\scriptsize\ttfamily]
package config1

case class ActionWildcard()
case class ActionDrop()
case class ActionForwardPacket(addr: ByteString, port: ByteString)

type FirewallAction = ActionDrop | ActionWildcard
case class FirewallTableEntry(fields: Option[FP4_LPM], action: FirewallAction)

type IPV4Action = ActionDrop | ActionForwardPacket | ActionWildcard
case class IPV4TableEntry(fields: (FP4_Exact, Option[FP4_LPM]), action: IPV4Action)

def connect(...): FP4Channel[FirewallTableEntry | IPV4TableEntry] = ...
\end{lstlisting}

\changeNoMargin{%
A program written with the resulting API would look like:%
}

\begin{lstlisting}[language=Scala,basicstyle=\scriptsize\ttfamily]
val s1 = config1.connect(0, "127.0.0.1", 50051)
insert(s1, config1.FirewallTableEntry(Some(FP4_LPM(...)), config1.ActionDrop))
\end{lstlisting}

\changeNoMargin{%
The type definitions outlined above are roughly as compact as the match types we
generate.\footnote{%
\changeNoMargin{%
These definitions may be more verbose in languages without the type union
``\texttt{|}'', going against requirement \ref{requirement:codegen}.  E.g.~in
F\# or OCaml, \texttt{FirewallAction} and \texttt{IPV4Action} would be rendered
as labelled sum types, and each action used in more than one table would result
in duplicated label definition (in this example, this would apply to
\texttt{ActionDrop} and \texttt{ActionWildcard}).%
}%
}
However, the main drawback of such type definitions
is that they are substantially more laborious to formalise: we would need to
extend the typing system of \ourFormalLang (\Cref{def:type-system}) with a
nominal environment to collect type definitions, and the formal encoding from
P4Info metadata to types would be significantly more complex than our
\Cref{def:server-config-encoding}.  As a consequence, stating and proving
results like our \Cref{theorem:preservation,theorem:progress} would be
considerably harder, hampering requirement \ref{requirement:formal-foundation}.

On the practical side, another drawback of the type definitions outlined above
is that they would make the API more cumbersome and limited:
e.g.~it would be hard or
impossible to write code like lines 7--11 in \Cref{fig:case-study-dsl}, where
the \texttt{insert} operation works on channels with different P4 configurations
\texttt{config1} and \texttt{config2}. The reason is that channels \texttt{s1}
and \texttt{s2} would only support table entries of type
\texttt{config1.FirewallTableEntry}, whereas channels \texttt{s3} and
\texttt{s4} would only support \texttt{config2.FirewallTableEntry}: such types
would be unrelated and could not be unified, hence a programmer would need to
duplicate the code of the \texttt{insert} operations.  One might try to mitigate
this duplication by leveraging structural typing (available e.g. in TypeScript,
or in OCaml \texttt{struct}s) --- but then, the signature of the API method
\texttt{insert} would become non-trivial and the feasibility of this approach
would require further research.  Instead, the match types produced by our
encoding in \Cref{def:server-config-encoding} allow the Scala compiler to verify
that the table entries for \texttt{"Process.firewall"} have the same type under
both \texttt{config1} and \texttt{config2}, hence the code in
\Cref{fig:case-study-dsl} type-checks.%
}


%% file: related.tex
\section{Related Work}\label{sec:related-work}
The programmability of SDNs comes at the cost of
complexity and attack surfaces of modern networks~\cite{kreutz2013towards}.
Several proposals address complementary problems
to our work by giving formal semantics to the data plane language~\cite{doenges2021petr4, alshnakat2022hol4p4, peterson2023p4cub} and by developing
static~\cite{liu2018p4v,stoenescu2018debugging,Eichholz2022POPL}%
\change{change:cite-eichholz}{}{} %
and dynamic~\cite{notzli2018p4pktgen, shukla2020p4consist} analysis tools for the data plane.

Several tools have been developed to verify various network
properties. Header Space Analysis~\cite{kazemian2012header} is a framework
that can analyse reachability and identify loops, among other properties, of dynamic networks. Both the data plane and the control plane must be represented in the abstract framework.
NetKAT~\cite{anderson2014netkat}, and the more recent DyNetKat~\cite{caltais2022dynetkat}, provides a network language
which encompasses both data plane and control plane,
with formal semantics and develops syntactic techniques for proving network reachability, non-interference, and correctness of program transformation.
Batfish~\cite{fogel2015general} uses a declarative approach to define network behavior via logical relations that represent both data and control plane. The framework allows to check if any reachable network configurations (and packets) can violate forwarding properties.

The main difference with our proposal is that these models are
non-executable specifications and are much more abstract than the languages used to program SDNs.
Therefore, they do not directly provide a method to program the
control plane. Verifying actual control software using these models requires to map software behavior to these specifications, which is extremely hard when the control plane is developed using a general-purpose language like Python or Scala.
Moreover, many of these models
assume a configurable, but not programmable,
data plane, which supports a limited and predefined
set of protocols (e.g., SDNs using OpenFlow \cite{OpenFlow08}).
Instead, our proposal provides a programming framework
for the control plane that can interact with arbitrary P4 data planes,
and that can statically prevent invalid table manipulations.








%% file: conclusion.tex
\section{Conclusion and Future Work}\label{sec:conclusions}

\change{change:conclusion-contribs}{}{%
We presented \ourDSL, a novel verified API for P4Runtime programs written in
Scala 3.  As a foundation for \ourDSL, we presented the first formal model of
P4Runtime networks, where servers interact with client applications written in
the calculus \ourFormalLang; we also developed a typing system for
\ourFormalLang (including match types and singleton types, inspired by Scala 3)
and proved that well-typed \ourFormalLang clients interact correctly with the
surrounding servers (\Cref{theorem:preservation,theorem:progress}). These
correctness results are inherited by actual P4 control programs that use
our \ourDSL API.

This paper is a stepping stones toward a broader objective: a fully-verified P4
development pipeline encompassing the verification of \emph{both} the control
plane and the data plane, ensuring that configuration updates applied by control
programs never compromise desired network properties. %
This objective determines our future work, outlined below.%
}

While our type system is sound in the sense that well-typed programs never get
stuck, a server may still in some cases reject an update by producing a
$\ttt{false}$ response value (for \ttt{Insert}, \ttt{Modify} or \ttt{Delete}).
Not all these cases can be statically verified (e.g.~trying to insert a table
entry that already exists in the server), but some cases may be prevented by
further typing constraints.
For example, instead of using the same $\PPPPEntity$ type for all of the
operations that handle table entries, we may adopt distinct formats or
restrictions on table entries for distinct operations --- e.g.~the \ttt{Insert}
operation does not in general accept entries where $\mathit{table\_matches} =
\wildcard$, but the \ttt{Read} operation always does.  A solution to this could
be to generate a distinct set of match types for each operation: this should not
drastically change the formalization nor the proofs.
Network properties like reachability of a node,
enforcement of access control list,
and presence of loops,
for systems with programmable data plane cannot be verified by looking only at
the control plane.
In order to verify these properties,
we plan to extend our semantics with P4Runtime stream messages
and integrate it with existing semantics of P4.
\change{change:conclusion-contribs-other-servers}{}{%
We may also need to formalise more detailed P4Runtime server semantics, e.g.~to
model P4 network elements that perform delayed table updates, have background
processes, or communicate with each other.  We expect that, thanks to our
adoption of an early LTS semantics for clients and servers
(\Cref{sec:p4runtime-client-semantics}), we will be able to adapt the server
semantics, while reusing most of the current proofs and results involving
\ourFormalLang clients.%
}%


%% file: appendix-defs.tex
\section{Additional Definitions}
\label{app:definitions}

\subsection{Substitutions}
\label{app:definitions:substitutions}

\subsubsection{Term Substitution in Terms}
\begin{align*}
    \subs{v}{x}{t}\                                                        &=\ v\\
    \subs{y}{x}{t}\                                                        &=\ t \quad (x = y)\\
    \subs{y}{x}{t}\                                                        &=\ y \quad (x \neq y)\\
    \subs{(\concat{t_1}{t_2})}{x}{t}\                                      &=\ \concat{(\subs{t_1}{x}{t})}{\subs{t_2}{x}{t}}\\
    \subs{(\listHead t_l)}{x}{t}\                                          &=\ \listHead \subs{t_l}{x}{t}\\
    \subs{(\listTail t_l)}{x}{t}\                                          &=\ \listTail \subs{t_l}{x}{t}\\
    \subs{\curl{f_1=t_1, ..., f_n=t_n}}{x}{t}\                             &=\ \curl{f_1=\subs{t_1}{x}{t}, ..., f_n=\subs{t_n}{x}{t}}\\
    \subs{(t_r.f)}{x}{t}\                                                  &=\ (\subs{t_r}{x}{t}).f\\
    \subs{(\abs{y}{T} t_0)}{x}{t}\                                         &=\ \abs{y}{T} \subs{t_0}{x}{t}\\
    \subs{(t_1\ t_2)}{x}{t}\                                               &=\ (\subs{t_1}{x}{t})\ (\subs{t_2}{x}{t})\\
    \subs{(\tabs{X}{T} t_0)}{x}{t}\                                        &=\ \tabs{X}{T} \subs{t_0}{x}{t}\\
    \subs{(t_1\ T_2)}{x}{t}\                                               &=\ (\subs{t_1}{x}{t})\ T_2\\
    \subs{(\letExp{y}{t_0} t_1)}{x}{t}\                                    &=\ \letExp{y}{(\subs{t_0}{x}{t})} (\subs{t_1}{x}{t})\\
    \subs{(\match{t_s}{\types{x_i}{T_i} \Rightarrow t_i}{i \in I})}{x}{t}\ &=\ \match{\subs{t_s}{x}{t}}{\types{x_i}{T_i} \Rightarrow \subs{t_i}{x}{t}}{i \in I}\\
    \subs{op(\overline{w})}{x}{t}\                                         &=\ op(\overline{\subs{w}{x}{t}})
\end{align*}

\subsubsection{Type Substitution in Terms}
\begin{align*}
    \subs{v}{X}{T}\                                                        &=\ v\\
    \subs{x}{X}{T}\                                                        &=\ x\\
    \subs{(\concat{t_1}{t_2})}{X}{T}\                                      &=\ \concat{(\subs{t_1}{x}{t})}{\subs{t_2}{X}{T}}\\
    \subs{(\listHead t_l)}{X}{T}\                                          &=\ \listHead \subs{t_l}{X}{T}\\
    \subs{(\listTail t_l)}{X}{T}\                                          &=\ \listTail \subs{t_l}{X}{T}\\
    \subs{\curl{f_1=t_1, ..., f_n=t_n}}{X}{T}\                             &=\ \curl{f_1=\subs{t_1}{X}{T}, ..., f_n=\subs{t_n}{X}{T}}\\
    \subs{(t_r.f)}{X}{T}\                                                  &=\ (\subs{t_r}{X}{T}).f\\
    \subs{(\abs{x}{T_1} t_0)}{X}{T}\                                       &=\ \abs{x}{\subs{T_1}{X}{T}} \subs{t_0}{X}{T}\\
    \subs{(t_1\ t_2)}{X}{T}\                                               &=\ (\subs{t_1}{X}{T})\ (\subs{t_2}{X}{T})\\
    \subs{(\tabs{Y}{T_1} t_0)}{X}{T}\                                      &=\ \tabs{Y}{\subs{T_1}{X}{T}} \subs{t_0}{X}{T}\\
    \subs{(t_1\ T_2)}{X}{T}\                                               &=\ (\subs{t_1}{X}{T})\ (\subs{T_2}{X}{T})\\
    \subs{(\letExp{x}{t_0} t_1)}{X}{T}\                                    &=\ \letExp{x}{(\subs{t_0}{X}{T})} (\subs{t_1}{X}{T})\\
    \subs{(\match{t_s}{\types{x_i}{T_i} \Rightarrow t_i}{i \in I})}{X}{T}\ &=\ \match{\subs{t_s}{X}{T}}{\types{x_i}{\subs{T_i}{X}{T}} \Rightarrow \subs{t_i}{X}{T}}{i \in I}\\
    \subs{op(\overline{w})}{X}{T}\                                         &=\ op(\subs{\overline{w}}{X}{T})
\end{align*}

\subsubsection{Type Substitution in Types}
\begin{align*}
    \subs{\top}{X}{T}\                                           &=\ \top\\
    \subs{\bigT}{X}{T}\                                          &=\ \bigT \quad \bigT \in \curl{\tx{Int}, \tx{Bool}, \tx{String}, ...}\\
    \subs{(\ServerRef{T_m,T_a,T_p})}{X}{T}\                      &=\ \ServerRef{\subs{T_m}{X}{T}, \subs{T_a}{X}{T}, \subs{T_p}{X}{T}}\\
    \subs{(\Chan{T_m,T_a,T_p})}{X}{T}\                           &=\ \Chan{\subs{T_m}{X}{T}, \subs{T_a}{X}{T}, \subs{T_p}{X}{T}}\\
    \subs{\curl{\types{f_1}{T_1}, ..., \types{f_n}{T_n}}}{X}{T}\ &=\ \curl{\types{f_1}{\subs{T_1}{X}{T}}, ..., \types{f_n}{\subs{T_n}{X}{T}}}6\\
    \subs{[T_l]}{X}{T}\                                          &=\ [\subs{T_l}{X}{T}]\\
    \subs{(T_1 \rightarrow T_2)}{X}{T}\                          &=\ \subs{T_1}{X}{T} \rightarrow \subs{T_2}{X}{T}\\
    \subs{Y}{X}{T}\                                              &=\ T \quad (X = Y)\\
    \subs{Y}{X}{T}\                                              &=\ Y \quad (X \neq Y)\\
    \subs{(\quant{Y}{T_1} T_0)}{X}{T}\                           &=\ \quant{Y}{\subs{T_1}{X}{T}} \subs{T_0}{X}{T}\\
    \subs{(T_1\ T_2)}{X}{T}\                                     &=\ \subs{T_1}{X}{T}\ \subs{T_2}{X}{T}\\
    \subs{(\typeUnion{T_1}{T_2})}{X}{T}\                         &=\ \typeUnion{\subs{T_1}{X}{T}}{\subs{T_2}{X}{T}}\\
    \subs{(\match{T_s}{T_i \Rightarrow T_i'}{i \in I})}{X}{T}\   &=\ \match{\subs{T_s}{X}{T}}{\subs{T_i}{X}{T} \Rightarrow \subs{T_i'}{X}{T}}{i \in I}\\
    \subs{\singleton{v}}{X}{T}\                                  &=\ \singleton{v}
\end{align*}

\subsubsection{"$\memberOf{v_G}{T}$" Relation}
\label{sec:memberof-relation}
\begin{gather*}
    \Rule
        {v_G \tx{ is an integer}}
        {\memberOf{v_G}{\tx{Int}}}
        {In-Int}\qquad
    \Rule
        {v_G \in \curl{\ttt{true}, \ttt{false}}}
        {\memberOf{v_G}{\tx{Bool}}}
        {In-Bool}\qquad
    \Rule
        {v_G = \unit}
        {\memberOf{v_G}{\tx{Unit}}}
        {In-Unit}\\\\
    \Rule
        {v_G \tx{ is a string}}
        {\memberOf{v_G}{\tx{String}}}
        {In-String}\qquad
    \Rule
        {v_G \tx{ is a byte sequence}}
        {\memberOf{v_G}{\tx{Bytes}}}
        {In-Bytes}\qquad
    \Rule
        {v_G = \nil}
        {\memberOf{v_G}{\listT{T}}}
        {In-List1}\\\\
    \Rule
        {\memberOf{v_{G1}}{T} \quad \memberOf{v_{G2}}{\listT{T}}}
        {\memberOf{\concat{v_{G1}}{v_{G2}}}{\listT{T}}}
        {In-List2}\qquad
    \Rule
        {\forall i \in I : \memberOf{v_{Gi}}{T_i}}
        {\memberOf{\curl{f_i = v_{Gi}}_{i \in I \cup J}}{\curl{f_i : T_i}_{i \in I}}}
        {In-Rec}\\\\
    \Rule
        {}
        {\memberOf{a_{T_m,T_a,T_p}}{\ServerRef{T_m,T_a,T_p}}}
        {In-ServerRef}\qquad
    \Rule
        {}
        {\memberOf{s_{T_m,T_a,T_p}}{\Chan{T_m,T_a,T_p}}}
        {In-Chan}
\end{gather*}

\subsection{Partial Internal Server Semantics}
\begin{figure}[H]
    \newcommand{\rowSep}{1.5mm}
    \begin{gather*}
        \Rule
            {v.\tx{table\_name} = \wildcard}
            {\serverRedInt{C}{E}{\tx{read}(v)}{E}}
            {}\qquad
        \Rule
            {v.\tx{table\_name} \neq \wildcard\\[-1.5mm]
             \serverRedInt{C}{\setcomp{e \in E}{e.\tx{table\_name} = v.\tx{table\_name}}}{\tx{read}_M(v)}{e'}}
            {\serverRedInt{C}{E}{\tx{read}(v)}{e'}}
            {}\\[\rowSep]
        \Rule
            {v.\tx{matches} = \wildcard \quad \serverRedInt{C}{E}{\tx{read}_A(v)}{e}}
            {\serverRedInt{C}{E}{\tx{read}_M(v)}{e}}
            {}\qquad
        \Rule
            {v.\tx{matches} \neq \wildcard \\[-1.5mm]
             \serverRedInt{C}{\setcomp{e \in E}{e.\tx{matches} = v.\tx{matches}}}{\tx{read}_A(v)}{e'}}
            {\serverRedInt{C}{E}{\tx{read}_M(v_m,v_a)}{e}}
            {}\\[\rowSep]
        \Rule
            {v.\tx{action} = \wildcard \quad \serverRedInt{C}{E}{\tx{read}_P(v)}{e}}
            {\serverRedInt{C}{E}{\tx{read}_A(v)}{e}}
            {}\qquad
        \Rule
            {v.\tx{action} \neq \wildcard \\[-1.5mm]
             \serverRedInt{C}{\setcomp{e \in E}{e.\tx{action} = v.\tx{action}}}{\tx{read}_P(v)}{e'}}
            {\serverRedInt{C}{E}{\tx{read}_A(v)}{e'}}
            {}\\[\rowSep]
        \Rule
            {v.\tx{priority} = \wildcard}
            {\serverRedInt{C}{E}{\tx{read}_P(v)}{E}}
            {}\qquad
        \Rule
            {v.\tx{priority} \neq \wildcard}
            {\serverRedInt{C}{E}{\tx{read}_P(v)}{\setcomp{e \in E}{e.\tx{priority} = v.\tx{priority}}}}
            {}
    \end{gather*}
    \caption{Internal evaluation semantics for ``read'' operations in a P4Runtime server (used in \Cref{fig:p4runtime-server-semantics}).}
    \label{fig:p4runtime-server-semantics-internal-read-eval}
\end{figure}

%% file: appendix.tex
\section{Proofs}
This appendix section presents the proofs for type safety.
The proof strategy closely mirrors the approach in \citeN{BlanvillainBKO22},
with the main deviations being the proof cases for new rules.
\todoin{Include}{Should the above text be included?}
\todoin{Move to main text}{For all the proofs, we assume the Barendregt variable convention.}

\begin{lemma}[Disjointness of values]\label{disjoint-value}
    $\subtypesTo{\Gamma\not}{\singleton{v}}{T} \Rightarrow \disjoint{\Gamma}{\singleton{v}}{T}$
\end{lemma}
\begin{proof}
The proof of the proposition is by contradiction.\\
Suppose that ${\subtypesTo{\Gamma\not}{\singleton{v}}{T}}$, but ${\exists T' : \subtypesTo{\Gamma}{T'}{\singleton{v}} \wedge \subtypesTo{\Gamma}{T'}{T}}$.\\
Since $\singleton{v}$ can only subtype to itself by \tsc{ST-Refl}, we have from the left conjunction that $T' = \singleton{v}$.\\
But then we can show from the right conjunction that $\subtypesTo{\Gamma}{\singleton{v}}{T}$, which is a contradiction.
\end{proof}

\begin{lemma}[Disjointness of values]\label{disjoint-value}
    $\subtypesTo{\Gamma\not}{\singleton{v}}{T} \Rightarrow \disjoint{\Gamma}{\singleton{v}}{T}$
\end{lemma}
\begin{proof}
The proof of the proposition is by contradiction.\\
Suppose that ${\subtypesTo{\Gamma\not}{\singleton{v}}{T}}$, but ${\exists T' : \subtypesTo{\Gamma}{T'}{\singleton{v}} \wedge \subtypesTo{\Gamma}{T'}{T}}$.\\
Since $\singleton{v}$ can only subtype to itself by \tsc{ST-Refl}, we have from the left conjunction that $T' = \singleton{v}$.\\
But then we can show from the right conjunction that $\subtypesTo{\Gamma}{\singleton{v}}{T}$, which is a contradiction.
\end{proof}

\subsection{Permutation}
\begin{lemma}[Permutation]\label{permutation}
If $\envJ{\Gamma}$, $\envJ{\Gamma'}$, and $\Gamma'$ is a permutation of $\Gamma$, then
\begin{enumerate}
    \item If $\typeJ{\Gamma}{T}$, then $\typeJ{\Gamma'}{T}$.
    \item If $\subtypesTo{\Gamma}{T}{T'}$, then $\subtypesTo{\Gamma'}{T}{T'}$.
    \item If $\typesTo{\Gamma}{t}{T}$, then $\typesTo{\Gamma'}{t}{T}$.
\end{enumerate}
\end{lemma}
\begin{proof}
(1) and (2) are proven simultaneously by induction on the two derivations $\typeJ{\Gamma}{T}$ and $\subtypesTo{\Gamma}{T}{T'}$.
\begin{enumerate}
    \item By induction on the derivation of $\typeJ{\Gamma}{T}$.
        \begin{itemize}
            \item \tsc{Type-Int}, \tsc{Type-Bool}, \tsc{Type-Unit}, \tsc{Type-String}, \tsc{Type-Bytes}\\
                Immediate from using the same rule with $\envJ{\Gamma'}$ in the premise.
            \item \tsc{Type-TV}\\
                where $\Gamma = \Gamma_1, \subtypes{X}{T}, \Gamma_2$\\
                If $\Gamma'$ is a permutation of $\Gamma$, then there must be some $\Gamma_1'$ and $\Gamma_2'$
                such that $\Gamma' = \Gamma_1', \subtypes{X}{T}, \Gamma_2$.
                The result then follows from \tsc{Type-TV}.
            \item \tsc{Type-SR}, \tsc{Type-Chan}, \tsc{Type-Abs}, \tsc{Type-Rec}, \tsc{Type-List}, \tsc{Type-Union}, \tsc{Type-Match}
                The result follows from the I.H. and the relevant rule.
            \item \tsc{Type-TAbs}
                \begin{equation}
                    \derive
                        {\typeJ{\Gamma, \subtypes{X}{T_1}}{T_2}}
                        {\typeJ{\Gamma}{\quant{X}{T_1} T_2}}
                \end{equation}
                where $T = \quant{X}{T_1} T_2$\\
                Before we can use the I.H., we must show that $\Gamma', \subtypes{X}{T_1}$ is indeed well-formed.
                The only rule that could have been used for well-formedness of the environment $\Gamma, \subtypes{X}{T_1}$ is \tsc{Env-TV}:
                \begin{equation}
                    \derive
                        {\typeJ{\Gamma}{T_1} \quad X \notin dom(\Gamma)}
                        {\envJ{\Gamma, \subtypes{X}{T_1}}}
                \end{equation}
                By the I.H., $\typeJ{\Gamma'}{T_1}$. If $\Gamma'$ is a permutation of $\Gamma$ and $X \notin dom(\Gamma)$, then $X \notin dom(\Gamma')$,
                so we can conclude by \tsc{Env-TV} that $\envJ{\Gamma', \subtypes{X}{T_1}}$.
                Then, we can use the I.H. with environment $(\Gamma', \subtypes{X}{T_1})$ to show $\typeJ{\Gamma', \subtypes{X}{T_1}}{T_2}$.
                The conclusion then follows from \tsc{Type-TAbs}.
            \item \tsc{Type-TApp}
                \begin{adjustwidth}{-2cm}{2cm}
                \begin{gather}
                    \derive
                        {\typeJ{\Gamma}{T_1} \quad \typeJ{\Gamma}{T_2}\\
                        \eqtypesTo{\Gamma}{T_1}{\quant{X}{T_{11}} T_{12}} \quad \typeJ{\Gamma}{\quant{X}{T_{11}} T_{12}} \quad \subtypesTo{\Gamma}{T_2}{T_{11}}}
                        {\typeJ{\Gamma}{T_1\ T_2}}
                \end{gather}
                \end{adjustwidth}
                where $T = T_1\ T_2$\\
                By the I.H., $\typeJ{\Gamma'}{T_1}$, $\typeJ{\Gamma'}{T_2}$ and $\typeJ{\Gamma'}{\quant{X}{T_{11}} T_{12}}$.
                By part (2) of the lemma, $\eqtypesTo{\Gamma'}{T_1}{\quant{X}{T_{11}} T_{12}}$ and $\subtypesTo{\Gamma'}{T_2}{T_{11}}$.
                The result follows from \tsc{Type-TApp}.
        \end{itemize}
    \item By induction on the derivation of $\subtypesTo{\Gamma}{T}$.
        \begin{itemize}
            \item \tsc{ST-Refl}, \tsc{ST-Val}, \tsc{ST-$\top$}, \tsc{ST-$\cup$Comm}, \tsc{ST-$\cup$}\\
                Immediate by using the same rule with environment $\Gamma'$.
            \item \tsc{ST-Var}\\
                where $X = T \quad \Gamma = \Gamma_1, \subtypes{X}{T'}, \Gamma_2$\\
                If $\Gamma'$ is a permutation of $\Gamma$, then there must be some $\Gamma_1'$ and $\Gamma_2'$
                such that $\Gamma' = \Gamma_1', \subtypes{X}{T'}, \Gamma_2$.
                The result then follows from \tsc{ST-Var}.
            \item \tsc{ST-Trans}, \tsc{ST-$\cup$L}, \tsc{ST-$\cup$R}, \tsc{ST-Rec}, \tsc{ST-List},
                  \tsc{ST-Match2}, \tsc{ST-Abs}, \tsc{ST-App}\\
                The result follows from applying the I.H. and the relevant rule.
            \item \tsc{ST-Match1}
                \begin{equation}
                    \derive
                        {\subtypesTo{\Gamma}{T_s}{T_k} \quad \forall i \in I : i < k \Rightarrow \disjoint{\Gamma}{T_s}{T_i}}
                        {\eqtypesTo{\Gamma}{\match{T_s}{T_i \Rightarrow T_i'}{i \in I}}{T_k'}}
                \end{equation}
                where $T = \match{T_s}{T_i \Rightarrow T_i'}{i \in I} \quad T' = T_k'$\\
                By the I.H., $\subtypesTo{\Gamma'}{T_s}{T_k}$. If we expand the disjointness premise, we get
                \begin{equation}
                    \forall i \in I: i < k \Rightarrow (\not\exists T_s' : \typeJ{\Gamma}{T_s'} \wedge \subtypesTo{\Gamma}{T_s'}{T_s} \wedge \subtypesTo{\Gamma}{T_s'}{T_i})
                \end{equation}
                By part (1) of the lemma, we have that $\typeJ{\Gamma'}{T_s'}$
                Now, by the I.H., we have
                \begin{align}
                    \forall i \in I: i < k \Rightarrow (\not\exists T_s' : & \typeJ{\Gamma'}{T_s'}            \nonumber\\
                                                                   \wedge\ & \subtypesTo{\Gamma'}{T_s'}{T_s}  \\
                                                                   \wedge\ & \subtypesTo{\Gamma'}{T_s'}{T_i}) \nonumber
                \end{align}
                The result follows from \tsc{ST-Match1}.
            \item \tsc{ST-TAbs}\\
                \begin{equation}
                    \derive
                        {\subtypesTo{\Gamma}{T_1'}{T_1} \quad \subtypesTo{\Gamma, \subtypes{X}{T_1'}}{T_2}{T_2'}}
                        {\subtypesTo{\Gamma}{(\quant{X}{T_1} T_2)}{(\quant{X}{T_1'} T_2')}}
                \end{equation}
                where $T = \quant{X}{T_1} T_2 \quad T' = \quant{X}{T_1'} T_2'$\\
                By the I.H., we have $\subtypesTo{\Gamma'}{T_1'}{T_1}$.
                Before we can use the I.H. on the right premise, we must show that $\Gamma', \subtypes{X}{T_1'}$ is indeed well-formed.
                The only rule that could have been used for well-formedness of the environment $\Gamma, \subtypes{X}{T_1'}$ is \tsc{Env-TV}:
                \begin{equation}
                    \derive
                        {\typeJ{\Gamma}{T_1'} \quad X \notin dom(\Gamma)}
                        {\envJ{\Gamma, \subtypes{X}{T_1'}}}
                \end{equation}
                By the I.H., $\typeJ{\Gamma'}{T_1'}$. If $\Gamma'$ is a permutation of $\Gamma$ and $X \notin dom(\Gamma)$, then $X \notin dom(\Gamma')$,
                so we can conclude by \tsc{Env-TV} that $\envJ{\Gamma', \subtypes{X}{T_1'}}$.
                Now we can use the I.H. to show that $\subtypesTo{\Gamma', \subtypes{X}{T_1'}}{T_2}{T_2'}$,
                after which the result follows from \tsc{ST-TAbs}.
        \end{itemize}
    \item By induction on the derivation of $\typesTo{\Gamma}{t}{T}$.
        \begin{itemize}
            \item \tsc{T-Val}\\
                Immediate by using \tsc{T-Val} with environment $\Gamma'$.
            \item \tsc{T-Rec}, \tsc{T-List}, \tsc{T-Head}, \tsc{T-Tail}, \tsc{T-App},
                  \tsc{T-Field}, \tsc{T-OpC}, \tsc{T-OpR}, \tsc{T-OpI},
                  \tsc{T-OpM}, \tsc{T-OpD}\\
                The result follows from applying the I.H. and the relevant rule.
            \item \tsc{T-Var}\\
                where $\Gamma = \Gamma_1, \types{x}{T}, \Gamma_2$\\
                If $\Gamma'$ is a permutation of $\Gamma$, there must be a $\Gamma_1'$ and $\Gamma_2'$ such that
                $\Gamma' = \Gamma_1', \types{x}{T}, \Gamma_2'$.
                The result then follows from \tsc{T-Var}.
            \item \tsc{T-Let}\\
                \begin{equation}
                    \derive
                        {\typesTo{\Gamma}{t_0}{T_0} \quad \typesTo{\Gamma, \typesTo{x}{T_0}}{t_1}{T_1}}
                        {\typesTo{\Gamma}{\letExp{x}{t_0} t_1}{T_1}}
                \end{equation}
                where $t = \letExp{x}{t_0} t_1 \quad T = T_1$\\
                By the I.H., $\typesTo{\Gamma'}{t_0}{T_0}$.
                Before we can use the I.H. on the right premise, we must show that $\Gamma', \typesTo{x}{T_0}$ is indeed well-formed.
                The only rule that could have been used for well-formedness of the environment $\Gamma, \subtypes{x}{T_0}$ is \tsc{Env-V}:
                \begin{equation}
                    \derive
                        {\typeJ{\Gamma}{T_0} \quad x \notin dom(\Gamma)}
                        {\envJ{\Gamma, \types{x}{T_0}}}
                \end{equation}
                By the I.H., $\typeJ{\Gamma'}{T_0}$. If $\Gamma'$ is a permutation of $\Gamma$ and $X \notin dom(\Gamma)$, then $X \notin dom(\Gamma')$,
                so we can conclude by \tsc{Env-V} that $\envJ{\Gamma', \types{x}{T_0}}$.
                Now we can use the I.H. to show that $\typesTo{\Gamma', \types{x}{T_0}}{t_1}{T_1}$,
                after which the result follows from \tsc{T-Let}.
            \item \tsc{T-Abs}\\
                The approach for this case is the same as the one for the right premise in \tsc{T-Let}.
            \item \tsc{T-TAbs}\\
                The approach for this case is analogous to the one for \tsc{T-Abs},
                the difference being that \tsc{Env-TV} is used in the well-formedness derivation instead of \tsc{Env-V}.
            \item \tsc{T-TApp}
                \begin{equation}
                    \derive
                        {\typesTo{\Gamma}{t_1}{\quant{X}{T_1}T_0} \quad \subtypesTo{\Gamma}{T_2}{T_1}}
                        {\typesTo{\Gamma}{t_1\ T_2}{(\quant{X}{T_1 T_0})\ T_2}}
                \end{equation}
                where $t = t_1\ T_2 \quad T = (\quant{X}{T_1 T_0})\ T_2$\\
                By the I.H. on the left premise, we have $\typesTo{\Gamma'}{t_1}{\quant{X}{T_1}T_0}$.
                By part (2) of the lemma, we have $\subtypesTo{\Gamma'}{T_2}{T_1}$.
                The result then follows from \tsc{T-TApp}.
            \item \tsc{T-Match}\\
                \begin{equation}
                    \derive
                        {\typesTo{\Gamma}{t_s}{T_s} \quad \forall i \in I : \typesTo{\Gamma, x_i : T_i}{t_i}{T_i'} \quad \subtypesTo{\Gamma}{T_s}{\cup_{i \in I} T_i}}
                        {\typesTo{\Gamma}{\match{t_s}{x_i : T_i \Rightarrow t_i}{i \in I}}{\match{T_s}{T_i \Rightarrow T_i'}{i \in I}}}
                \end{equation}
                where $t = \match{t_s}{x_i : T_i \Rightarrow t_i}{i \in I} \quad T = \match{T_s}{T_i \Rightarrow T_i'}{i \in I}$\\
                By the I.H., $\typesTo{\Gamma'}{t_s}{T_s}$.
                Using the approach outlined in the previous cases, we show that $\envJ{\Gamma', x_i : T_i}$.
                We then use the I.H. to show $\forall i \in I : \typesTo{\Gamma', x_i : T_i}{t_i}{T_i'}$.
                By part (2) of the lemma, we have $\subtypesTo{\Gamma'}{T_s}{\cup_{i \in I} T_i}$,
                after which the result follows from \tsc{T-Match}.
        \end{itemize}
\end{enumerate}
\end{proof}

\subsection{Weakening}
\begin{lemma}[Weakening]\label{weakening}\hfill
\begin{enumerate}
    \item If $\typeJ{\Gamma}{T}$ and $\envJ{\Gamma, \subtypes{X}{T'}}$, then $\typeJ{\Gamma, \subtypes{X}{T'}}{T}$.
    \item If $\typeJ{\Gamma}{T}$ and $\envJ{\Gamma, \types{x}{T'}}$, then $\typeJ{\Gamma, \subtypes{x}{T'}}{T}$.
    \item If $\subtypesTo{\Gamma}{T}{T'}$ and $\envJ{\Gamma, \subtypes{X}{T''}}$, then $\subtypesTo{\Gamma, \subtypes{X}{T''}}{T}{T'}$.
    \item If $\subtypesTo{\Gamma}{T}{T'}$ and $\envJ{\Gamma, \types{x}{T''}}$, then $\subtypesTo{\Gamma, \types{x}{T''}}{T}{T'}$.
    \item If $\typesTo{\Gamma}{t}{T}$ and $\envJ{\Gamma, \subtypes{X}{T'}}$, then $\typesTo{\Gamma, \subtypes{X}{T'}}{t}{T}$.
    \item If $\typesTo{\Gamma}{t}{T}$ and $\envJ{\Gamma, \types{x}{T'}}$, then $\typesTo{\Gamma, \types{x}{T'}}{t}{T}$.
\end{enumerate}
\end{lemma}
\begin{proof}
(1) and (3) are proven simultaneously by induction on the two derivations $\typeJ{\Gamma}{T}$ and $\subtypesTo{\Gamma}{T}{T'}$.
\begin{enumerate}
    \item By induction on the derivation of $\typeJ{\Gamma}{T}$.
        \begin{itemize}
            \item \tsc{Type-Int}, \tsc{Type-Bool}, \tsc{Type-Unit}, \tsc{Type-String}, \tsc{Type-Bytes}\\
                The conclusion is immediate by using $\envJ{\Gamma, \subtypes{X}{T'}}$ with the relevant rule.
            \item \tsc{Type-TV}
                \begin{equation}
                    \derive
                        {\envJ{\Gamma', \subtypes{X'}{T''}}}
                        {\typeJ{\Gamma', \subtypes{X'}{T''}}{X'}}
                \end{equation}
                where $T = X' \quad \Gamma = \Gamma', \subtypes{X'}{T''}$\\
                We know that $X \neq X'$, as otherwise $(\Gamma', \subtypes{X'}{T''}, \subtypes{X}{T'})$ would not be well-formed.
                Then, we can simply use the same rule \tsc{Type-TV} with the new environment to get the desired result.
            \item \tsc{Type-SR}, \tsc{Type-Chan}, \tsc{Type-Abs}, \tsc{Type-Rec}, \tsc{Type-List}, \tsc{Type-Union}, \tsc{Type-Match}\\
                The result follows from the I.H. on the premises and the relevant rule.
            \item \tsc{Type-TAbs}
                \begin{equation}
                    \derive
                        {\typeJ{\Gamma, \subtypes{X'}{T_1}}{T_2}}
                        {\typeJ{\Gamma}{\quant{X'}{T_1} T_2}}
                \end{equation}
                where $T = X'$\\
                By the I.H. (using environment $(\Gamma, \subtypes{X'}{T_1})$), we have ${\typeJ{\Gamma, \subtypes{X'}{T_1}, \subtypes{X}{T'}}{T_2}}$,
                from which we can use \tsc{Type-TAbs} to get the result.
            \item \tsc{Type-TApp}\\
                \begin{adjustwidth}{-2cm}{2cm}
                \begin{gather}
                    \derive
                        {\typeJ{\Gamma}{T_1} \quad \typeJ{\Gamma}{T_2}\\
                         \eqtypesTo{\Gamma}{T_1}{\quant{X'}{T_{11}} T_{12}} \quad \typeJ{\Gamma}{\quant{X'}{T_{11}} T_{12}} \quad \subtypesTo{\Gamma}{T_2}{T_{11}}}
                        {\typeJ{\Gamma}{T_1\ T_2}}
                \end{gather}
                \end{adjustwidth}
                where $T = T_1\ T_2$\\
                By the I.H., we have ${\typeJ{\Gamma, \subtypes{X}{T'}}{T_1}}$, ${\typeJ{\Gamma, \subtypes{X}{T'}}{T_2}}$
                and ${\typeJ{\Gamma, \subtypes{X}{T'}}{\quant{X'}{T_{11}} T_{12}}}$.
                Now, from part (3) of the lemma, we get ${\eqtypesTo{\Gamma, \subtypes{X}{T'}}{T_1}{\quant{X'}{T_{11}} T_{12}}}$
                and ${\subtypesTo{\Gamma, \subtypes{X}{T'}}{T_2}{T_{11}}}$.
                Then we simply use \tsc{Type-TApp} to obtain the conclusion.
        \end{itemize}
    \item None of the type well-formedness rules relate to typing assumptions, so the conclusion trivially holds in all cases.
    \item By induction on the derivation of $\subtypesTo{\Gamma}{T}{T'}$.
        \begin{itemize}
            \item \tsc{ST-Refl}, \tsc{ST-$\top$}, \tsc{ST-Val}, \tsc{ST-$\cup$Comm}, \tsc{ST-$\cup$Assoc}\\
                Using the applicable rule with the new context $(\Gamma, \subtypes{X}{T''})$ gives the conclusion directly.
            \item \tsc{ST-Var}\\
                where $T = X' \quad \Gamma = \Gamma', \subtypes{X'}{T'}$\\
                We know that $X \neq X'$, as otherwise $(\Gamma', \subtypes{X'}{T'}, \subtypes{X}{T''})$ would not be well-formed.
                Then, we can simply use the same rule \tsc{ST-Var} with the new environment to get the desired result.
            \item \tsc{ST-Trans}, \tsc{ST-$\cup$L}, \tsc{ST-$\cup$R}, \tsc{ST-Rec}, \tsc{ST-List},
                  \tsc{ST-Match2}, \tsc{ST-Abs}, \tsc{ST-App}\\
                The conclusion follows from applying the I.H. to the premises, and then applying the relevant rule.
            \item \tsc{ST-Match1}\\
                where $T = T_s\ttt{ match }\curl{T_i \Rightarrow T_i'}_{i \in I} \quad T' = T_k'$\\
                From the I.H., $\subtypesTo{\Gamma, \subtypes{X}{T''}}{T_s}{T_k}$. As for the second premise, it can be expanded top
                \begin{equation}
                    \forall i \in I: i < k \Rightarrow (\not\exists T_s' : \typeJ{\Gamma}{T_s'} \wedge \subtypesTo{\Gamma}{T_s'}{T_s} \wedge \subtypesTo{\Gamma}{T_s'}{T_i})
                \end{equation}
                By part (1) of the lemma, we have that $\typeJ{\Gamma, \subtypes{X}{T''}}{T_s'}$
                Now, by the I.H., we have
                \begin{align}
                    \forall i \in I: i < k \Rightarrow (\not\exists T_s' : & \typeJ{\Gamma, \subtypes{X}{T''}}{T_s'}            \nonumber\\
                                                                   \wedge\ & \subtypesTo{\Gamma, \subtypes{X}{T''}}{T_s'}{T_s}  \\
                                                                   \wedge\ & \subtypesTo{\Gamma, \subtypes{X}{T''}}{T_s'}{T_i}) \nonumber
                \end{align}
                The conclusion follows from \tsc{ST-Match1}.
            \item \tsc{ST-TAbs}\\
                where $T = \quant{X'}{T_1} T_2 \quad T' = \quant{X'}{T_1'} T_2'$\\
                By the I.H., ${\subtypesTo{\Gamma, \subtypes{X}{T''}}{T_1'}{T_1}}$.
                Also by the I.H. with environment $(\Gamma, \subtypes{X'}{T_1'})$,
                we have ${\subtypesTo{\Gamma, \subtypes{X'}{T_1'}, \subtypes{X}{T''}}{T_2}{T_2'}}$.
                Then by \cref{permutation} we also have
                ${\subtypesTo{\Gamma, \subtypes{X}{T''}, \subtypes{X'}{T_1'}}{T_2}{T_2'}}$.
                The result follows from \tsc{ST-TAbs}.
        \end{itemize}
    \item None of the subtyping rules relate to typing assumptions, so the conclusion trivially holds in all cases.
    \item By induction on $\typesTo{\Gamma}{t}{T}$:
        \begin{itemize}
            \item \tsc{T-Val}
                \begin{equation}
                    \derive
                        {v\tx{ is not an abstraction}}
                        {\typesTo{\Gamma}{v}{\singleton{v}}}
                \end{equation}
                where $t = v \quad T = \singleton{v}$\\
                The conclusion is obtained directly from using the context $(\Gamma, \subtypes{X}{T'})$ with \tsc{T-Val}.
            \item \tsc{T-Rec}, \tsc{T-List}, \tsc{T-Head}, \tsc{T-Tail}, \tsc{T-App},
                  \tsc{T-Field}, \tsc{T-OpC}, \tsc{T-OpR}, \tsc{T-OpI},
                  \tsc{T-OpM}, \tsc{T-OpD}\\
                The conclusion follows from applying the I.H. on the premises and then applying the relevant rule.
            \item \tsc{T-Var}
                where $t = x \quad \Gamma = \Gamma', \types{x}{T}$\\
                We take the environment to be $(\Gamma', \types{x}{T}, \subtypes{X}{T'})$
                and get the conclusion directly from \tsc{T-Var}.
            \item \tsc{T-Let}
                \begin{equation}
                    \derive
                        {\typesTo{\Gamma}{t_0}{T_0} \quad \typesTo{\Gamma,\types{x}{T_0}}{t_1}{T_1}}
                        {\typesTo{\Gamma}{\letExp{x}{t_0} t_1}{T_1}}
                \end{equation}
                where $t = \letExp{x}{t_0}{t_1} \quad T = T_1$\\
                By applying the I.H. (using the environment $(\Gamma, \types{x}{T_0})$ in the second premise),
                we have $\typesTo{\Gamma, \subtypes{X}{T'}}{t_0}{T_0}$ and $\typesTo{\Gamma, \types{x}{T_0}, \subtypes{X}{T'}}{t_1}{T_1}$.
                From \cref{permutation}, we have $\typesTo{\Gamma, \subtypes{X}{T'}, \types{x}{T_0}}{t_1}{T_1}$,
                which lets us use \tsc{T-Let} to obtain the result.
            \item \tsc{T-Abs}
                \begin{equation}
                    \derive
                        {\typesTo{\Gamma,\types{x}{T_1}}{t_0}{T_2}}
                        {\typesTo{\Gamma}{\abs{x}{T_1} t_0}{T_1 \rightarrow T_2}}
                \end{equation}
                where $t = \abs{x}{T_1} t_0 \quad T = T_1 \rightarrow T_2$\\
                By applying the I.H. with environment $(\Gamma, \types{x}{T_1})$,
                we get $\typesTo{\Gamma, \types{x}{T_1}, \subtypes{X}{T'}}{t_0}{T_2}$.
                From \cref{permutation}, we have $\typesTo{\Gamma, \subtypes{X}{T'}, \types{x}{T_1}}{t_0}{T_2}$,
                which lets us use \tsc{T-Abs} to obtain the result.
            \item \tsc{T-TAbs}
                \begin{equation}
                    \derive
                        {\typesTo{\Gamma,\subtypes{X'}{T_1}}{t_0}{T_2}}
                        {\typesTo{\Gamma}{(\tabs{X'}{T_1} t_0)}{(\quant{X'}{T_1} T_2)}}
                \end{equation}
                where $t = \tabs{X'}{T_1} t_0 \quad T = \quant{X'}{T_1} T_2$\\
                After applying the I.H. (using the environment $(\Gamma,\subtypes{X'}{T_1})$),
                we get $\typesTo{\Gamma, \subtypes{X'}{T_1}, \subtypes{X}{T'}}{t_0}{T_2}$.
                Then by \cref{permutation}
                we also have $\typesTo{\Gamma, \subtypes{X}{T'}, \subtypes{X'}{T_1}}{t_0}{T_2}$.
                The result follows from \tsc{T-TAbs}.
            \item \tsc{T-TApp}
                \begin{equation}
                    \derive
                        {\typesTo{\Gamma}{t_1}{\quant{X'}{T_1} T_0} \quad \subtypesTo{\Gamma}{T_2}{T_1}}
                        {\typesTo{\Gamma}{t_1\ T_2}{(\quant{X'}{T_1} T_0)\ T_2}}
                \end{equation}
                where $t = t_1\ T_2 \quad T = (\quant{X}{T_1} T_0)\ T_2$\\
                By the I.H. on the left premise, we have $\typesTo{\Gamma, \subtypes{X}{T'}}{t_1}{\quant{X'}{T_1} T_0}$.
                By part (3) of the lemma on the right premise, we have $\subtypesTo{\Gamma, \subtypes{X}{T'}}{T_2}{T_1}$.
                The result then follows from \tsc{T-TApp}.
            \item \tsc{T-Match}
                \begin{equation}
                    \derive
                        {\typesTo{\Gamma}{t_s}{T_s} \quad \forall i \in I: \typesTo{\Gamma, \types{x_i}{T_i}}{t_i}{T_i'} \quad \subtypesTo{\Gamma}{T_s}{\cup_{i \in I} T_i}}
                        {\typesTo{\Gamma}{t_s \ttt{ match }\curl{x_i : T_i \Rightarrow t_i}_{i \in I}}{T_s \ttt{ match } \curl{T_i \Rightarrow T_i'}_{i \in I}}}
                \end{equation}
                where $t = t_s \ttt{ match }\curl{x_i : T_i \Rightarrow t_i}_{i \in I} \quad T_s \ttt{ match } \curl{T_i \Rightarrow T_i'}_{i \in I}$\\
                By applying the I.H. (using the environment $(\Gamma, \types{x_i}{T_i})$ in the second premise),
                we have $\typesTo{\Gamma, \subtypes{X}{T'}}{t_s}{T_s}$ and $\forall i \in I: \typesTo{\Gamma, \types{x_i}{T_i}, subtypes{X}{T'}}{t_i}{T_i'}$.
                By \cref{permutation}, we have $\forall i \in I: \typesTo{\Gamma, subtypes{X}{T'}, \types{x_i}{T_i}}{t_i}{T_i'}$,
                and by part (3) of the lemma, we have $\subtypesTo{\Gamma, \subtypesTo{X}{T'}}{T_s}{\cup_{i \in I} T_i}$,
                after which we can use \tsc{T-Match} to obtain the result.
            \item \tsc{T-Sub}
                \begin{equation}
                    \derive
                        {\typesTo{\Gamma}{t}{T''} \quad \subtypesTo{\Gamma}{T''}{T}}
                        {\typesTo{\Gamma}{t}{T}}
                \end{equation}
                By the I.H., we have ${\typesTo{\Gamma, \subtypes{X}{T'}}{t}{T''}}$.
                From part (3) of the lemma, we also have ${\subtypesTo{\Gamma, \subtypes{X}{T'}}{T''}{T}}$.
                Then we get the conclusion by applying \tsc{T-Sub}.
        \end{itemize}
    \item The proof is by induction on $\typesTo{\Gamma}{t}{T}$:
        \begin{itemize}
            \item \tsc{T-Val}\\
                where $t = v \quad T = \singleton{v}$\\
                The conclusion is obtained directly from using the context $(\Gamma, \subtypes{x}{T'})$ with \tsc{T-Val}.
            \item \tsc{T-Rec}, \tsc{T-List}, \tsc{T-Head}, \tsc{T-Tail}, \tsc{T-App},
                  \tsc{T-Field}, \tsc{T-OpC}, \tsc{T-OpR}, \tsc{T-OpI},
                  \tsc{T-OpM}, \tsc{T-OpD}\\
                The conclusion follows from applying the I.H. on the premises and then applying the relevant rule.
            \item \tsc{T-Var}\\
                where $t = x' \quad \Gamma = \Gamma', \typesTo{x}{T}$\\
                We take the environment to be $(\Gamma', \types{x}{T}, \types{x}{T'})$
                and get the conclusion directly from \tsc{T-Var}.
            \item \tsc{T-Let}\\
                \begin{equation}
                    \derive
                        {\typesTo{\Gamma}{t_0}{T_0} \quad \typesTo{\Gamma,\types{x'}{T_0}}{t_1}{T_1}}
                        {\typesTo{\Gamma}{\letExp{x'}{t_0} t_1}{T_1}}
                \end{equation}
                where $t = \letExp{x'}{t_0}{t_1} \quad T = T_1$\\
                By applying the I.H. (using the environment $(\Gamma, \types{x'}{T_0})$ in the second premise),
                we have $\typesTo{\Gamma, \types{x}{T'}}{t_0}{T_0}$ and $\typesTo{\Gamma,\types{x'}{T_0},\types{x}{T'}}{t_1}{T_1}$.
                By \cref{permutation}, we have $\typesTo{\Gamma,\types{x}{T'},\types{x'}{T_0}}{t_1}{T_1}$,
                from which we can use \tsc{T-Let} to obtain the conclusion.
            \item \tsc{T-Abs}
                \begin{equation}
                    \derive
                        {\typesTo{\Gamma,\types{x'}{T_1}}{t}{T_2}}
                        {\typesTo{\Gamma}{\abs{x'}{T_1} t}{T_1 \rightarrow T_2}}
                \end{equation}
                where $t = \abs{x'}{T_1} t_0 \quad T = T_1 \rightarrow T_2$\\
                By applying the I.H. with environment $(\Gamma, \types{x'}{T_1})$,
                we have $\typesTo{\Gamma,\types{x'}{T_1},\types{t}{T'}}{t}{T_2}$.
                By \cref{permutation}, we have $\typesTo{\Gamma,\types{x}{T'},\types{x'}{T_1}}{t}{T_2}$,
                from which we can use \tsc{T-Abs} to obtain the conclusion.
            \item \tsc{T-TAbs}
                \begin{equation}
                    \derive
                        {\typesTo{\Gamma,\subtypes{X}{T_1}}{t}{T_2}}
                        {\typesTo{\Gamma}{(\tabs{X}{T_1} t)}{(\quant{X}{T_1} T_2)}}
                \end{equation}
                where $t = \tabs{X}{T_1} t_0 \quad T = \quant{X}{T_1} T_2$\\
                After applying the I.H. (using the environment $(\Gamma,\subtypes{X}{T_1})$),
                we have $\typesTo{\Gamma,\subtypes{X}{T_1},\types{x}{T'}}{t}{T_2}$.
                By \cref{permutation}, we have $\typesTo{\Gamma,\types{x}{T'},\subtypes{X}{T_1}}{t}{T_2}$,
                from which we can use \tsc{T-TAbs} to obtain the conclusion.
            \item \tsc{T-TApp}
                \begin{equation}
                    \derive
                        {\typesTo{\Gamma}{t_1}{\quant{X}{T_1} T_0} \quad \subtypesTo{\Gamma}{T_2}{T_1}}
                        {\typesTo{\Gamma}{t_1\ T_2}{(\quant{X}{T_1} T_0)\ T_2}}
                \end{equation}
                where $t = t_1\ T_2 \quad T = (\quant{X}{T_1} T_0)\ T_2$\\
                By the I.H. on the left premise, we have $\typesTo{\Gamma, \types{x}{T'}}{t_1}{\quant{X}{T_1} T_0}$.
                By part (4) of the lemma on the right premise, we have $\subtypesTo{\Gamma, \subtypes{x}{T'}}{T_2}{T_1}$.
                The result then follows from \tsc{T-TApp}.
            \item \tsc{T-Match}
                \begin{equation}
                    \derive
                        {\typesTo{\Gamma}{t_s}{T_s} \quad \forall i \in I: \typesTo{\Gamma, \types{x_i}{T_i}}{t_i}{T_i'} \quad \subtypesTo{\Gamma}{T_s}{\cup_{i \in I} T_i}}
                        {\typesTo{\Gamma}{t_s \ttt{ match }\curl{x_i : T_i \Rightarrow t_i}_{i \in I}}{T_s \ttt{ match } \curl{T_i \Rightarrow T_i'}_{i \in I}}}
                \end{equation}
                where $t = t_s \ttt{ match }\curl{x_i : T_i \Rightarrow t_i}_{i \in I} \quad T_s \ttt{ match } \curl{T_i \Rightarrow T_i'}_{i \in I}$\\
                By applying the I.H. (using the environment $(\Gamma, \types{x_i}{T_i})$ in the second premise),
                we have $\typesTo{\Gamma, \types{x}{T'}}{t_s}{T_s}$ and $\forall i \in I: \typesTo{\Gamma, \types{x_i}{T_i}, \types{x}{T'}}{t_i}{T_i'}$.
                By \cref{permutation}, we have $\forall i \in I: \typesTo{\Gamma, \types{x}{T'}, \types{x_i}{T_i}}{t_i}{T_i'}$,
                and by part (4) of the lemma, we have $\subtypesTo{\Gamma, \types{x}{T'}}{T_s}{\cup_{i \in I} T_i}$,
                after which we can use \tsc{T-Match} to obtain the conclusion.
            \item \tsc{T-Sub}
                \begin{equation}
                    \derive
                        {\typesTo{\Gamma}{t}{T''} \quad \subtypesTo{\Gamma}{T''}{T}}
                        {\typesTo{\Gamma}{t}{T}}
                \end{equation}
                By the I.H., we have ${\typesTo{\Gamma, \types{x}{T'}}{t}{T''}}$.
                From part (4) of the lemma, we also have ${\subtypesTo{\Gamma, \types{x}{T'}}{T''}{T}}$.
                Then we get the conclusion by applying \tsc{T-Sub}.
        \end{itemize}
\end{enumerate}
\end{proof}

\subsection{Strengthening}
\begin{theorem}[Strengthening]
\label{strengthening}
If $\subtypesTo{\Gamma, \typesTo{x}{T'}}{T}{T'}$, then $\subtypesTo{\Gamma}{T}{T'}$
\end{theorem}
\begin{proof}
Immediate from inspection of the subtyping rules, as none of them are related to typing of term-level variables.
\end{proof}

\subsection{Variable absence}
\begin{lemma}[Variable absence]\label{var-absence}
If $\envJ{\Gamma, \subtypes{X}{T}, \Gamma'}$, then $X$ does not appear anywhere in $\Gamma$.
\end{lemma}
\begin{proof}
By induction on the derivation of $\envJ{\Gamma, \subtypes{X}{T}, \Gamma'}$.
\begin{itemize}
    \item \tsc{Env-$\emptyset$}\\
        Could not have been used, as the environment is non-empty.
    \item \tsc{Env-V}\\
        \begin{equation}
            \derive
                {\typeJ{\Gamma, \subtypes{X}{T}, \Gamma''}{T'} \quad x \notin dom(\Gamma)}
                {\envJ{\Gamma, \subtypes{X}{T}, \Gamma'', \types{x}{T'}}}
        \end{equation}
        where $\Gamma' = \Gamma'', \types{x}{T''}$\\
        By the definition of the type well-formedness judgement, we have that $\envJ{\Gamma, \subtypes{X}{T}, \Gamma''}$.
        Then, by the I.H., we know that $X$ does not appear anywhere in $\Gamma$.
    \item \tsc{Env-TV}\\
        From the rule, we have $\envJ{\Gamma, \subtypes{X}{T}, \Gamma'}$.
        There are two possibilities:
        \begin{itemize}
            \item $\Gamma' = \emptyset$\\
                In this case, we have from the premises that $\typeJ{\Gamma}{T}$ and $X \notin dom(\Gamma)$.
                If $X$ is not in $dom(\Gamma)$, then it cannot appear \textit{anywhere} in $\Gamma$, as otherwise $X$ would be unbound.
            \item $\Gamma' = \Gamma'', \subtypes{X'}{T'}$\\
                In this case, we have from the premises that $\typeJ{\Gamma, \subtypes{X}{T}, \Gamma''}{T'}$.
                By the definition of the type well-formedness judgement, we have that $\envJ{\Gamma, \subtypes{X}{T}, \Gamma''}$.
                Then, by the I.H., we know that $X$ does not appear anywhere in $\Gamma$.
        \end{itemize}
\end{itemize}
\end{proof}

\subsection{Narrowing}
\begin{lemma}[Narrowing]\label{narrowing}\hfill
    \begin{enumerate}
        \item If $\envJ{\Gamma, \subtypes{X}{T}, \Gamma'}$ and $\subtypesTo{\Gamma}{T'}{T}$, then $\envJ{\Gamma, \subtypes{X}{T'}, \Gamma'}$.
        \item If $\typeJ{\Gamma, \subtypes{X}{T'}, \Gamma'}{T}$ and $\subtypesTo{\Gamma}{T''}{T'}$, then $\typeJ{\Gamma, \subtypes{X}{T''}, \Gamma'}{T}$.
        \item If $\subtypesTo{\Gamma, \subtypes{X}{T''}, \Gamma'}{T}{T'}$ and $\subtypesTo{\Gamma}{T'''}{T''}$, then $\subtypesTo{\Gamma, \subtypes{X}{T'''}, \Gamma'}{T}{T'}$.
        \item If $\typesTo{\Gamma, \subtypes{X}{T'}, \Gamma'}{t}$ and $\subtypesTo{\Gamma}{T''}{T'}$, then $\typesTo{\Gamma, \subtypes{X}{T''}, \Gamma'}{t}$.
    \end{enumerate}
\end{lemma}
\begin{proof}
(1), (2) and (3) are proven simultaneously by induction on the derivations
$\envJ{\Gamma, \subtypes{X}{T}, \Gamma'}$, $\typeJ{\Gamma, \subtypes{X}{T'}, \Gamma'}{T}$ and $\subtypesTo{\Gamma}{T}{T'}$.
\begin{itemize}
    \item By induction on the derivation of $\envJ{\Gamma, \subtypes{X}{T}, \Gamma'}$.
        \begin{itemize}
            \item \tsc{Env-$\emptyset$}\\
                Does not apply, as the environment is non-empty.
            \item \tsc{Env-V}\\
                \begin{equation}
                    \derive
                        {\typeJ{\Gamma, \subtypes{X}{T}, \Gamma'}{T''} \quad x \notin dom(\Gamma, \subtypes{X}{T}, \Gamma')}
                        {\envJ{\Gamma, \subtypes{X}{T}, \Gamma', \types{x}{T''}}}
                \end{equation}
                By part (2) of the lemma, we have $\typeJ{\Gamma, \subtypes{X}{T'}, \Gamma'}{T''}$.
                The domain does not change, so we simply use \tsc{Env-V} to obtain the result.
            \item \tsc{Env-TV}\\
                \begin{equation}
                    \derive
                        {\typeJ{\Gamma, \subtypes{X}{T}, \Gamma'}{T''} \quad Y \notin dom(\Gamma, \subtypes{X'}{T}, \Gamma')}
                        {\envJ{\Gamma, \subtypes{X}{T}, \Gamma', \subtypes{Y}{T''}}}
                \end{equation}
                By part (2) of the lemma, we have $\typeJ{\Gamma, \subtypes{X}{T'}, \Gamma'}{T''}$.
                The domain does not change, so we simply use \tsc{Env-V} to obtain the result.
        \end{itemize}
    \item By induction on the derivation of $\typeJ{\Gamma, \subtypes{X}{T'}, \Gamma'}{T}$.
        \begin{itemize}
            \item \tsc{Type-Val}, \tsc{Type-Int}, \tsc{Type-Bool}, \tsc{Type-Unit},
                \tsc{Type-String}, \tsc{Type-Bytes}, \tsc{Type-TV}\\
                The conclusion follows from applying part (1) of the lemma, and then applying the relevant rule.
            \item \tsc{Type-SR}, \tsc{Type-Chan}, \tsc{Type-Abs}, \tsc{Type-Rec},
                \tsc{Type-List}, \tsc{Type-Union}, \tsc{Type-Match}\\
                The conclusion follows from applying the I.H. on the premises, and then applying the relevant rule.
            \item \tsc{Type-TAbs}\\
                By the I.H. (using environment $\Gamma, \subtypes{X}{T'}, \Gamma'', \subtypes{Y}{T_1}$), we have
                $\subtypesTo{\Gamma, \subtypes{X}{T''}, \Gamma'', \subtypes{Y}{T_1}}{T_2}$.
                The result then follows from \tsc{Type-TAbs}.
            \item \tsc{Type-TApp}\\
                By the I.H., we have $\typeJ{\Gamma, \subtypes{X}{T''}, \Gamma'}{T_1}$,
                $\typeJ{\Gamma, \subtypes{X}{T''}, \Gamma'}{T_2}$ and
                $\typeJ{\Gamma, \subtypes{X}{T''}, \Gamma'}{\quant{Y}{T_{11}} T_{12}}$.
                By part (3) of the lemma, we have
                $\eqtypesTo{\Gamma, \subtypes{X}{T''}, \Gamma'}{T_1}{\quant{X}{T_{11}} T_{12}}$ and
                $\subtypesTo{\Gamma, \subtypes{X}{T''}, \Gamma'}{T_2}{T_{11}}$.
                The result then follows from \tsc{Type-TApp}.
        \end{itemize}
    \item By induction on the derivation of $\subtypesTo{\Gamma, \subtypes{X}{T''}, \Gamma'}{T}{T'}$.
        \begin{itemize}
            \item \tsc{ST-Refl}, \tsc{ST-$\top$}, \tsc{ST-Val}, \tsc{ST-$\cup$Comm}, \tsc{ST-$\cup$Assoc}\\
                The conclusion is immediate by using the relevant rule with environment $(\Gamma, \subtypesTo{X}{T'''}, \Gamma')$.
            \item \tsc{ST-Var}\\
                \begin{equation}
                    \derive
                        {\envJ{\Gamma, \subtypes{X}{T''}, \Gamma'}}
                        {\subtypesTo{\Gamma, \subtypes{X}{T''}, \Gamma'}{Y}{T'}}
                \end{equation}
                where $T = Y$\\
                There are three possible cases, depending on the environment:
                \begin{itemize}
                    \item $\Gamma = \Gamma'', \subtypes{Y}{T'}, \Gamma'''$\\
                        By part (1) of the lemma, we have $\envJ{\Gamma'', \subtypes{Y}{T'}, \Gamma''', \subtypes{X}{T'''}, \Gamma'}$.
                        The result follows from \tsc{ST-Var}.
                    \item $\Gamma' = \Gamma'', \subtypes{Y}{T'}, \Gamma'''$\\
                        Analogous to the previous case.
                    \item $X = Y \quad T' = T''$\\
                        By part (1) of the lemma, we have $\envJ{\Gamma, \subtypesTo{Y}{T'''}, \Gamma'}$.
                        By \tsc{ST-Var}, we have $\subtypesTo{\Gamma, \subtypesTo{Y}{T'''}, \Gamma'}{Y}{T'''}$.
                        By \cref{weakening} with $\subtypesTo{\Gamma}{T'''}{T''}$, we have
                        $\subtypesTo{\Gamma, \subtypes{Y}{T'''}, \Gamma'}{T'''}{T''}$.
                        Now, from \tsc{ST-Trans} we have $\subtypesTo{\Gamma, \subtypes{Y}{T'''}, \Gamma'}{Y}{T''}$,
                        which, since $X = Y$ and $T' = T''$, is the desired result.
                \end{itemize}
            \item \tsc{ST-Trans}, \tsc{ST-List}, \tsc{ST-$\cup$R}, \tsc{ST-Rec},
                \tsc{ST-Match2}, \tsc{ST-Abs}, \tsc{ST-App}\\
                The conclusion is obtained by applying the I.H. on the premises and then applying the relevant rule.
            \item \tsc{ST-Match1}\\
                \begin{equation}
                    \derive
                        {\subtypesTo{\Gamma, \subtypes{X}{T''}, \Gamma'}{T_s}{T_k} \quad \forall i \in I : i < k \Rightarrow \disjoint{\Gamma, \subtypes{X}{T''}, \Gamma'}{T_s}{T_i}}
                        {\eqtypesTo{\Gamma, \subtypes{X}{T''}, \Gamma'}{\match{T_s}{T_i \Rightarrow T_i'}{i \in I}}{T_k'}}
                \end{equation}
                where $T = \match{T_s}{T_i \Rightarrow T_i'}{i \in I} \quad T' = T_k'$\\
                By the I.H., we have $\subtypesTo{\Gamma, \subtypes{X}{T'''}, \Gamma'}{T_s}{T_k}$.
                We expand the second premise to:
                \begin{align}
                    \forall i \in I: i < k \Rightarrow (\not\exists T_s' : & \typeJ{\Gamma, \subtypes{X}{T''}, \Gamma'}{T_s'}\\
                                                                   \wedge\ & \subtypesTo{\Gamma. \subtypes{X}{T''}, \Gamma'}{T_s'}{T_s}\\
                                                                   \wedge\ & \subtypesTo{\Gamma, \subtypes{X}{T''}, \Gamma'}{T_s'}{T_i})
                \end{align}
                By part (2) of the lemma, we have $\typeJ{\Gamma, \subtypes{X}{T'''}, \Gamma'}{T_s'}$.
                Then we can use the I.H. to show:
                \begin{align}
                    \forall i \in I: i < k \Rightarrow (\not\exists T_s' : & \typeJ{\Gamma, \subtypes{X}{T'''}, \Gamma'}{T_s'}\\
                                                                   \wedge\ & \subtypesTo{\Gamma, \subtypes{X}{T'''}, \Gamma'}{T_s'}{T_s}\\
                                                                   \wedge\ & \subtypesTo{\Gamma, \subtypes{X}{T'''}, \Gamma'}{T_s'}{T_i})
                \end{align}
                The result follows from \tsc{ST-Match1}.
            \item \tsc{ST-TAbs}
                \begin{equation}
                    \derive
                        {\subtypesTo{\Gamma, \subtypes{X}{T''}, \Gamma'}{T_1'}{T_1} \quad \subtypesTo{\Gamma, \subtypes{X}{T''}, \Gamma', \subtypes{Y}{T_1'}}{T_2}{T_2'}}
                        {\subtypesTo{\Gamma, \subtypes{X}{T''}, \Gamma'}{(\quant{Y}{T_1} T_2)}{(\quant{Y}{T_1'} T_2')}}
                \end{equation}
                where $T = (\quant{Y}{T_1} T_2) \quad T' = (\quant{Y}{T_1'} T_2')$\\
                The conclusion is obtained by applying the I.H. on the premises
                (using the environment $(\Gamma, \subtypes{X}{T''}, \Gamma', \subtypes{Y}{T_1'})$ with the right premise)
                and then applying \tsc{ST-TAbs}.
        \end{itemize}
    \item By induction on the derivation of $\typesTo{\Gamma, \subtypes{X}{T'}, \Gamma'}{t}{T}$.
        \begin{itemize}
            \item \tsc{T-Val}\\
                The conclusion is immediate by using \tsc{T-Val} with the environment $\Gamma, \subtypes{X}{T'}, \Gamma'$.
            \item \tsc{T-Rec}, \tsc{T-List}, \tsc{T-Head}, \tsc{T-Tail},
                \tsc{T-App} \tsc{T-Field}, \tsc{T-Sub}, \tsc{T-OpC},
                \tsc{T-OpR}, \tsc{T-OpI}, \tsc{T-OpM}, \tsc{T-OpD}\\
                The conclusion follows by applying the I.H. to the premises and then applying the relevant rule.
            \item \tsc{T-Var}\\
                \begin{equation}
                    \derive
                        {\envJ{\Gamma, \subtypes{X}{T'}, \Gamma'}}
                        {\typesTo{\Gamma, \subtypes{X}{T'}, \Gamma'}{x}{T}}
                \end{equation}
                where $t = x$\\
                There are two subcases:
                \begin{itemize}
                    \item $\Gamma = \Gamma'', \types{x}{T}, \Gamma'''$\\
                        We apply the 1st part of the lemma to obtain $envJ{\\Gamma'', \types{x}{T}, \Gamma''', \subtypes{X}{T'}, \Gamma'}$.
                        The result then follows from \tsc{T-Var}.
                    \item $\Gamma' = \Gamma'', \types{x}{T}, \Gamma'''$\\
                        Analogous to the previous case.
                \end{itemize}
            \item \tsc{T-Let}\\
                \begin{equation}
                    \derive
                        {\typesTo{\Gamma, \subtypes{X}{T'}, \Gamma'}{t_0}{T_0} \quad \typesTo{\Gamma, \subtypes{X}{T'}, \Gamma', \types{x}{T_0}}{t_1}{T_1}}
                        {\typesTo{\Gamma, \subtypes{X}{T'}, \Gamma'}{\letExp{x}{t_0} t_1}{T_1}}
                \end{equation}
                where $t = \letExp{x}{t_0} t_1 \quad T = T_1$\\
                By applying the I.H. on the premises
                (using the environment $\Gamma, \subtypes{X}{T'}, \Gamma', \typesTo{x}{T_0}$ in the right premise)
                the result then follows from \tsc{T-Let}.
            \item \tsc{T-Abs}\\
                \begin{equation}
                    \derive
                        {\typesTo{\Gamma, \subtypes{X}{T'}, \Gamma', \types{x}{T_1}}{t_0}{T_2}}
                        {\typesTo{\Gamma, \subtypes{X}{T'}, \Gamma'}{\abs{x}{T_1} t_0}{T_1 \rightarrow T_2}}
                \end{equation}
                where $t = \abs{x}{T_1} t_0 \quad T = T_1 \rightarrow T_2$\\
                By applying the I.H. on the premise
                (using the environment $\Gamma, \subtypes{X}{T'}, \Gamma', \typesTo{x}{T_1}$)
                the result then follows from \tsc{T-Abs}.
            \item \tsc{T-TAbs}\\
                \begin{equation}
                    \derive
                        {\typesTo{\Gamma, \subtypes{X}{T'}, \Gamma', \subtypes{Y}{T_1}}{t_0}{T_2}}
                        {\typesTo{\Gamma, \subtypes{X}{T'}, \Gamma'}{\tabs{Y}{T_1} t_0}{\quant{Y}{T_1} T_2}}
                \end{equation}
                where $t = \tabs{Y}{T_1} t_0 \quad T = \quant{Y}{T_1} T_2$\\
                By applying the I.H. on the premise
                (using the environment $\Gamma, \subtypes{X}{T'}, \Gamma', \subtypes{Y}{T_1}$)
                the result then follows from \tsc{T-TAbs}.
            \item \tsc{T-TApp}
                \begin{equation}
                    \derive
                        {\typesTo{\Gamma, \subtypes{X}{T'}, \Gamma'}{t_1}{\quant{X'}{T_1} T_0} \quad \subtypesTo{\Gamma, \subtypes{X}{T'}, \Gamma'}{T_2}{T_1}}
                        {\typesTo{\Gamma, \subtypes{X}{T'}, \Gamma'}{t_1\ T_2}{(\quant{X'}{T_1} T_0)\ T_2}}
                \end{equation}
                where $t = t_1\ T_2 \quad T = (\quant{X}{T_1} T_0)\ T_2$\\
                By the I.H. on the left premise, we have $\typesTo{\Gamma, \subtypes{X}{T''}, \Gamma'}{t_1}{\quant{X'}{T_1} T_0}$.
                By part (3) of the lemma on the right premise, we have $\subtypesTo{\Gamma, \subtypes{X}{T''}, \Gamma'}{T_2}{T_1}$.
                The result then follows from \tsc{T-TApp}.
            \item \tsc{T-Match}\\
                \begin{equation}
                    \derive
                        {\typesTo{\Gamma, \subtypes{X}{T'}, \Gamma'}{t_s}{T_s} \quad \forall i \in I : \typesTo{\Gamma, \subtypes{X}{T'}, \Gamma', \types{x_i}{T_i}}{t_i}{T_i'} \\
                         \subtypesTo{\Gamma, \subtypes{X}{T'}, \Gamma'}{T_s}{\cup_{i \in I} T_i}}
                        {\typesTo{\Gamma, \subtypes{X}{T'}, \Gamma'}{\match{t_s}{x_i : T_i \Rightarrow t_i}{i \in I}}{\match{T_s}{T_i \Rightarrow T_i'}{i \in I}}}
                \end{equation}
                where $t = \match{t_s}{x_i : T_i \Rightarrow t_i}{i \in I} \quad T = \match{T_s}{T_i \Rightarrow T_i'}{i \in I}$\\
                By applying the I.H. on the first two premises
                (using the environment $\Gamma, \subtypes{X}{T'}, \Gamma', \types{x_i}{T_i}$ on each of the case derivations)
                and by part (3) of the lemma on the last premise, the result follows from \tsc{T-Match}.
        \end{itemize}
\end{itemize}
\end{proof}

\subsection{Substitution}
\begin{theorem}[Substitution]
\label{substitution}\hfill
\begin{enumerate}
    \item If ${\envJ{\Gamma, \subtypes{X}{T}, \Gamma'}}$ and ${\subtypesTo{\Gamma}{T'}{T}}$,\\
        then ${\envJ{\Gamma, \subs{\Gamma'}{X}{T'}}}$.
    \item If ${\typeJ{\Gamma, \subtypes{X}{T'}, \Gamma'}{T}}$ and ${\subtypesTo{\Gamma}{T''}{T'}}$,\\
        then ${\typeJ{\Gamma, \subs{\Gamma'}{X}{T''}}{\subs{T}{X}{T''}}}$.
    \item If ${\subtypesTo{\Gamma, \subtypes{X}{T''}, \Gamma'}{T}{T'}}$ and ${\subtypesTo{\Gamma}{T'''}{T''}}$,\\
        then ${\subtypesTo{\Gamma, \subs{\Gamma'}{X}{T'''}}{\subs{T}{X}{T'''}}{\subs{T'}{X}{T'''}}}$.
    \item If ${\typesTo{\Gamma, \subtypes{X}{T'}, \Gamma'}{t}{T}}$ and ${\subtypesTo{\Gamma}{T''}{T'}}$,\\
        then ${\typesTo{\Gamma, \subs{\Gamma'}{X}{T''}}{\subs{t}{X}{T''}}{\subs{T}{X}{T''}}}$.
    \item If ${\typesTo{\Gamma, \types{x}{T'}, \Gamma'}{t}{T}}$ and ${\typesTo{\Gamma}{v}{T'}}$,\\
        then ${\typesTo{\Gamma, \Gamma'}{\subs{t}{x}{v}}{T}}$
\end{enumerate}
\end{theorem}
\begin{proof}
(1), (2) and (3) are proven simultaneously by induction on the derivations
${\envJ{\Gamma, \subtypes{X}{T}, \Gamma'}}$, ${\typeJ{\Gamma, \subtypes{X}{T'}, \Gamma'}{T}}$
and ${\subtypesTo{\Gamma, \subtypes{X}{T''}, \Gamma'}{T}{T'}}$.
\begin{enumerate}
    \item By induction on the derivation of ${\envJ{\Gamma, \subtypes{X}{T}, \Gamma'}}$.
        \begin{itemize}
            \item \tsc{Env-$\emptyset$}\\
                Could not have been used, as the environment is not empty.
            \item \tsc{Env-V}
                \begin{equation}
                    \derive
                        {\typeJ{\Gamma', \subtypes{X}{T}, \Gamma''}{T''} \quad x \notin dom(\Gamma, \subtypes{X}{T}, \Gamma'')}
                        {\envJ{\Gamma, \subtypes{X}{T}, \Gamma'', \types{x}{T''}}}
                \end{equation}
                where $\Gamma' = \Gamma'', \types{x}{T''}$\\
                By using part (2) of the lemma, we have ${\typeJ{\Gamma, \subs{\Gamma''}{X}{T'}}{\subs{T''}{X}{T'}}}$.
                None of the term-variables in $\Gamma'$ are altered by substitution, i.e.
                \begin{equation}
                    x \notin dom(\Gamma, \subtypes{X}{T}, \Gamma'') \Rightarrow x \notin dom(\Gamma, \subs{\Gamma''}{X}{T'})
                \end{equation}
                The result follows from \tsc{Env-V}.
            \item \tsc{Env-TV}\\
                \begin{equation}
                    \derive
                        {\typeJ{\Gamma, \subtypes{X}{T}, \Gamma''}{T''} \quad X' \notin dom(\Gamma, \subtypes{X}{T}, \Gamma'')}
                        {\envJ{\Gamma, \subtypes{X}{T}, \Gamma'', \subtypes{X'}{T''}}}
                \end{equation}
                where $\Gamma' = \Gamma'', \subtypes{X'}{T''}$\\
                By using part (2) of the lemma, we have ${\typeJ{\Gamma, \subs{\Gamma''}{X}{T}}{T''}}$.
                We know that $X \neq X'$ (by the right premise), so
                \begin{equation}
                    X' \notin dom(\Gamma, \subtypes{X}{T}, \Gamma') \Rightarrow X' \notin dom(\Gamma, \subs{\Gamma''}{X}{T'})
                \end{equation}
                The result follows from \tsc{Env-TV}.
        \end{itemize}
    \item By induction on the derivation of ${\typeJ{\Gamma, \subtypes{X}{T'}, \Gamma'}{T}}$.
        \begin{itemize}
            \item \tsc{Type-Int}, \tsc{Type-Bool}, \tsc{Type-Unit}, \tsc{Type-String}, \tsc{Type-Bytes}\\
                We apply part (1) of the lemma, and the result follows afterwards by applying the relevant rule.
            \item \tsc{Type-TV}\\
                \begin{equation}
                    \derive
                        {\envJ{\Gamma, \subtypes{X'}{T}, \Gamma'}}
                        {\typeJ{\Gamma, \subtypes{X'}{T}, \Gamma'}{X'}}
                \end{equation}
                We know that $X \neq X'$, as otherwise we could use part (1) of lemma to obtain
                $\envJ{\Gamma, \subs{\Gamma'}{X'}{T''}}$, an environment in which $X'$ is unbound.\\
                Then, either $\subtypes{X}{T'} \in \Gamma$ or $\subtypes{X}{T'} \in \Gamma'$.\\
                If $X \in \Gamma$, there is a $\Gamma_1$ and $\Gamma_2$ such that
                \begin{equation}
                    \envJ{\Gamma_1, \subtypes{X}{T'}, \Gamma_2, \subtypes{X'}{T}, \Gamma'}
                \end{equation}
                Then by part (1) of the lemma, we have
                \begin{equation}
                    \envJ{\Gamma_1, \subs{\Gamma_2}{X}{T''}, \subtypes{X'}{\subs{T}{X}{T''}, \subs{\Gamma'}{X}{T''}}}
                \end{equation}
                The result then follows from \tsc{Type-TV}.\\
                If $X \in \Gamma'$, there is a $\Gamma_1'$ and $\Gamma_2'$ such that
                \begin{equation}
                    \envJ{\Gamma, \subtypes{X'}{T}, \Gamma_1', \subtypes{X}{T'}, \Gamma_2'}
                \end{equation}
                Then by the 1st part of the lemma, we have
                \begin{equation}
                    \envJ{\Gamma, \subtypes{X'}{T}, \Gamma_1', \subs{\Gamma_2'}{X}{T''}}
                \end{equation}
                The result then also follows from \tsc{Type-TV}.\\
            \item \tsc{Type-SR}, \tsc{Type-Chan}, \tsc{Type-Abs}, \tsc{Type-Rec}, \tsc{Type-List},
                \tsc{Type-Union}, \tsc{Type-Match}\\
                The conclusion follows from applying the I.H. on the premises, and then applying the relevant rule.
            \item \tsc{Type-TAbs}\\
                \begin{equation}
                    \derive
                        {\typeJ{\Gamma, \subtypes{X}{T'}, \Gamma', \subtypes{X'}{T_1}}{T_2}}
                        {\typeJ{\Gamma, \subtypes{X}{T'}, \Gamma'}{\quant{X'}{T_1} T_2}}
                \end{equation}
                where $T = \quant{X'}{T_1} T_2$\\
                By the I.H. with environment $\Gamma, \subtypes{X}{T'}, \Gamma', \subtypes{X'}{T_1}$,
                we have $\typeJ{\Gamma, \subs{\Gamma'}{X}{T''}, \subtypes{X'}{\subs{T_1}{X}{T''}}}{\subs{T_2}{X}{T''}}$.
                With \tsc{Type-TAbs}, we get the desired conclusion.
            \item \tsc{Type-TApp}\\
                \begin{adjustwidth}{-3cm}{3cm}
                \begin{equation}
                    \derive
                        {\typeJ{\Gamma, \subtypes{X}{T'}, \Gamma'}{T_1} \quad \typeJ{\Gamma, \subtypes{X}{T'}, \Gamma'}{T_2}\\
                         \eqtypesTo{\Gamma, \subtypes{X}{T'}, \Gamma'}{T_1}{\quant{X'}{T_{11}} T_{12}} \quad \typeJ{\Gamma, \subtypes{X}{T'}, \Gamma'}{\quant{X'}{T_{11}} T_{12}} \quad \subtypesTo{\Gamma, \subtypes{X}{T'}, \Gamma'}{T_2}{T_{11}}}
                        {\typeJ{\Gamma, \subtypes{X}{T'}, \Gamma'}{T_1\ T_2}}
                \end{equation}
                \end{adjustwidth}
                where $T = T_1\ T_2$\\
                By the I.H., ${\typeJ{\Gamma, \subs{\Gamma'}{X}{T''}}{\subs{T_1}{X}{T''}}}$, ${\typeJ{\Gamma, \subs{\Gamma'}{X}{T''}}{\subs{T_2}{X}{T''}}}$
                and ${\typeJ{\Gamma, \subs{\Gamma'}{X}{T''}}{\quant{X'}{\subs{T_{11}}{X}{T''}} \subs{T_{12}}{X}{T''}}}$.
                Then, from part (3) of the lemma, we have\\
                ${\eqtypesTo{\Gamma, \subs{\Gamma'}{X}{T''}}{\subs{T_1}{X}{T''}}{\quant{X'}{\subs{T_{11}}{X}{T''}} \subs{T_{12}}{X}{T''}}}$
                and ${\subtypesTo{\Gamma, \subs{\Gamma'}{X}{T''}}{\subs{T_2}{X}{T''}}{\subs{T_1}{X}{T''}}}$.
                Finally, we just use \tsc{Type-TApp} to obtain the conclusion.
        \end{itemize}
    \item By induction on the derivation of ${\subtypesTo{\Gamma, \subtypes{X}{T''}, \Gamma'}{T}{T'}}$.
        \begin{itemize}
            \item \tsc{ST-Refl}, \tsc{ST-$\top$}, \tsc{ST-$\cup$Comm}, \tsc{ST-$\cup$Assoc}\\
                Immediate from using the relevant rule with the environment $(\Gamma, \subs{\Gamma'}{X}{T'''})$
            \item \tsc{ST-Val}\\
                where $T = \singleton{v_G}$\\
                Since $v_G$ is a ground value, $\singleton{v_G} = \subs{\singleton{v_G}}{X}{T'''}$.
                Now, we have to consider if $\memberOf{v_G}{T'} \Rightarrow \memberOf{v_G}{\subs{T'}{X}{T''}}$.
                Looking at the rules for $\in_G$, it is clear that type substitution has no effect on the judgement,
                so the conclusion follows from \tsc{ST-Val}.
            \item \tsc{ST-Var}\\
                \begin{equation}
                    \derive
                        {\envJ{\Gamma_1, \subtypes{X'}{T'}, \Gamma_2}}
                        {\subtypesTo{\Gamma_1, \subtypes{X'}{T'}, \Gamma_2}{X'}{T'}}
                \end{equation}
                where $T = X'$\\
                Either $\subtypes{X}{T''} \in \Gamma_1$ or $\subtypes{X}{T''} \in \Gamma_2$.\\
                If $\subtypes{X}{T''} \in \Gamma_1$, then there is a $\Gamma_{11}$ and $\Gamma_{12}$ such that
                \begin{equation}
                    \envJ{\Gamma_{11}, \subtypes{X}{T''}, \Gamma_{12}, \subtypes{X'}{T'}, \Gamma_2}
                \end{equation}
                where $\Gamma = \Gamma_{11} \quad \Gamma' = \Gamma_{12}, \subtypes{X'}{T'}, \Gamma_2$\\
                By part (1) of the lemma, we have
                \begin{equation}
                    \envJ{\Gamma_{11}, \subs{\Gamma_{12}}{X}{T'''}, \subtypes{X'}{\subs{T'}{X}{T'''}}, \subs{\Gamma_2}{X}{T'''}}
                \end{equation}
                The result then follows from \tsc{ST-Var} ($X = \subs{X}{X'}{T'''}$).\\
                If $\subtypes{X}{T''} \in \Gamma_2$, then there is a $\Gamma_{21}$ and $\Gamma_{22}$ such that
                \begin{equation}
                    \envJ{\Gamma_1, \subtypes{X'}{T'}, \Gamma_{21}, \subtypes{X}{T''}, \Gamma_{22}}
                \end{equation}
                where $\Gamma = \Gamma_1, \subtypes{X'}{T'}, \Gamma_{21} \quad \Gamma' = \Gamma_{22}$\\
                By part (1) of the lemma, we have
                \begin{equation}
                    \envJ{\Gamma_1, \subtypes{X'}{T'}, \Gamma_{21}, \subs{\Gamma_{22}}{X}{T'''}}
                \end{equation}
                Then from \tsc{ST-Var}, we can show $\subtypesTo{\Gamma_1, \subtypes{X'}{T'}, \Gamma_{21}, \subs{\Gamma_{22}}{X}{T'''}}{X'}{T'}$.
                We know that $X' = \subs{X'}{X}{T'''}$ and by \cref{var-absence}, that $T' = \subs{T'}{X}{T''}$, so
                $\subtypesTo{\Gamma_1, \subtypes{X'}{T'}, \Gamma_{21}, \subs{\Gamma_{22}}{X}{T'''}}{\subs{X'}{X}{T'''}}{\subs{T'}{X}{T'''}}$,
                which is what we wanted to show.
            \item \tsc{ST-Trans}, \tsc{ST-$\cup$L}, \tsc{ST-$\cup$R}, \tsc{ST-Rec}, \tsc{ST-List},
                  \tsc{ST-Match2}, \tsc{ST-Abs}, \tsc{ST-App}\\
                The conclusion follows from applying the I.H. on the premises and then applying the relevant rule.
            \item \tsc{ST-Match1}\\
                where $T = {T_s\ttt{ match }\curl{T_i \Rightarrow T_i'}_{i \in I}} \quad T' = {T_k'}$\\
                By the I.H. on the first premise, we have $\subtypesTo{\Gamma, \subs{\Gamma'}{X}{T'''}}{\subs{T_s}{X}{T'''}}{\subs{T_k}{X}{T'''}}$.
                We expand the second premise to:
                \begin{align}
                    \forall i \in I: i < k \Rightarrow (\not\exists T_s' : & \typeJ{\Gamma, \subtypes{X}{T''}, \Gamma'}{T_s'}\\
                                                                   \wedge\ & \subtypesTo{\Gamma. \subtypes{X}{T''}, \Gamma'}{T_s'}{T_s}\\
                                                                   \wedge\ & \subtypesTo{\Gamma, \subtypes{X}{T''}, \Gamma'}{T_s'}{T_i})
                \end{align}
                By part (2) of the lemma, we have $\typeJ{\Gamma, \subs{\Gamma'}{X}{T'''}}$.
                Then we can use the I.H. to show:
                \begin{align}
                    \forall i \in I: i < k \Rightarrow (\not\exists T_s' : & \typeJ{\Gamma, \subs{\Gamma'}{X}{T'''}}{T_s'}\\
                                                                   \wedge\ & \subtypesTo{\Gamma, \subs{\Gamma'}{X}{T'''}}{T_s'}{T_s}\\
                                                                   \wedge\ & \subtypesTo{\Gamma, \subs{\Gamma'}{X}{T'''}}{T_s'}{T_i})
                \end{align}
                The result follows from \tsc{ST-Match1}.
            \item \tsc{ST-TAbs}\\
                where $T = \quant{X'}{T_1} T_2 \quad T' = \quant{X'}{T_1'}{T_2'}$\\
                By the I.H. (with the environment $\Gamma, \subtypes{X}{T''}, \Gamma', \subtypes{X'}{T_1'}$ in the right premise), we have that
                ${\subtypesTo{\Gamma, \subs{\Gamma'}{X}{T'''}}{\subs{T_1'}{X}{T'''}}{\subs{T_1}{X}{T'''}}}$
                and ${\subtypesTo{\Gamma, \subs{\Gamma'}{X}{T'''}, \subtypes{X'}{\subs{T_1'}{X}{T'''}}}{\subs{T_2}{X}{T'''}}{\subs{T_2'}{X}{T'''}}}$.
                The result follows from \tsc{ST-TAbs}.
        \end{itemize}
    \item By induction on the derivation of ${\typesTo{\Gamma, \subtypes{X}{T'}, \Gamma'}{t}{T}}$.
        \begin{itemize}
            \item \tsc{T-Val}\\
                where $t = v \quad T = \singleton{v}$.
                $v = \subs{v}{X}{T''}$ and $\singleton{v} = \subs{\singleton{v}}{X}{T''}$,
                so the result is immediate from using \tsc{T-Val} with the new environment.
            \item \tsc{T-Rec}, \tsc{T-List}, \tsc{T-Head}, \tsc{T-Tail}, \tsc{T-App},
                  \tsc{T-Field}, \tsc{T-OpC}, \tsc{T-OpR}, \tsc{T-OpI}, \tsc{T-OpM}, \tsc{T-OpD}\\
                The result follows from the I.H. and the relevant rule.
            \item \tsc{T-Var}\\
                \begin{equation}
                    \derive
                        {\envJ{\Gamma_1, \types{x}{T}, \Gamma_2}}
                        {\typesTo{\Gamma_1, \types{x}{T}, \Gamma_2}{x}{T}}
                \end{equation}
                Either $\subtypes{X}{T'} \in \Gamma_1$ or $\subtypes{X}{T'} \in \Gamma_2$.\\
                If $\subtypes{X}{T'} \in \Gamma_1$, then there is a $\Gamma_{11}$ and $\Gamma_{12}$ such that
                \begin{equation}
                    \envJ{\Gamma_{11}, \subtypes{X}{T'}, \Gamma_{12}, \types{x}{T}, \Gamma_2}
                \end{equation}
                where $\Gamma = \Gamma_{11} \quad \Gamma' = \Gamma_{12}, \types{x}{T}, \Gamma_2$\\
                By part (1) of the lemma:
                \begin{equation}
                    \envJ{\Gamma_{11}, \subs{\Gamma_{12}}{X}{T''}, \types{x}{\subs{T}{X}{T''}}, \subs{\Gamma_2}{X}{T''}}
                \end{equation}
                The result then follows from \tsc{T-Var} (since $x = \subs{x}{X}{T''}$).
                If $\subtypes{X}{T'} \in \Gamma_2$, then there is a $\Gamma_{21}$ and $\Gamma_{22}$ such that
                \begin{equation}
                    \envJ{\Gamma_1, \types{x}{T}, \Gamma_{21}, \subtypes{X}{T'}, \Gamma_{22}}
                \end{equation}
                where $\Gamma = \Gamma_1, \types{x}{T}, \Gamma_{21} \quad \Gamma' = \Gamma_{22}$\\
                By part (1) of the lemma:
                \begin{equation}
                    \envJ{\Gamma_1, \types{x}{T}, \Gamma_{21}, \subs{\Gamma_{22}}{X}{T''}}
                \end{equation}
                From \tsc{T-Var}, we have $\typesTo{\Gamma_1, \types{x}{T}, \Gamma_{21}, \subs{\Gamma_{22}}{X}{T''}}{x}{T}$.
                By \cref{var-absence}, we know that $X$ is not in $T$, so $T = \subs{T}{X}{T''}$.
                This gives us $\typesTo{\Gamma_1, \types{x}{T}, \Gamma_{21}, \subs{\Gamma_{22}}{X}{T''}}{\subs{x}{X}{T''}}{\subs{T}{X}{T''}}$,
                which is what we wanted to show.
            \item \tsc{T-Let}\\
                By the I.H. (using the environment $\Gamma, \subtypes{X}{T'}, \Gamma', \types{x}{T_0}$ on the right premise),
                we have $\typesTo{\Gamma, \subs{\Gamma'}{X}{T''}}{\subs{t_0}{X}{T''}}{\subs{T_0}{X}{T''}}$
                and $\typesTo{\Gamma, \subs{\Gamma'}{X}{T''}, \typesTo{x}{\subs{T_0}{X}{T''}}}{\subs{t_1}{X}{T''}}{\subs{T_1}{X}{T''}}$.
                The result follows from \tsc{T-Let}.
            \item \tsc{T-Abs}\\
                By the I.H. (using the environment $\Gamma, \subtypes{X}{T'}, \Gamma', \types{x}{T_1}$),
                we have $\typesTo{\Gamma, \subs{\Gamma'}{X}{T''}, \types{x}{\subs{T_1}{X}{T''}}}{\subs{t_0}{X}{T''}}{\subs{T_2}{X}{T''}}$.
                The result follows from \tsc{T-Abs}.
            \item \tsc{T-TAbs}\\
                By the I.H. (using the environment $\Gamma, \subtypes{X}{T'}, \Gamma', \subtypes{X}{T_1}$),
                we have $\typesTo{\Gamma, \subs{\Gamma'}{X}{T''}, \subtypes{X}{\subs{T_1}{X}{T''}}}{\subs{t_0}{X}{T''}}{\subs{T_2}{X}{T''}}$.
                The result follows from \tsc{T-TAbs}.
            \item \tsc{T-TApp}\\
                By the I.H., we have $\typesTo{\Gamma, \subs{\Gamma}{X}{T''}}{\subs{t_1}{X}{T''}{\subs{(\quant{X'}{T_1} T_0)}{X}{T''}}}$.
                By part (3) of the lemma, $\subtypesTo{\Gamma, \subs{\Gamma'}{X}{T''}}{\subs{T_2}{X}{T''}}{\subs{T_1}{X}{T''}}$.
                The result follows from \tsc{T-TApp}.
            \item \tsc{T-Match}\\
                By the I.H. (using the environment $\Gamma, \subtypes{X}{T'}, \Gamma', \types{x_i}{T_i}$ on the right premise),
                we have $\typesTo{\Gamma, \subs{\Gamma}{X}{T''}}{\subs{t_s}{X}{T''}}{\subs{T_s}{X}{T''}}$
                and $\forall i \in I : \typesTo{\Gamma, \subs{\Gamma'}{X}{T''}, \types{x_i}{\subs{T_i}{X}{T''}}}{\subs{t_i}{X}{T''}}{\subs{T_i}{X}{T''}}$.
                By part (3) of the lemma, $\subtypesTo{\Gamma, \subs{\Gamma'}{X}{T''}}{T_s}{\cup_{i \in I} T_i}$,
                and the result then follows from \tsc{T-Match}.
            \item \tsc{T-Sub}
                By the I.H. on the left premise, we have $\typesTo{\Gamma, \subs{\Gamma'}{X}{T''}}{\subs{t}{X}{T''}}{\subs{T'''}{X}{T''}}$.
                By part (3) of the lemma, $\subtypesTo{\Gamma, \subs{\Gamma}{X}{T''}}{\subs{T'''}{X}{T''}}{\subs{T}{X}{T''}}$.
                The result follows from \tsc{T-Sub}.
        \end{itemize}
    \item By induction on the derivation of ${\typesTo{\Gamma, \types{x}{T'}}{t}{T}}$.
        \begin{itemize}
            \item \tsc{T-Val}\\
                where $t = v \quad T = \singleton{v}$.
                $v = \subs{v}{x}{T'}$ and $\singleton{v} = \subs{\singleton{v}}{x}{T'}$,
                so the result is immediate from using \tsc{T-Val} with the new environment.
            \item \tsc{T-Rec}, \tsc{T-List}, \tsc{T-Head}, \tsc{T-Tail}, \tsc{T-App},
                  \tsc{T-Field}, \tsc{T-OpC}, \tsc{T-OpR}, \tsc{T-OpI}, \tsc{T-OpM}, \tsc{T-OpD}\\
                The result follows from the I.H. and the relevant rule.
            \item \tsc{T-Var}\\
                \begin{equation}
                    \derive
                        {\envJ{\Gamma_1, \types{x'}{T}, \Gamma_2}}
                        {\typesTo{\Gamma_1, \types{x'}{T}, \Gamma_2}{x'}{T}}
                \end{equation}
                where ether $\types{x}{T'} \in \Gamma_1$ or $\types{x}{T'} \in \Gamma_2$.\\
                In any case, $x' = \subs{x'}{x}{T'}$, so the result follows from \tsc{T-Var}.
            \item \tsc{T-Let}\\
                By the I.H. (using the environment $\Gamma, \types{x}{T'}, \Gamma', \types{x'}{T_0}$ on the right premise),
                we have $\typesTo{\Gamma, \Gamma'}{\subs{t_0}{x}{T'}}{T_0}$
                and $\typesTo{\Gamma, \Gamma', \typesTo{x'}{T_0}}{\subs{t_1}{x}{T'}}{T_1}$.
                The result follows from \tsc{T-Let}.
            \item \tsc{T-Abs}\\
                By the I.H. (using the environment $\Gamma, \types{x}{T'}, \Gamma', \types{x'}{T_1}$),
                we have $\typesTo{\Gamma, \Gamma', \types{x'}{T_1}}{\subs{t_0}{x}{T'}}{T_2}$.
                The result follows from \tsc{T-Abs}.
            \item \tsc{T-TAbs}\\
                By the I.H. (using the environment $\Gamma, \types{x}{T'}, \Gamma', \types{x}{T_1}$),
                we have $\typesTo{\Gamma, \Gamma', \types{x'}{T_1}}{\subs{t_0}{x}{T'}}{T_2}$.
                The result follows from \tsc{T-TAbs}.
            \item \tsc{T-TApp}\\
                By the I.H., $\typesTo{\Gamma, \Gamma'}{\subs{t_1}{x}{T'}}{\quant{X}{T_1} T_0}$.
                The result follows from \cref{strengthening} and \tsc{T-TApp}
            \item \tsc{T-Match}\\
                By the I.H. (using the environment $\Gamma, \types{x}{T'}, \Gamma', \types{x_i}{T_i}$ on the right premise),
                we have $\typesTo{\Gamma, \Gamma}{\subs{t_s}{x}{T'}}{T_s}$
                and $\forall i \in I : \typesTo{\Gamma, \Gamma', \types{x_i}{T_i}}{\subs{t_i}{x}{T'}}{T_i}$.
                The result follows from \cref{strengthening} and \tsc{T-Match}.
            \item \tsc{T-Sub}
                By the I.H. on the left premise, we have $\typesTo{\Gamma, \Gamma'}{\subs{t}{x}{T'}}{T''}$.
                The right premise is $\subtypesTo{\Gamma, \types{x}{T'}, \Gamma'}{T''}{T'}$,
                but by \cref{strengthening}, $\subtypesTo{\Gamma, \Gamma'}{T''}{T'}$.
                The result follows from \tsc{T-Sub}.
        \end{itemize}
\end{enumerate}
\end{proof}

\subsection{Exclusivity}
\begin{lemma}[Exclusivity]\label{exclusivity}
Subtyping ($\subtypesTo{\Gamma}{T}{T'}$) and disjointness ($\disjoint{\Gamma}{T}{T'}$) are mutually exclusive.
\begin{enumerate}
    \item $\subtypesTo{\Gamma}{T}{T'} \Rightarrow \disjoint{\Gamma\not}{T}{T'}$
    \item $\disjoint{\Gamma}{T}{T'} \Rightarrow \subtypesTo{\Gamma\not}{T}{T'}$
\end{enumerate}
\end{lemma}
\begin{proof}
\begin{enumerate}
    To reiterate, the expanded definition of disjointness is:
    \begin{equation}
        \not\exists T'' : \typeJ{\Gamma}{T''} \wedge \subtypesTo{\Gamma}{T''}{T} \wedge \subtypesTo{\Gamma}{T''}{T'}
    \end{equation}
    \item $\subtypesTo{\Gamma}{T}{T'} \Rightarrow \disjoint{\Gamma\not}{T}{T'}$\\
        If $\subtypesTo{\Gamma}{T}{T'}$, then we know that $\typeJ{\Gamma}{T}$.
        From \tsc{ST-Refl}, we have:
        \begin{equation}
            \typeJ{\Gamma}{T} \wedge \subtypesTo{\Gamma}{T}{T} \wedge \subtypesTo{\Gamma}{T}{T'}
        \end{equation}
        which contradicts the definition of disjointness.
    \item $\disjoint{\Gamma}{T}{T'} \Rightarrow \subtypesTo{\Gamma\not}{T}{T'}$\\
        Suppose for the sake of contradiction that $\subtypesTo{\Gamma}{T}{T'}$.
        We then have that $\typeJ{\Gamma}{T}$, and from \tsc{ST-Refl}:
        \begin{equation}
            \typeJ{\Gamma}{T} \wedge \subtypesTo{\Gamma}{T}{T} \wedge \subtypesTo{\Gamma}{T}{T'}
        \end{equation}
        which contradicts our initial assumption, so it must be the case that $\subtypesTo{\Gamma\not}{T}{T'}$.
\end{enumerate}
\end{proof}

\subsection{Minimal types}
\begin{lemma}[Minimal types]\label{minimal-types}\hfill
\begin{enumerate}
    \item If $\typesTo{\Gamma}{v_G}{T}$, then $\subtypesTo{\Gamma}{\underline{v_G}}{T}$
    \item If $\typesTo{\Gamma}{\abs{x}{T_1} t_2}{T}$, then there is a $T_2$ such that $\subtypesTo{\Gamma}{T_1 \rightarrow T_2}{T}$
    \item If $\typesTo{\Gamma}{\tabs{X}{T_1} t_2}{T}$, then there is a $T_2$ such that $\subtypesTo{\Gamma}{\quant{X}{T_1} T_2}{T}$
\end{enumerate}
\end{lemma}
\begin{proof}
\begin{enumerate}
    \item By induction on the derivation of $\typesTo{\Gamma}{v_G}{T}$.
        \begin{itemize}
            \item \tsc{T-Val}\\
                where $T = \singleton{v_G}$\\
                Immediate from \tsc{ST-Refl}.
            \item \tsc{T-Sub}\\
                \begin{equation}
                    \derive
                        {\typesTo{\Gamma}{v_G}{T'} \quad \subtypesTo{\Gamma}{T'}{T}}
                        {\typesTo{\Gamma}{v_G}{T}}
                \end{equation}
                By the I.H., we have that $\subtypesTo{\Gamma}{\singleton{v_G}}{T'}$.
                We obtain the result by applying \tsc{ST-Trans}.
            \item The rest of the rules are not applicable to ground values.
        \end{itemize}
    \item By induction on the derivation of $\typesTo{\Gamma}{\abs{x}{T_1} t_2}{T}$.
        \begin{itemize}
            \item \tsc{T-Abs}\\
                where $T = T_1 \rightarrow T_2$\\
                Immediate from \tsc{ST-Refl}.
            \item \tsc{T-Sub}\\
                \begin{equation}
                    \derive
                        {\typesTo{\Gamma}{\abs{x}{T_1} t_2}{T'} \quad \subtypesTo{\Gamma}{T'}{T}}
                        {\typesTo{\Gamma}{\abs{x}{T_1} t_2}{T}}
                \end{equation}
                By the I.H., we have that there is a $T_2$ such that $\subtypesTo{\Gamma}{T_1 \rightarrow T_2}{T'}$.
                We obtain the result by applying \tsc{ST-Trans}.
            \item The rest of the rules are not applicable to term abstractions.
        \end{itemize}
    \item By induction on the derivation of $\typesTo{\Gamma}{\tabs{X}{T_1} t_2}{T}$.
        \begin{itemize}
            \item \tsc{T-TAbs}\\
                where $T = \quant{X}{T_1} T_2$\\
                Immediate from \tsc{ST-Refl}.
            \item \tsc{T-Sub}\\
                \begin{equation}
                    \derive
                        {\typesTo{\Gamma}{\tabs{X}{T_1} t_2}{T'} \quad \subtypesTo{\Gamma}{T'}{T}}
                        {\typesTo{\Gamma}{\tabs{X}{T_1} t_2}{T}}
                \end{equation}
                By the I.H., we have that there is a $T_2$ such that $\subtypesTo{\Gamma}{\quant{X}{T_1} T_2}{T'}$.
                We obtain the result by applying \tsc{ST-Trans}.
            \item The rest of the rules are not applicable to type abstractions.
        \end{itemize}
\end{enumerate}
\end{proof}

\subsection{Inversion}

We first introduce the following relation
\begin{gather*}
    \Rule
        {\Rule
            {\vdots}
            {\eqtypesTo{\Gamma}{T_1}{T_2}}
            {ST-App}}
        {\mutualTo{\Gamma}{T_1}{T_2}}
        {M-TApp}\qquad
    \Rule
        {\Rule
            {\vdots}
            {\eqtypesTo{\Gamma}{T_1}{T_2}}
            {ST-Match1}}
        {\mutualTo{\Gamma}{T_1}{T_2}}
        {M-Match}\\\\
    \Rule
        {\mutualTo{\Gamma}{T_1}{T_2} \quad \mutualTo{\Gamma}{T_2}{T_3}}
        {\mutualTo{\Gamma}{T_1}{T_3}}
        {M-Trans}
\end{gather*}
The relation $\mutualTo{\Gamma}{T}{T'}$ shows mutual subtyping between types $T$ and $T'$,
with the constraint that $T$ must be either a match type or a type application, and that
the derivation is constructed with pairwise application of \tsc{ST-App}, \tsc{ST-Match1}
and \tsc{ST-Trans}.

\begin{lemma}[Subtyping inversion]
    \label{subtyping-inversion}
    We prove a number of properties regarding subtyping inversion:
    \begin{enumerate}
        \item If $\mutualTo{\Gamma}{T}{T'}$, then either:
            \begin{enumerate}
                \item $T = \match{T_s}{T_i \Rightarrow T_i'}{i \in I}$, $\subtypesTo{\Gamma}{T_s}{T_k}$,
                    $\forall i \in I : i < k \Rightarrow \disjoint{\Gamma}{T_s}{T_i}$ and
                    $\mutualTo{\Gamma}{T_k'}{T'}$,
                \item $T = \match{T_s}{T_i \Rightarrow T_i'}{i \in I}$, $\subtypesTo{\Gamma}{T_s}{T_k}$,
                    $\forall i \in I : i < k \Rightarrow \disjoint{\Gamma}{T_s}{T_i}$ and $T_k' = T'$,
                \item $T = (\quant{X}{T_1} T_0)\ T_2$, $\subtypesTo{\Gamma}{T_2}{T_1}$ and
                    $\mutualTo{\Gamma}{\subs{T_0}{X}{T_2}}{T'}$,
                \item $T = (\quant{X}{T_1} T_0)\ T_2$, $\subtypesTo{\Gamma}{T_2}{T_1}$ and
                    $\subs{T_0}{X}{T_2} = T'$.
            \end{enumerate}
        \item If $\subtypesTo{\Gamma}{T}{\singleton{v_G}}$, or $\subtypesTo{\Gamma}{T}{T'}$ where $\mutualTo{\Gamma}{T'}{\singleton{v_G}}$,
            then either:
            \begin{enumerate}
                \item $\mutualTo{\Gamma}{T}{\singleton{v_G}}$,
                \item $T$ is ground type.
            \end{enumerate}
        \item If $\subtypesTo{\Gamma}{T}{X}$, or $\subtypesTo{\Gamma}{T}{T'}$ where $\mutualTo{\Gamma}{T'}{X}$,
            then either:
            \begin{enumerate}
                \item $\mutualTo{\Gamma}{T}{Y}$ for some type variable $Y$,
                \item $T$ is a type variable.
            \end{enumerate}
        \item If $\subtypesTo{\Gamma}{T}{T_1' \rightarrow T_2'}$
            or $\subtypesTo{\Gamma}{T}{T'}$ and $\mutualTo{\Gamma}{T'}{T_1' \rightarrow T_2'}$ then either:
            \begin{enumerate}
                \item $\mutualTo{\Gamma}{T}{T_1 \rightarrow T_2}$ for some $T_1$ and $T_2$
                    such that $\subtypesTo{\Gamma}{T_1'}{T_1}$ and $\subtypesTo{\Gamma}{T_2}{T_2'}$,
                \item $\mutualTo{\Gamma}{T}{X}$ for some type variable $X$,
                \item $T$ is a type variable,
                \item $T = T_1 \rightarrow T_2$ where $\subtypesTo{\Gamma}{T_1'}{T_1}$ and $\subtypesTo{\Gamma}{T_2}{T_2'}$.
            \end{enumerate}
        \item If $\subtypesTo{\Gamma}{T}{(\quant{X}{T_1'} T_2')}$ or $\subtypesTo{\Gamma}{T}{T'}$
            and $\mutualTo{\Gamma}{T'}{(\quant{X}{T_1'} T_2')}$, then either:
            \begin{enumerate}
                \item $\mutualTo{\Gamma}{T}{(\quant{X}{T_1} T_2)}$, for some $T_1$ and $T_2$ such that
                    $\subtypesTo{\Gamma}{T_1'}{T_1}$ and $\subtypesTo{\Gamma, \subtypes{X}{T_1'}}{T_2}{T_2'}$,
                \item $\mutualTo{\Gamma}{T}{X}$, for some type variable $X$,
                \item $T$ is a type variable,
                \item $T = \quant{X}{T_1''}{T_2}$ and $\subtypesTo{\Gamma, \subtypes{X}{T_1''}}{T_2}{T_2'}$.
            \end{enumerate}
        \item If $\subtypesTo{\Gamma}{T}{\curl{f_i : T_i'}_{i \in I}}$ or $\subtypesTo{\Gamma}{T}{T'}$
            and $\mutualTo{\Gamma}{T'}{\curl{f_i : T_i'}_{i \in I}}$, then either:
            \begin{enumerate}
                \item $\mutualTo{\Gamma}{T}{\curl{f_i : T_i}_{i \in I}}$ such that
                    $\forall i \in I : \subtypesTo{\Gamma}{T_i}{T_i'}$.
                \item $\mutualTo{\Gamma}{T}{X}$, for some type variable $X$,
                \item $T$ is a type variable,
                \item $\mutualTo{\Gamma}{T}{\singleton{v_G}}$ for some ground type $\singleton{v_G}$,
                \item $T$ is a ground type,
                \item $T = \curl{f_i : T_i}$ and $\forall i \in I : \subtypesTo{\Gamma}{T_i}{T_i'}$.
            \end{enumerate}
        \item If $\subtypesTo{\Gamma}{T}{\listT{T_l'}}$ or $\subtypesTo{\Gamma}{T}{T'}$
            and $\mutualTo{\Gamma}{T'}{\listT{T_l'}}$, then either:
            \begin{enumerate}
                \item $\mutualTo{\Gamma}{T}{\listT{T_l}}$, for some $T_l$ such that
                    $\subtypesTo{\Gamma}{T_l}{T_l'}$.
                \item $\mutualTo{\Gamma}{T}{X}$, for some type variable $X$,
                \item $T$ is a type variable,
                \item $\mutualTo{\Gamma}{T}{\singleton{v_G}}$ for some ground type $\singleton{v_G}$,
                \item $T$ is a ground type,
                \item $T = \listT{T_l}$ and $\subtypesTo{\Gamma}{T_l}{T_l'}$.
            \end{enumerate}
    \end{enumerate}
\end{lemma}
\begin{proof}
All proofs are by induction on the size of the derivation.
\begin{enumerate}
    \item $\mutualTo{\Gamma}{T}{T'}$ 
        \begin{itemize}
            \item \tsc{M-TApp}\\
                where $T = (\quant{X}{T_1} T_0)\ T_2 \quad T' = \subs{T_0}{X}{T_2}$\\
                We have case $d$ immediately from the premises.
            \item \tsc{M-Match}
                where $T = \match{T_s}{T_i \Rightarrow T_i'}{i \in I} \quad T' = T_k'$\\
                We have case $b$ immediately from the premises.
            \item \tsc{M-Trans}\\
                \begin{equation}
                    \derive
                        {\mutualTo{\Gamma}{T}{T''} \quad \mutualTo{\Gamma}{T''}{T'}}
                        {\mutualTo{\Gamma}{T}{T'}}
                \end{equation}
                By the I.H. on the left premise, we have four subcases:
                \begin{enumerate}
                    \item $T = \match{T_s}{T_i \Rightarrow T_i'}{i \in I}$, $\subtypesTo{\Gamma}{T_s}{T_k}$,
                        $\forall i \in I : i < k \Rightarrow \disjoint{\Gamma}{T_s}{T_i}$ and
                        $\mutualTo{\Gamma}{T_k'}{T''}$.\\
                        By \tsc{M-Trans} we have case $a$.
                    \item $T = \match{T_s}{T_i \Rightarrow T_i'}{i \in I}$, $\subtypesTo{\Gamma}{T_s}{T_k}$,
                        $\forall i \in I : i < k \Rightarrow \disjoint{\Gamma}{T_s}{T_i}$ and $T_k' = T''$\\
                        We immediately have case $a$.
                    \item $T = (\quant{X}{T_1} T_0)\ T_2$, $\subtypesTo{\Gamma}{T_2}{T_1}$ and
                        $\mutualTo{\Gamma}{\subs{T_0}{X}{T_2}}{T''}$.\\
                        By \tsc{M-Trans} we have case $c$.
                    \item $T = (\quant{X}{T_1} T_0)\ T_2$, $\subtypesTo{\Gamma}{T_2}{T_1}$ and
                        $\subs{T_0}{X}{T_2} = T''$.
                        We immediately have case $c$.
                \end{enumerate}
        \end{itemize}
    \item $\subtypesTo{\Gamma}{T}{T'}$ 
        \begin{itemize}
            \item \tsc{ST-Refl}
                where $T = T'$\\
                If $T'$ is a ground type, we immediately have case $b$.
                If $\mutualTo{\Gamma}{T'}{\singleton{v_G}}$, then we have case $a$.
            \item \tsc{ST-Val}\\
                where $T = \singleton{v_G}$\\
                Case $b$ is immediate.
            \item \tsc{ST-Var}\\
                From looking at the rules for the "$\in_G$" relation, it is clear that the type on the right-hand
                side cannot be a ground type, so this rule could not have been used.
            \item \tsc{ST-Trans}
                \begin{equation}
                    \derive
                        {\subtypesTo{\Gamma}{T}{T''} \quad \subtypesTo{\Gamma}{T''}{T'}}
                        {\subtypesTo{\Gamma}{T}{T'}}
                \end{equation}
                If $T'$ is a ground type $\singleton{v_G}$, then we can use the I.H. on the right premise
                to get that either $\mutualTo{\Gamma}{T''}{\singleton{v_G}}$ or $T''$ is a ground type.
                In any case, we can use the I.H. on the left premise to obtain the result.
                If $\mutualTo{\Gamma}{T'}{\singleton{v_G}}$ we can also use the I.H. on the right, then
                left premise to obtain the result.
            \item \tsc{ST-Match1}
                First case:
                \begin{equation}
                    \derive
                        {\subtypesTo{\Gamma}{T_s}{T_k} \quad \forall i \in I : i < k \Rightarrow \disjoint{\Gamma}{T_s}{T_i}}
                        {\subtypesTo{\Gamma}{\match{T_s}{T_i \Rightarrow T_i'}{i \in I}}{T_k'}}
                \end{equation}
                where $T = \match{T_s}{T_i \Rightarrow T_i'}{i \in I} \quad T' = T_k'$\\
                If $T' = \singleton{v_G}$, we have case $a$ from \tsc{M-Match}.\\
                If $\mutualTo{\Gamma}{T'}{\singleton{v_G}}$, we have case $a$ from \tsc{M-Match} and \tsc{M-Trans}.\\\\
                Second case:
                \begin{equation}
                    \derive
                        {\subtypesTo{\Gamma}{T_s}{T_k} \quad \forall i \in I : i < k \Rightarrow \disjoint{\Gamma}{T_s}{T_i}}
                        {\subtypesTo{\Gamma}{T_k'}{\match{T_s}{T_i \Rightarrow T_i'}{i \in I}}}
                \end{equation}
                where $T = T_k' \quad T' = \match{T_s}{T_i \Rightarrow T_i'}{i \in I}$\\
                We use the 1st part of the lemma on $\mutualTo{\Gamma}{T'}{\singleton{v_G}}$ to get four subcases:
                \begin{enumerate}
                    \item $T' = \match{T_s}{T_i \Rightarrow T_i'}{i \in I}$, $\subtypesTo{\Gamma}{T_s}{T_k}$,
                        $\forall i \in I : i < k \Rightarrow \disjoint{\Gamma}{T_s}{T_i}$ and
                        $\mutualTo{\Gamma}{T_k'}{\singleton{v_G}}$\\
                        By $T = T_k'$ we have $\mutualTo{\Gamma}{T}{\singleton{v_G}}$, which gives us case $a$.
                    \item $T' = \match{T_s}{T_i \Rightarrow T_i'}{i \in I}$, $\subtypesTo{\Gamma}{T_s}{T_k}$,
                        $\forall i \in I : i < k \Rightarrow \disjoint{\Gamma}{T_s}{T_i}$ and $T_k' = \singleton{v_G}$\\
                        Case $b$ is immediate.
                    \item Not applicable since $T'$ is a match type.
                    \item Not applicable since $T'$ is a match type.
                \end{enumerate}
            \item \tsc{ST-Match2}
                \begin{equation}
                    \derive
                        {\subtypesTo{\Gamma}{T_s}{T_s'} \quad \forall i \in I : \subtypesTo{\Gamma}{T_i'}{T_i''}}
                        {\subtypesTo{\Gamma}{\match{T_s}{T_i \Rightarrow T_i'}{i \in I}}{\match{T_s'}{T_i \Rightarrow T_i''}{i \in I}}}
                \end{equation}
                where $T = \match{T_s}{T_i \Rightarrow T_i'}{i \in I} \quad T' = \match{T_s'}{T_i \Rightarrow T_i''}{i \in I}$\\
                We use the 1st part of the lemma on $\mutualTo{\Gamma}{T'}{\singleton{v_G}}$ to get four subcases:
                \begin{enumerate}
                    \item $T' = \match{T_s'}{T_i \Rightarrow T_i''}{i \in I}$, $\subtypesTo{\Gamma}{T_s'}{T_k}$,
                        $\forall i \in I : i < k \Rightarrow \disjoint{\Gamma}{T_s'}{T_i}$ and
                        $\mutualTo{\Gamma}{T_k''}{\singleton{v_G}}$\\
                        By the definition of disjointness and \tsc{ST-Trans}, we have
                        $\forall i \in I : i < k \Rightarrow \disjoint{\Gamma}{T_s}{T_i}$.
                        Also from \tsc{ST-Trans}, we have $\subtypesTo{\Gamma}{T_s}{T_k}$.
                        Then, from \tsc{ST-Match1} and \tsc{M-Match}, we have $\mutualTo{\Gamma}{T}{T_k'}$.
                        Now, since $\subtypesTo{\Gamma}{T_k'}{T_k''}$ and $\mutualTo{\Gamma}{T_k''}{\singleton{v_G}}$,
                        we can use the I.H. to get that either $\mutualTo{\Gamma}{T_k'}{\singleton{v_G}}$,
                        or $T_k'$ is a ground type.
                        In the first case, we use \tsc{M-Trans} to get $\mutualTo{\Gamma}{T}{\singleton{v_G}}$ for case $a$.
                        In the second case, we have case $a$ from $\mutualTo{\Gamma}{T}{T_k'}$.
                    \item $T' = \match{T_s'}{T_i \Rightarrow T_i''}{i \in I}$, $\subtypesTo{\Gamma}{T_s'}{T_k}$,
                        $\forall i \in I : i < k \Rightarrow \disjoint{\Gamma}{T_s'}{T_i}$ and $T_k'' = \singleton{v_G}$\\
                        By the definition of disjointness and \tsc{ST-Trans}, we have
                        $\forall i \in I : i < k \Rightarrow \disjoint{\Gamma}{T_s}{T_i}$.
                        Also from \tsc{ST-Trans}, we have $\subtypesTo{\Gamma}{T_s}{T_k}$.
                        Then, from \tsc{ST-Match1} and \tsc{M-Match}, we have $\mutualTo{\Gamma}{T}{T_k'}$.
                        Since $T_k'' = \singleton{v_G}$, we have $\subtypesTo{\Gamma}{T_k'}{\singleton{v_G}}$, and then
                        we can use the I.H. to get that either $\mutualTo{\Gamma}{T_k'}{\singleton{v_G}}$,
                        or $T_k'$ is a ground type.
                        In the first case, we use \tsc{M-Trans} to get $\mutualTo{\Gamma}{T}{\singleton{v_G}}$ for case $a$.
                        In the second case, we have case $a$ from $\mutualTo{\Gamma}{T}{T_k'}$.
                    \item Not applicable since $T'$ is a match type.
                    \item Not applicable since $T'$ is a match type.
                \end{enumerate}
            \item \tsc{ST-App}
                First case:
                \begin{equation}
                    \derive
                        {\subtypesTo{\Gamma}{T_2}{T_1}}
                        {\subtypesTo{\Gamma}{(\quant{X}{T_1} T_0)\ T_2}{\subs{T_0}{X}{T_2}}}
                \end{equation}
                where $T = (\quant{X}{T_1} T_0)\ T_2 \quad T' = \subs{T_0}{X}{T_2}$\\
                If $T' = \singleton{v_G}$, then we have case $a$ from \tsc{M-TApp}.\\
                If $\mutualTo{\Gamma}{T'}{\singleton{v_G}}$, we have case $a$ from \tsc{M-TApp} and \tsc{M-Trans}.\\\\
                Second case:
                \begin{equation}
                    \derive
                        {\subtypesTo{\Gamma}{T_2}{T_1}}
                        {\subtypesTo{\Gamma}{\subs{T_0}{X}{T_2}}{(\quant{X}{T_1} T_0)\ T_2}}
                \end{equation}
                where $T = \subs{T_0}{X}{T_2} \quad T' = (\quant{X}{T_1} T_0)\ T_2$\\
                We use the 1st part of the lemma on $\mutualTo{\Gamma}{T'}{\singleton{v_G}}$ to get four subcases:
                \begin{enumerate}
                    \item Not applicable since $T'$ is a type application.
                    \item Not applicable since $T'$ is a type application
                    \item $T' = (\quant{X}{T_1} T_0)\ T_2$, $\subtypesTo{\Gamma}{T_2}{T_1}$ and
                        $\mutualTo{\Gamma}{\subs{T_0}{X}{T_2}}{\singleton{v_G}}$\\
                        Since $T = \subs{T_0}{X}{T_2}$, we have case $a$ immediately.
                    \item $T' = (\quant{X}{T_1} T_0)\ T_2$, $\subtypesTo{\Gamma}{T_2}{T_1}$ and
                        $\subs{T_0}{X}{T_2} = \singleton{v_G}$\\
                        Case $b$ is immediate.
                \end{enumerate}
            \item \tsc{ST-$\top$}, \tsc{ST-$\cup$Comm}, \tsc{ST-$\cup$Assoc}, \tsc{ST-List},
                \tsc{ST-$\cup$R}, \tsc{ST-Rec}, \tsc{ST-Abs}, \tsc{ST-TAbs}\\
                These rules could not have been used.
                $T'$ cannot be a ground type (by the structure), nor can we derive $\mutualTo{\Gamma}{T'}{X}$
                since $T'$ must be either a match type or type application.
        \end{itemize}
    \item $\subtypesTo{\Gamma}{T}{T'}$ 
        \begin{itemize}
            \item \tsc{ST-Refl}\\
                where $T = T'$\\
                If $T'$ is a type variable, then we immediately have case $b$.\\
                If $\mutualTo{\Gamma}{T'}{X}$, then we have case $a$.
            \item \tsc{ST-Val}\\
                where $T = \singleton{v_G}$\\
                From looking at the rules for the "$\in_G$" relation, it is clear that the type on the right-hand
                side cannot be a type variable, so this rule could not have been used.
            \item \tsc{ST-Var}\\
                where $T = X \quad \Gamma = \Gamma', \subtypes{X}{T'}, \Gamma''$\\
                Case 2 is immediate.
            \item \tsc{ST-Trans}\\
                \begin{equation}
                    \derive
                        {\subtypesTo{\Gamma}{T}{T''} \quad \subtypesTo{\Gamma}{T''}{T'}}
                        {\subtypesTo{\Gamma}{T}{T'}}
                \end{equation}
                If $T'$ is a type variable $X$ then we use the I.H. on the right premise to get that either
                $\mutualTo{\Gamma}{T''}{X}$ or $T''$ is a type variable. In any case, we can use the I.H.
                on the left premise to obtain the result.\\
                If $\mutualTo{\Gamma}{T'}{X}$ we can also use the I.H. on the right, then left premise to
                obtain the result.
            \item \tsc{ST-Match1}
                First case:
                \begin{equation}
                    \derive
                        {\subtypesTo{\Gamma}{T_s}{T_k} \quad \forall i \in I : i < k \Rightarrow \disjoint{\Gamma}{T_s}{T_i}}
                        {\subtypesTo{\Gamma}{\match{T_s}{T_i \Rightarrow T_i'}{i \in I}}{T_k'}}
                \end{equation}
                where $T = \match{T_s}{T_i \Rightarrow T_i'}{i \in I} \quad T' = T_k'$\\
                If $T' = X$, we have case 1 from \tsc{M-Match}.\\
                If $\mutualTo{\Gamma}{T'}{X}$, we have case $a$ from \tsc{M-Match} and \tsc{M-Trans}.\\\\
                Second case:
                \begin{equation}
                    \derive
                        {\subtypesTo{\Gamma}{T_s}{T_k} \quad \forall i \in I : i < k \Rightarrow \disjoint{\Gamma}{T_s}{T_i}}
                        {\subtypesTo{\Gamma}{T_k'}{\match{T_s}{T_i \Rightarrow T_i'}{i \in I}}}
                \end{equation}
                where $T = T_k' \quad T' = \match{T_s}{T_i \Rightarrow T_i'}{i \in I}$\\
                We use the 1st part of the lemma on $\mutualTo{\Gamma}{T'}{X}$ to get four subcases:
                \begin{enumerate}
                    \item $T' = \match{T_s}{T_i \Rightarrow T_i'}{i \in I}$, $\subtypesTo{\Gamma}{T_s}{T_k}$,
                        $\forall i \in I : i < k \Rightarrow \disjoint{\Gamma}{T_s}{T_i}$ and
                        $\mutualTo{\Gamma}{T_k'}{X}$\\
                        By $T = T_k'$ we have $\mutualTo{\Gamma}{T}{X}$, which gives us case $a$.
                    \item $T' = \match{T_s}{T_i \Rightarrow T_i'}{i \in I}$, $\subtypesTo{\Gamma}{T_s}{T_k}$,
                        $\forall i \in I : i < k \Rightarrow \disjoint{\Gamma}{T_s}{T_i}$ and $T_k' = X$\\
                        Case $b$ is immediate.
                    \item Not applicable since $T'$ is a match type.
                    \item Not applicable since $T'$ is a match type.
                \end{enumerate}
            \item \tsc{ST-Match2}
                \begin{equation}
                    \derive
                        {\subtypesTo{\Gamma}{T_s}{T_s'} \quad \forall i \in I : \subtypesTo{\Gamma}{T_i'}{T_i''}}
                        {\subtypesTo{\Gamma}{\match{T_s}{T_i \Rightarrow T_i'}{i \in I}}{\match{T_s'}{T_i \Rightarrow T_i''}{i \in I}}}
                \end{equation}
                where $T = \match{T_s}{T_i \Rightarrow T_i'}{i \in I} \quad T' = \match{T_s'}{T_i \Rightarrow T_i''}{i \in I}$\\
                We use the 1st part of the lemma on $\mutualTo{\Gamma}{T'}{X}$ to get four subcases:
                \begin{enumerate}
                    \item $T' = \match{T_s'}{T_i \Rightarrow T_i''}{i \in I}$, $\subtypesTo{\Gamma}{T_s'}{T_k}$,
                        $\forall i \in I : i < k \Rightarrow \disjoint{\Gamma}{T_s'}{T_i}$ and
                        $\mutualTo{\Gamma}{T_k''}{X}$\\
                        By the definition of disjointness and \tsc{ST-Trans}, we have
                        $\forall i \in I : i < k \Rightarrow \disjoint{\Gamma}{T_s}{T_i}$.
                        Also from \tsc{ST-Trans}, we have $\subtypesTo{\Gamma}{T_s}{T_k}$.
                        Then, from \tsc{ST-Match1} and \tsc{M-Match}, we have $\mutualTo{\Gamma}{T}{T_k'}$.
                        Now, since $\subtypesTo{\Gamma}{T_k'}{T_k''}$ and $\mutualTo{\Gamma}{T_k''}{X}$,
                        we can use the I.H. to get that either $\mutualTo{\Gamma}{T_k'}{Y}$ for some type variable $Y$,
                        or $T_k'$ is a type variable.
                        In the first case, we use \tsc{M-Trans} to get $\mutualTo{\Gamma}{T}{Y}$ for case $a$.
                        In the second case, we have case $a$ from $\mutualTo{\Gamma}{T}{T_k'}$.
                    \item $T' = \match{T_s'}{T_i \Rightarrow T_i''}{i \in I}$, $\subtypesTo{\Gamma}{T_s'}{T_k}$,
                        $\forall i \in I : i < k \Rightarrow \disjoint{\Gamma}{T_s'}{T_i}$ and $T_k'' = X$\\
                        By the definition of disjointness and \tsc{ST-Trans}, we have
                        $\forall i \in I : i < k \Rightarrow \disjoint{\Gamma}{T_s}{T_i}$.
                        Also from \tsc{ST-Trans}, we have $\subtypesTo{\Gamma}{T_s}{T_k}$.
                        Then, from \tsc{ST-Match1} and \tsc{M-Match}, we have $\mutualTo{\Gamma}{T}{T_k'}$.
                        Since $T_k'' = X$, we have $\subtypesTo{\Gamma}{T_k'}{X}$, and then
                        we can use the I.H. to get that either $\mutualTo{\Gamma}{T_k'}{Y}$ for some type variable $Y$,
                        or $T_k'$ is a type variable.
                        In the first case, we use \tsc{M-Trans} to get $\mutualTo{\Gamma}{T}{Y}$ for case $a$.
                        In the second case, we have case $a$ from $\mutualTo{\Gamma}{T}{T_k'}$.
                    \item Not applicable since $T'$ is a match type.
                    \item Not applicable since $T'$ is a match type.
                \end{enumerate}
            \item \tsc{ST-App}
                First case:
                \begin{equation}
                    \derive
                        {\subtypesTo{\Gamma}{T_2}{T_1}}
                        {\subtypesTo{\Gamma}{(\quant{Y}{T_1} T_0)\ T_2}{\subs{T_0}{Y}{T_2}}}
                \end{equation}
                where $T = (\quant{Y}{T_1} T_0)\ T_2 \quad T' = \subs{T_0}{Y}{T_2}$\\
                If $T' = X$, then we have case $a$ from \tsc{M-TApp}.\\
                If $\mutualTo{\Gamma}{T'}{X}$, we have case $a$ from \tsc{M-TApp} and \tsc{M-Trans}.\\\\
                Second case:
                \begin{equation}
                    \derive
                        {\subtypesTo{\Gamma}{T_2}{T_1}}
                        {\subtypesTo{\Gamma}{\subs{T_0}{Y}{T_2}}{(\quant{Y}{T_1} T_0)\ T_2}}
                \end{equation}
                where $T = \subs{T_0}{Y}{T_2} \quad T' = (\quant{Y}{T_1} T_0)\ T_2$\\
                We use the 1st part of the lemma on $\mutualTo{\Gamma}{T'}{X}$ to get four subcases:
                \begin{enumerate}
                    \item Not applicable since $T'$ is a type application.
                    \item Not applicable since $T'$ is a type application
                    \item $T' = (\quant{Y}{T_1} T_0)\ T_2$, $\subtypesTo{\Gamma}{T_2}{T_1}$ and
                        $\mutualTo{\Gamma}{\subs{T_0}{Y}{T_2}}{X}$\\
                        Since $T = \subs{T_0}{Y}{T_2}$, we have case $a$ immediately.
                    \item $T' = (\quant{Y}{T_1} T_0)\ T_2$, $\subtypesTo{\Gamma}{T_2}{T_1}$ and
                        $\subs{T_0}{Y}{T_2} = X$\\
                        Case $b$ is immediate.
                \end{enumerate}
            \item \tsc{ST-$\top$}, \tsc{ST-$\cup$Comm}, \tsc{ST-$\cup$Assoc}, \tsc{ST-List},
                \tsc{ST-$\cup$R}, \tsc{ST-Rec}, \tsc{ST-Abs}, \tsc{ST-TAbs}\\
                These rules could not have been used.
                $T'$ cannot be a type variable (by the structure), nor can we derive $\mutualTo{\Gamma}{T'}{X}$
                since $T'$ must be either a match type or type application.
        \end{itemize}
    \item $\subtypesTo{\Gamma}{T}{T'}$ 
        \begin{itemize}
            \item \tsc{ST-Refl}\\
                where $T = T'$\\
                If $T' = T_1' \rightarrow T_2'$ we have case $d$ since by \tsc{ST-Refl} since $T_1 = T_1'$ and $T_2 = T_2'$.\\
                If $\mutualTo{\Gamma}{T'}{T_1' \rightarrow T_2'}$ we have case $a$ by \tsc{ST-Refl}.
            \item \tsc{ST-Val}\\
                By looking at the $\in_G$ relation rules, $T'$ cannot be an arrow type, a match type or a type application,
                so this rule could not have been used.
            \item \tsc{ST-Var}\\
                where $X = T \quad \Gamma = \Gamma', \subtypes{X}{T'}, \Gamma''$\\
                Case $c$ is immediate.
            \item \tsc{ST-Trans}
                \begin{equation}
                    \derive
                        {\subtypesTo{\Gamma}{T}{T''} \quad \subtypesTo{\Gamma}{T''}{T'}}
                        {\subtypesTo{\Gamma}{T}{T'}}
                \end{equation}
                By the I.H. on the right premise, we have four subcases:
                \begin{enumerate}
                    \item $\mutualTo{\Gamma}{T''}{T_1'' \rightarrow T_2''}$ for some $T_1''$ and $T_2''$
                        such that $\subtypesTo{\Gamma}{T_1'}{T_1''}$ and $\subtypesTo{\Gamma}{T_2''}{T_2'}$.\\
                        The result follows by using the I.H. on the left premise. If the result is case $a$ or $d$,
                        we have $\subtypesTo{\Gamma}{T_1'}{T_1}$ and $\subtypesTo{\Gamma}{T_2}{T_2'}$ by
                        \tsc{ST-Trans}.
                    \item $\mutualTo{\Gamma}{T''}{X}$ for some type variable $X$.\\
                        The result follows from using the 3rd part of the lemma on the left premise.
                    \item $T''$ is a type variable.
                        The result follows from using the 3rd part of the lemma on the left premise.
                    \item $T'' = T_1'' \rightarrow T_2''$ for some $T_1''$ and $T_2''$
                        such that $\subtypesTo{\Gamma}{T_1'}{T_1''}$ and $\subtypesTo{\Gamma}{T_2''}{T_2'}$.\\
                        The result follows by using the I.H. on the left premise. If the result is case $a$ or $d$,
                        we have $\subtypesTo{\Gamma}{T_1'}{T_1}$ and $\subtypesTo{\Gamma}{T_2}{T_2'}$ by
                        \tsc{ST-Trans}.
                \end{enumerate}
            \item \tsc{ST-Abs}
                \begin{equation}
                    \derive
                        {\subtypesTo{\Gamma}{T_1'}{T_1} \quad \subtypesTo{\Gamma}{T_2}{T_2'}}
                        {\subtypesTo{\Gamma}{T_1 \rightarrow T_2}{T_1' \rightarrow T_2'}}
                \end{equation}
                where $T = T_1 \rightarrow T_2 \quad T' = T_1' \rightarrow T_2'$\\
                Case $d$ is immediate.
            \item \tsc{ST-Match1}
                First case:
                \begin{equation}
                    \derive
                        {\subtypesTo{\Gamma}{T_s}{T_k} \quad \forall i \in I : i < k \Rightarrow \disjoint{\Gamma}{T_s}{T_i}}
                        {\subtypesTo{\Gamma}{\match{T_s}{T_i \Rightarrow T_i'}{i \in I}}{T_k'}}
                \end{equation}
                where $T = \match{T_s}{T_i \Rightarrow T_i'}{i \in I} \quad T' = T_k'$\\
                If $T' = T_1' \rightarrow T_2'$, we have case $a$ from \tsc{M-Match} and \tsc{ST-Refl}.\\
                If $\mutualTo{\Gamma}{T'}{T_1' \rightarrow T_2'}$, we have $\mutualTo{\Gamma}{T}{T_1' \rightarrow T_2'}$
                from \tsc{M-Match} and \tsc{M-Trans}. We then have case $a$ from \tsc{ST-Refl}.\\\\
                Second case:
                \begin{equation}
                    \derive
                        {\subtypesTo{\Gamma}{T_s}{T_k} \quad \forall i \in I : i < k \Rightarrow \disjoint{\Gamma}{T_s}{T_i}}
                        {\subtypesTo{\Gamma}{T_k'}{\match{T_s}{T_i \Rightarrow T_i'}{i \in I}}}
                \end{equation}
                where $T = T_k' \quad T' = \match{T_s}{T_i \Rightarrow T_i'}{i \in I}$\\
                We use the 1st part of the lemma on $\mutualTo{\Gamma}{T'}{T_1' \rightarrow T_2'}$ to get four subcases:
                \begin{enumerate}
                    \item $T' = \match{T_s}{T_i \Rightarrow T_i'}{i \in I}$, $\subtypesTo{\Gamma}{T_s}{T_k}$,
                        $\forall i \in I : i < k \Rightarrow \disjoint{\Gamma}{T_s}{T_i}$ and
                        $\mutualTo{\Gamma}{T_k'}{T_1' \rightarrow T_2'}$\\
                        By $T = T_k'$ we have $\mutualTo{\Gamma}{T}{T_1' \rightarrow T_2'}$.
                        Case $a$ then follows from \tsc{ST-Refl}.
                    \item $T' = \match{T_s}{T_i \Rightarrow T_i'}{i \in I}$, $\subtypesTo{\Gamma}{T_s}{T_k}$,
                        $\forall i \in I : i < k \Rightarrow \disjoint{\Gamma}{T_s}{T_i}$ and $T_k' = T_1' \rightarrow T_2'$\\
                        Case $d$ follows from \tsc{ST-Refl} with $T_1 = T_1'$ and $T_2 = T_2'$.
                    \item Not applicable since $T'$ is a match type.
                    \item Not applicable since $T'$ is a match type.
                \end{enumerate}
            \item \tsc{ST-Match2}
                \begin{equation}
                    \derive
                        {\subtypesTo{\Gamma}{T_s}{T_s'} \quad \forall i \in I : \subtypesTo{\Gamma}{T_i'}{T_i''}}
                        {\subtypesTo{\Gamma}{\match{T_s}{T_i \Rightarrow T_i'}{i \in I}}{\match{T_s'}{T_i \Rightarrow T_i''}{i \in I}}}
                \end{equation}
                where $T = \match{T_s}{T_i \Rightarrow T_i'}{i \in I} \quad T' = \match{T_s'}{T_i \Rightarrow T_i''}{i \in I}$\\
                We use the 1st part of the lemma on $\mutualTo{\Gamma}{T'}{T_1' \rightarrow T_2'}$ to get four subcases:
                \begin{enumerate}
                    \item $T' = \match{T_s'}{T_i \Rightarrow T_i''}{i \in I}$, $\subtypesTo{\Gamma}{T_s'}{T_k}$,
                        $\forall i \in I : i < k \Rightarrow \disjoint{\Gamma}{T_s'}{T_i}$ and
                        $\mutualTo{\Gamma}{T_k''}{T_1' \rightarrow T_2'}$\\
                        By the definition of disjointness and \tsc{ST-Trans}, we have
                        $\forall i \in I : i < k \Rightarrow \disjoint{\Gamma}{T_s}{T_i}$
                        Also by \tsc{ST-Trans}, $\subtypesTo{\Gamma}{T_s}{T_k}$.
                        We then have from \tsc{ST-Match1} and \tsc{M-Match} that $\mutualTo{\Gamma}{T}{T_k'}$.
                        Now, we use the I.H. on $\subtypesTo{\Gamma}{T_k'}{T_k''}$ and $\mutualTo{\Gamma}{T_k''}{T_1' \rightarrow T_2'}$
                        and obtain 4 subsubcases:
                        \begin{enumerate}
                            \item $\mutualTo{\Gamma}{T_k'}{T_1 \rightarrow T_2}$ for some $T_1$ and $T_2$ such that
                                $\subtypesTo{\Gamma}{T_1'}{T_1}$ and $\subtypesTo{\Gamma}{T_2}{T_2'}$\\
                                We have case $a$ from \tsc{M-Trans}.
                            \item $\mutualTo{\Gamma}{T_k'}{X}$\\
                                We have case $b$ from \tsc{M-Trans}.
                            \item $T_k'$ is a type variable\\
                                Case $b$ is immediate.
                            \item $T_k' = T_1 \rightarrow T_2$ for some $T_1$ and $T_2$ such that
                                $\subtypesTo{\Gamma}{T_1'}{T_1}$ and $\subtypesTo{\Gamma}{T_2}{T_2'}$\\
                                Case $a$ is immediate.
                        \end{enumerate}
                    \item $T' = \match{T_s'}{T_i \Rightarrow T_i''}{i \in I}$, $\subtypesTo{\Gamma}{T_s'}{T_k}$,
                        $\forall i \in I : i < k \Rightarrow \disjoint{\Gamma}{T_s'}{T_i}$ and $T_k'' = T_1' \rightarrow T_2'$\\
                        By the definition of disjointness and \tsc{ST-Trans}, we have
                        $\forall i \in I : i < k \Rightarrow \disjoint{\Gamma}{T_s}{T_i}$
                        Also by \tsc{ST-Trans}, $\subtypesTo{\Gamma}{T_s}{T_k}$.
                        We then have from \tsc{ST-Match1} and \tsc{M-Match} that $\mutualTo{\Gamma}{T}{T_k'}$.
                        Now, we use the I.H. on $\subtypesTo{\Gamma}{T_k'}{T_k''}$ and obtain four subcases:
                        \begin{enumerate}
                            \item $\mutualTo{\Gamma}{T_k'}{T_1 \rightarrow T_2}$ for some $T_1$ and $T_2$ such that
                                $\subtypesTo{\Gamma}{T_1'}{T_1}$ and $\subtypesTo{\Gamma}{T_2}{T_2'}$\\
                                We have case $a$ from \tsc{M-Trans}.
                            \item $\mutualTo{\Gamma}{T_k'}{X}$\\
                                We have case $b$ from \tsc{M-Trans}.
                            \item $T_k'$ is a type variable\\
                                Case $b$ is immediate.
                            \item $T_k' = T_1 \rightarrow T_2$ for some $T_1$ and $T_2$ such that
                                $\subtypesTo{\Gamma}{T_1'}{T_1}$ and $\subtypesTo{\Gamma}{T_2}{T_2'}$\\
                                Case $a$ is immediate.
                        \end{enumerate}
                    \item Not applicable since $T'$ is a match type.
                    \item Not applicable since $T'$ is a match type.
                \end{enumerate}
            \item \tsc{ST-App}
                First case:
                \begin{equation}
                    \derive
                        {\subtypesTo{\Gamma}{T_2}{T_1}}
                        {\subtypesTo{\Gamma}{(\quant{X}{T_1} T_0)\ T_2}{\subs{T_0}{X}{T_2}}}
                \end{equation}
                where $T = (\quant{X}{T_1} T_0)\ T_2 \quad T' = \subs{T_0}{X}{T_2}$\\
                If $T' = T_1' \rightarrow T_2'$, then we have case $a$ from \tsc{M-TApp} and \tsc{ST-Refl}.\\
                If $\mutualTo{\Gamma}{T'}{T_1' \rightarrow T_2'}$, we have case $a$ from
                \tsc{M-TApp}, \tsc{M-Trans} and \tsc{ST-Refl}.\\\\
                Second case:
                \begin{equation}
                    \derive
                        {\subtypesTo{\Gamma}{T_2}{T_1}}
                        {\subtypesTo{\Gamma}{\subs{T_0}{X}{T_2}}{(\quant{X}{T_1} T_0)\ T_2}}
                \end{equation}
                where $T = \subs{T_0}{X}{T_2} \quad T' = (\quant{X}{T_1} T_0)\ T_2$\\
                We use the 1st part of the lemma on $\mutualTo{\Gamma}{T'}{T_1' \rightarrow T_2'}$ to get four subcases:
                \begin{enumerate}
                    \item Not applicable since $T'$ is a type application.
                    \item Not applicable since $T'$ is a type application
                    \item $T' = (\quant{X}{T_1} T_0)\ T_2$, $\subtypesTo{\Gamma}{T_2}{T_1}$ and
                        $\mutualTo{\Gamma}{\subs{T_0}{Y}{T_2}}{T_1' \rightarrow T_2'}$\\
                        Since $T = \subs{T_0}{Y}{T_2}$, we have case $a$ by \tsc{ST-Refl}.
                    \item $T' = (\quant{X}{T_1} T_0)\ T_2$, $\subtypesTo{\Gamma}{T_2}{T_1}$ and
                        $\subs{T_0}{X}{T_2} = T_1' \rightarrow T_2'$\\
                        We have case $d$ from \tsc{ST-Refl}.
                \end{enumerate}
            \item \tsc{ST-$\top$}, \tsc{ST-$\cup$Comm}, \tsc{ST-$\cup$Assoc}, \tsc{ST-List},
                \tsc{ST-$\cup$R}, \tsc{ST-Rec}, \tsc{ST-TAbs}\\
                These rules could not have been used.
                $T'$ cannot be a term abstraction (by the structure), nor can we derive $\mutualTo{\Gamma}{T'}{T_1' \rightarrow T_2}$
                since $T'$ must be either a match type or type application.
        \end{itemize}
    \item $\subtypesTo{\Gamma}{T}{T'}$ 
        \begin{itemize}
            \item \tsc{ST-Refl}
                where $T = T'$\\
                If $T' = (\quant{X}{T_1'} T_2')$, case $d$ follows from \tsc{ST-Refl}.
                If $\mutualTo{\Gamma}{T'}{(\quant{X}{T_1'} T_2')}$, case $a$ follows from \tsc{ST-Refl}.
            \item \tsc{ST-Val}
                By looking at the $\in_G$ relation rules, $T'$ cannot be a universal quantification, a match type or a type application,
                so this rule could not have been used.
            \item \tsc{ST-Var}
                where $T = X \quad \Gamma = \Gamma', \subtypes{X}{T'}, \Gamma''$\\
                Case $c$ is immediate.
            \item \tsc{ST-Trans}
                \begin{equation}
                    \derive
                        {\subtypesTo{\Gamma}{T}{T''} \quad \subtypesTo{\Gamma}{T''}{T'}}
                        {\subtypesTo{\Gamma}{T}{T'}}
                \end{equation}
                By the I.H. on the right premise, we have four subcases:
                \begin{enumerate}
                    \item $\mutualTo{\Gamma}{T''}{(\quant{X}{T_1''} T_2'')}$ for some $T_1''$ and $T_2''$
                        such that $\subtypesTo{\Gamma}{T_1'}{T_1''}$ and $\subtypesTo{\Gamma, \subtypes{X}{T_1'}}{T_2''}{T_2'}$.\\
                        The result follows by using the I.H. on the left premise.
                        If the result is case $a$ or $d$,
                        we have $\subtypesTo{\Gamma}{T_1'}{T_1}$ and $\subtypesTo{\Gamma, \subtypes{X}{T_1''}}{T_2}{T_2'}$ by
                        by \cref{narrowing} and \tsc{ST-Trans}.
                    \item $\mutualTo{\Gamma}{T''}{X}$ for some type variable $X$.\\
                        The result follows from using the 3rd part of the lemma on the left premise.
                    \item $T''$ is a type variable.
                        The result follows from using the 3rd part of the lemma on the left premise.
                    \item $T'' = (\quant{X}{T_1''} T_2'')$ for some $T_1''$ and $T_2''$
                        such that $\subtypesTo{\Gamma}{T_1'}{T_1''}$ and $\subtypesTo{\Gamma}{T_2''}{T_2'}$.\\
                        The result follows by using the I.H. on the left premise. If the result is case $a$ or $d$,
                        we have $\subtypesTo{\Gamma}{T_1'}{T_1}$ and $\subtypesTo{\Gamma}{T_2}{T_2'}$ by
                        \cref{narrowing} and \tsc{ST-Trans}.
                \end{enumerate}
            \item \tsc{ST-TAbs}\\
                Case $d$ is immediate.
            \item \tsc{ST-Match1}
                First case:
                \begin{equation}
                    \derive
                        {\subtypesTo{\Gamma}{T_s}{T_k} \quad \forall i \in I : i < k \Rightarrow \disjoint{\Gamma}{T_s}{T_i}}
                        {\subtypesTo{\Gamma}{\match{T_s}{T_i \Rightarrow T_i'}{i \in I}}{T_k'}}
                \end{equation}
                where $T = \match{T_s}{T_i \Rightarrow T_i'}{i \in I} \quad T' = T_k'$\\
                If $T' = \quant{X}{T_1'} T_2'$, we have case $a$ from \tsc{M-Match} and \tsc{ST-Refl}.\\
                If $\mutualTo{\Gamma}{T'}{\quant{X}{T_1'} T_2'}$, we have $\mutualTo{\Gamma}{T}{\quant{X}{T_1'} T_2'}$
                from \tsc{M-Match} and \tsc{M-Trans}. We then have case $a$ from \tsc{ST-Refl}.\\\\
                Second case:
                \begin{equation}
                    \derive
                        {\subtypesTo{\Gamma}{T_s}{T_k} \quad \forall i \in I : i < k \Rightarrow \disjoint{\Gamma}{T_s}{T_i}}
                        {\subtypesTo{\Gamma}{T_k'}{\match{T_s}{T_i \Rightarrow T_i'}{i \in I}}}
                \end{equation}
                where $T = T_k' \quad T' = \match{T_s}{T_i \Rightarrow T_i'}{i \in I}$\\
                We use the 1st part of the lemma on $\mutualTo{\Gamma}{T'}{\quant{X}{T_1'} T_2'}$ to get four subcases:
                \begin{enumerate}
                    \item $T' = \match{T_s}{T_i \Rightarrow T_i'}{i \in I}$, $\subtypesTo{\Gamma}{T_s}{T_k}$,
                        $\forall i \in I : i < k \Rightarrow \disjoint{\Gamma}{T_s}{T_i}$ and
                        $\mutualTo{\Gamma}{T_k'}{\quant{X}{T_1'} T_2'}$\\
                        By $T = T_k'$ we have $\mutualTo{\Gamma}{T}{\quant{X}{T_1'} T_2'}$.
                        Case $a$ then follows from \tsc{ST-Refl}.
                    \item $T' = \match{T_s}{T_i \Rightarrow T_i'}{i \in I}$, $\subtypesTo{\Gamma}{T_s}{T_k}$,
                        $\forall i \in I : i < k \Rightarrow \disjoint{\Gamma}{T_s}{T_i}$ and $T_k' = \quant{X}{T_1'} T_2'$\\
                        Case $d$ follows from \tsc{ST-Refl} with $T_1 = T_1'$ and $T_2 = T_2'$.
                    \item Not applicable since $T'$ is a match type.
                    \item Not applicable since $T'$ is a match type.
                \end{enumerate}
            \item \tsc{ST-Match2}
                \begin{equation}
                    \derive
                        {\subtypesTo{\Gamma}{T_s}{T_s'} \quad \forall i \in I : \subtypesTo{\Gamma}{T_i'}{T_i''}}
                        {\subtypesTo{\Gamma}{\match{T_s}{T_i \Rightarrow T_i'}{i \in I}}{\match{T_s'}{T_i \Rightarrow T_i''}{i \in I}}}
                \end{equation}
                where $T = \match{T_s}{T_i \Rightarrow T_i'}{i \in I} \quad T' = \match{T_s'}{T_i \Rightarrow T_i''}{i \in I}$\\
                We use the 1st part of the lemma on $\mutualTo{\Gamma}{T'}{\quant{X}{T_1'} T_2'}$ to get four subcases:
                \begin{enumerate}
                    \item $T' = \match{T_s'}{T_i \Rightarrow T_i''}{i \in I}$, $\subtypesTo{\Gamma}{T_s'}{T_k}$,
                        $\forall i \in I : i < k \Rightarrow \disjoint{\Gamma}{T_s'}{T_i}$ and
                        $\mutualTo{\Gamma}{T_k''}{\quant{X}{T_1'} T_2'}$\\
                        By the definition of disjointness and \tsc{ST-Trans}, we have
                        $\forall i \in I : i < k \Rightarrow \disjoint{\Gamma}{T_s}{T_i}$
                        Also by \tsc{ST-Trans}, $\subtypesTo{\Gamma}{T_s}{T_k}$.
                        We then have from \tsc{ST-Match1} and \tsc{M-Match} that $\mutualTo{\Gamma}{T}{T_k'}$.
                        Now, we use the I.H. on $\subtypesTo{\Gamma}{T_k'}{T_k''}$ and $\mutualTo{\Gamma}{T_k''}{\quant{X}{T_1'} T_2'}$
                        and obtain four subsubcases:
                        \begin{enumerate}
                            \item $\mutualTo{\Gamma}{T_k'}{\quant{X}{T_1} T_2}$ for some $T_1$ and $T_2$ such that
                                $\subtypesTo{\Gamma}{T_1'}{T_1}$ and $\subtypesTo{\Gamma, \subtypes{X}{T_1'}}{T_2}{T_2'}$\\
                                We have case $a$ from \tsc{M-Trans}.
                            \item $\mutualTo{\Gamma}{T_k'}{X}$\\
                                We have case $b$ from \tsc{M-Trans}.
                            \item $T_k'$ is a type variable\\
                                Case $b$ is immediate.
                            \item $T_k' = \quant{X}{T_1} T_2$ for some $T_1$ and $T_2$ such that
                                $\subtypesTo{\Gamma}{T_1'}{T_1}$ and $\subtypesTo{\Gamma, \subtypes{X}{T_1'}}{T_2}{T_2'}$\\
                                Case $a$ is immediate.
                        \end{enumerate}
                    \item $T' = \match{T_s'}{T_i \Rightarrow T_i''}{i \in I}$, $\subtypesTo{\Gamma}{T_s'}{T_k}$,
                        $\forall i \in I : i < k \Rightarrow \disjoint{\Gamma}{T_s'}{T_i}$ and $T_k'' = \quant{X}{T_1'} T_2'$\\
                        By the definition of disjointness and \tsc{ST-Trans}, we have
                        $\forall i \in I : i < k \Rightarrow \disjoint{\Gamma}{T_s}{T_i}$
                        Also by \tsc{ST-Trans}, $\subtypesTo{\Gamma}{T_s}{T_k}$.
                        We then have from \tsc{ST-Match1} and \tsc{M-Match} that $\mutualTo{\Gamma}{T}{T_k'}$.
                        Now, we use the I.H. on $\subtypesTo{\Gamma}{T_k'}{T_k''}$ and obtain 4 subcases:
                        \begin{enumerate}
                            \item $\mutualTo{\Gamma}{T_k'}{\quant{X}{T_1} T_2}$ for some $T_1$ and $T_2$ such that
                                $\subtypesTo{\Gamma}{T_1'}{T_1}$ and $\subtypesTo{\Gamma, \subtypes{X}{T_1'}}{T_2}{T_2'}$\\
                                We have case $a$ from \tsc{M-Trans}.
                            \item $\mutualTo{\Gamma}{T_k'}{X}$\\
                                We have case $b$ from \tsc{M-Trans}.
                            \item $T_k'$ is a type variable\\
                                Case $b$ is immediate.
                            \item $T_k' = \quant{X}{T_1} T_2$ for some $T_1$ and $T_2$ such that
                                $\subtypesTo{\Gamma}{T_1'}{T_1}$ and $\subtypesTo{\Gamma, \subtypes{X}{T_1'}}{T_2}{T_2'}$\\
                                Case $a$ is immediate.
                        \end{enumerate}
                    \item Not applicable since $T'$ is a match type.
                    \item Not applicable since $T'$ is a match type.
                \end{enumerate}
            \item \tsc{ST-App}
                First case:
                \begin{equation}
                    \derive
                        {\subtypesTo{\Gamma}{T_2}{T_1}}
                        {\subtypesTo{\Gamma}{(\quant{X}{T_1} T_0)\ T_2}{\subs{T_0}{X}{T_2}}}
                \end{equation}
                where $T = (\quant{X}{T_1} T_0)\ T_2 \quad T' = \subs{T_0}{X}{T_2}$\\
                If $T' = \quant{Y}{T_1'} T_2'$, then we have case $a$ from \tsc{M-TApp} and \tsc{ST-Refl}.\\
                If $\mutualTo{\Gamma}{T'}{T_1' \rightarrow T_2'}$, we have case $a$ from
                \tsc{M-TApp}, \tsc{M-Trans} and \tsc{ST-Refl}.\\\\
                Second case:
                \begin{equation}
                    \derive
                        {\subtypesTo{\Gamma}{T_2}{T_1}}
                        {\subtypesTo{\Gamma}{\subs{T_0}{X}{T_2}}{(\quant{X}{T_1} T_0)\ T_2}}
                \end{equation}
                where $T = \subs{T_0}{X}{T_2} \quad T' = (\quant{X}{T_1} T_0)\ T_2$\\
                We use the 1st part of the lemma on $\mutualTo{\Gamma}{T'}{T_1' \rightarrow T_2'}$ to get four subcases:
                \begin{enumerate}
                    \item Not applicable since $T'$ is a type application.
                    \item Not applicable since $T'$ is a type application
                    \item $T' = (\quant{X}{T_1} T_0)\ T_2$, $\subtypesTo{\Gamma}{T_2}{T_1}$ and
                        $\mutualTo{\Gamma}{\subs{T_0}{Y}{T_2}}{T_1' \rightarrow T_2'}$\\
                        Since $T = \subs{T_0}{Y}{T_2}$, we have case $a$ by \tsc{ST-Refl}.
                    \item $T' = (\quant{X}{T_1} T_0)\ T_2$, $\subtypesTo{\Gamma}{T_2}{T_1}$ and
                        $\subs{T_0}{X}{T_2} = T_1' \rightarrow T_2'$\\
                        We have case $d$ from \tsc{ST-Refl}.
                \end{enumerate}
            \item \tsc{ST-$\top$}, \tsc{ST-$\cup$Comm}, \tsc{ST-$\cup$Assoc}, \tsc{ST-List},
                \tsc{ST-$\cup$R}, \tsc{ST-Rec}, \tsc{ST-Abs}
                These rules could not have been used.
                $T'$ cannot be a type abstraction (by the structure), nor can we derive $\mutualTo{\Gamma}{T'}{\quant{X}{T_1'} T_2'}$
                since $T'$ must be either a match type or type application.
        \end{itemize}
    \item $\subtypesTo{\Gamma}{T}{T'}$ 
        \begin{itemize}
            \item \tsc{ST-Refl}\\
                where $T = T'$\\
                If $T' = \curl{f_i : T_i'}_{i \in I}$, case $f$ follows from \tsc{ST-Refl}.
                If $\mutualTo{\Gamma}{T'}{\curl{f_i : T_i'}}$, case $a$ follows from \tsc{ST-Refl}.
            \item \tsc{ST-Val}\\
                \begin{equation}
                    \derive
                        {\singleton{v_G} \in T'}
                        {\subtypesTo{\Gamma}{\singleton{v_G}}{T'}}
                \end{equation}
                where $T = \singleton{v_G}$\\
                Case $e$ is immediate.
            \item \tsc{ST-Var}\\
                where $T = X \quad \Gamma = \Gamma', \subtypes{X}{T'}, \Gamma''$\\
                Case $c$ is immediate.
            \item \tsc{ST-Trans}\\
                \begin{equation}
                    \derive
                        {\subtypesTo{\Gamma}{T}{T''} \quad \subtypesTo{\Gamma}{T''}{T'}}
                        {\subtypesTo{\Gamma}{T}{T'}}
                \end{equation}
                By the I.H. on the right premise, we have six subcases:
                \begin{enumerate}
                    \item $\mutualTo{\Gamma}{T''}{\curl{f_i : T_i''}_{i \in I}}$ such that
                        $\forall i \in I : \subtypesTo{\Gamma}{T_i''}{T_i'}$\\
                        The result follows by using the I.H. on the left premise.
                        If the result is case $a$ or $f$,
                        we have $\forall i \in I : \subtypesTo{\Gamma}{T_i}{T_i'}$ by \tsc{ST-Trans}.
                    \item $\mutualTo{\Gamma}{T''}{X}$ for some type variable $X$.\\
                        The result follows from using the 3rd part of the lemma on the left premise.
                    \item $T''$ is a type variable.
                        The result follows from using the 3rd part of the lemma on the left premise.
                    \item $\mutualTo{\Gamma}{T''}{\singleton{v_G}}$ for some ground type $\singleton{v_G}$.\\
                        The result follows from using the 2nd part of the lemma on the left premise.
                    \item $T''$ is a ground type.\\
                        The result follows from using the 2nd part of the lemma on the left premise.
                    \item $T'' = \curl{f_i : T_i''}_{i \in I}$ such that
                        $\forall i \in I : \subtypesTo{\Gamma}{T_i''}{T_i'}$\\
                        The result follows by using the I.H. on the left premise.
                        If the result is case $a$ or $f$,
                        we have $\forall i \in I : \subtypesTo{\Gamma}{T_i}{T_i'}$ by \tsc{ST-Trans}.
                \end{enumerate}
            \item \tsc{ST-Rec}\\
                Case $f$ is immediate.
            \item \tsc{ST-Match1}\\
                First case:
                \begin{equation}
                    \derive
                        {\subtypesTo{\Gamma}{T_s}{T_k} \quad \forall j \in J : j < k \Rightarrow \disjoint{\Gamma}{T_s}{T_j}}
                        {\subtypesTo{\Gamma}{\match{T_s}{T_j \Rightarrow T_j'}{j \in J}}{T_k'}}
                \end{equation}
                where $T = \match{T_s}{T_j \Rightarrow T_j'}{j \in J} \quad T' = T_k'$\\
                If $T' = \curl{f_i : T_i'}_{i \in I}$, we have case $a$ from \tsc{M-Match} and \tsc{ST-Refl}.\\
                If $\mutualTo{\Gamma}{T'}{\curl{f_i : T_i'}_{i \in I}}$, we have $\mutualTo{\Gamma}{T}{\curl{f_i : T_i'}_{i \in I}}$
                from \tsc{M-Match} and \tsc{M-Trans}. We then have case $a$ from \tsc{ST-Refl}.\\\\
                Second case:
                \begin{equation}
                    \derive
                        {\subtypesTo{\Gamma}{T_s}{T_k} \quad \forall j \in J : j < k \Rightarrow \disjoint{\Gamma}{T_s}{T_j}}
                        {\subtypesTo{\Gamma}{T_k'}{\match{T_s}{T_j \Rightarrow T_j'}{j \in J}}}
                \end{equation}
                where $T = T_k' \quad T' = \match{T_s}{T_j \Rightarrow T_j'}{j \in J}$\\
                We use the 1st part of the lemma on $\mutualTo{\Gamma}{T'}{\curl{f_i : T_i'}_{i \in I}}$ to get four subcases:
                \begin{enumerate}
                    \item $T' = \match{T_s}{T_j \Rightarrow T_j'}{j \in J}$, $\subtypesTo{\Gamma}{T_s}{T_k}$,
                        $\forall j \in J : j < k \Rightarrow \disjoint{\Gamma}{T_s}{T_j}$ and
                        $\mutualTo{\Gamma}{T_k'}{\curl{f_i : T_i'}_{i \in I}}$\\
                        By $T = T_k'$ we have $\mutualTo{\Gamma}{T}{\curl{f_i : T_i'}_{i \in I}}$.
                        Case $a$ then follows from \tsc{ST-Refl}.
                    \item $T' = \match{T_s}{T_j \Rightarrow T_j'}{j \in J}$, $\subtypesTo{\Gamma}{T_s}{T_k}$,
                        $\forall j \in J : j < k \Rightarrow \disjoint{\Gamma}{T_s}{T_j}$ and $T_k' = \curl{f_i : T_i'}_{i \in I}$\\
                        Case $f$ follows from \tsc{ST-Refl}.
                    \item Not applicable since $T'$ is a match type.
                    \item Not applicable since $T'$ is a match type.
                \end{enumerate}
            \item \tsc{ST-Match2}\\
                \begin{equation}
                    \derive
                        {\subtypesTo{\Gamma}{T_s}{T_s'} \quad \forall j \in J : \subtypesTo{\Gamma}{T_j'}{T_j''}}
                        {\subtypesTo{\Gamma}{\match{T_s}{T_j \Rightarrow T_j'}{j \in J}}{\match{T_s'}{T_j \Rightarrow T_j''}{j \in J}}}
                \end{equation}
                where $T = \match{T_s}{T_j \Rightarrow T_j'}{j \in J} \quad T' = \match{T_s'}{T_j \Rightarrow T_j''}{j \in J}$\\
                We use the 1st part of the lemma on $\mutualTo{\Gamma}{T'}{\curl{f_i : T_i'}_{i \in I}}$ to get four subcases:
                \begin{enumerate}
                    \item $T' = \match{T_s'}{T_j \Rightarrow T_j''}{j \in J}$, $\subtypesTo{\Gamma}{T_s'}{T_k}$,
                        $\forall j \in J : j < k \Rightarrow \disjoint{\Gamma}{T_s'}{T_j}$ and
                        $\mutualTo{\Gamma}{T_k''}{\curl{f_i : T_i'}_{i \in I}}$\\
                        By the definition of disjointness and \tsc{ST-Trans}, we have
                        $\forall j \in J : j < k \Rightarrow \disjoint{\Gamma}{T_s}{T_j}$
                        Also by \tsc{ST-Trans}, $\subtypesTo{\Gamma}{T_s}{T_k}$.
                        We then have from \tsc{ST-Match1} and \tsc{M-Match} that $\mutualTo{\Gamma}{T}{T_k'}$.
                        Now, we use the I.H. on $\subtypesTo{\Gamma}{T_k'}{T_k''}$ and $\mutualTo{\Gamma}{T_k''}{\curl{f_i : T_i'}_{i \in I}}$
                        and obtain six subsubcases:
                        \begin{enumerate}
                            \item $\mutualTo{\Gamma}{T_k'}{\curl{f_i : T_i}_{i \in I}}$ such that
                                $\forall i \in I : \subtypesTo{\Gamma}{T_i}{T_i'}$\\
                                We have case $a$ from \tsc{M-Trans}.
                            \item $\mutualTo{\Gamma}{T_k'}{X}$\\
                                We have case $b$ from \tsc{M-Trans}.
                            \item $T_k'$ is a type variable\\
                                Case $b$ is immediate.
                            \item $\mutualTo{\Gamma}{T_k'}{\singleton{v_G}}$\\
                                We have case $d$ from \tsc{M-Trans}.
                            \item $T_k'$ is a ground variable\\
                                Case $d$ is immediate.
                            \item $T_k' = \curl{f_i : T_i}_{i \in I}$ such that
                                $\forall i \in I : \subtypesTo{\Gamma}{T_i}{T_i'}$\\
                                Case $a$ is immediate.
                        \end{enumerate}
                    \item $T' = \match{T_s'}{T_j \Rightarrow T_j''}{j \in I}$, $\subtypesTo{\Gamma}{T_s'}{T_k}$,
                        $\forall j \in J : j < k \Rightarrow \disjoint{\Gamma}{T_s'}{T_j}$ and $T_k'' = \curl{f_i : T_i}_{i \in I}$\\
                        By the definition of disjointness and \tsc{ST-Trans}, we have
                        $\forall j \in J : j < k \Rightarrow \disjoint{\Gamma}{T_s}{T_j}$
                        Also by \tsc{ST-Trans}, $\subtypesTo{\Gamma}{T_s}{T_k}$.
                        We then have from \tsc{ST-Match1} and \tsc{M-Match} that $\mutualTo{\Gamma}{T}{T_k'}$.
                        Now, we use the I.H. on $\subtypesTo{\Gamma}{T_k'}{T_k''}$ and obtain six subcases:
                        \begin{enumerate}
                            \item $\mutualTo{\Gamma}{T_k'}{\curl{f_i : T_i}_{i \in I}}$ such that
                                $\forall i \in I : \subtypesTo{\Gamma}{T_i}{T_i'}$\\
                                We have case $a$ from \tsc{M-Trans}.
                            \item $\mutualTo{\Gamma}{T_k'}{X}$\\
                                We have case $b$ from \tsc{M-Trans}.
                            \item $T_k'$ is a type variable\\
                                Case $b$ is immediate.
                            \item $\mutualTo{\Gamma}{T_k'}{\singleton{v_G}}$\\
                                We have case $d$ from \tsc{M-Trans}.
                            \item $T_k'$ is a ground variable\\
                                Case $d$ is immediate.
                            \item $T_k' = \curl{f_i : T_i}_{i \in I}$ such that
                                $\forall i \in I : \subtypesTo{\Gamma}{T_i}{T_i'}$\\
                                Case $a$ is immediate.
                        \end{enumerate}
                    \item Not applicable since $T'$ is a match type.
                    \item Not applicable since $T'$ is a match type.
                \end{enumerate}
            \item \tsc{ST-App}\\
                First case:
                \begin{equation}
                    \derive
                        {\subtypesTo{\Gamma}{T_2}{T_1}}
                        {\subtypesTo{\Gamma}{(\quant{X}{T_1} T_0)\ T_2}{\subs{T_0}{X}{T_2}}}
                \end{equation}
                where $T = (\quant{X}{T_1} T_0)\ T_2 \quad T' = \subs{T_0}{X}{T_2}$\\
                If $T' = \curl{f_i : T_i'}_{i \in I}$, then we have case $a$ from \tsc{M-TApp} and \tsc{ST-Refl}.\\
                If $\mutualTo{\Gamma}{T'}{\curl{f_i : T_i'}_{i \in I}}$, we have case $a$ from
                \tsc{M-TApp}, \tsc{M-Trans} and \tsc{ST-Refl}.\\\\
                Second case:
                \begin{equation}
                    \derive
                        {\subtypesTo{\Gamma}{T_2}{T_1}}
                        {\subtypesTo{\Gamma}{\subs{T_0}{X}{T_2}}{(\quant{X}{T_1} T_0)\ T_2}}
                \end{equation}
                where $T = \subs{T_0}{X}{T_2} \quad T' = (\quant{X}{T_1} T_0)\ T_2$\\
                We use the 1st part of the lemma on $\mutualTo{\Gamma}{T'}{\curl{f_i : T_i'}_{i \in I}}$ to get four subcases:
                \begin{enumerate}
                    \item Not applicable since $T'$ is a type application.
                    \item Not applicable since $T'$ is a type application
                    \item $T' = (\quant{X}{T_1} T_0)\ T_2$, $\subtypesTo{\Gamma}{T_2}{T_1}$ and
                        $\mutualTo{\Gamma}{\subs{T_0}{X}{T_2}}{\curl{f_i : T_i'}_{i \in I}}$\\
                        Since $T = \subs{T_0}{X}{T_2}$, we have case $a$ by \tsc{ST-Refl}.
                    \item $T' = (\quant{X}{T_1} T_0)\ T_2$, $\subtypesTo{\Gamma}{T_2}{T_1}$ and
                        $\subs{T_0}{X}{T_2} = \curl{f_i : T_i'}_{i \in I}$\\
                        We have case $f$ from \tsc{ST-Refl}.
                \end{enumerate}
            \item \tsc{ST-$\top$}, \tsc{ST-$\cup$Comm}, \tsc{ST-$\cup$Assoc}, \tsc{ST-List},
                \tsc{ST-$\cup$R}, \tsc{ST-Abs}, \tsc{ST-TAbs}
                These rules could not have been used.
                $T'$ cannot be a record (by the structure), nor can we derive $\mutualTo{\Gamma}{T'}{\curl{f_i : T_i'}_{i \in I}}$
                since $T'$ must be either a match type or type application.
        \end{itemize}
    \item $\subtypesTo{\Gamma}{T}{T'}$ 
        \begin{itemize}
            \item \tsc{ST-Refl}\\
                where $T = T'$\\
                If $T' = \listT{T_l'}$, case $f$ follows from \tsc{ST-Refl}.
                If $\mutualTo{\Gamma}{T'}{\listT{T_l'}}$, case $a$ follows from \tsc{ST-Refl}.
            \item \tsc{ST-Val}\\
                where $T = \singleton{v_G}$\\
                Case $e$ is immediate.
            \item \tsc{ST-Var}\\
                Case $c$ is immediate.
            \item \tsc{ST-Trans}\\
                \begin{equation}
                    \derive
                        {\subtypesTo{\Gamma}{T}{T''} \quad \subtypesTo{\Gamma}{T''}{T'}}
                        {\subtypesTo{\Gamma}{T}{T'}}
                \end{equation}
                By the I.H. on the right premise, we have six subcases:
                \begin{enumerate}
                    \item $\mutualTo{\Gamma}{T''}{\listT{T_l''}}$ such that
                        $\subtypesTo{\Gamma}{T_l''}{T_l'}$\\
                        The result follows by using the I.H. on the left premise.
                        If the result is case $a$ or $f$,
                        we have $\subtypesTo{\Gamma}{T_l}{T_l'}$ by \tsc{ST-Trans}.
                    \item $\mutualTo{\Gamma}{T''}{X}$ for some type variable $X$.\\
                        The result follows from using the 3rd part of the lemma on the left premise.
                    \item $T''$ is a type variable.
                        The result follows from using the 3rd part of the lemma on the left premise.
                    \item $\mutualTo{\Gamma}{T''}{\singleton{v_G}}$ for some ground type $\singleton{v_G}$.\\
                        The result follows from using the 2nd part of the lemma on the left premise.
                    \item $T''$ is a ground type.\\
                        The result follows from using the 2nd part of the lemma on the left premise.
                    \item $T'' = \listT{T_l''}$ such that
                        $\subtypesTo{\Gamma}{T_l''}{T_l'}$\\
                        The result follows by using the I.H. on the left premise.
                        If the result is case $a$ or $f$,
                        we have $\subtypesTo{\Gamma}{T_l}{T_l'}$ by \tsc{ST-Trans}.
                \end{enumerate}
            \item \tsc{ST-List}\\
                Case $f$ is immediate.
            \item \tsc{ST-Match1}\\
                First case:
                \begin{equation}
                    \derive
                        {\subtypesTo{\Gamma}{T_s}{T_k} \quad \forall i \in I : i < k \Rightarrow \disjoint{\Gamma}{T_s}{T_i}}
                        {\subtypesTo{\Gamma}{\match{T_s}{T_i \Rightarrow T_i'}{i \in I}}{T_k'}}
                \end{equation}
                where $T = \match{T_s}{T_i \Rightarrow T_i'}{i \in I} \quad T' = T_k'$\\
                If $T' = \listT{T_l'}$, we have case 1 from \tsc{M-Match} and \tsc{ST-Refl}.\\
                If $\mutualTo{\Gamma}{T'}{\listT{T_l'}}$, we have $\mutualTo{\Gamma}{T}{\listT{T_l'}}$
                from \tsc{M-Match} and \tsc{M-Trans}. We then have case $a$ from \tsc{ST-Refl}.\\\\
                Second case:
                \begin{equation}
                    \derive
                        {\subtypesTo{\Gamma}{T_s}{T_k} \quad \forall i \in I : i < k \Rightarrow \disjoint{\Gamma}{T_s}{T_i}}
                        {\subtypesTo{\Gamma}{T_k'}{\match{T_s}{T_i \Rightarrow T_i'}{i \in I}}}
                \end{equation}
                where $T = T_k' \quad T' = \match{T_s}{T_i \Rightarrow T_i'}{i \in I}$\\
                We use the 1st part of the lemma on $\mutualTo{\Gamma}{T'}{\listT{T_l'}}$ to get four subcases:
                \begin{enumerate}
                    \item $T' = \match{T_s}{T_i \Rightarrow T_i'}{i \in I}$, $\subtypesTo{\Gamma}{T_s}{T_k}$,
                        $\forall i \in I : i < k \Rightarrow \disjoint{\Gamma}{T_s}{T_i}$ and
                        $\mutualTo{\Gamma}{T_k'}{\listT{T_l'}}$\\
                        By $T = T_k'$ we have $\mutualTo{\Gamma}{T}{\listT{T_l'}}$.
                        Case $a$ then follows from \tsc{ST-Refl}.
                    \item $T' = \match{T_s}{T_i \Rightarrow T_i'}{i \in I}$, $\subtypesTo{\Gamma}{T_s}{T_k}$,
                        $\forall i \in I : i < k \Rightarrow \disjoint{\Gamma}{T_s}{T_i}$ and $T_k' = \listT{T_l'}$\\
                        Case $f$ follows from \tsc{ST-Refl}.
                    \item Not applicable since $T'$ is a match type.
                    \item Not applicable since $T'$ is a match type.
                \end{enumerate}
            \item \tsc{ST-Match2}\\
                \begin{equation}
                    \derive
                        {\subtypesTo{\Gamma}{T_s}{T_s'} \quad \forall i \in I : \subtypesTo{\Gamma}{T_i'}{T_i''}}
                        {\subtypesTo{\Gamma}{\match{T_s}{T_i \Rightarrow T_i'}{i \in I}}{\match{T_s'}{T_i \Rightarrow T_i''}{i \in I}}}
                \end{equation}
                where $T = \match{T_s}{T_i \Rightarrow T_i'}{i \in I} \quad T' = \match{T_s'}{T_i \Rightarrow T_i''}{i \in I}$\\
                We use the 1st part of the lemma on $\mutualTo{\Gamma}{T'}{\listT{T_l'}}$ to get four subcases:
                \begin{enumerate}
                    \item $T' = \match{T_s'}{T_i \Rightarrow T_i''}{i \in I}$, $\subtypesTo{\Gamma}{T_s'}{T_k}$,
                        $\forall i \in I : i < k \Rightarrow \disjoint{\Gamma}{T_s'}{T_i}$ and
                        $\mutualTo{\Gamma}{T_k''}{\listT{T_l'}}$\\
                        By the definition of disjointness and \tsc{ST-Trans}, we have
                        $\forall i \in I : i < k \Rightarrow \disjoint{\Gamma}{T_s}{T_i}$
                        Also by \tsc{ST-Trans}, $\subtypesTo{\Gamma}{T_s}{T_k}$.
                        We then have from \tsc{ST-Match1} and \tsc{M-Match} that $\mutualTo{\Gamma}{T}{T_k'}$.
                        Now, we use the I.H. on $\subtypesTo{\Gamma}{T_k'}{T_k''}$ and $\mutualTo{\Gamma}{T_k''}{\listT{T_l'}}$
                        and obtain six subsubcases:
                        \begin{enumerate}
                            \item $\mutualTo{\Gamma}{T_k'}{\listT{T_l}}$ for some $T_l$ such that
                                $\subtypesTo{\Gamma}{T_l}{T_l'}$\\
                                We have case $a$ from \tsc{M-Trans}.
                            \item $\mutualTo{\Gamma}{T_k'}{X}$\\
                                We have case $b$ from \tsc{M-Trans}.
                            \item $T_k'$ is a type variable\\
                                Case $b$ is immediate.
                            \item $\mutualTo{\Gamma}{T_k'}{\singleton{v_G}}$\\
                                We have case $d$ from \tsc{M-Trans}.
                            \item $T_k'$ is a ground variable\\
                                Case $d$ is immediate.
                            \item $T_k' = \listT{T_l}$ for some $T_l$ such that $\subtypesTo{\Gamma}{T_l}{T_l'}$\\
                                Case $1$ is immediate.
                        \end{enumerate}
                    \item $T' = \match{T_s'}{T_i \Rightarrow T_i''}{i \in I}$, $\subtypesTo{\Gamma}{T_s'}{T_k}$,
                        $\forall i \in I : i < k \Rightarrow \disjoint{\Gamma}{T_s'}{T_i}$ and $T_k'' = \listT{T_l'}$\\
                        By the definition of disjointness and \tsc{ST-Trans}, we have
                        $\forall i \in I : i < k \Rightarrow \disjoint{\Gamma}{T_s}{T_i}$
                        Also by \tsc{ST-Trans}, $\subtypesTo{\Gamma}{T_s}{T_k}$.
                        We then have from \tsc{ST-Match1} and \tsc{M-Match} that $\mutualTo{\Gamma}{T}{T_k'}$.
                        Now, we use the I.H. on $\subtypesTo{\Gamma}{T_k'}{T_k''}$ and obtain six subcases:
                        \begin{enumerate}
                            \item $\mutualTo{\Gamma}{T_k'}{\listT{T_l}}$ for some $T_l$ such that
                                $\subtypesTo{\Gamma}{T_l}{T_l'}$\\
                                We have case $a$ from \tsc{M-Trans}.
                            \item $\mutualTo{\Gamma}{T_k'}{X}$\\
                                We have case $b$ from \tsc{M-Trans}.
                            \item $T_k'$ is a type variable\\
                                Case $b$ is immediate.
                            \item $\mutualTo{\Gamma}{T_k'}{\singleton{v_G}}$\\
                                We have case $d$ from \tsc{M-Trans}.
                            \item $T_k'$ is a ground variable\\
                                Case $d$ is immediate.
                            \item $T_k' = \listT{T_l}$ for some $T_l$ such that $\subtypesTo{\Gamma}{T_l}{T_l'}$\\
                                Case $a$ is immediate.
                        \end{enumerate}
                    \item Not applicable since $T'$ is a match type.
                    \item Not applicable since $T'$ is a match type.
                \end{enumerate}
            \item \tsc{ST-App}\\
            First case:
                \begin{equation}
                    \derive
                        {\subtypesTo{\Gamma}{T_2}{T_1}}
                        {\subtypesTo{\Gamma}{(\quant{X}{T_1} T_0)\ T_2}{\subs{T_0}{X}{T_2}}}
                \end{equation}
                where $T = (\quant{X}{T_1} T_0)\ T_2 \quad T' = \subs{T_0}{X}{T_2}$\\
                If $T' = \listT{T_l'}$, then we have case $a$ from \tsc{M-TApp} and \tsc{ST-Refl}.\\
                If $\mutualTo{\Gamma}{T'}{\listT{T_l'}}$, we have case $a$ from
                \tsc{M-TApp}, \tsc{M-Trans} and \tsc{ST-Refl}.\\\\
                Second case:
                \begin{equation}
                    \derive
                        {\subtypesTo{\Gamma}{T_2}{T_1}}
                        {\subtypesTo{\Gamma}{\subs{T_0}{X}{T_2}}{(\quant{X}{T_1} T_0)\ T_2}}
                \end{equation}
                where $T = \subs{T_0}{X}{T_2} \quad T' = (\quant{X}{T_1} T_0)\ T_2$\\
                We use the 1st part of the lemma on $\mutualTo{\Gamma}{T'}{\listT{T_l'}}$ to get four subcases:
                \begin{enumerate}
                    \item Not applicable since $T'$ is a type application.
                    \item Not applicable since $T'$ is a type application
                    \item $T' = (\quant{X}{T_1} T_0)\ T_2$, $\subtypesTo{\Gamma}{T_2}{T_1}$ and
                        $\mutualTo{\Gamma}{\subs{T_0}{X}{T_2}}{\listT{T_l'}}$\\
                        Since $T = \subs{T_0}{X}{T_2}$, we have case $a$ by \tsc{ST-Refl}.
                    \item $T' = (\quant{X}{T_1} T_0)\ T_2$, $\subtypesTo{\Gamma}{T_2}{T_1}$ and
                        $\subs{T_0}{X}{T_2} = \listT{T_l'}$\\
                        We have case $f$ from \tsc{ST-Refl}.
                \end{enumerate}
            \item \tsc{ST-$\top$}, \tsc{ST-$\cup$Comm}, \tsc{ST-$\cup$Assoc}, \tsc{ST-Rec},
                \tsc{ST-$\cup$R}, \tsc{ST-Abs}, \tsc{ST-TAbs}
                These rules could not have been used.
                $T'$ cannot be a list (by the structure), nor can we derive $\mutualTo{\Gamma}{T'}{[T_l']}$
                since $T'$ must be either a match type or type application.
        \end{itemize}
\end{enumerate}
\end{proof}

\noindent
\begin{lemma}[Typing inversion]
\label{typing-inversion}\hfill
\begin{enumerate}
    \item If $\typesTo{\Gamma}{\abs{x}{T_1} t_2}{T}$ and $\subtypesTo{\Gamma}{T}{T_1' \rightarrow T_2'}$
        then $\subtypesTo{\Gamma}{T_1'}{T_1}$ and there is a $T_2$ such that $\typesTo{\Gamma, \types{x}{T_1}}{t_2}{T_2}$
        and $\subtypesTo{\Gamma}{T_2}{T_2'}$.
    \item If $\typesTo{\Gamma}{\tabs{X}{T_1} t_2}{T}$ and $\subtypesTo{\Gamma}{T}{\quant{X}{T_1'} T_2'}$
        then $T_1 = T_1'$ and there is a $T_2$ such that $\typesTo{\Gamma, \subtypes{X}{T_1}}{t_2}{T_2}$
        and $\subtypesTo{\Gamma, \subtypes{X}{T_1}}{T_2}{T_2'}$.
    \item If $\typesTo{\Gamma}{\concat{t}{t'}}{T}$ and $\subtypesTo{\Gamma}{T}{\listT{T'}}$
        then there is a $T''$ such that $\subtypesTo{\Gamma}{T''}{T'}$, ${\typesTo{\Gamma}{t}{T''}}$
        and ${\typesTo{\Gamma}{t'}{\listT{T''}}}$.
    \item If $\typesTo{\Gamma}{\curl{f_i = t_i}_{i \in I}}{T}$ and $\subtypesTo{\Gamma}{T}{\curl{\types{f_i}{T_i'}}_{i \in I}}$
        then $\forall i \in I, \exists T_i : \typesTo{\Gamma}{t_i}{T_i} \wedge \subtypesTo{\Gamma}{T_i}{T_i'}$.
\end{enumerate}
\end{lemma}
\begin{proof}
\begin{enumerate}
    \item By induction on the derivation of $\typesTo{\Gamma}{\abs{x}{T_1} t_2}{T}$.
        \begin{itemize}
            \item \tsc{T-Abs}
                \begin{equation}
                    \derive
                        {\typesTo{\Gamma, \types{x}{T_1}}{t_2}{T_2''}}
                        {\typesTo{\Gamma}{\abs{x}{T_1} t_2}{T_1 \rightarrow T_2''}}
                \end{equation}
                We have that $\subtypesTo{\Gamma}{T_1 \rightarrow T_2''}{T_1' \rightarrow T_2'}$.
                By \cref{subtyping-inversion}, we have $\subtypesTo{\Gamma}{T_1'}{T_1}$ and $\subtypesTo{\Gamma}{T_2''}{T_2'}$.
                We take $T_2 = T_2''$ to obtain the desired result.
            \item \tsc{T-Sub}
                \begin{equation}
                    \derive
                        {\typesTo{\Gamma}{\abs{x}{T_1} t_2}{T'} \subtypes{\Gamma}{T'}{T}}
                        {\typesTo{\Gamma}{\abs{x}{T_1} t_2}{T}}
                \end{equation}
                By \tsc{ST-Trans}, $\subtypesTo{\Gamma}{T'}{T_1' \rightarrow T_2'}$.
                The conclusion is then immediate after applying the I.H..
            \item In any other case, the term cannot be an abstraction.
        \end{itemize}
    \item By induction on the derivation of $\typesTo{\Gamma}{\tabs{X}{T_1} t_2}{T}$.
        \begin{itemize}
            \item \tsc{T-TAbs}
                \begin{equation}
                    \derive
                        {\typesTo{\Gamma, \subtypes{X}{T_1}}{t_2}{T_2''}}
                        {\typesTo{\Gamma}{(\tabs{X}{T_1} t_2)}{(\quant{X}{T_1} T_2'')}}
                \end{equation}
                We have that $\subtypesTo{(\quant{X}{T_1} T_2'')}{(\quant{X}{T_1'} T_2')}$.
                By \cref{subtyping-inversion}, we have $T_1 = T_1'$ and $\subtypesTo{\Gamma, \subtypes{X}{T_1}}{T_2''}{T_2'}$.
                We take $T_2 = T_2''$ to obtain the desired result.
            \item \tsc{T-Sub}
                \begin{equation}
                    \derive
                        {\typesTo{\Gamma}{\tabs{X}{T_1} t_2}{T'} \subtypes{\Gamma}{T'}{T}}
                        {\typesTo{\Gamma}{\tabs{X}{T_1} t_2}{T}}
                \end{equation}
                By \tsc{ST-Trans}, $\subtypesTo{\Gamma}{T'}{\quant{X}{T_1'} T_2'}$.
                The conclusion is then immediate after applying the I.H..
            \item In any other case, the term cannot be a type abstraction.
        \end{itemize}
    \item By induction on the derivation of $\typesTo{\Gamma}{\concat{t}{t'}}{T}$.
        \begin{itemize}
            \item \tsc{T-List}
                \begin{equation}
                    \derive
                        {\typesTo{\Gamma}{t}{T'''} \quad \typesTo{\Gamma}{t'}{\listT{T'''}}}
                        {\typesTo{\Gamma}{\concat{t}{t'}}{\listT{T'''}}}
                \end{equation}
                We have $\subtypesTo{\Gamma}{\listT{T'''}}{\listT{T'}}$.
                By \cref{subtyping-inversion}, we have $\subtypesTo{\Gamma}{T'''}{T'}$,
                $\typesTo{\Gamma}{t}{T'''}$ and $\typesTo{\Gamma}{t'}{\listT{T'''}}$.
                We take $T'' = T'''$ to obtain the desired result.
            \item \tsc{T-Sub}
                \begin{equation}
                    \derive
                        {\typesTo{\Gamma}{\concat{t}{t'}}{T'''} \quad \subtypesTo{\Gamma}{T'''}{T}}
                        {\typesTo{\Gamma}{\concat{t}{t'}}{T}}
                \end{equation}
                By \tsc{ST-Trans}, $\subtypesTo{\Gamma}{T'''}{\listT{T'}}$.
                The conclusion is then immediate after applying the I.H..
            \item In any other case, the term cannot be a list concatenation.
        \end{itemize}
    \item By induction on the derivation of $\typesTo{\Gamma}{\curl{f_i = t_i}_{i \in I}}{T}$.
        \begin{itemize}
            \item \tsc{T-Rec}
                \begin{equation}
                    \derive
                        {\forall i \in I : \typesTo{\Gamma}{t_i}{T_i''}}
                        {\typesTo{\Gamma}{\curl{f_i = t_i}_{i \in I}}{\curl{f_i : T_i''}_{i \in I}}}
                \end{equation}
                We have $\subtypesTo{\Gamma}{\curl{f_i : T_i''}_{i \in I}}{\curl{\types{f_i}{T_i}}_{i \in I}}$.
                By \cref{subtyping-inversion}, we have
                $\forall i \in I : \typesTo{\Gamma}{t_i}{T_i''} \wedge \subtypesTo{\Gamma}{T_i''}{T_i'}$.
                By taking each $T_i$ to be $T_i''$, we obtain the conclusion.
            \item \tsc{T-Sub}
                \begin{equation}
                    \derive
                        {\typesTo{\Gamma}{\curl{f_i = t_i}_{i \in I}}{T'} \quad \subtypesTo{\Gamma}{T'}{T}}
                        {\typesTo{\Gamma}{\curl{f_i = t_i}_{i \in I}}{T}}
                \end{equation}
                By \tsc{ST-Trans}, $\subtypesTo{\Gamma}{T'}{\curl{\types{f_i}{T_i'}}_{i \in I}}$.
                The conclusion is then immediate after applying the I.H..
            \item In any other case, the term cannot be a record.
        \end{itemize}
\end{enumerate}
\end{proof}

\begin{lemma}[Type preservation for \ourFormalLang terms]
    \label{term-pres}
    If ${\typesTo{\Gamma}{t}{T}}$ and ${t \rightarrow t'}$, then ${\typesTo{\Gamma}{t'}{T}}$
\end{lemma}
\begin{proof}
The proof is by induction on the derivation of ${\typesTo{\Gamma}{t}{T}}$.
\begin{itemize}
    \item \tsc{T-Val}, \tsc{T-Var}, \tsc{T-Abs}, \tsc{T-TAbs}\\
        These rules type values which have no applicable reduction rules,
        so the premise is false and the conclusion holds vacuously.
    \item \tsc{T-Rec}
        \begin{equation}
            \derive
                {\forall i \in I : \typesTo{\Gamma}{t_i}{T_i}}
                {\typesTo{\Gamma}{\curl{f_i = t_i}_{i \in I}}{\curl{f_i : T_i}_{i \in I}}}
        \end{equation}
        where $t = \curl{f_i = t_i}_{i \in I} \quad T = \curl{f_i : T_i}_{i \in I}$\\
        If all fields $t_i$ are values, there is no reduction rule, so the premise is false.\\
        If there is at least one field which is not a value, $t$ must have been reduced using the context reduction rule \tsc{E-$\mbC$}.
        Let this field be $t_k$ such that there is no $j \in I, j < k$ where $t_j$ is not a value.
        Then the reduction is:
        \begin{equation}
            \derive
                {t'' \reduce{\alpha} t'''}
                {\mbC[t''] \reduce{\alpha} \mbC[t''']}
        \end{equation}
        where $t'' = t_k$, $\mbC[t''] = t$, $\mbC[t'''] = t'$ and
        \begin{equation*}
            \mbC = \curl{f_i = \gamma_i}_{i \in I} \quad \tx{where } \begin{array}{c}\forall i \in i\end{array}:\begin{cases}i < k \Rightarrow \gamma_i = v_i\\i = k \Rightarrow \gamma_i = \hole\\i > k \Rightarrow \gamma_i = t_i\end{cases}\\
        \end{equation*}
        By the I.H., we have that $\typesTo{\Gamma}{t'''}{T_k}$. Then, we can immediately show from \tsc{T-Rec} that:
        \begin{equation}
            \derive
                {\forall i \in I : \typesTo{\Gamma}{t_i}{T_i}}
                {\typesTo{\Gamma}{t'}{\curl{f_i : T_i}_{i \in I}}}
        \end{equation}
    \item \tsc{T-List}
        \begin{equation}
            \derive
                {\typesTo{\Gamma}{t_0}{T_0} \quad \typesTo{\Gamma}{t_1}{[T_0]}}
                {\typesTo{\Gamma}{\concat{t_0}{t_1}}{[T_0]}}
        \end{equation}
        where $t = \concat{t_0}{t_1} \quad T = [T_0]$\\
        If both $t_0$ and $t_1$ are values, there is no reduction rule, so the premise is false.\\
        Otherwise, we must be able to reduce either $t_0$ or $t_1$ further, which can only be done using the context reduction rule \tsc{E-$\mbC$}:
        \begin{equation}
            \derive
                {t'' \reduce{\alpha} t'''}
                {\mbC[t''] \reduce{\alpha} \mbC[t''']}
        \end{equation}
        where one of the following assignments apply:
        \begin{align}
            t'' = t_0 \quad \mbC[t''] = t \quad \mbC[t'''] = t' \quad \mbC = \concat{\hole}{t_1} \label{T-list-1:term-prog}\\
            t'' = t_1 \quad \mbC[t''] = t \quad \mbC[t'''] = t' \quad \mbC = \concat{t_0}{\hole} \label{T-list-2:term-prog}
        \end{align}
        If we use \cref{T-list-1:term-prog}, we have by the I.H. that $\typesTo{\Gamma}{t'''}{T_0}$. Then, we can show by \tsc{T-List} that:
        \begin{equation}
            \derive
                {\typesTo{\Gamma}{t'''}{T_0} \quad \typesTo{\Gamma}{t_1}{[T_0]}}
                {\typesTo{\Gamma}{t'}{[T_0]}}
        \end{equation}
        Otherwise, if we use \cref{T-list-2:term-prog}, we have by the I.H. that $\typesTo{\Gamma}{t'''}{[T_0]}$.
        Then, we can show by \tsc{T-List} that:
        \begin{equation}
            \derive
                {\typesTo{\Gamma}{t_0}{T_0} \quad \typesTo{\Gamma}{t'''}{[T_0]}}
                {\typesTo{\Gamma}{t'}{[T_0]}}
        \end{equation}
    \item \tsc{T-Head}
        \begin{equation}
            \derive
                {\typesTo{\Gamma}{t_0}{\listT{T}}}
                {\typesTo{\Gamma}{\listHead{t_0}}{\tx{Option}\ T}}
        \end{equation}
        where ${t = \listHead{t_0}}$\\
        If $t_0$ is a value, the only reduction rule that applies is either \tsc{E-Head1} or \tsc{E-Head2}.\\
        If \tsc{E-Head1} was used, we have the following reduction:
        \begin{equation}
            \derive{}{\listHead{\concat{v_0}{v_0'}} \reduce{\tau} \curl{\tx{some} = v_0}}
        \end{equation}
        where $t_0 = \concat{v_0}{v_0'}$\\
        Now, from \cref{typing-inversion} on $\typesTo{\Gamma}{t_0}{\listT{T}}$, we have that
        ${\typesTo{\Gamma}{v_0}{T}}$ and ${\typesTo{\Gamma}{v_0'}{\listT{T}}}$.
        From \tsc{T-Rec} we have that ${\typesTo{\Gamma}{\curl{\tx{some} = v_0}}{\curl{\types{\tx{some}}{T}}}}$.
        From \tsc{ST-Refl}, \tsc{ST-$\cup$R} and \tsc{T-Sub}, we also have ${\typesTo{\Gamma}{\curl{\tx{some} = v_0}}{\tx{Option}\ T}}$,
        which gives us the desired result.\\
        If \tsc{E-Head2} was used, we have the following reduction:
        \begin{equation}
            \derive{}{\listHead{\nil} \reduce{\tau} \curl{\tx{none} = \unit}}
        \end{equation}
        where $t_0 = \nil$\\
        From \tsc{T-Val}, \tsc{ST-Val} and \tsc{Type-Unit} we have ${\typesTo{\Gamma}{\unit}{\tx{Unit}}}$.
        From \tsc{T-Rec}, ${\typesTo{\Gamma}{\curl{\tx{none} = \unit}}{\curl{\types{\tx{none}}{\tx{Unit}}}}}$.
        From \tsc{ST-Refl}, \tsc{ST-$\cup$R} and \tsc{T-Sub}, we have ${\typesTo{\Gamma}{\curl{\tx{none} = \unit}}{\tx{Option}\ T}}$,
        which is the desired result.\\
        Otherwise, if $t_0$ is not a value, we must be able to reduce it further, which can only be done using the context reduction rule \tsc{E-$\mbC$}:
        \begin{equation}
            \derive
                {t'' \reduce{\alpha} t'''}
                {\mbC[t''] \reduce{\alpha} \mbC[t''']}
        \end{equation}
        where $t'' = t_0 \quad \mbC[t''] = t \quad \mbC[t'''] = t' \quad \mbC = \listHead{\hole}$\\
        By the I.H., $\typesTo{\Gamma}{t'''}{[T]}$. Then, we can show by \tsc{T-Head} that:
        \begin{equation}
            \derive
                {\typesTo{\Gamma}{t'''}{[T]}}
                {\typesTo{\Gamma}{t'}{T}}
        \end{equation}
    \item \tsc{T-Tail}\\
        Analogous to the case for \tsc{T-Head}.
    \item \tsc{T-Let}
        \begin{equation}
            \derive
                {\typesTo{\Gamma}{t_0}{T_0} \quad \typesTo{\Gamma, x : T_0}{t_1}{T_1}}
                {\typesTo{\Gamma}{\letExp{x}{t_0} t_1}{T_1}} \label{T-let-1:term-prog}
        \end{equation}
        where $t = \letExp{x}{t_0} t_1 \quad T = T_1$\\
        If $t_0$ is a value, then the only applicable reduction rule is \tsc{E-Let}:
        \begin{equation}
            \derive{}{\letExp{x}{t_0} t_1 \reduce{\tau} \subs{t_1}{x}{t_0}}
        \end{equation}
        From \cref{T-let-1:term-prog}, we have by \cref{substitution} that $\typesTo{\Gamma}{\subs{t_1}{x}{t_0}}{T_1}$,
        which is what we want to show.\\
        Otherwise, we must be able to reduce $t_0$ further, which can only be done using the context reduction rule \tsc{E-$\mbC$}:
        \begin{equation}
            \derive
                {t'' \reduce{\alpha} t'''}
                {\mbC[t''] \reduce{\alpha} \mbC[t''']}
        \end{equation}
        where $t'' = t_0 \quad \mbC[t''] = t \quad \mbC[t'''] = t' \quad \mbC = \letExp{x}{\hole} t_1$\\
        By the I.H., $\typesTo{\Gamma}{t'''}{T_0}$. Then, we can show by \tsc{T-Let} that
        \begin{equation}
            \derive
                {\typesTo{\Gamma}{t'''}{T_0} \quad \typesTo{\Gamma, x : T_0}{t_1}{T_1}}
                {\typesTo{\Gamma}{t'}{T_1}}
        \end{equation}
    \item \tsc{T-App}
        \begin{equation}
            \derive
                {\typesTo{\Gamma}{t_1}{T_1 \rightarrow T_2} \quad \typesTo{\Gamma}{t_2}{T_1}}
                {\typesTo{\Gamma}{t_1\ t_2}{T_2}}
        \end{equation}
        where $t = t_1\ t_2 \quad T = T_2$\\
        We first consider the case where both $t_1$ and $t_2$ are values.
        Then, the only applicable reduction rule is \tsc{E-App}:
        \begin{equation}
            \derive{}{(\abs{x}{T_1'} t_0)\ t_2 \reduce{\tau} \subs{t_0}{x}{t_2}}
        \end{equation}
        where $t_1 = \abs{x}{T_1'} t_0$\\
        From \cref{typing-inversion} on ${\typesTo{\Gamma}{\abs{x}{T_1'} t_0}{T_1 \rightarrow T_2}}$ we have that
        ${\subtypesTo{\Gamma}{T_1}{T_1'}}$ and ${\typesTo{\Gamma, \types{x}{T_1'}}{t_0}{T_2}}$.
        First from \tsc{T-Sub} we have ${\typesTo{\Gamma}{t_2}{T_1'}}$.
        Now by \cref{substitution} we have that ${\typesTo{\Gamma}{\subs{t_0}{x}{t_2}}{T_2}}$, which is what we wanted to show.\\
        Otherwise, either $t_1$ or $t_2$ must be able to reduce further. Here, we could only have used the context reduction rule \tsc{E-$\mbC$}:
        \begin{equation}
            \derive
                {t'' \reduce{\alpha} t'''}
                {\mbC[t''] \reduce{\alpha} \mbC[t''']}
        \end{equation}
        where one of the following assignments applies:
        \begin{align}
            t'' = t_1 \quad \mbC[t''] = t \quad \mbC[t'''] = t' \quad \mbC = \hole\ t_2 \label{T-app-1:term-prog}\\
            t'' = t_2 \quad \mbC[t''] = t \quad \mbC[t'''] = t' \quad \mbC = (\abs{x}{T_1'} t_0)\ \hole \label{T-app-2:term-prog}
        \end{align}
        If \cref{T-app-1:term-prog} applies, we have by the I.H. that $\typesTo{\Gamma}{t'''}{T_1 \rightarrow T_2}$.
        Then, we can show by \tsc{T-App} that:
        \begin{equation}
            \derive
                {\typesTo{\Gamma}{t'''}{T_1 \rightarrow T_2} \quad \typesTo{\Gamma}{t_2}{T_1}}
                {\typesTo{\Gamma}{t'}{T_2}}
        \end{equation}
        Otherwise, if \cref{T-app-2:term-prog} applies, we have by the I.H. that $\typesTo{\Gamma}{t'''}{T_1}$.
        Then, we can again show by \tsc{T-App} that:
        \begin{equation}
            \derive
                {\typesTo{\Gamma}{t'''}{T_1 \rightarrow T_2} \quad \typesTo{\Gamma}{t_2}{T_1}}
                {\typesTo{\Gamma}{t'''\ t_2}{T_2}}
        \end{equation}
    \item \tsc{T-TApp}
        \begin{equation}
            \derive
                {\typesTo{\Gamma}{t_1}{\quant{X}{T_1} T_0} \quad \subtypesTo{\Gamma}{T_2}{T_1}}
                {\typesTo{\Gamma}{t_1\ T_2}{\subs{T_0}{X}{T_2}}}
        \end{equation}
        where $t = t_1\ T_2 \quad T = \subs{T_0}{X}{T_2}$\\
        If $t_1$ is a value, the only reduction rule that could have been used is \tsc{E-TApp}:
        \begin{equation}
            \derive{}{(\tabs{X}{T_1'} t_0)\ T_2 \rightarrow \subs{t_0}{X}{T_2}}
        \end{equation}
        where $t_1 = \tabs{X}{T_1'} t_0$\\
        From \cref{typing-inversion} on $\typesTo{\Gamma}{\tabs{X}{T_1'} t_0}{\quant{X}{T_1} T_0}$ we have that
        $T_1 = T_1'$ and $\typesTo{\Gamma, \subtypes{X}{T_1'}}{t_0}{T_0}$.
        By \cref{substitution}, we have ${\typesTo{\subs{\Gamma}{X}{T_2}}{\subs{t_0}{X}{T_2}}{\subs{T_0}{X}{T_2}}}$.
        Since $X$ cannot be in $\Gamma$, we have that $\typesTo{\Gamma}{\subs{t_0}{X}{T_2}}{\subs{T_0}{X}{T_2}}$, which is the desired result.\\
        Otherwise, if $t_1$ is not a value, it must be able to reduce further. The only applicable rule is \tsc{E-$\mbC$}:
        \begin{equation}
            \derive
                {t'' \reduce{\alpha} t'''}
                {\mbC[t''] \reduce{\alpha} \mbC[t''']}
        \end{equation}
        where $t'' = t_1 \quad \mbC[t''] = t \quad \mbC[t'''] = t' \quad \mbC = \hole\ T_2$\\
        We have by the I.H. that ${\typesTo{\Gamma}{t'''}{\quant{X}{T_1} T_0}}$, then we can show by \tsc{T-TApp} that
        \begin{equation}
            \derive
                {\typesTo{\Gamma}{t'''}{\quant{X}{T_1} T_0} \quad \subtypesTo{\Gamma}{T_2}{T_1}}
                {\typesTo{\Gamma}{t'}{\subs{T_0}{X}{T_2}}}
        \end{equation}
    \item \tsc{T-Field}
        \begin{equation}
            \derive
                {\typesTo{\Gamma}{t_r}{\curl{f_i : T_i}_{i \in I}} \quad k \in I}
                {\typesTo{\Gamma}{t_r.f_k}{T_k}}
        \end{equation}
        where $t = t_r.f_k \quad T = T_k$\\
        If $t_r$ is a value, then the only applicable reduction rule is \tsc{E-Field}:
        \begin{equation}
            \derive
                {t_r = \curl{f_i = v_i}_{i \in I} \quad k \in I}
                {t_r.f_k \reduce{\tau} v_k}
        \end{equation}
        From \cref{typing-inversion}, we have that $\typesTo{\Gamma}{v_k}{T_k}$, which is what we wanted to show.\\
        Otherwise, if $t_r$ is not a value, then the only rule that could have been used is \tsc{E-$\mbC$}:
        \begin{equation}
            \derive
                {t'' \reduce{\alpha} t'''}
                {\mbC[t''] \reduce{\alpha} \mbC[t''']}
        \end{equation}
        where $t'' = t_r \quad \mbC[t''] = t \quad \mbC[t'''] = t' \quad \mbC = \hole.f_k$\\
        By the I.H., $\typesTo{\Gamma}{t'''}{\curl{f_i : T_i}_{i \in I}}$. Then we have from \tsc{T-Field} that
        \begin{equation}
            \derive
                {\typesTo{\Gamma}{t'''}{\curl{f_i : T_i}_{i \in I}} \quad k \in I}
                {\typesTo{\Gamma}{t'}{T_k}}
        \end{equation}
    \item \tsc{T-Match}
        \begin{equation}
            \derive
                {\typesTo{\Gamma}{t_s}{T_s} \quad \forall i \in I: \typesTo{\Gamma, x_i:T_i}{t_i}{T_i'} \quad \subtypesTo{\Gamma}{T_s}{\cup_{i \in I} T_i}}
                {\typesTo{\Gamma}{t_s \ttt{ match }\curl{x_i : T_i \Rightarrow t_i}_{i \in I}}{T_s \ttt{ match } \curl{T_i \Rightarrow T_i'}_{i \in I}}}
        \end{equation}
        where ${t = t_s \ttt{ match }\curl{x_i : T_i \Rightarrow t_i}_{i \in I} \quad T = T_s \ttt{ match } \curl{T_i \Rightarrow T_i'}_{i \in I}}$\\
        For convenience, we enumerate each judgement from the typing derivation (omitting the subtyping derivation as it is unnecessary):
        \begin{align}
            &\typesTo{\Gamma}{t_s}{T_s\ttt{ match }\curl{T_i \Rightarrow T_i'}_{i \in I}} \label{T-eq-match:pres-match}\\
            &\typesTo{\Gamma}{t_s}{T_s} \label{v-type-Ts:pres-match} \\
            &\forall i \in I : (\typesTo{\Gamma, x_i : T_i}{t_i}{T_i'}) \label{all-ti-type-TiP:pres-match}
        \end{align}
        If $t_s$ is a value $v_s$, then the only applicable reduction rule is \tsc{E-Match}:
        \begin{equation}
            \derive
                {k \in I \quad \subtypesTo{\emptyset}{\singleton{v_s}}{T_k} \quad \forall j \in I: j < k \Rightarrow \subtypesTo{\emptyset\not}{\singleton{v_s}}{T_j}}
                {v_s \ttt{ match }\curl{x_i:T_i \Rightarrow t_i}_{i \in I} \reduce{\tau} \subs{t_k}{x_k}{v_s}}
        \end{equation}
        First from the premise of \tsc{E-Match}, we have that
        \begin{align}
            \subtypesTo{\Gamma}{\underline{v_s}}{T_k} \label{v-sub-Tk:pres-match}
        \end{align}
        Then, from \cref{minimal-types} and \cref{v-type-Ts:pres-match}, we have
        \begin{align}
            \subtypesTo{\Gamma}{\underline{v_s}}{T_s} \label{v-sub-Ts:pres-match}
        \end{align}
        Suppose that ${T' = \underline{v_s}\ttt{ match }\curl{T_i \Rightarrow T_i'}_{i \in I}}$.
        Now, from \cref{v-sub-Tk:pres-match} and \cref{disjoint-value} with the premises of \tsc{E-Match},
        we can use \tsc{ST-Match1} to conclude
        \begin{align}
            \subtypesTo{\Gamma}{T_k'}{T'} \label{TkP-sub-TP:pres-match}
        \end{align}
        We can then use \cref{v-sub-Ts:pres-match} with \tsc{ST-Refl} and \tsc{ST-Match2} to conclude
        \begin{align}
            \subtypesTo{\Gamma}{T'}{T_s\ttt{ match }\curl{T_i \Rightarrow T_i'}_{i \in I}} \label{TP-sub-match:pres-match}
        \end{align}
        Which means by \tsc{ST-Trans} and \cref{T-eq-match:pres-match}
        \begin{align}
            \subtypesTo{\Gamma}{T'}{T} \label{TP-sub-T:pres-match}
        \end{align}
        Also by \tsc{ST-Trans} with \cref{TkP-sub-TP:pres-match}:
        \begin{align}
            \subtypesTo{\Gamma}{T_k'}{T} \label{TkP-sub-T:pres-match}
        \end{align}
        Now, from \cref{v-sub-Tk:pres-match} with \tsc{T-Val} and \tsc{T-Sub}, we have that
        \begin{align}
            \typesTo{\Gamma}{v_s}{T_k}
        \end{align}
        Then, we can take $k \in I$ and \cref{all-ti-type-TiP:pres-match} such that ${\typesTo{x_k : T_k}{t_k}{T_k'}}$. Together with \cref{substitution} we can show
        \begin{align}
            \typesTo{\Gamma}{\subs{t_k}{x_k}{v_s}}{T_k'} \label{tksub-type-TkP:pres-match}
        \end{align}
        By applying \tsc{T-Sub} to \cref{tksub-type-TkP:pres-match} and \cref{TkP-sub-T:pres-match}, we can finally show that
        \begin{align}
            \typesTo{\Gamma}{\subs{t_k}{x_k}{v_s}}{T}
        \end{align}
        Otherwise, if $t_s$ is not a value, the only applicable reduction rule is \tsc{E-$\mbC$}:
        \begin{equation}
            \derive
                {t'' \reduce{\alpha} t'''}
                {\mbC[t''] \reduce{\alpha} \mbC[t''']}
        \end{equation}
        where ${t'' = t_s \quad \mbC[t''] = t \quad \mbC[t'''] = t' \quad \mbC = \match{\hole}{x_i : T_i \Rightarrow t_i}{i \in I}}$\\
        By the I.H., we have that $\typesTo{\Gamma}{t'''}{T_s}$. Then we have from \tsc{T-Match} that:
        \begin{equation}
            \derive
                {\typesTo{\Gamma}{t'''}{T_s} \quad \forall i \in I: \typesTo{\Gamma, x_i:T_i}{t_i}{T_i'}}
                {\typesTo{\Gamma}{t'}{T_s \ttt{ match } \curl{T_i \Rightarrow T_i'}_{i \in I}}}
        \end{equation}
    \item \tsc{T-Sub}
        \begin{equation}
            \derive
                {\typesTo{\Gamma}{t}{T'} \quad \subtypesTo{\Gamma}{T'}{T}}
                {\typesTo{\Gamma}{t}{T}}
        \end{equation}
        We can apply the I.H. on $\typesTo{\Gamma}{t}{T'}$ and get that $\typesTo{\Gamma}{t'}{T'}$.
        Then, from \tsc{T-Sub}, we have $\typesTo{\Gamma}{t'}{T}$.
    \item \tsc{T-OpC}
        \begin{equation}
            \derive
                {\typesTo{\Gamma}{a}{\tx{ServerRef}[T_m,T_a,T_p]}}
                {\typesTo{\Gamma}{\ttt{Connect}(a)}{\tx{Chan}[T_m,T_a,T_p]}}
        \end{equation}
        where $t = \ttt{Connect}(a) \quad T = \tx{Chan}[T_m,T_a,T_p]$\\
        $t$ is not a value, so there must be a reduction. The only applicable rule is \tsc{E-Connect}:
        \begin{equation}
            \derive
                {a \in \tx{ServerRef}[T_m,T_a,T_p] \quad s \in \tx{Chan}[T_m,T_a,T_p]}
                {\ttt{Connect}(a) \reduceTwo{\tx{connect}(a)}{s} s}
        \end{equation}
        where $t' = s$\\
        From \tsc{T-Val}, \tsc{ST-Val} and \tsc{T-Sub}, we have that ${\typesTo{\Gamma}{t'}{\tx{Chan}[T_m,T_a,T_p]}}$.
    \item \tsc{T-OpR}
        \begin{equation}
            \derive
                {\typesTo{\Gamma}{s}{\tx{Chan}[T_m,T_a,T_p]} \quad \typesTo{\Gamma}{t}{\PPPPEntity\ T_m\ T_a\ T_p\ X_n\ X_a}}
                {\typesTo{\Gamma}{\ttt{Read}(s,t)}{[\PPPPEntity\ T_m\ T_a\ T_p\ X_n\ X_a]}}
        \end{equation}
        where $t = \ttt{Read}(s,t) \quad T = [\PPPPEntity\ T_m\ T_a\ T_p\ X_n\ X_a]$\\
        $t$ is not a value, so there must be a reduction. The only applicable rule is \tsc{E-Read}:
        \begin{equation}
            \derive
                {s \in \tx{Chan}[T_m,T_a,T_p] \quad t \in \PPPPEntity\ T_m\ T_a\ T_p\ X_n\ X_a}
                {\ttt{Read}(s,t) \reduceTwo{\tx{read}(s,t)}{r} r}
        \end{equation}
        where $t' = r$\\
        From \tsc{T-Val}, \tsc{ST-Val} and \tsc{T-Sub}, we have that ${\typesTo{\Gamma}{t'}{\tx{Chan}[T_m,T_a,T_p]}}$.
    \item \tsc{T-OpI}
        \begin{equation}
            \derive
                {\typesTo{\Gamma}{s}{\tx{Chan}[T_m,T_a,T_p]} \quad \typesTo{\Gamma}{t}{\PPPPEntity\ T_m\ T_a\ T_p\ X_n\ X_a}}
                {\typesTo{\Gamma}{\ttt{Insert}(s,t)}{\tx{Bool}}}
        \end{equation}
        where $t = \ttt{Insert}(s,t) \quad T = \tx{Bool}$\\
        $t$ is not a value, so there must be a reduction. The only applicable rule is \tsc{E-Insert}:
        \begin{equation}
            \derive
                {b \in \tx{Bool}}
                {\ttt{Insert}(s,t) \reduceTwo{\tx{insert}(s,t)}{b} b}
        \end{equation}
        where $t' = b$\\
        From \tsc{T-Val}, \tsc{ST-Val} and \tsc{T-Sub}, we have that ${\typesTo{\Gamma}{t'}{\tx{Bool}}}$.
    \item \tsc{T-OpM}, \tsc{T-OpD}\\
        Analogous to \tsc{T-OpI}.
\end{itemize}
\end{proof}

\begin{theorem}[Type preservation with one client and one server]
    \label{theorem:preservation-one-client-server}
    If ${\typesTo{\Gamma}{t}{T}}$ and ${S | t \rightarrow S' | t'}$, then ${\typesTo{\Gamma}{t'}{T}}$
\end{theorem}
\begin{proof}
The proof is by cases on the derivation of ${t | S \rightarrow t' | S'}$:
\begin{itemize}
    \item \tsc{Net-$\alpha$}
        \begin{equation}
            \derive
                {t \reduce{\tau} t'}
                {S | t \rightarrow S | t'}
        \end{equation}
        where $S' = S$\\
        Immediate by applying \cref{term-pres}.
    \item \tsc{Net-Comm}
        \begin{equation}
            \derive
                {S \reduce{\dual{\alpha}} S' \quad t \reduce{\alpha} t'}
                {S | t \rightarrow S' | t'}
        \end{equation}
        Immediate by applying \cref{term-pres}.
\end{itemize}
\end{proof}

\begin{proposition}[Network semantics and composition]
  \label{lem:network-semantics-comp}
  Assume $N \reduce{\beta} N'$, where $\beta$ renges over $\alpha$ (from
  \Cref{def:p4runtime-client-semantics}) and $\dual{\alpha}$ (from
  \Cref{def:p4runtime-server-semantics}).  Then, we have either:
  \begin{enumerate}
  \item\label{item:network-semantics-comp-tau}
    $\beta = \tau$, and either:
    \begin{itemize}
      \item $\exists t, t', N_0$ such that $N \equiv t \mid N_0$ and $t \reduce{\tau} t'$ and $N' \equiv t' \mid N_0$, or
      \item $\exists t, S, t', S', N_0$ such that $N \equiv t \mid S \mid N_0$ and $t \mid S \reduce{\tau} t' \mid S'$ and $N' \equiv t' \mid S' \mid N_0$;
    \end{itemize}
  \item \label{item:network-semantics-comp-not-tau}
    $\beta \neq \tau$, and either:
    \begin{itemize}
       \item $\exists t, t', N_0$ such that $N \equiv t \mid N_0$ and $t \reduce{\beta} t'$ and $N' \equiv t' \mid N_0$, or
       \item $\exists S, S', N_0$ such that $N \equiv S \mid N_0$ and $S \reduce{\beta} S'$ and $N' \equiv t' \mid S' \mid N_0$;
    \end{itemize}
  \end{enumerate}
\end{proposition}
\begin{proof}
  By induction on the derivation of the transition $N \reduce{\alpha} N'$,
  according to \Cref{def:p4runtime-network-semantics}, and by
  \Cref{def:network-congruence}.
  \notein{Note}{More details would not hurt!}
\end{proof}

\begin{corollary}[Type preservation (one step)]
    \label{theorem:preservation-one-step}
    If\, ${\typesTo{\Gamma}{t}{T}}$ and $t \mid N \reduce{} t' \mid N'$, then ${\typesTo{\Gamma}{t'}{T}}$.
\end{corollary}
\begin{proof}
  Direct consequence of \Cref{theorem:preservation-one-client-server} and
  \Cref{lem:network-semantics-comp}.
  \notein{Note}{More details would not hurt!}
\end{proof}

\lemPreservation*
\begin{proof}
  By hypothesis we have, for some $n \ge 0$:
  \[
    t \mid N = N_0 \reduce{} N_1 \reduce{} \cdots \reduce{} N_n = t' \mid N'
  \]
  We prove the thesis by induction on $n$.
  \begin{itemize}
    \item Base case $n = 0$. We have $t = t'$, hence the statement holds trivially;
    \item Inductive case $n = m + 1$. We have $t \mid N = N_0 \reduce{} N_1 \reduce{} \cdots \reduce{} N_m = t'' \mid N'' \reduce{} t' \mid N'$.
      By the induction hypothesis, we have ${\typesTo{\Gamma}{t''}{T}}$.
      Therefore, by \Cref{theorem:preservation-one-step}, we conclude ${\typesTo{\Gamma}{t'}{T}}$.
  \end{itemize}
\end{proof}

\begin{lemma}[Channel consistency]
    \label{chan-cons}
    If ${\typeJ{\emptyset}{t}}$, $\typesTo{\emptyset}{t}{T}$ and $t \reduce{\alpha} t'$\\
    and $\forall a \in t : \exists T^a_m, T^a_a, T^a_p : a \in \tx{ServerRef}[T^a_m, T^a_a, T^a_p]$\\
    and $\forall s \in t : \exists T^s_m, T^s_a, T^s_p : s \in \tx{Chan}[T^s_m, T^s_a, T^s_p]$\\\\
    Then $\alpha$ must be one of the following labels:
    \begin{align*}
        \cdot\ & \alpha = "\tau"\\
        & \tx{where } \forall a : a \in t' \Rightarrow a \in t \wedge \forall s : s \in t' \Rightarrow s \in t\\
        \cdot\ & \alpha = "\tx{connect}(a') = s'"\\
        & \tx{where } \forall a \in t' : a \in t \wedge \forall s \in t' : (s \in t \vee s = s') \wedge s' \in \tx{Chan}[T^{a'}_m, T^{a'}_a, T^{a'}_p]\\
        \cdot\ & \alpha = ("\tx{read}(s',v) = r"\ |\ "\tx{insert}(s',v) = b"\ |\ "\tx{update}(s',v) = b"\ |\ "\tx{delete}(s',v) = b")\\
        & \tx{where } \exists X_n, X_a : \typesTo{\emptyset}{v}{\PPPPEntity\ T^{s'}_m\ T^{s'}_a\ T^{s'}_p\ X_n\ X_a}\\
        & \tx{and } s' \in t \wedge \forall a : a \in t \Leftrightarrow a \in t' \wedge \forall s \in t' : s \in t
    \end{align*}
\end{lemma}
\begin{proof}
The proof is by induction on the derivation of $t \reduce{\alpha} t'$. We consider each case:
\begin{itemize}
    \item \tsc{E-App}
        \begin{equation}
            \derive
                {}
                {(\lambda x : T_1 . t_0) v \reduce{\tau} \subs{t_0}{x}{v}}
        \end{equation}
        where $t = (\lambda x : T_1 . t_0) v \quad t' = \subs{t_0}{x}{v} \quad \alpha = \tau$\\
        $t_0$ contains a subset of the values in $t$, in fact, the only possible change is $v$ subsituting some variable. As such, we have that:
        \begin{align}
            &\forall a : a \in t_0 \cup \curl{v} \Leftrightarrow a \in t\\
        \wedge\ &\forall s : s \in t_0 \cup \curl{v} \Leftrightarrow s \in t
        \end{align}
        The substitution $\subs{t_0}{x}{v}$ only replaces free variables; any existing values are unchanged from $t_0$.
        However, if there are no free variables $x$ in $t_0$, $\subs{t_0}{x}{v}$ will not (necessarily) contain any $v$.
        We can then conclude:
        \begin{align}
            &\forall a : a \in \subs{t_0}{x}{v} \Rightarrow a \in t_0 \cup \curl{v}\\
        \wedge\ &\forall s : s \in \subs{t_0}{x}{v} \Rightarrow s \in t_0 \cup \curl{v}
        \end{align}
        From this we have:
        \begin{equation}
            \forall a : a \in t' \Rightarrow a \in t \wedge \forall s : s \in t' \Rightarrow s \in t
        \end{equation}
        which is what we want to show for $\alpha = \tau$.
    \item \tsc{E-TApp}
        \begin{equation}
            \derive
                {}
                {(\lambda X<:T_1. t_0)\ T_2 \reduce{\tau} \subs{t_0}{X}{T_2}}
        \end{equation}
        where $t = (\lambda X <: T_1. t_0)\ T_2 \quad t' = \subs{t_0}{X}{T_2} \quad \alpha = \tau$\\
        By the structure of $t$, we know that
        \begin{align}
            &\forall a : a \in t_0 \Leftrightarrow a \in t\\
        \wedge\ &\forall s : s \in t_0 \Leftrightarrow s \in t
        \end{align}
        Type substitution does not change the values in $t_0$, so all values remain unchanged:
        \begin{align}
            &\forall a : a \in \subs{t_0}{X}{T_1} \Leftrightarrow a \in t_0\\
        \wedge\ &\forall s : s \in \subs{t_0}{X}{T_1} \Leftrightarrow s \in t_0
        \end{align}
        Now we can directly show that:
        \begin{align}
            &\forall a : a \in t \Leftrightarrow a \in t'\\
        \wedge\ &\forall s : s \in t \Leftrightarrow s \in t'
        \end{align}
        from which we obtain the conclusion we want to show for $\alpha = \tau$.
    \item \tsc{E-Let}
        \begin{equation}
            \derive
                {}
                {\letExp{x}{v} t_0 \reduce{\tau} \subs{t_0}{x}{v}}
        \end{equation}
        where $t = \letExp{x}{v} t \quad \subs{t}{x}{v} \quad \alpha = \tau$\\
        This case is analogous to the one for \tsc{E-App}.
    \item \tsc{E-Field}
        \begin{equation}
            \derive
                {v = \curl{f_i=v_i}_{i \in I} \quad k \in I}
                {v.f_k \reduce{\tau} v_k}
        \end{equation}
        where $t = v.f_k \quad t' = v_k \quad \alpha = \tau$\\
        Since $v_k$ is the value of one of the fields in $v$, and $v$ may have more fields, we have that:
        \begin{align}
            &\forall a : a \in v_k \Rightarrow a \in v\\
        \wedge\ &\forall s : s \in v_k \Rightarrow s \in v
        \end{align}
        Which gives us what we want to show.
    \item \tsc{E-Head}
        \begin{equation}
            \derive{}{\listHead{\concat{v_0}{v_1}} \reduce{\tau} v_0}
        \end{equation}
        where $t = \listHead{\concat{v_0}{v_1}} \quad t' = v_0 \quad \alpha = \tau$\\
        Since the entire list $\concat{v_0}{v_1}$ must have at least as many values as $v_0$, we have:
        \begin{align}
            &\forall a : a \in v_0 \Rightarrow a \in \concat{v_0}{v_1}\\
        \wedge\ &\forall s : s \in v_0 \Rightarrow s \in \concat{v_0}{v_1}
        \end{align}
        Which is what we want to show.
    \item \tsc{E-Tail}
        \begin{equation}
            \derive{}{\listTail{\concat{v_0}{v_1}} \reduce{\tau} v_1}
        \end{equation}
        where $t = \listTail{\concat{v_0}{v_1}} \quad t' = v_1 \quad \alpha = \tau$\\
        Same argument as for \tsc{E-Head}, but with $t' = v_1$ instead.
    \item \tsc{E-Match}
        \begin{equation}
            \derive
                {k \in I \quad \subtypesTo{\emptyset}{\singleton{v}}{T_k} \quad \forall j \in I: j < k \Rightarrow \subtypesTo{\emptyset\not}{\singleton{v}}{T_j}}
                {v \ttt{ match }\curl{x_i:T_i \Rightarrow t_i}_{i \in I} \reduce{\tau} \subs{t_k}{x_k}{v}}
        \end{equation}
        where $t = v \ttt{ match }\curl{x_i:T_i \Rightarrow t_i}_{i \in I} \quad t' = \subs{t_k}{x_k}{v} \quad \alpha = \tau$\\
        $t_k$ is only one of the branches of $t$,
        and the substitution $\subs{t_k}{x_k}{v}$ does not introduce any new variables,
        so we immediately have
        \begin{align}
            &\forall a : a \in \subs{t_k}{x_k}{v} \Rightarrow a \in t\\
        \wedge\ &\forall s : s \in \subs{t_k}{x_k}{v} \Rightarrow s \in t
        \end{align}
        Which is what we wanted to show.
    \item \tsc{E-Connect}
        \begin{equation}
            \derive
                {a' \in \tx{ServerRef}[T_m, T_a, T_p] \quad s' \in \tx{Chan}[T_m, T_a, T_p]}
                {\ttt{Connect}(a') \reduce{\tx{connect}(a')\ =\ s'} s'}
        \end{equation}
        where $t = \ttt{Connect}(a') \quad t' = s' \quad \alpha = "\tx{connect}(a') = s'"$\\
        We want to make the following conclusion:
        \begin{align}
            & \forall a \in t' : a \in t \label{econn-1:chan-cons}\\
            \wedge\ & \forall s \in t' : (s \in t \vee s = s') \label{econn-2:chan-cons}\\
            \wedge\ & s' \in \tx{Chan}[T^{a'}_m, T^{a'}_a, T^{a'}_p] \label{econn-3:chan-cons}
        \end{align}
        \Cref{econn-1:chan-cons} holds vacuously, as there is no $a$ in $t'$.\\
        \Cref{econn-2:chan-cons} holds, as the only $s$ in $t'$ is $s'$.\\
        Now, from the premise of the lemma, we have that
        \begin{equation}
            \forall a \in t : \exists T^a_m, T^a_a, T^a_p : a \in \tx{ServerRef}[T^a_m, T^a_a, T^a_p] \label{econn-prem:chan-cons}
        \end{equation}
        Since $s'$ uses the same internal channel types as $a'$, we must have (from the premises of \tsc{E-Connect})
        that $s' \in \tx{Chan}[T^{a'}_m, T^{a'}_a, T^{a'}_p]$, which shows that \cref{econn-3:chan-cons} holds.
    \item \tsc{E-Read}
        \begin{equation}
            \derive
                {s' \in \tx{Chan}[T_m,T_a,T_p] \quad v \in \PPPPEntity\ T_m\ T_a\ T_p\ X_n\ X_a}
                {\ttt{Read}(s',v) \reduce{\tx{read}(s',v)\ =\ r} r}
        \end{equation}
        where $t = \ttt{Read}(s',v) \quad t' = r \quad {\alpha = "\tx{read}(s',v)\ =\ r"}$\\
        We want to make the following conclusion:
        \begin{align}
            & \typesTo{\emptyset}{v}{\PPPPEntity\ T^{s'}_m\ T^{s'}_a\ T^{s'}_p\ X_n\ X_a} \label{eread-1:chan-cons}\\
            \wedge\ & s' \in t \label{eread-0:chan-cons}\\
            \wedge\ & \forall a : a \in t \Leftrightarrow a \in t' \label{eread-2:chan-cons}\\
            \wedge\ & \forall s \in t' : s \in t \label{eread-3:chan-cons}
        \end{align}
        \Cref{eread-0:chan-cons} is immediate, as $s' \in \ttt{Read}(s', v)$.\\
        From \tsc{T-Val}, \tsc{ST-Val} and \tsc{T-Sub} we have that
        \begin{align}
            &\typesTo{\emptyset}{s'}{\tx{Chan}[T_m, T_a, T_p]} \label{eread-inv-1:chan-cons}\\
            &\typesTo{\emptyset}{v}{\PPPPEntity\ T^{s'}_m\ T^{s'}_a\ T^{s'}_p\ X_n\ X_a} \label{eread-inv-2:chan-cons}
        \end{align}
        Now, from the premise of the lemma, we have that
        \begin{equation}
            \forall s \in t : \exists T^s_m, T^s_a, T^s_p : s \in \tx{Chan}[T^s_m, T^s_a, T^s_p] \label{econn-prem:chan-cons}
        \end{equation}
        So from \cref{eread-inv-1:chan-cons} it must be the case that
        \begin{equation}
            \typesTo{\emptyset}{s'}{\tx{Chan}[T^{s'}_m, T^{s'}_a, T^{s'}_p]}
        \end{equation}
        \Cref{eread-1:chan-cons} is identical to \cref{eread-inv-2:chan-cons}.\\
        Since there are no $a$ in $t$ or $t'$, \cref{eread-2:chan-cons} holds vacuously.\\
        Also, since there are no $s$ in $t'$, \cref{eread-3:chan-cons} also holds vacuously.
    \item \tsc{E-Insert}
        \begin{equation}
            \derive
                {}
                {\ttt{Insert}(s',v) \reduce{\tx{insert}(s',v)\ =\ b} b}
        \end{equation}
        where $t = \ttt{Insert}(s',v) \quad t' = b \quad {\alpha = "\tx{insert}(s',v)\ =\ b"}$\\
        This case is analogous to the one for \tsc{E-Read}.
    \item \tsc{E-Modify}
        \begin{equation}
            \derive
                {}
                {\ttt{Modify}(s',v) \reduce{\tx{modify}(s',v)\ =\ b} b}
        \end{equation}
        where $t = \ttt{Modify}(s',v) \quad t' = b \quad {\alpha = "\tx{modify}(s',v)\ =\ b"}$\\
        This case is analogous to the one for \tsc{E-Read}.
    \item \tsc{E-Delete}
        \begin{equation}
            \derive
                {}
                {\ttt{Delete}(s',v) \reduce{\tx{delete}(s',v)\ =\ b} b}
        \end{equation}
        where $t = \ttt{Delete}(s',v) \quad t' = b \quad {\alpha = "\tx{delete}(s',v)\ =\ b"}$\\
        This case is analogous to the one for \tsc{E-Read}.
    \item \tsc{E-$\mbC$}
        \begin{equation}
            \derive
                {t'' \reduce{\alpha} t'''}
                {\mbC[t''] \reduce{\alpha} \mbC[t''']}
        \end{equation}
        where $t = \mbC[t''] \quad t' = \mbC[t''']$\\
        Since the context $\mbC$ surrounding $t''$ and $t'''$ is the same before and after reduction, we know that all $a$ and $s$ in $\mbC$ are unchanged. Formally,
        \begin{align}
            \forall s : s \in \mbC \Rightarrow s \in \mbC[t''] \wedge s \in \mbC[t'''] \label{ctx-same-s:chan-cons}\\
            \forall a : a \in \mbC \Rightarrow a \in \mbC[t''] \wedge a \in \mbC[t'''] \label{ctx-same-a:chan-cons}
        \end{align}
        By the I.H. on the derivation of ${t'' \reduce{\alpha} t'''}$, we know that $t'''$ satisfies the properties determined by $\alpha$.
        We consider each case of $\alpha$
        \begin{itemize}
            \item $\alpha = "\tau"$\\
                By the I.H. on $t'' \reduce{\alpha} t'''$:
                \begin{align}
                    &\forall s : s \in t''' \Rightarrow s \in t'' \label{ctx-tau-s:chan-cons}\\
                    \wedge\ &\forall a : a \in t''' \Rightarrow a \in t'' \label{ctx-tau-a:chan-cons}
                \end{align}
                From \cref{ctx-tau-s:chan-cons} and \cref{ctx-same-s:chan-cons} we get
                \begin{equation}
                    \forall s : s \in \mbC[t'''] \Rightarrow s \in \mbC[t'']
                \end{equation}
                And likewise with \cref{ctx-tau-a:chan-cons} and \cref{ctx-same-a:chan-cons}:
                \begin{equation}
                    \forall a : a \in \mbC[t'''] \Rightarrow a \in \mbC[t'']
                \end{equation}
                which is what we wanted to show.
            \item $\alpha = "\tx{connect}(a) = s'"$\\
                By the I.H. on $t'' \reduce{\alpha} t'''$:
                \begin{align}
                    \forall s \in t''' : (s \in t'' \vee s = s') \label{ctx-conn-1:chan-cons}\\
                    \wedge\ & s' \in \tx{Chan}[T^a_m, T^a_a, T^a_p] \label{ctx-conn-2:chan-cons}\\
                    \wedge\ & \forall a \in t''' : a \in t'' \label{ctx-conn-3:chan-cons}
                \end{align}
                From \cref{ctx-conn-1:chan-cons} and \cref{ctx-same-s:chan-cons}, we have
                \begin{equation}
                    \forall s \in \mbC[t''']: (s \in \mbC[t''] \vee s = s')
                \end{equation}
                \Cref{ctx-conn-2:chan-cons} is exactly the same as the conclusion we want to show.\\
                From \cref{ctx-conn-3:chan-cons} and \cref{ctx-same-a:chan-cons}, we have
                \begin{equation}
                    \forall a \in \mbC[t'''] : a \in \mbC[t'']
                \end{equation}
                which concludes this case.
            \item ${\alpha = ("\tx{read}(s',v) = r"\ |\ "\tx{insert}(s,v) = b"\ |\ "\tx{update}(s',v) = b"\ |\ "\tx{delete}(s',v) = b")}$
                By the I.H. on $t'' \reduce{\alpha} t'''$:
                \begin{align}
                    & \typesTo{\emptyset}{v}{\PPPPEntity\ T^{s'}_m\ T^{s'}_a\ T^{s'}_p\ X_n\ X_a} \label{ctx-op-1:chan-cons}\\
                    \wedge\ &s' \in t''  \label{ctx-op-0:chan-cons}\\
                    \wedge\ &\forall a : a \in t'' \Leftrightarrow a \in t'''\label{ctx-op-2:chan-cons}\\
                    \wedge\ &\forall s \in t''' : s \in t'' \label{ctx-op-3:chan-cons}
                \end{align}
                Clearly, if $s' \in t''$ from \cref{ctx-op-0:chan-cons} then $s' \in \mbC[t'']$.\\
                \Cref{ctx-op-1:chan-cons} is the same as the desired conclusion.\\
                From \cref{ctx-op-2:chan-cons} and \cref{ctx-same-a:chan-cons}, we have
                \begin{equation}
                    \forall a : a \in \mbC[t''] \Leftrightarrow a \in \mbC[t''']
                \end{equation}
                From \cref{ctx-op-3:chan-cons} and \cref{ctx-same-s:chan-cons}, we have
                \begin{equation}
                    \forall s \in \mbC[t'''] : s \in \mbC[t'']
                \end{equation}
        \end{itemize}
\end{itemize}
\end{proof}

\begin{theorem}[Term progress]
    \label{term-progress}
    If ${\typeJ{\emptyset}{t}}$ and ${\typesTo{\emptyset}{t}{T}}$\\
    and $\forall a \in t : \exists T^a_m, T^a_a, T^a_p : a \in \tx{ServerRef}[T^a_m, T^a_a, T^a_p]$\\
    and $\forall s \in t : \exists T^s_m, T^s_a, T^s_p : s \in \tx{Chan}[T^s_m, T^s_a, T^s_p]$\\\\
    Then, either $t$ is a value or there is a reduction $t \reduce{\alpha} t'$ such that:
    \begin{align*}
        \cdot\ & \alpha = "\tau"\\
        & \tx{where } \forall a : a \in t' \Rightarrow a \in t \wedge \forall s : s \in t' \Rightarrow s \in t\\
        \cdot\ & \alpha = "\tx{connect}(a') = s'"\\
        & \tx{where } \forall a \in t' : a \in t \wedge \forall s \in t' : (s \in t \vee s = s') \wedge s' \in \tx{Chan}[T^{a'}_m, T^{a'}_a, T^{a'}_p]\\
        \cdot\ & \alpha = ("\tx{read}(s',v) = r"\ |\ "\tx{insert}(s',v) = b"\ |\ "\tx{update}(s',v) = b"\ |\ "\tx{delete}(s',v) = b")\\
        & \tx{where } \typesTo{\emptyset}{v}{\PPPPEntity\ T^{s'}_m\ T^{s'}_a\ T^{s'}_p\ X_n\ X_a}\\
        & \tx{and } s' \in t \wedge \forall a : a \in t \Leftrightarrow a \in t' \wedge \forall s \in t' : s \in t
    \end{align*}
\end{theorem}
\begin{proof}
We show the theorem by a proof by induction on the derivation of $\typesTo{\emptyset}{t}{T}$.
\begin{itemize}
    \item \tsc{T-Val}\\
        \begin{align}
            \derive{}{\typesTo{\emptyset}{v}{\singleton{v}}}
        \end{align}
        where $t = v \quad T = \singleton{v}$\\
        $t$ is always a value, so the conclusion is immediate.
    \item \tsc{T-Rec}\\
        \begin{align}
            \derive
                {\forall i \in I : \typesTo{\emptyset}{t_i}{T_i}}
                {\typesTo{\emptyset}{\curl{f_i = t_i}_{i \in I}}{\curl{f_i : T_i}_{i \in I}}}
        \end{align}
        where $t = \curl{f_i = t_i}_{i \in I} \quad T = \curl{f_i : T_i}_{i \in I}$\\
        By the I.H., all $t_i$ are either values or reduce to some $t_i'$.\\
        If all of them are values, $t$ is also a value by definition, and so we have the conclusion.\\
        If at least one of them, say $t_k$ where $k \in I$, is not a value, then by the I.H. there is a
        reduction $t_k \reduce{\alpha} t_k'$.
        We can then apply the context rule \tsc{E-$\mbC$} to get a reduction
        \begin{equation}
            \derive
                {t_k \reduce{\alpha} t_k'}
                {\mbC[t_k] \reduce{\alpha} \mbC[t_k']}
        \end{equation}
        where $t = \mbC[t_k], t' = \mbC[t_k']$ and
        \begin{equation}
            \mbC\ =\ \curl{f_i = \gamma_i}_{i \in I}\tx{ where }\forall i \in I :
            \begin{cases}
                i < k \Rightarrow \gamma_i = v_i\\
                i = k \Rightarrow \gamma_i = \hole\\
                i > k \Rightarrow \gamma_i = t_i
            \end{cases}
        \end{equation}
        Since we have a reduction, we can use \cref{chan-cons} to obtain the desired conclusion.
    \item \tsc{T-List}
        \begin{equation}
            \derive
                {\typesTo{\emptyset}{t_0}{T_0} \quad \typesTo{\emptyset}{t_1}{[T_0]}}
                {\typesTo{\emptyset}{\concat{t_0}{t_1}}{[T_0]}}
        \end{equation}
        where $t = \concat{t_0}{t_1} \quad T = [T_0]$\\
        By the I.H., $t_0$ and $t_1$ are either values or there exists reductions
        $t_0 \reduce{\alpha} t_0'$ and $t_1 \reduce{\alpha} t_1'$ respectively.
        If $t_0$ is not a value, we can apply \tsc{E-$\mbC$} to get a reduction of $t = \mbC[t_0]$, where:
        \begin{equation}
            \mbC\ =\ \concat{\hole}{t_1}
        \end{equation}
        If $t_0$ is a value and $t_1$ is not, we can again use the context rule \tsc{E-$\mbC$} to get a reduction of $t = \mbC[t_1]$, where:
        \begin{equation}
            \mbC\ =\ \concat{t_0}{\hole}
        \end{equation}
        In any reduction case, we can use \cref{chan-cons} to obtain the desired result.\\
        If both are values, $t$ is also a value, and so we again have the desired result.
    \item \tsc{T-Head}
        \begin{equation}
            \derive
                {\typesTo{\emptyset}{t_0}{[T_0]}}
                {\typesTo{\emptyset}{\listHead{t_0}}{T_0}}
        \end{equation}
        where $t = \listHead{t_0} \quad T = T_0$\\
        By the I.H., either $t_0$ is a value or there exists a reduction $t_0 \reduce{\alpha} t_0'$.
        If $t_0$ is not a value, we can apply \tsc{E-$\mbC$} to get a reduction of $t = \mbC[t_0]$, where:
        \begin{equation}
            \mbC = \listHead{\hole}
        \end{equation}
        Otherwise, we can use \tsc{E-Head} to obtain a reduction.
        In both reduction cases, we apply \cref{chan-cons} to obtain the desired result.
    \item \tsc{T-Tail}
        Analogous to the case for \tsc{T-Head}.
    \item \tsc{T-Var}
        \begin{align}
            \derive
                {x : T \in \emptyset}
                {\typesTo{\emptyset}{x}{T}}
        \end{align}
        where $t = x$\\
        Vacuously holds, as $x : T \not\in \emptyset$.
    \item \tsc{T-Let}
        \begin{align}
            \derive
                {\typesTo{\emptyset}{t_0}{T_0} \quad \typesTo{x : T_0}{t_1}{T_1}}
                {\typesTo{\emptyset}{\letExp{x}{t_0} t_1}{T_1}}
        \end{align}
        where $t = \letExp{x}{t_0} t_1 \quad T = T_1$\\
        By the I.H., either $t_0$ is a value or reduces to some $t_0'$.\\
        If $t_0$ is not a value, there is a reduction $t_0 \reduce{\alpha} t_0'$.
        We can apply \tsc{E-$\mbC$} to get a reduction of $t = \mbC[t_0]$, where
        \begin{equation}
            \mbC\ =\ \letExp{x}{\hole} t_1
        \end{equation}
        Otherwise, if $t_0$ is a value, $t$ can be reduced with the \tsc{E-Let} rule.\\
        Since $t$ can always be reduced, we can apply \cref{chan-cons} to obtain the conclusion.
    \item \tsc{T-Abs}
        \begin{align}
            \derive
                {\typesTo{x : T_1}{t_2}{T_2}}
                {\typesTo{\emptyset}{\lambda x : T_1. t_2}{T_1 \rightarrow T_2}}
        \end{align}
        where $t = \lambda x : T_1. t_2 \quad T = T_1 \rightarrow T_2$\\
        $t$ is a value, so the conclusion is immediate.
    \item \tsc{T-App}
        \begin{align}
            \derive
                {\typesTo{\emptyset}{t_1}{T_1 \rightarrow T_2} \quad \typesTo{\emptyset}{t_2}{T_1}}
                {\typesTo{\emptyset}{t_1\ t_2}{T_2}}
        \end{align}
        where $t = t_1\ t_2 \quad T = T_2$\\
        By the I.H., either $t_1$ is a value or reduces to some $t_1'$.\\
        If $t_1$ is not a value, we can apply \tsc{E-$\mbC$} to get a reduction of $t = \mbC[t_1]$, where
        \begin{equation}
            \mbC\ =\ \hole\ t_2
        \end{equation}
        If $t_1$ is a value, it must have the form $v_1 = \lambda x : T_1'. t_0$,
        where ${\subtypesTo{\emptyset}{T_1'}{T_1}}$.
        Then, by the I.H., either $t_2$ is a value or it reduces to some $t_2'$.\\
        If $t_2$ is not a value, we can again apply \tsc{E-$\mbC$} to get a reduction of $t = \mbC[t_2]$, where
        \begin{equation}
            \mbC\ =\ v_1\ \hole
        \end{equation}
        If $t_2$ is a value, $t$ can be reduced with the \tsc{E-App} rule.\\
        Since $t$ can always be reduced, we can apply \cref{chan-cons} to obtain the conclusion.
    \item \tsc{T-TAbs}
        \begin{align}
            \derive
                {\typesTo{X <: T_1}{t_2}{T_2}}
                {\typesTo{\emptyset}{(\lambda X <: T_1. t_2)}{(\forall X <: T_1 . T_2)}}
        \end{align}
        where $t = \lambda X <: T_1. t_2 \quad T = \forall X <: T_1 . T_2$\\
        $t$ is a value, so the conclusion is immediate.
    \item \tsc{T-TApp}
        \begin{align}
            \derive
                {\typesTo{\emptyset}{t_1}{\forall X <: T_1. T_0} \quad \subtypesTo{T_2}{T_1}}
                {\typesTo{\emptyset}{t_1\ T_2}{\subs{T_0}{X}{T_2}}}
        \end{align}
        where $t = t_1\ T_2 \quad T = \subs{T_0}{X}{T_1'}$\\
        By the I.H., either $t_1$ is a value or reduces to some $t_1'$.\\
        If $t_1$ is not a value, we can apply \tsc{E-$\mbC$} to get a reduction of $t = \mbC[t_1]$, where
        \begin{equation}
            \mbC\ =\ \hole\ T_2
        \end{equation}
        If $t_1$ is a value, it must have the form $v_1 = \lambda X <: T_1'. t_0$,
        where ${\subtypesTo{\emptyset}{T_1}{T_1'}}$.
        Then $t$ can be reduced with the \tsc{E-TApp} rule.\\
        Since $t$ can always be reduced, we can apply \cref{chan-cons} to get the conclusion.
    \item \tsc{T-Field}
        \begin{align}
            \derive
                {\typesTo{\emptyset}{t_r}{\curl{f_i : T_i}_{i \in I}} \quad k \in I}
                {\typesTo{\emptyset}{t_r.f_k}{T_k}}
        \end{align}
        where $t = t_r.f_k \quad T_k$\\
        By the I.H., either $t_r$ is a value (all fields are values) or it reduces to some $t_r'$
        (at least one field is not a value).\\
        If $t_r$ is not a value, we can apply \tsc{E-$\mbC$} to get a reduction of $t = \mbC[t_r]$, where
        \begin{equation}
            \mbC\ =\ \hole.f_k
        \end{equation}
        If $t_0$ is a value, $t$ can be reduced with the \tsc{E-Field} rule.\\
        Since $t$ can always be reduced, we can apply \cref{chan-cons} to obtain the conclusion.
    \item \tsc{T-Match}
        \begin{align}
            \derive
                {\typesTo{\emptyset}{t_s}{T_s} \quad \forall i \in I : \typesTo{x_i : T_i}{t_i}{T_i} \quad \subtypesTo{\Gamma}{T_s}{\cup_{i \in I} T_i}}
                {\typesTo{\emptyset}{\match{t_s}{x_i : T_i \Rightarrow t_i}{i \in I}}{\match{T_s}{T_i \Rightarrow T_i'}{i \in I}}}
        \end{align}
        where $t = \match{t_s}{x_i : T_i \Rightarrow t_i}{i \in I} \quad T = \match{T_s}{T_i \Rightarrow T_i'}{i \in I}$\\
        By the I.H., either $t_s$ is a value or it reduces to some $t_s'$.\\
        If $t_s$ is not a value, we can apply \tsc{E-$\mbC$} to get a reduction of $t = \mbC[t_s]$, where
        \begin{equation}
            \mbC\ =\ \match{\hole}{x_i : T_i \Rightarrow t_i}{i \in I}
        \end{equation}
        If $t_s$ is a value ($v_s$), the only other rule that could potentially be applied is \tsc{E-Match}.
        We now show that this rule \textit{can} in fact be applied.\\
        First, from \cref{minimal-types} and we have that
        \begin{equation}
            \subtypesTo{\emptyset}{\singleton{v_s}}{T_s} \label{v-sub-Ts:prog-match}
        \end{equation}
        From the premises and \tsc{ST-Trans}:
        \begin{equation}
            \subtypesTo{\emptyset}{\singleton{v_s}}{\cup_{i \in I} T_i} \label{v-sub-UTi:prog-match}
        \end{equation}
        The subtyping derivation of \cref{v-sub-UTi:prog-match} must have used \tsc{ST-$\cup$R},
        so we must be able to go up the derivation tree and find at least one of the unioned types which alone is a supertype of $\singleton{v_s}$.
        We pick the supertype $T_k$ such that there is no $j < k, j \in I$ for which $T_j$ is a supertype of $T_s$. Now, we have that
        \begin{equation}
            \subtypesTo{\emptyset}{\singleton{v_s}}{T_k} \label{v-sub-Tk:prog-match}
        \end{equation}
        and
        \begin{align}
            \forall j \in I : j < k \Rightarrow \subtypesTo{\emptyset\not}{\singleton{v_s}}{T_j}
        \end{align}
        Now, we can apply \tsc{E-Match} to achieve the reduction.\\
        Since we have now shown that $t$ can always be reduced, we can apply \cref{chan-cons} to obtain the desired conclusion.
    \item \tsc{T-Sub}
        \begin{align}
            \derive
                {\typesTo{\emptyset}{t}{T'} \quad \subtypesTo{\emptyset}{T'}{T}}
                {\typesTo{\emptyset}{t}{T}}
        \end{align}
        Since the derivation tree is finite, we can apply the I.H. to the derivation $\typesTo{\emptyset}{t}{T'}$.
        This tells us that either $t$ is a value or it can be reduced to some $t'$, which is what we want to show.
    \item \tsc{T-OpC}
        \begin{align}
            \derive
                {\typesTo{\emptyset}{a}{\tx{ServerRef}[T_m,T_a,T_p]}}
                {\typesTo{\emptyset}{\ttt{Connect}(a)}{\tx{Chan}[T_m,T_a,T_p]}}
        \end{align}
        where $t = \ttt{Connect}(a) \quad T = \tx{Chan}[T_m,T_a,T_p]$\\
        By definition, $a$ must be a value (or a variable, but this would immediately make the premise false).
        $\ttt{Connect}(a)$ is not a value, so there must be a reduction to some $t'$.
        Since ${\typesTo{\emptyset}{a}{\tx{ServerRef}[T_m,T_a,T_p]}}$, we have ${a \in \tx{ServerRef}[T_m,T_a,T_p]}$ by definition.
        Also, from \cref{term-pres}, we have that ${\typesTo{\emptyset}{t'}{\tx{Chan}[T_m,T_a,T_p]}}$, which again gives us that ${t' \in \tx{Chan}[T_m,T_a,T_p]}$.
        Now, we can use \tsc{E-Connect} for the reduction $t \reduce{\tx{connect}(a)\ =\ s} t'$, which allows us to use \cref{chan-cons} to get the conclusion.
    \item \tsc{T-OpR}
        \begin{equation}
            \derive
                {\typesTo{\emptyset}{s}{\tx{Chan}[T_m,T_a,T_p]} \quad \typesTo{\emptyset}{v}{\PPPPEntity\ T_m\ T_a\ T_p\ X_n\ X_a}}
                {\typesTo{\emptyset}{\ttt{Read}(s,v)}{[\PPPPEntity\ T_m\ T_a\ T_p\ X_n\ X_a]}}
        \end{equation}
        where $t = \ttt{Read}(s,v) \quad T = [\PPPPEntity\ T_m\ T_a\ T_p\ X_n\ X_a]$\\
        By definition, $s$ and $v$ are both values, but $\ttt{Read}(s,v)$ is not a value, so there must be a reduction to some $t'$.
        Since $\typesTo{\emptyset}{s}{\tx{Chan}[T_m,T_a,T_p]}$, we have ${s \in \tx{Chan}[T_m,T_a,T_p]}$ by definition.\\
        The same applies for ${\typesTo{\emptyset}{v}{\PPPPEntity\ T_m\ T_a\ T_p\ X_n\ X_a}}$ and\\
        ${v \in \PPPPEntity\ T_m\ T_a\ T_p\ X_n\ X_a}$.
        Now, the only applicable reduction is \tsc{E-Read}, which gives us the reduction ${t \reduceTwo{\tx{read}(s,v)}{r}}$.
        We then use \cref{chan-cons} to get the desired result.
    \item \tsc{T-OpI}
        \begin{equation}
            \derive
                {\typesTo{\emptyset}{s}{\tx{Chan}[T_m,T_a,T_p]} \quad \typesTo{\emptyset}{v}{\PPPPEntity\ T_m\ T_a\ T_p\ X_n\ X_a}}
                {\typesTo{\emptyset}{\ttt{Insert}(s,v)}{\tx{Bool}}}
        \end{equation}
        where $t = \ttt{Insert}(s,v) \quad T = \tx{Unit}$\\
        By definition, all of the arguments must be values (otherwise the premise is false).
        $\ttt{Insert}(s,v)$ is not a value, so there must be a reduction to some $t'$.
        The only applicable reduction rule is \tsc{E-Insert}, so we use it to obtain the reduction ${t \reduce{\tx{insert}(s,v)\ =\ \unit} t'}$.
        We then use \cref{chan-cons} to get the conclusion.
    \item \tsc{T-OpM}
        \begin{align}
            \derive
                {\typesTo{\emptyset}{s}{\tx{Chan}[T_m,T_a,T_p]} \quad \typesTo{\emptyset}{v}{\PPPPEntity\ T_m\ T_a\ T_p\ X_n\ X_a}}
                {\typesTo{\emptyset}{\ttt{Modify}(s,v)}{\tx{Bool}}}
        \end{align}
        where $t = \ttt{Modify}(s,v) \quad T = \tx{Unit}$\\
        This case is analogous to the one for \tsc{T-OpI}.
    \item \tsc{T-OpD}
        \begin{align}
            \derive
                {\typesTo{\emptyset}{s}{\tx{Chan}[T_m,T_a,T_p]} \quad \typesTo{\emptyset}{v}{\PPPPEntity\ T_m\ T_a\ T_p\ X_n\ X_a}}
                {\typesTo{\emptyset}{\ttt{Delete}(s,v)}{\tx{Bool}}}
        \end{align}
        where $t = \ttt{Delete}(s,v) \quad T = \tx{Unit}$\\
        This case is analogous to the one for \tsc{T-OpI}.
\end{itemize}
\end{proof}

\newcommand{\tm}[1]{#1.\mathit{tm}}
\newcommand{\ta}[1]{#1.\mathit{ta}}
\newcommand{\ap}[1]{#1.\mathit{ap}}
In order to reduce verbosity, we define:
\begin{align*}
    C.\mathit{table\_matches} &= \tm{C}\\
    C.\mathit{table\_actions} &= \ta{C}\\
    C.\mathit{action\_params} &= \ap{C}
\end{align*}

\begin{theorem}[Progress (with explicit invariants)]
    \label{theorem:progress-general}
    Assume $\typesTo{\emptyset}{t}{T}$ and $\S = \langle C, E, a', K \rangle$\\
    and $\forall a \in t : a = a' \wedge a \in \tx{ServerRef}[\encoded{\tm{C}}, \encoded{\ta{C}}, \encoded{\ap{C}}]$\\
    and $\forall s \in t : s \in K \wedge s \in \tx{Chan}[\encoded{\tm{C}}, \encoded{\ta{C}}, \encoded{\ap{C}}]$\\\\
    Then, either $t$ is a value or $t | S \rightarrow t' | S'$ where $\mathcal{S}' = \langle C', E', a'', K' \rangle$\\
    and $\forall a \in t' : a = a'' \wedge a \in \tx{ServerRef}[\encoded{\tm{C}}, \encoded{\ta{C}}, \encoded{\ap{C}}]$\\
    and $\forall s \in t' : s \in K' \wedge s \in \tx{Chan}[\encoded{\tm{C}}, \encoded{\ta{C}}, \encoded{\ap{C}}]$
\end{theorem}
\begin{proof}
If we take
\begin{align}
    T^a_m, T^s_m &= \encoded{\tm{C}} \nonumber \\
    T^a_a, T^s_a &= \encoded{\ta{C}} \label{encoding:server-prog} \\
    T^a_p, T^s_p &= \encoded{\ap{C}} \nonumber
\end{align}
Then we have that
\begin{gather}
    \forall a \in t : \exists T^a_m, T^a_a, T^a_p : a \in \tx{ServerRef}[T^a_m, T^a_a, T^a_p] \label{a-premise:server-prog}\\
    \forall s \in t : \exists T^s_m, T^s_a, T^s_p : s \in \tx{Chan}[T^s_m, T^s_a, T^s_p] \label{s-premise:server-prog}
\end{gather}
Which means that we can use \cref{term-progress} to show that either $t$ is a value or there is a reduction ${t \reduce{\alpha} t'}$.\\
If $t$ is a value, we are done. Otherwise, we must show that there exists a reduction ${t | S \rightarrow t' | S'}$.
The proof is by cases on the $\alpha$ label of the $t$ reduction:
\begin{itemize}
    \item $\alpha = "\tau"$\\
        where
        \begin{equation}
            \forall a : a \in t' \Rightarrow a \in t \wedge \forall s : s \in t' \Rightarrow s \in t \label{tau:server-prog}
        \end{equation}
        Since the server cannot reduce with label $\dual{\tau}$ by any rule,
        the reduction must have used the \tsc{Net-$\alpha$} rule with an
        internal $\tau$-transition $t \reduce{\tau} t'$ of the client,
        hence we have $S = S'$.
        From \cref{tau:server-prog}, $t'$ does not gain any new $a$ or $s$,
        so the properties in the premise which hold for $t$ must also hold for $t'$ in the conclusion.
    \item $\alpha = "\tx{connect}(a') = s'"$\\
        where
        \begin{align}
            &\forall a \in t' : a \in t \label{conn-1:server-prog}\\
    \wedge\ &\forall s \in t' : (s \in t \vee s = s') \label{conn-2:server-prog}\\
    \wedge\ &s' \in \tx{Chan}[T^{a'}_m, T^{a'}_a, T^{a'}_p] \label{conn-3:server-prog}
        \end{align}
        The only rule that can reduce the server with label $\dual{\tx{connect}(a') = s'}$ is \tsc{Sv-Connect}:
        \begin{equation}
            \derive
                {s' \tx{ is a fresh variable} \quad s' \in \tx{Chan}[\encoded{\tm{C}}, \encoded{\ta{C}}, \encoded{\ap{C}}]}
                {\serverRed{C, E, a', K}{\tx{connect}(a')}{s'}{C, E, a', K \cup \curl{s'}}}
        \end{equation}
        such that $C = C' \quad E = E' \quad a' = a'' \quad K' = K \cup \curl{s'}$\\
        Then, we can conclude:
        \begin{adjustwidth}{-20mm}{20mm}
        \begin{align}
             &\forall a \in t : \exists T^a_m, T^a_a, T^a_p : a \in \tx{ServerRef}[T^a_m, T^a_a, T^a_p]\qquad&\\
\Rightarrow\ &\forall a \in t' : \exists T^a_m, T^a_a, T^a_p : a \in \tx{ServerRef}[T^a_m, T^a_a, T^a_p]\qquad& (\tx{by \cref{conn-1:server-prog}})\\
\Rightarrow\ &\forall a \in t' : a \in \tx{ServerRef}[\encoded{\tm{C}}, \encoded{\ta{C}}, \encoded{\ap{C}}]\qquad& (\tx{by \cref{encoding:server-prog}})\\
\Rightarrow\ &\forall a \in t' : a = a'' \wedge a \in \tx{ServerRef}[\encoded{\tm{C'}}, \encoded{\ta{C'}}, \encoded{\ap{C'}}]\qquad& (\tx{by \tsc{Sv-Connect}})
        \end{align}
        \end{adjustwidth}
        We can also conclude:
        \begin{adjustwidth}{-20mm}{20mm}
        \begin{align}
             &\forall s \in t' : (s \in t \vee s = s')\qquad&\\
\Rightarrow\ &\forall s \in t' : (s \in K \vee s = s')\qquad& \tx{by the theorem premise}\\
\Rightarrow\ &\forall s \in t' : s \in K' \qquad& \tx{by the definition of }K'\\
\Rightarrow\ &\forall s \in t' : s \in K' \wedge s \in \tx{Chan}[\encoded{\tm{C}}, \encoded{\ta{C}}, \encoded{\ap{C}}] \qquad& \tx{by \tsc{Sv-Connect}}\\
\Rightarrow\ &\forall s \in t' : s \in K' \wedge s \in \tx{Chan}[\encoded{\tm{C'}}, \encoded{\ta{C'}}, \encoded{\ap{C'}}] \qquad& \tx{by the definition of }C'\\
        \end{align}
        \end{adjustwidth}
    \item $\alpha = "\tx{read}(s',v) = r"$\\
        where
        \begin{align}
            &\typesTo{\emptyset}{v}{\PPPPEntity\ T^{s'}_m\ T^{s'}_a\ T^{s'}_p\ X_n\ X_a} \label{read-1:server-prog}\\
    \wedge\ &s' \in t \label{read-0:server-prog}\\
    \wedge\ &\forall a : a \in t \Leftrightarrow a \in t' \label{read-2:server-prog}\\
    \wedge\ &\forall s \in t' : s \in t \label{read-3:server-prog}
        \end{align}
        The only rule that can reduce the server with label $\dual{\tx{read}(s', v) = r}$ is \tsc{Sv-Read}:
        \begin{equation}
            \derive
                {s' \in K \quad \mathit{Conforms}(v,C) \quad \serverRedInt{C}{E}{\tx{read}(v)}{v'}}
                {\serverRed{C, E, a', K}{\tx{read}(s',v)}{v'}{C, E, a', K}}
        \end{equation}
        such that $S = S'$\\
        From \cref{read-0:server-prog} and the premises of the theorem, we have $s' \in K$.\\
        Then, we also have that $s \in \tx{Chan}[\encoded{\tm{C}}, \encoded{\ta{C}}, \encoded{\ap{C}}]$,
        which means that\\
        $\typesTo{\emptyset}{v}{\PPPPEntity\ \encoded{\tm{C}}, \encoded{\ta{C}}, \encoded{\ap{C}}\ X_n\ X_a}$.
        From \cref{def:p4runtime-client-semantics}, we have\\
        $v \in \PPPPEntity\ \encoded{\tm{C}}, \encoded{\ta{C}}, \encoded{\ap{C}}\ X_n\ X_a$,
        which is the same as $\mathit{Conforms}(v,C)$.\\
        $\serverRedInt{C}{E}{\tx{read}(v)}{v'}$ always yields a well-typed response,
        so this premise always holds.\\
        \todoin{Proof wording}{Does the reasoning above look OK?}
        Then, we can conclude:
        \begin{adjustwidth}{-20mm}{20mm}
        \begin{align}
             &\forall a \in t : \exists T^a_m, T^a_a, T^a_p : a \in \tx{ServerRef}[T^a_m, T^a_a, T^a_p]\qquad&\\
\Rightarrow\ &\forall a \in t' : \exists T^a_m, T^a_a, T^a_p : a \in \tx{ServerRef}[T^a_m, T^a_a, T^a_p]\qquad& (\tx{by \cref{read-2:server-prog}})\\
\Rightarrow\ &\forall a \in t' : a \in \tx{ServerRef}[\encoded{C.m_m}, \encoded{C.m_a}, \encoded{C.m_p}]\qquad& (\tx{by \cref{encoding:server-prog}})\\
\Rightarrow\ &\forall a \in t' : a = a'' \wedge a \in \tx{ServerRef}[\encoded{C'.m_m}, \encoded{C'.m_a}, \encoded{C'.m_p}]\qquad& (\tx{by \tsc{Sv-Read}})
        \end{align}
        \end{adjustwidth}
        We can also conclude:
        \begin{adjustwidth}{-20mm}{20mm}
        \begin{align}
             &\forall s \in t' : s \in t \qquad&\\
\Rightarrow\ &\forall s \in t' : s \in K \wedge s \in \tx{Chan}[,\encoded{C.m_m}, \encoded{C.m_a}, \encoded{C.m_p}] \qquad& \tx{by the theorem premise}\\
\Rightarrow\ &\forall s \in t' : s \in K' \wedge s \in \tx{Chan}[\encoded{C'.m_m}, \encoded{C'.m_a}, \encoded{C'.m_p}] \qquad& \tx{by the definition of }S'\\
        \end{align}
        \end{adjustwidth}
    \item $\alpha = "\tx{insert}(s',v) = b"$\\
        Analogous to the case for $\tx{read}(s',v) = r$, only using the \tsc{Sv-Insert} rule.
    \item $\alpha = "\tx{modify}(s',v) = b"$\\
        Analogous to the case for $\tx{read}(s',v) = r$, only using the \tsc{Sv-Modify} rule.
    \item $\alpha = "\tx{delete}(s',v) = b"$\\
        Analogous to the case for $\tx{read}(s',v) = r$, only using the \tsc{Sv-Delete} rule.
\end{itemize}
\end{proof}

\progress*
\begin{proof}
  Direct consequence of \Cref{theorem:progress-general} (which embeds in its
  statements the conditions of \Cref{def:well-typed-network}),
  \Cref{lem:network-semantics-comp} and \Cref{def:network-congruence}.
  \notein{Note}{More details would not hurt!}
\end{proof}